%% file: thesis-main.tex
\documentclass[b5paper,10pt]{book}

\usepackage{phdthesis}
\usepackage[
  paperwidth=17cm
  ,paperheight=24cm
  ,inner=2.6cm
  ,outer=2.3cm
  ,top=2.3cm
  ,bottom=2.3cm
]{geometry}

\usepackage[english]{babel}
\usepackage[hyphens]{url}
\usepackage[noend]{algorithmic}
\usepackage{caption}
\usepackage[square,comma,numbers,sort&compress,sectionbib]{natbib}
\usepackage[usenames,dvipsnames,table]{xcolor}
\usepackage{accents}
\usepackage{acronym}
\usepackage[chapter]{algorithm}
\usepackage{amssymb}
\usepackage{balance}
\usepackage{booktabs} 
\usepackage{colortbl}
\usepackage{color}
\usepackage[inline]{enumitem}
\usepackage{flushend}
\usepackage{graphicx}
\usepackage{import}
\usepackage{microtype}
\usepackage{multicol}
\usepackage{multirow}
\usepackage{rotating}
\usepackage{tabularx}
\usepackage{times}
\usepackage{url}
\usepackage{wrapfig}
\usepackage{subfig}
\usepackage[np, autolanguage]{numprint}
\usepackage{amsthm}
\usepackage{rotating}
\usepackage{bbm}
\usepackage{bm}
\usepackage{dsfont}
\usepackage{mleftright}
\usepackage{appendix}
\usepackage{chngcntr}
\usepackage{etoolbox}
\usepackage{placeins}
\usepackage{pdfpages}

\AtBeginEnvironment{subappendices}{%
\clearpage
}

\AtEndEnvironment{subappendices}{%
}

\usepackage[backref=page]{hyperref}
\hypersetup{%
pdfborder = {0 0 0},
pdftitle = {Learning from User Interactions with Rankings: A Unification of the Field},
pdfsubject = {PhD Thesis},
pdfkeywords = {learning to rank},
pdfauthor = {Harrie Oosterhuis}
}
\usepackage{bookmark}

\renewcommand*{\backref}[1]{} 
\renewcommand*{\backrefalt}[4]{%
\ifcase #1
\or (Cited on page~#2.)  %
\else 
(Cited on pages~#2.)  
\fi
}

\let\svthefootnote\thefootnote
\newcommand\blankfootnote[1]{%
  \let\thefootnote\relax\footnotetext{#1}%
  \let\thefootnote\svthefootnote%
}
\let\svfootnote\footnote
\renewcommand\footnote[2][?]{%
  \if\relax#1\relax%
    \blankfootnote{#2}%
  \else%
    \if?#1\svfootnote{#2}\else\svfootnote[#1]{#2}\fi%
  \fi
}
\DeclareMathOperator*{\argmax}{arg\,max}
\DeclareMathOperator*{\argmin}{arg\,min}
\DeclareMathOperator{\sgn}{sgn}
\DeclareMathOperator{\rank}{rank}

\newcommand{\enkelop}{$^{\vartriangle}$}
\newcommand{\dubbelop}{$^{\blacktriangle}$}
\newcommand{\enkelneer}{$^{\triangledown}$}
\newcommand{\dubbelneer}{$^{\blacktriangledown}$}

\newcommand{\ubar}[1]{\underaccent{\bar}{#1}}

\newcommand{\data}{\mathcal{D}}
\newcommand{\traindata}{\mathcal{D}^\textit{train}}
\newcommand{\bounddata}{\mathcal{D}^\textit{sel}}

\newtheorem{theorem}{Theorem}[chapter]

\preto\chapter\acresetall

\begin{document}

\frontmatter
\input{thesis-front}
\input{acknowledgements}
\tableofcontents

\acrodef{IR}{Information Retrieval}
\acrodef{LTR}{Learning to Rank}
\acrodef{OLTR}{Online Learning to Rank}
\acrodef{TDM}{Team-Draft Multileaving}
\acrodef{OM}{Optimized Multileaving}
\acrodef{PM}{Probabilistic Multileaving}
\acrodef{SOSM}{Sample-Scored-Only Multileaving}
\acrodef{PPM}{Pairwise Preference Multileaving}
\acrodef{LM}{Lambda Multileaving}
\acrodef{PBI}{Preference-Based Balanced Interleaving}
\acrodef{TDI}{Team Draft Interleaving}
\acrodef{OI}{Optimized Interleaving}
\acrodef{PI}{Probabilistic Interleaving}
\acrodef{DBGD}{Dueling Bandit Gradient Descent}
\acrodef{MGD}{Multileave Gradient Descent}
\acrodef{RL}{Reinforcement Learning}
\acrodef{PL}{Plackett-Luce}
\acrodef{PDGD}{Pairwise Differentiable Gradient Descent}
\acrodef{IR}{Information Retrieval}
\acrodef{ARP}{Average Relevance Position}
\acrodef{DCG}{Discounted Cumulative Gain}
\acrodef{EM}{Expectation Maximization}
\acrodef{GENSPEC}{Generalization and Specialization}
\acrodef{IPS}{Inverse Propensity Scoring}
\acrodef{CTR}{Click-Through-Rate}
\acrodef{SEA}{Safe Exploration Algorithm}
\acrodef{PBM}{Position-Based Model algorithm}
\acrodef{LogOpt}{Logging-Policy Optimization Algorithm}
\acrodef{3PR}{Perturbed Preference Perceptron for Ranking}
\acrodef{COLTR}{Counterfactual Online Learning to Rank}
\acrodef{NDCG}{Normalized DCG}
\mainmatter
\input{01-introduction/main}

\part{Novel Online Methods for Learning and Evaluating}
\input{04-multileave/main}

\input{05-pdgd/main}

\input{06-oltrcomparison/main}

\part{A Single Framework for Online and Counterfactual Learning to Rank}
\input{07-topk/main}

\input{08-genspec/main}
\input{09-onlinecountereval/main}
\input{10-onlinecounterltr/main}

\bookmarksetup{startatroot} %
\addtocontents{toc}{\bigskip} %

\input{11-conclusions/main}

\backmatter

\input{thesis-back}

\end{document}

%% file: thesis-front.tex
{\pagestyle{empty}
\newcommand{\printtitle}{%
{\Huge\bf Learning from User Interactions with Rankings:\\ A Unification of the Field \\[0.8cm]
}}

\begin{titlepage}
\includepdf{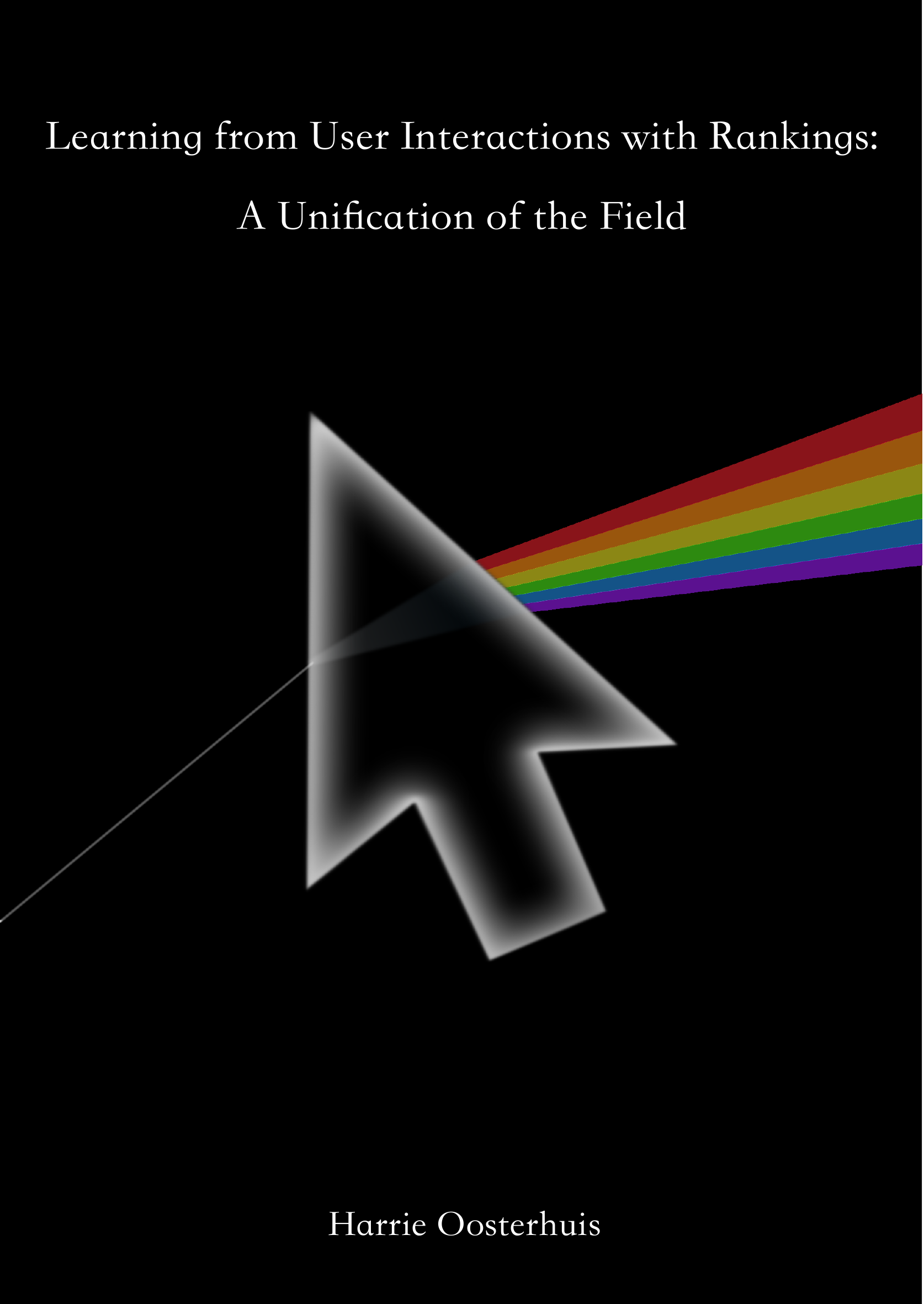}
\par\vskip 2cm
\begin{center}
\printtitle
\vfill
{\LARGE\bf Harrie Oosterhuis}
\vskip 2cm
\end{center}
\end{titlepage}

\mbox{}\newpage
\setcounter{page}{1}

\clearpage
\par\vskip 2cm
\begin{center}
\printtitle
\par\vspace {4cm}
{\large \sc Academisch Proefschrift}
\par\vspace {1cm}
{\large ter verkrijging van de graad van doctor aan de \\
Universiteit van Amsterdam\\
op gezag van de Rector Magnificus\\
prof.\ dr.\ ir.\ K.I.J. Maex\\
ten overstaan van een door het College voor Promoties ingestelde \\
commissie, in het openbaar te verdedigen in \\
de Agnietenkapel\\
op vrijdag 27 november 2020, te 16:00 uur \\ } %
\par\vspace {1cm} {\large door}
\par \vspace {1cm}
{\Large Hendrikus Roelof Oosterhuis}
\par\vspace {1cm}
{\large geboren te Schaijk} %
\end{center}

\clearpage
\noindent%
\textbf{Promotiecommissie} \\\\
\begin{tabular}{@{}l l l}
Promotor: \\
& Prof.\ dr.\ M.\ de Rijke & Universiteit van Amsterdam \\  %
Co-promotor: \\
& Prof.\ dr.\ E.\ Kanoulas & Universiteit van Amsterdam  \\
Overige leden: \\
& Prof.\ dr.\ H.\ Haned & Universiteit van Amsterdam \\
& Prof.\ dr.\ T.\ Joachims & Cornell University \\
& Dr.\ ir.\ J.\ Kamps & Universiteit van Amsterdam \\
& Prof.\ dr.\ C.G.M.\ Snoek & Universiteit van Amsterdam \\
& Prof.\ dr.\ ir.\ A.P.\ de Vries & Radboud Universiteit Nijmegen \\
\end{tabular}

\bigskip\noindent%
Faculteit der Natuurwetenschappen, Wiskunde en Informatica\\

\vfill

\noindent
The research was supported by the Netherlands Organisation for Scientific Research (NWO) under project number 612.001.551.  \\
\bigskip

\noindent
Copyright \copyright~2020 Harrie Oosterhuis, Amsterdam, The Netherlands\\
Cover by Harrie Oosterhuis\\
Printed by Offpage, Amsterdam\\
\\
ISBN: 978-94-93197-36-7\\

\clearpage
}

%% file: acknowledgements.tex
{\pagestyle{empty}

{
\begin{center}

\noindent
\textbf{Acknowledgements} \\ \vskip .5cm
\end{center}
}

\noindent 
Over six years ago, I was invited to join the ILPS research group as an honours MSc.\ A.I. student.
This was the start of an amazing period where I was able to learn, explore and develop myself into a true researcher.
Many people have helped me on this journey, and I am truly grateful for all the support and friendship I have received along the way.
I hope to inspire future students in the same way that you have all inspired me,
and
I will now try my best to thank each and everyone of you.

First and foremost, I want to thank Maarten de Rijke, my supervisor and promotor.
Maarten, I have learned more from you than I thought possible, you have taught me how to do research, how to become a better teacher and supervisor, and how to develop a research career.
Your help was always there when needed and without fail you have always gone above and beyond.
You have set a great example for me and all of your students.
Thank you so much.

Second, I wish to thank my co-promotor Evangelos Kanoulas.
In our annual meetings you have always given me great advice.
You are a truly compassionate supervisor, who cares greatly about his students.
I am very happy to know that you will continue to have an amazing and caring influence on the future of the research group.

Third, I thank Hinda Haned, Thorsten Joachims, Jaap Kamps, Cees Snoek, and Arjen de Vries, I am truly honoured that you are all part of my PhD committee.

My special thanks to Ana and Tom for being my paranymphs.
Through the peaks and valleys of my PhD life you have always been there for me, and I am very honoured to defend this thesis with you on my side.

Another special thanks to Anne Schuth for accepting and supervising me in ILPS when I had only just started the MSc.\ A.I.
In the end, I hold you responsible for my interest in ranking systems and user interactions and I cannot thank you enough.
Also I want to thank Petra in particular, your amazing work has made all of this possible.
Without contest, I consider you the true ILPS MVP, thank you Petra.

Further thanks to everyone who has been part of ILPS during my journey:
Adith, Alexey, Ali, Ali, Amir, Ana, Anna, Anne, Antonis, Arezoo, Arianna, Artem, Bob, Boris, Chang, Christof, Christophe, Chuan, Cristina, Daan, Damien, Dan, Dat, David, David, Dilek, Evgeny, Georgios, Hamid, Hendra, Hinda, Hosein, Ilya, Isaac, Ivan, Jiahuan, Jie, Jin, Julia, Julien, Kaspar, Katya, Ke, Maarten, Maarten, Maartje, Mahsa, Mariya, Marlies, Marzieh, Masrour, Maurits, Mohammad, Mostafa, Mozhdeh, Nikos, Olivier, Peilei, Pengjie, Petra, Praveen, Richard, Ridho, Rolf, Sam, Sami, Sebastian, Shangsong, Shaojie, Spyretta, Svitlana, Thorsten, Tobias, Tom, Trond, Vera, Wanyu, Xiaohui, Xiaojuan, Xinyi, Yangjun, Yaser, Yifan, Zhaochun, and Ziming.
Together, you have all made ILPS a wonderful group to be part of, I am very grateful to call all of you my colleagues.
I was also very happy to part of several sub-groups that discussed ranking systems, thank you Ali, Arezoo, Artem, Chang, Jin, Julia, Maarten, Rolf, and Wanyu, for the great discussions, hopefully there will be many more discussions to come.
In addition, special thanks to Ana, Antonis, Bob, Chang, Hosein, Maartje, Maurits, Nikos, Rolf, and Tom for being great friends as well.

Furthermore, I want to thank the all the people that welcomed me abroad.
For the great experiences I had during Google internships I thank
Ajay,
Ariel,
Bo,
Eugene,
George,
Guan-Lin,
Heng-Tze,
Larry,
Maxime,
Michael,
Mustafa,
Roger,
Sujith,
Vihan, and
Yi-fan.
For the absolute amazing time I had in Australia, I want to thank
Andrew,
Binsheng,
Brian,
Falk,
Joel,
Luke,
Mark,
Sarah, and
Shane.
Especially Joel and Binsheng for literally travelling to the other side of the world with me.
I am really grateful to have met you all and hope the future will allow me to visit you a great many times.

Dan wil ik nog mijn studiegenoten bedanken: Carla, Dasyel, Fabian, Jelle en Wietze, voor de vele leuke herinneringen aan mijn studie.
Verder ben ik ook dankbaar voor mijn lange vriendschap met Chiel, Don, Kit, Stefan, Mark en Luuk, het is mij erg dierbaar om vrienden te hebben die ik al sinds de kleuterklas ken.

Als laatste wil ik mijn familie bedanken, de belangrijkste mensen in mijn leven.
Marianna, Roelof, Anna en Jeroen, bedankt voor alle steun,
zonder jullie was het mij nooit gelukt om zo ver te komen.
Ik bedank Helena en Nathalie omdat zij ons altijd zo warm verwelkomen.
Erg dankbaar ben ik ook voor mijn lieve oma Anna, die altijd zo geduldig luistert als ik weer eens probeer uit te leggen wat ik nu eigenlijk bij de universiteit doe.
Het meest bedank ik mijn grote liefde Emily omdat zij altijd voor mij klaar staat en de dagen zoveel mooier maakt.

\vskip .5cm
\noindent
\begin{flushright}
Harrie Oosterhuis\\
Amsterdam\\
October 2020
\end{flushright}

\clearpage
}

%% file: 01-introduction/main.tex
\chapter{Introduction}
\label{chapter:introduction}

Search engines allow users to efficiently navigate through the enormous numbers of documents available online~\citep{albert1999diameter}.
Underlying every search engine is a ranking system that processes documents in order to present a ranking to the user~\citep{liu2009learning}.
Over the years, the role of ranking systems has only become more important, as they are now used in a wide variety of settings.
Users rely on them to search through many large collections of content, including images~\citep{gordo2016deep}, scientific articles~\citep{joachims2017unbiased}, e-commerce products~\citep{karmaker2017application}, streaming videos~\citep{chelaru2014useful}, job applications/applicants~\citep{geyik2018talent}, emails~\citep{wang2016learning}, and legal documents~\citep{roegiest2015trec}.
Similarly, ranking systems are used for recommendation as well, where they help to suggest content to users that matches their interests~\citep{resnick1997recommender}.
This may even be content of which users are not aware that they have an interest in~\citep{schafer1999recommender}.
In all these rankings scenarios, 
the best user experience is provided when the items that users prefer most are on top of the produced rankings~\citep{jarvelin2002cumulated}.
In other words, the ranking should help the user find what they are looking for with the minimal amount of effort~\citep{sanderson2010user}.

Without a ranking system finding the right information in any sizeable collection becomes an impossible task.
Furthermore, without recommendations many online services would lack a lot of user engagement~\citep{dias2008value}.
Thus, ranking systems drive both user satisfaction -- providing users with the content they prefer -- and user engagement -- bringing providers of content or services to interested consumers~\citep{gomez2015netflix}.
Therefore, the performance of a ranking system is very important to both the users of a service and its providers.
Due to this importance, a lot of interest has gone to the evaluation of ranking systems~\citep{sanderson2010user, hofmann-online-2016, sanderson2010, joachims2002unbiased, chapelle2012large, he2009evaluation} and the field of \ac{LTR} which covers methods for optimizing ranking systems~\citep{liu2009learning, joachims2017unbiased, karatzoglou2013learning, karmaker2017application, burges2010ranknet, wang2018lambdaloss}.

Traditionally, ranking evaluation and \ac{LTR} methods made use of human judgements in the form of expert annotations~\citep{sanderson2010}:
for given pairs of queries and documents, experts are asked to annotate the relevance of a document w.r.t.\ a specific query.
This costly process results in an annotated dataset: a collection of query and document pairs with corresponding expert annotations~\citep{letor, Chapelle2011, dato2016fast}.
For an annotated dataset to be useful it should accurately capture:
\begin{enumerate*}[label=(\roman*)]
\item the queries users typically issue;
\item the documents that have to be ranked; and
\item the relevance preferences of the user~\citep{sanderson2010user}.
\end{enumerate*}
With such a dataset, the optimization of a ranking system can be done through supervised \ac{LTR} methods.
These methods optimize ranking metrics such as Precision, \ac{ARP} or \ac{DCG}, based on the provided relevance annotations~\citep{jarvelin2002cumulated, liu2009learning}.
While very important to the \ac{LTR} field, some severe limitations of this supervised approach have become apparent over the years:
\begin{enumerate*}[label=(\roman*)]
\item Expert annotations are expensive and time-consuming to obtain~\citep{letor, Chapelle2011}.
\item In sensitive settings acquiring experts annotations can be unethical, for instance, when gathering data for optimizing systems for search over personal documents such as emails~\citep{wang2016learning}.
\item For specific settings there may be no experts that can judge what is relevant, for instance, in the context of  personalized recommendations.
\item What users perceive as relevant is known to change over time, thus a dataset would have to be updated regularly, further increasing the associated costs~\citep{dumais-web-2010,lefortier-online-2014}.
\item Actual user preferences and expert annotations are often misaligned~\citep{sanderson2010}.
\end{enumerate*}
Consequently, the supervised approach is infeasible for many \ac{LTR} practitioners because they do not have the resources to create an annotated dataset or gathering annotations is not possible in their ranking setting.
Moreover, even if a dataset can be obtained, it may not lead to the optimal ranking system.
Thus there is a need for an alternative to the supervised approach to ranking evaluation and \ac{LTR}.

An alternative approach that has received a lot of attention is to base evaluation and optimization on user interactions~\citep{radlinski2008does, joachims2003evaluating}.
For rankings this usually means that user clicks are used to compare and improve ranking systems.
At first glance user interactions seem to solve the problems with annotations:
\begin{enumerate*}[label=(\roman*)]
\item If a service has enough active users, interactions are virtually free and available at a large scale.
\item Gathering interactions can be done without showing sensitive items to experts for annotation.
\item Unlike annotations, interactions are an indication of the actual individual user preferences.
\end{enumerate*}
Thus there appears to be a lot of potential for using user interactions, however, there are also drawbacks specific to using them:
\begin{enumerate*}[label=(\roman*)]
\item It requires keeping track of large amounts of user behavior, something users may not consent with~\citep{politou2018forgetting}.
\item User behavior is very unpredictable, clicks in particular are known to be a very noisy signal~\citep{chuklin-click-2015}.
\item Clicks are a form of implicit feedback; there are other factors beside user preference that also affect whether a click takes place, making clicks a biased signal of relevance~\citep{chuklin-click-2015, craswell2008experimental}.
\end{enumerate*}
This thesis will not explore the first drawback and will instead focus on settings where acquiring user interactions is done with consent, in a privacy-respecting and ethical manner.
Mainly, we will consider how methods of evaluation and optimization based on clicks can mitigate the negative effects from click-related noise and bias.

Existing methods for ranking evaluation and optimization from user interactions can roughly be divided into two families: the \emph{online} family that deals with bias through direct interaction and result-randomization~\citep{yue2009interactively, raman2013stable}; and the \emph{counterfactual} family that first models click behavior and then uses the inferred model to correct for bias in logged click data~\citep{wang2016learning, joachims2017unbiased}.
A further division can be made.
For this thesis a decomposition into five areas is relevant.
We will divide the online family into three areas:
\begin{enumerate}[label=(\roman*),leftmargin=*]
\item Online Evaluation -- methods like A/B testing and interleaving~\citep{joachims2003evaluating} that interact directly with users to compare ranking systems and randomize displayed results to mitigate biases~\citep{schuth2015predicting, chapelle2012large,hofmann2013fidelity, schuth2016mgd}.
\item Feature-Based Online \ac{LTR} -- methods like \ac{DBGD}~\citep{yue2009interactively} and the \ac{3PR}~\citep{raman2013stable} that optimize feature-based ranking models by direct interaction with users, often relying on online evaluation~\citep{schuth2016mgd, wang2019variance, hofmann12:balancing}.
\item Tabular Online \ac{LTR} -- methods like Cascading Bandits~\citep{kveton2015cascading} and the \ac{PBM}~\citep{lagree2016multiple} that optimize a single ranking for a single ranking setting, by learning from direct interactions and result randomization~\citep{zoghi2016click, lattimore2019bandit, Komiyama2015, zoghi-online-2017}.
Characteristic about tabular methods is that they do not use any feature-based prediction model but instead memorize the best ranking.
\end{enumerate}
For the counterfactual family, we will use the following division into two areas:
\begin{enumerate}[label=(\roman*),leftmargin=*,resume]
\item Counterfactual Evaluation -- methods that evaluate rankings based on historically logged clicks. They require an inferred model of click behavior and use that model to correct for biases using, for instance, \ac{IPS}~\citep{joachims2017unbiased, swaminathan2015self, carterette2018offline, agarwal2019estimating, ovaisi2020correcting}.
\item Counterfactual \ac{LTR} -- methods that use counterfactual evaluation to estimate performance based on historical click-logs, and that optimize ranking models to maximize the estimated a system's performance~\citep{agarwal2019counterfactual, joachims2017unbiased, hu2019unbiased, wang2016learning, agarwal2019addressing, ovaisi2020correcting}.
\end{enumerate}
This division reveals a rich diversity in approaches that all share the same goal of evaluating or optimizating ranker performance based on user interactions.

On the one hand, this diversity is understandable, since in some settings only one area of methods is applicable.
For instance, one cannot add randomization to data that is already logged, making the counterfactual approach the only available option if only logged data is available.
On the other hand, the diversity of approaches is also unexpected and raises some questions.
For instance, why would online approaches not benefit from an accurate model of click behavior if one is available, similar to the counterfactual approach?

In this thesis, we investigate whether this online/counterfactual division is truly necessary.
We will introduce several novel \ac{LTR} methods that improve over the efficiency of existing methods, and increase the applicability of \ac{LTR} from user clicks.
In particular, we focus on finding \ac{LTR} methods that bridge the online/counterfactual division and find methods that are highly effective both when applied online or counterfactually.
An important result of our thesis on the \ac{LTR} field, is that we offer a unified perspective and set of \ac{LTR} methods.

\input{01-introduction/rqs}

\input{01-introduction/contributions}

\input{01-introduction/overview}

\input{01-introduction/origins}

%% file: 01-introduction/rqs.tex
\section{Research Outline and Questions}
\label{section:introduction:rqs}

The overarching question this thesis aims to answer is:
\begin{itemize}
\item[] \em Could there be a single general theoretically-grounded approach that has competitive performance for both evaluation and \ac{LTR} from user clicks on rankings, in both the counterfactual and online settings?
\end{itemize}
Our aim is to progress the \ac{LTR} field towards answering this question in the affirmation.
In this thesis, we will explore two directions in search of a single general theoretically-grounded approach.
Firstly, by introducing novel online \ac{LTR} methods that outperform existing online methods in optimization and large scale optimization in the online setting.
Secondly, by introducing novel counterfactual \ac{LTR} methods that build on the original \ac{IPS}-based counterfactual \ac{LTR} approach~\citep{joachims2017unbiased}.
Our novel counterfactual \ac{LTR} methods expand the original counterfactual approach and make it applicable to more tasks and settings.
As a result, these novel methods bridge several gaps between counterfactual \ac{LTR} and the areas of supervised \ac{LTR} and online \ac{LTR}.
Furthermore, all our novel counterfactual \ac{LTR} methods are compatible with each other, and can be seen as part of a novel counterfactual \ac{LTR} framework.
At the end of the thesis, our proposed framework has taken the original counterfactual \ac{LTR} approach and greatly increased its applicability and effectiveness for both online and counterfactual evaluation and optimization.
This leads to a more unified perspective of the \ac{LTR} field, where areas that were previously largely independent are now connected.

\subsection{Novel online methods for learning and evaluating}
In the first part of the thesis, we introduce two methods that greatly increase the efficiency of large scale online evaluation and online \ac{LTR}.
Additionally, we take a critical look at several existing methods for online evaluation and online \ac{LTR}.

Interleaving was introduced as an efficient evaluation paradigm designed for evaluating whether one ranking system outperforms another~\citep{joachims2003evaluating}.
Interleaving methods take the rankings produced by two systems and combine them into an interleaved ranking~\citep{hofmann2011probabilistic, radlinski2013optimized, radlinski2008does}.
Clicks on the interleaved ranking are interpreted directly as preference signals between the two systems, resulting in a more data-efficient approach~\citep{schuth2015predicting}.
Thus allowing one to efficiently estimate if an alteration leads to an improved system.
Later, the interleaving approach was extended to multileaving which allows for comparisons that include more than two systems at once~\citep{Schuth2014a, brost2016improved, schuth2015probabilistic},
thereby enabling efficient comparing large numbers of systems with each other.

In Chapter~\ref{chapter:01-online-evaluation} we look at such multileaving methods for large scale online ranking evaluation.
Specifically, we investigate the following question:
\begin{enumerate}[label=\textbf{RQ\arabic*},ref=\textbf{RQ\arabic*}]
\item Does the effectiveness of online ranking evaluation methods scale to large comparisons? \label{thesisrq:multileaving}
\end{enumerate}
We examine existing multileaving methods in terms of \emph{fidelity} -- are they provably unbiased in unambiguous cases~\citep{hofmann2013fidelity} -- and \emph{considerateness} -- are they safe w.r.t.\ the user experience during the gathering of clicks.
From our theoretical analysis, we find that no existing multileaving method manages to meet both criteria.
Furthermore, our empirical analysis reveals that their performance decreases as comparisons involve more ranking systems at once.
As a novel alternative, we introduce the \ac{PPM} algorithm, \ac{PPM} bases evaluation on inferred pairwise item preferences.
We prove that it meets both the \emph{fidelity} and \emph{considerateness} criteria.
Furthermore, our empirical results indicate that using \ac{PPM} leads to a much smaller number of errors especially in large scale comparisons.

Besides evaluation, optimization is also very important to obtain effective ranking systems~\citep{liu2009learning}.
The idea of optimizing ranking systems based on clicks is long-established.
One of the first-theoretically grounded approaches was \acf{DBGD}~\citep{yue2009interactively}.
For every incoming query, \ac{DBGD} samples a variation on a ranking system and then uses interleaving to estimate whether this variation is an improvement.
If so, it updates the ranking system to be more similar to the variation.
Over time this process is supposed to oscillate towards the optimal ranking system.
Numerous extensions have been proposed but all have kept the overall \ac{DBGD} approach of sampling variations and using online evaluation~\citep{schuth2016mgd, wang2019variance, hofmann12:balancing, raman2013stable}.
This is somewhat puzzling, since this sampling approach is in stark contrast with all other \ac{LTR} methods that use gradient-based optimization.

In Chapter~\ref{chapter:02-pdgd} we explore alternatives to the \ac{DBGD} approach and ask ourselves the following question:
\begin{enumerate}[label=\textbf{RQ\arabic*},ref=\textbf{RQ\arabic*},resume]
\item Is online \ac{LTR} possible without relying on model-sampling and online evaluation?%
\label{thesisrq:pdgd}
\end{enumerate}
We answer this question in the affirmative by proposing a novel online \ac{LTR} method: \acf{PDGD}.
Unlike \ac{DBGD}, \ac{PDGD} does not require model-sampling nor does it make use of any online evaluation.
Instead, \ac{PDGD} optimizes a stochastic Plackett-Luce ranking model and bases its updates on inferred pairwise item preferences.
\ac{PDGD} weights the gradients w.r.t.\ item-pairs to mitigate the effect of position bias.
We prove, under very mild assumptions, that the weighted gradient of \ac{PDGD} is unbiased w.r.t.\ item-preferences.
Our experimental results show that \ac{PDGD} requires far fewer interactions to reach the same level of performance as \ac{DBGD}.
Furthermore, we show that even in ideal settings \ac{DBGD} may not be able to find the optimal model and is ineffective at optimizing neural models.
In contrast, \ac{PDGD} does converge to near optimal models, and reaches even higher performance when applied to neural networks.

The large improvements of \ac{PDGD} over \ac{DBGD} observed in Chapter~\ref{chapter:02-pdgd}, made us wonder whether \ac{DBGD} is actually a reliable choice for online \ac{LTR}.
In response to this question, Chapter~\ref{chapter:03-oltr-comparison} tackles the following question:
\begin{enumerate}[label=\textbf{RQ\arabic*},ref=\textbf{RQ\arabic*},resume]
\item Are \ac{DBGD} \ac{LTR} methods reliable in terms of theoretical soundness and empirical performance?
\label{thesisrq:dbgd}
\end{enumerate}
First, we take a critical look at the theory underlying the \ac{DBGD} approach, and find that its assumptions do not hold for deterministic ranking systems and common ranking metrics.
Consequently, we conclude that its theory is not applicable to the large majority of existing research that utilizes the \ac{DBGD} approach~\citep{hofmann2013reusing, schuth2016mgd, oosterhuis2016probabilistic, hofmann12:balancing, wang2018efficient, zhao2016constructing}.
Second, we perform an empirical analysis where \ac{DBGD} and \ac{PDGD} are compared in circumstances ranging from near-ideal -- where interactions contain little noise and no position bias -- to extremely difficult -- where interactions contain extreme amounts of noise and position bias.
The difference in performance between \ac{PDGD} and \ac{DBGD} is so large, that we conclude that \ac{PDGD} is by far the more reliable choice.

For the field of online \ac{LTR} this leads us to question the relevancy of \ac{DBGD} and its extentions, as we have found theoretical weaknesses and empirical inferiority.
The fact that virtually all previous methods in the online \ac{LTR} field are extensions of \ac{DBGD} raises profound questions.

\subsection{A single framework for online and counterfactual learning to rank}

In the second part of the thesis, we expand the existing \ac{IPS}-based counterfactual \ac{LTR} approach~\citep{joachims2017unbiased} to create a unified framework for both online and counterfactual \ac{LTR} and ranking evaluation based on clicks.

The conclusions of the first part of the thesis revealed that \ac{DBGD}, which forms the basis of most previous work in online \ac{LTR}, has problems in terms of performance and its theoretical basis.
It is concerning that these conclusions could have been made much earlier: previous work could have taken a critical look at the theory at any moment;
furthermore, if previous work had compared \ac{DBGD} performance with supervised \ac{LTR} in the prevalent simulated setups, it would have observed the convergence problems of \ac{DBGD}.
To avoid similar issues, we chose to build upon the Counterfactual \ac{LTR} approach because it has a strong theoretical basis, and additionally, all experimental comparisons in the second part include optimal ranking models to detect potential convergence issues.

In contrast with online \ac{LTR} approaches, counterfactual \ac{LTR} and evaluation makes explicit assumptions about user behavior~\citep{wang2016learning, joachims2017unbiased}.
By making such assumptions, the unbiasedness of counterfactual methods can be proven.
Thus guaranteeing optimal convergence, given that the assumptions are correct.
While this provides a strong foundation for learning from historically logged clicks, the counterfactual approach is not always applicable nor always the most effective option~\citep{jagerman2019comparison}.
The following research questions consider whether counterfactual \ac{LTR} could overcome its limitations and become the best choice for \ac{LTR} from clicks in general.

One of the requirements for the unbiasedness of the original counterfactual \ac{LTR} method is that it requires every relevant item to be displayed at every query~\citep{joachims2017unbiased}.
This is a problem in top-$k$ ranking settings where not all items can be displayed at once~\citep{ovaisi2020correcting}.
Hence, Chapter~\ref{chapter:04-topk} concerns the question:
\begin{enumerate}[label=\textbf{RQ\arabic*},ref=\textbf{RQ\arabic*},resume]
\item Can counterfactual \ac{LTR} be extended to top-$k$ ranking settings?
\label{thesisrq:topk}
\end{enumerate}
We introduce the Policy-Aware estimator that corrects for position bias while taking into account the behavior of a stochastic logging policy.
As a result, the policy-aware estimator is unbiased even when learning from top-$k$ feedback, if the policy gives every relevant item a non-zero chance of appearing in the top-$k$.
Thus with this extension counterfactual \ac{LTR} is also applicable to the top-$k$ setting which is especially prevalent in recommendation.

Existing work has considered how to optimize ranking metrics such as \ac{DCG} using counterfactual \ac{LTR}~\citep{agarwal2019counterfactual, hu2019unbiased}.
Interestingly, the solutions for counterfactual \ac{LTR} are very different than those in supervised \ac{LTR}~\citep{burges2010ranknet, wang2018lambdaloss}.
To investigate whether this difference is really necessary, Chapter~\ref{chapter:04-topk} also addresses the question:
\begin{enumerate}[label=\textbf{RQ\arabic*},ref=\textbf{RQ\arabic*},resume]
\item Is it possible to apply state-of-the-art supervised \ac{LTR} methods to the counterfactual \ac{LTR} problem?
\label{thesisrq:lambdaloss}
\end{enumerate}
We find that the LambdaLoss framework~\citep{wang2018lambdaloss}, which includes the famous LambdaMART method~\citep{burges2010ranknet}, can also be applied to counterfactual estimates of ranking metrics.
Thus we show that there does not need to be a divide between state-of-the-art supervised \ac{LTR} and counterfactual \ac{LTR}.

So far we have not considered the area of tabular online \ac{LTR}: methods that find the optimal ranking for a single query based on result randomization and direct interaction~\citep{kveton2015cascading, lagree2016multiple, lattimore2019bandit, Komiyama2015, zoghi-online-2017}.
While these methods need a lot of click data to reach decent performance, they can always find the optimal ranking since they optimize a memorized ranking, instead of using a feature-based model~\citep{zoghi2016click}.
The downside is that when few clicks are available for a query, tabular \ac{LTR} methods are highly sensitive to noise.
Thus these approaches are good for specialization: they have great performance on queries where numerous clicks have been observed, while also having an initial period of poor performance.
Conversely, counterfactual \ac{LTR} commonly optimizes feature-based models for generalization to have a robust performance on previously unseen queries, while often not reaching perfect performance at convergence.

Inspired by this contrast, in Chapter~\ref{chapter:05-genspec} we ask ourselves:
\begin{enumerate}[label=\textbf{RQ\arabic*},ref=\textbf{RQ\arabic*},resume]
\item Can the specialization ability of tabular online \ac{LTR} be combined with the robust feature-based approach of counterfactual \ac{LTR}? \label{thesisrq:genspec}
\end{enumerate}
Our answer is in the form of the novel \ac{GENSPEC} algorithm, it optimizes a single robust generalized policy and numerous specialized policies each optimized for a single query.
Then the \ac{GENSPEC} meta-policy uses high-confidence bounds to safely decide per query which policy to deploy.
Consequently, for previously unseen queries \ac{GENSPEC} choose the generalized policy which utilizes the robust feature-based ranking model. 
While for other queries it can decide to deploy a specialized policy, i.e., if it has enough data to confidently determine that the specialized policy has found the better ranking.
For the \ac{LTR} field, \ac{GENSPEC} shows that specialization does not need to be unique to tabular online \ac{LTR}, instead it can be a property of counterfactual \ac{LTR} as well.
Moreover, overall it shows that specialization and generalization are not mutually exclusive abilities.

While counterfactual evaluation methods are designed for using historical clicks, they can be applied online by simply treating newly gathered data as historical~\citep{carterette2018offline, jagerman2019comparison}.
In contrast with online evaluation methods, counterfactual evaluation is completely passive: its methods do not prescribe which rankings should be displayed.
This difference leads us to ask the following question in Chapter~\ref{chapter:06-onlinecountereval}:
\begin{enumerate}[label=\textbf{RQ\arabic*},ref=\textbf{RQ\arabic*},resume]
\item Can counterfactual evaluation methods for ranking be extended to perform efficient and effective online evaluation?
\label{thesisrq:onlineeval}
\end{enumerate}
We answer this question positively by introducing the novel \ac{LogOpt} which uses available clicks to optimize the logging policy to minimize the variance of counterfactual estimates of ranking metrics.
By minimizing variance, \ac{LogOpt} increases the data-efficiency of counterfactual evaluation, leading to more accurate estimates from fewer logged clicks.
\ac{LogOpt} is applied when data is still been gathered and changes what rankings will be displayed for future queries.
Thus, with the addition of \ac{LogOpt}, counterfactual evaluation is transformed into an online approach that is actively involved with how data is gathered.
Our experimental results suggest that \ac{LogOpt} is at least as efficient as interleaving methods, while also being proven to be unbiased under the common assumptions of counterfactual \ac{LTR}.

The results in Chapter~\ref{chapter:01-online-evaluation} and Chapter~\ref{chapter:06-onlinecountereval} did not show any online evaluation method converge on a zero error.
This lead us to also ask the following question in Chapter~\ref{chapter:06-onlinecountereval}:
\begin{enumerate}[label=\textbf{RQ\arabic*},ref=\textbf{RQ\arabic*},resume]
\item Are existing interleaving methods truly capable of unbiased evaluation w.r.t.\ position bias?
\label{thesisrq:interleaving}
\end{enumerate}
We prove that, under the assumption of basic position bias, interleaving methods are not unbiased.
Furthermore, our results in Chapter~\ref{chapter:06-onlinecountereval} indicate that interleaving methods have a systematic error.
Unfortunately, we are unable to estimate the impact this systematic error has on real-world comparisons.
To the best of our knowledge, no empirical studies have been performed that could measure such a bias, our findings strongly show that such a study would be highly valuable to the field.

In Chapter~\ref{chapter:06-onlinecountereval} we have shown that counterfactual ranking evaluation can be as efficient as online evaluation methods, while also having the theoretical justification of counterfactual methods.
Naturally this leads to a similar question regarding \ac{LTR}:
\begin{enumerate}[label=\textbf{RQ\arabic*},ref=\textbf{RQ\arabic*},resume]
\item Can the counterfactual \ac{LTR} approach be extended to perform highly effective online \ac{LTR}?
\label{thesisrq:onlinecounterltr}
\end{enumerate}
In Chapter~\ref{chapter:06-onlinecounterltr} we answer this question by introducing the intervention-aware estimator for online/counterfactual \ac{LTR}.
The intervention-aware estimator corrects for position-bias and trust-bias while also taking into account the effect of online interventions.
This means that if an intervention takes place -- i.e., the logging policy changes during the gathering of data -- the intervention-aware estimator takes its effect on the interaction biases into account.
The result is an estimator that, one the one hand, is just as efficient as other counterfactual estimators when applied to historical data.
While on the other hand, it is much more efficient when applied online than existing estimators.
Moreover, its performance is comparable to online \ac{LTR} methods.
In contrast with online methods, including \ac{DBGD} and \ac{PDGD}, the intervention-aware estimator is proven to be unbiased w.r.t.\ ranking metrics under the standard assumptions.
In other words, it is the only method that is proven to converge on the optimal model, while also being as efficient as the others.
Therefore, we consider the intervention-aware estimator a bridge between online and counterfactual \ac{LTR} as it is a most-reliable choice in both scenarios.

%% file: 01-introduction/contributions.tex
\section{Main Contributions}
\label{section:introduction:contributions}

This section will now summarize the main contributions of this thesis.
We differentiate between algorithmic contributions -- novel algorithms introduced in the thesis -- and theoretical contributions -- findings that are important to the field, both in the form of formal proofs and empirical observations.

\subsection{Algorithmic contributions}

\begin{enumerate}
\item The \acf{PPM} algorithm for large scale comparisons in online evaluation.
\item The \acf{PDGD} algorithm for fast and efficient online \ac{LTR}.
\item The policy-aware estimator that can perform unbiased counterfactual \ac{LTR} from top-$k$ settings. \label{contrib:policyaware}
\item Three loss functions for optimizing top-$k$ metrics with counterfactual \ac{LTR}, including an adaption of the supervised \ac{LTR} LambdaLoss method. \label{contrib:loss}
\item The \acf{GENSPEC} algorithm that combines the specialization ability of tabular models with the generalization ability of feature-based models. \label{contrib:genspec}
\item The \acf{LogOpt} algorithm that turns counterfactual evaluation into online evaluation so as to minimize variance by updating the logging policy during the gathering of data. \label{contrib:logopt}
\item The intervention-aware estimator that bridges the gap between counterfactual and online \ac{LTR}, by extending the policy-aware estimator to take into account the effect of online interventions. \label{contrib:interaware}
\item An overarching framework for both online and counterfactual \ac{LTR} evaluation and optimization, by combining the existing counterfactual approach with the contributions of the second part of the thesis.
For counterfactual/online evaluation contributions \ref{contrib:policyaware}, \ref{contrib:logopt}, and \ref{contrib:interaware} can be applied simultaneously, similarly for counterfactual/online \ac{LTR} the same can be done with contributions \ref{contrib:policyaware}, \ref{contrib:loss}, \ref{contrib:genspec}, and \ref{contrib:interaware}.
\end{enumerate}

\subsection{Theoretical contributions}

\begin{enumerate}[resume]
\item An extension of the definition of \emph{fidelity} and \emph{considerateness} for multileaving; in addition, we show that no existing multileave method meets the criteria of both simultaneously.
\item A formal proof that \ac{PDGD} is unbiased w.r.t.\ pairwise item preferences under mild assumptions.
\item A formal proof that the assumptions of \ac{DBGD} are not sound for deterministic ranking models, thus invalidating some claims of unbiasedness in previous online \ac{LTR} work.
\item An extensive comparison of \ac{DBGD} and \ac{PDGD} under circumstances ranging from ideal to near worst-case, revealing that even in ideal circumstances \ac{DBGD} is often unable to approximate the optimal model.
\item A formal proof for the unbiasedness of the policy-aware and intervention-aware estimators, proving that the former is unbiased w.r.t.\ position bias and item-selection bias and the latter w.r.t.\ position bias, item-selection bias, and trust bias respectively.
\item A formal demonstration how \ac{LTR} loss functions can be adapted to bound top-$k$ metrics, including a description of how LambdaLoss can be adapted for counterfactual \ac{LTR}.
\item An extension of existing bounds in order to bound the relative performance of two policies, with an additional proof that this bound is more efficient than comparing the bounds of individual policies.
\item A formal proof that interleaving methods are not unbiased w.r.t.\ position bias.
\item An empirical analysis that reveals that \ac{PDGD} is not unbiased w.r.t.\ position bias, item-selection bias, and trust bias, when not applied fully online.
\end{enumerate}

\noindent
In addition to these contributions, the source code used to perform the experiments in each published chapter has been shared publicly to enable reproducibility.

%% file: 01-introduction/overview.tex
\section{Thesis Overview}
\label{section:introduction:overview}

This section will provide an overview of the thesis, and provide some recommendations for reading directions.
This thesis consists of an introduction chapter, seven research chapters divided into two parts, and a conclusion.
Each research chapter answers one or two of the thesis research questions put forward in Section~\ref{section:introduction:rqs}, in addition to several chapter-specific research questions.
The thesis research questions are important to the overarching story of the thesis, whereas the chapter-specific research questions only consider the individual contributions of the chapters.

The first chapter, which you are currently reading, introduces the subject of this thesis: \ac{LTR} and ranking evaluation from user clicks.
Furthermore, it lays out the thesis research questions this thesis answers, and provides an overview of its contributions and its origins.

Part I titled \emph{Novel Online Methods for Learning and Evaluating} contains three research chapters that all consider online methods for \ac{LTR} and ranking evaluation.
Chapter~\ref{chapter:01-online-evaluation} looks at multileaving methods for online evaluation, evaluates existing methods and introduces a novel multileaving method.
Chapter~\ref{chapter:02-pdgd} considers online \ac{LTR} and introduces \ac{PDGD}, a novel debiased pairwise method.
Chapter~\ref{chapter:03-oltr-comparison} performs an extensive comparison of the previous state-of-the-art online \ac{LTR} method \ac{DBGD} and our novel \ac{PDGD}, in terms of theoretical guarantees and an experimental analysis.

Part II titled \emph{A Single Framework for Online and Counterfactual Learning to Rank} contains four research chapters that build on the counterfactual approach to \ac{LTR} and ranker evaluation.
The chapters in this part of the thesis are complementary, most of their contributions can be applied together or build upon each other.
Chapter~\ref{chapter:04-topk} extends counterfactual \ac{LTR} to top-$k$ settings; it introduces a novel estimator to learn from top-$k$ feedback and extends supervised \ac{LTR} methods to optimize counterfactual estimates of top-$k$ ranking metrics.
Chapter~\ref{chapter:05-genspec} looks at both tabular and feature-based ranking models, and introduces an algorithm that optimizes both types of models and safely deploys different models per query.
Thus combining the specialization abilities of tabular models with the robust performance of feature-based models in previously unseen circumstances.
Chapter~\ref{chapter:06-onlinecountereval} aims to unify counterfactual and online ranking evaluation; it introduces a method that updates the logging policy during the gathering of data, turning counterfactual evaluation into efficient online evaluation.
Similarly, Chapter~\ref{chapter:06-onlinecounterltr} seeks to unify counterfactual and online \ac{LTR}; it proposes a novel estimator that takes into account the effect of online interventions but can also be applied counterfactually.
As a result, the estimator is effective for both counterfactual \ac{LTR} and online \ac{LTR}. 

Lastly, the thesis is concluded in Chapter~\ref{chapter:conclusions}, where we summarize the findings of the thesis; in particular, we discus whether the division between the families of online and counterfactual \ac{LTR} methods has been bridged.
We end the chapter with a discussion of possible future research directions.

The research chapters in this thesis are self-contained, therefore, a reader can read any single chapter independently if they desire.
The research chapters grew out of published papers.
We wanted to avoid creating alternate versions of published work that deviate from the originals.
As a result, the notation between some chapters differ somewhat; to help the reader, we have added a table at the end of each chapter detailing the notation it uses.
For the best experience, we recommend reading all the chapters in part II because they build on each other.
For the same reason, Chapter~\ref{chapter:02-pdgd} and Chapter~\ref{chapter:03-oltr-comparison} are best read together.

%% file: 01-introduction/origins.tex
\section{Origins}
\label{section:introduction:origins}

We will now list the publications on which the research chapters were based.
Each of the publications is a conference paper written by Harrie Oosterhuis and Maarten de Rijke.
In all cases, Oosterhuis came up with the main research ideas, performed all experiments, and wrote the majority of text.
De Rijke lead the discussions on how each paper should be structured and contributed significantly to the text.
In total, this thesis is built on 6 publications~
\citep{oosterhuis2017sensitive, oosterhuis2018differentiable, oosterhuis2019optimizing, oosterhuis2020topkrankings, oosterhuis2021genspec, oosterhuis2020taking, oosterhuis2021onlinecounterltr}.

\begin{enumerate}[align=left, leftmargin=*]
\item[\textbf{Chapter~\ref{chapter:01-online-evaluation}}] is based on \emph{Sensitive and scalable online evaluation with theoretical guarantees} published at CIKM '17 by \citep{oosterhuis2017sensitive}.
\item[\textbf{Chapter~\ref{chapter:02-pdgd}}]  is based on \emph{Differentiable Unbiased Online Learning to Rank} published at CIKM '18 by \citet{oosterhuis2018differentiable}.
\item[\textbf{Chapter~\ref{chapter:03-oltr-comparison}}]  is based on \emph{Optimizing Ranking Models in an Online Setting} published at ECIR '19 by \citet{oosterhuis2019optimizing}.
\item[\textbf{Chapter~\ref{chapter:04-topk}}] is based on \emph{Policy-Aware Unbiased Learning to Rank for Top-k Rankings} published at SIGIR '20 by \citet{oosterhuis2020topkrankings}.
\item[\textbf{Chapter~\ref{chapter:05-genspec}}] is based on \emph{Robust Generalization and Safe Query-Specialization in Counterfactual Learning to Rank} submitted to WWW '21 by \citet{oosterhuis2021genspec}.
\item[\textbf{Chapter~\ref{chapter:06-onlinecountereval}}]  is based on \emph{Taking the Counterfactual Online: Efficient and Unbiased Online Evaluation for Ranking} published at ICTIR '20 by \citet{oosterhuis2020taking}.
\item[\textbf{Chapter~\ref{chapter:06-onlinecounterltr}}]  is based on \emph{Unifying Online and Counterfactual Learning to Rank} published at WSDM '21 by \citet{oosterhuis2021onlinecounterltr}.
\end{enumerate}

\noindent
In addition, this thesis also indirectly benefitted from the following publications:
\begin{itemize}[align=left, leftmargin=*]
\item \emph{Probabilistic Multileave for Online Retrieval Evaluation} published at SIGIR '15 by \citet{schuth2015probabilistic}.
\item \emph{Multileave Gradient Descent for Fast Online Learning to Rank} published at WSDM '16 by \citet{schuth2016mgd}.
\item \emph{Probabilistic Multileave Gradient Descent} published at ECIR '16 by \citet{oosterhuis2016probabilistic}.
\item \emph{Balancing Speed and Quality in Online Learning to Rank for Information Retrieval} published at CIKM '17 by \citet{oosterhuis2017balancing}.
\item \emph{Query-level Ranker Specialization} published at CEUR '17 by \citet{jagerman2017query}.
\item \emph{Ranking for Relevance and Display Preferences in Complex Presentation Layouts} published at SIGIR '18 by \citet{oosterhuis2018ranking}.
\item \emph{The Potential of Learned Index Structures for Index Compression} published at ADCS '18 by \citet{oosterhuis2018potential}.
\item \emph{To Model or to Intervene: A Comparison of Counterfactual and Online Learning to Rank from User Interactions} published at SIGIR '19 by \citet{jagerman2019comparison}.
\item \emph{When Inverse Propensity Scoring does not Work: Affine Corrections for Unbiased Learning to Rank} published at CIKM '20 by \citet{vardasbi2020trust}.
\item \emph{Keeping Dataset Biases out of the Simulation: A Debiased Simulator for Reinforcement Learning based Recommender Systems} published at RecSys '20 by \citet{huang2020keeping}.
\end{itemize}

\noindent
Furthermore, other work helped with gaining broader research insights, without being directly related to the thesis topic:
\begin{itemize}[align=left, leftmargin=*]
\item \emph{Semantic Video Trailers} by \citet{oosterhuis2016semantic}.
\item \emph{Optimizing Interactive Systems with Data-Driven Objectives} by \citet{li2018optimizing}.
\item \emph{Actionable Interpretability through Optimizable Counterfactual Explanations for Tree Ensembles} by \citet{lucic2019actionable}.
\end{itemize}

%% file: 04-multileave/main.tex
\chapter{Sensitive and Scalable Online Evaluation with Theoretical Guarantees}
\label{chapter:01-online-evaluation}

\footnote[]{This chapter was published as~\citep{oosterhuis2017sensitive}. Appendix~\ref{notation:01-online-evaluation} gives a reference for the notation used in this chapter.}

Multileaved comparison methods generalize interleaved comparison methods to provide a scalable approach for comparing ranking systems based on regular user interactions.
Such methods enable the increasingly rapid research and development of search engines. 
However, existing multileaved comparison methods that provide reliable outcomes do so by degrading the user experience during evaluation.
Conversely, current multileaved comparison methods that maintain the user experience cannot guarantee correctness.

In this chapter, we address the following thesis research question:
\begin{itemize}
\item[\ref{thesisrq:multileaving}] \emph{Does the effectiveness of online evaluation methods scale to large comparisons?}
\end{itemize}

\noindent
Our answer comes in a two-fold contribution;
First, we propose a theoretical framework for systematically comparing multileaved comparison methods using the notions of \emph{considerateness}, which concerns maintaining the user experience, and \emph{fidelity}, which concerns reliable correct outcomes.
Second, we introduce a novel multileaved comparison method, \ac{PPM}, that performs comparisons based on document-pair preferences, and prove that it is {considerate} and has {fidelity}.
We show empirically that, compared to previous multileaved comparison methods, \acs{PPM} is more {sensitive} to user preferences and {scalable} with the number of rankers being compared.

\input{04-multileave/01-intro}
\input{04-multileave/02-related}

\input{04-multileave/03-framework}

\input{04-multileave/04-assessment}

\input{04-multileave/05-method}

\input{04-multileave/06-experiments}

\input{04-multileave/07-results}

\input{04-multileave/08-conclusion}

\begin{subappendices}
\input{04-multileave/notation}
\end{subappendices}

%% file: 04-multileave/01-intro.tex
\section{Introduction}
\label{sec:multileave:intro}

Evaluation is of tremendous importance to the development of modern search engines. Any proposed change to the system should be verified to ensure it is a true improvement.
Online approaches to evaluation aim to measure the actual utility of an \ac{IR} system in a natural usage
environment~\citep{hofmann-online-2016}. 
Interleaved comparison methods are a within-subject setup for online experimentation in \ac{IR}. For interleaved comparison, two experimental conditions (``control'' and ``treatment'') are typical. Recently, multileaved comparisons have been introduced for the purpose of efficiently comparing large numbers of rankers~\citep{Schuth2014a,brost2016improved}. These multileaved comparison methods were introduced as an extension to interleaving and the majority are directly derived from their interleaving counterparts \cite{Schuth2014a, schuth2015probabilistic}. The effectiveness of these methods has thus far only been measured using simulated experiments on public datasets. While this gives some insight into the general \emph{sensitivity} of a method, there is no work that assesses under what circumstances these methods provide correct outcomes and when they break. Without knowledge of the theoretical properties of multileaved comparison methods we are unable to identify when their outcomes are reliable.

In prior work on interleaved comparison methods a theoretical framework has been introduced that provides explicit requirements that an interleaved comparison method should satisfy~\citep{hofmann2013fidelity}. We take this approach as our starting point and adapt and extend it to the setting of multileaved comparison methods. Specifically, the notion of \emph{fidelity} is central to \citet{hofmann2013fidelity}'s previous work; Section~\ref{sec:multileaving} describes the framework with its requirements of \emph{fidelity}. In the setting of multileaved comparison methods, this means that a multileaved comparison method should always recognize an unambiguous winner of a comparison. We also introduce a second notion, \emph{considerateness}, which says that a comparison method should not degrade the user experience, e.g., by allowing all possible permutations of documents to be shown to the user. In this chapter we examine all existing multileaved comparison methods and find that none satisfy both the \emph{considerateness} and \emph{fidelity} requirements. In other words, no existing multileaved comparison method is correct without sacrificing the user experience. 

To address this gap, we propose a novel multileaved comparison method, \acf{PPM}. \ac{PPM} differs from existing multileaved comparison methods as its comparisons are based on inferred pairwise document preferences, whereas existing multileaved comparison methods either use some form of document assignment \cite{Schuth2014a, schuth2015probabilistic} or click credit functions \cite{Schuth2014a, brost2016improved}. We prove that \ac{PPM} meets both the \emph{considerateness} and the \emph{fidelity} requirements, thus \ac{PPM} guarantees correct winners in unambiguous cases while maintaining the user experience at all times. Furthermore, we show empirically that \ac{PPM} is more \emph{sensitive} than existing methods, i.e., it makes fewer errors in the preferences it finds. Finally, unlike other multileaved comparison methods, \ac{PPM} is computationally efficient and \emph{scalable}, meaning that it maintains most of its \emph{sensitivity} as the number of rankers in a comparison increases.

In this chapter we address thesis research question \ref{thesisrq:multileaving} by answering the following more specific research questions:
 \begin{enumerate}[label={\bf RQ2.\arabic*},leftmargin=*]
    \item Does \ac{PPM} meet the \emph{fidelity} and \emph{considerateness} requirements?\label{rq:theory}
    \item Is \ac{PPM} more sensitive than existing methods when comparing multiple rankers?\label{rq:sensitive}
\end{enumerate}
To summarize, our contributions in this chapter are:
 \begin{enumerate}[leftmargin=*]
    \item A theoretical framework for comparing multileaved comparison methods;
    \item A comparison of all existing multileaved comparison methods in terms of \emph{considerateness}, \emph{fidelity} and \emph{sensitivity};
    \item A novel multileaved comparison method that is \emph{considerate} and has \emph{fidelity} and is more \emph{sensitive} than existing methods.
\end{enumerate}

%% file: 04-multileave/02-related.tex
\section{Related Work}
\label{sec:related}

Evaluation of information retrieval systems is a core problem in IR. Two types of approach are common to designing reliable methods for measuring an IR system's effectiveness.  Offline approaches such as the Cranfield paradigm~\cite{sanderson2010} are effective for measuring topical relevance, but have difficulty taking into account contextual information including the user's current situation, fast
changing information needs, and past interaction history with the
system~\citep{hofmann-online-2016}. In contrast, online approaches to evaluation aim to measure
the actual utility of an \ac{IR} system in a natural usage
environment. User feedback in online
evaluation is usually implicit, in the form of clicks, dwell time, etc.

By far the most common type of controlled experiment on the web is A/B testing~\citep{kohavi2009controlled,Kohavi2013}. This is a classic between-subject experiment, where each subject is exposed to one of two conditions, \emph{control}---the current system---and \emph{treatment}---an experimental system that is assumed to outperform the control.

An alternative experimental design uses a within-subject setup, where all study participants are exposed to both experimental conditions. Interleaved comparisons \cite{Joachims2002,radlinski2008does} have been developed specifically for online experimentation in IR. Interleaved comparison methods have two main ingredients. First, a method for constructing interleaved result lists specifies how to select documents from the original rankings (``control'' and ``treatment''). Second, a method for inferring comparison outcomes based on observed user interactions with the interleaved result list. Because of their within-subject nature, interleaved comparisons can be up to two orders of magnitude more efficient than A/B tests in effective sample size for studies
of comparable dependent variables~\citep{chapelle2012large}.

For interleaved comparisons, two experimental conditions are typical. Extensions to multiple conditions have been introduced by \citet{Schuth2014a}. Such \emph{multileaved} comparisons are an efficient online evaluation method for comparing multiple rankers simultaneously. Similar to interleaved comparison methods~\cite{hofmann2011probabilistic, radlinski2013optimized, radlinski2008does, joachims2003evaluating}, a multileaved comparison infers preferences between rankers. Interleaved comparisons do this by presenting users with interleaved result lists; these represent two rankers in such a way that a preference between the two can be inferred from clicks on their documents. Similarly, for multileaved comparisons multileaved result lists are created that allow more than two rankers to be represented in the result list. As a consequence, multileaved comparisons can infer preferences between multiple rankers from a single click. Due to this property multileaved comparisons require far fewer interactions than interleaved comparisons to achieve the same accuracy when multiple rankers are involved \cite{Schuth2014a,  schuth2015probabilistic}.

\begin{algorithm}[h]
\begin{algorithmic}[1]
\STATE \textbf{Input}: set of rankers $\mathcal{R}$, documents $D$, no. of timesteps $T$.
\STATE $P \leftarrow \mathbf{0}$ \hfill\textit{\small// initialize $|\mathcal{R}|\times |\mathcal{R}|$ preference matrix} \label{alg:genmul:init}
\FOR{$t = 1,\ldots ,T$}
    \STATE $q_t \leftarrow \text{wait\_for\_user()}$ \hfill \textit{\small // receive query from user} \label{alg:genmul:query}
    \FOR{$i = 1,\ldots,|\mathcal{R}|$}
    	\STATE $\mathbf{l}_i \leftarrow \mathbf{r}_i(q,D)$ \hfill \textit{\small // create ranking for query per ranker}\label{alg:genmul:rank}
    \ENDFOR
    \STATE $\mathbf{m}_t \leftarrow \text{combine\_lists}(\mathbf{l}_1,\ldots,\mathbf{l}_{R})$\hfill\textit{\small// combine into multileaved list} \label{alg:genmul:multileaving}
    \STATE $\mathbf{c} \leftarrow \text{display}(\mathbf{m}_t)$ \hfill \textit{\small // display to user and record interactions} \label{alg:genmul:display}
    \FOR{$i = 1,\ldots,|\mathcal{R}|$}
    	\FOR{$j = 1,\ldots,|\mathcal{R}|$}
		\STATE $P_{ij} \leftarrow P_{ij} + \text{infer}( i,j, \mathbf{c}, \mathbf{m}_t)$ \hfill \textit{\small // infer pref. between rankers}
		\label{alg:genmul:infer}
	\ENDFOR
    \ENDFOR
\ENDFOR
\RETURN $P$
\end{algorithmic}
\caption{General pipeline for multileaved comparisons.}
\label{alg:generalmultileaving}
\end{algorithm}

The general approach for every multileaved comparison method is described in Algorithm~\ref{alg:generalmultileaving}; here, a comparison of a set of rankers $\mathcal{R}$ is performed over $T$ user interactions. After the user submits a query $q$ to the system (Line~\ref{alg:genmul:query}), a ranking $\mathbf{l}_i$ is generated for each ranker $\mathbf{r}_i$ in $\mathcal{R}$ (Line~\ref{alg:genmul:rank}). These rankings are then combined into a single result list by the multileaving method (Line~\ref{alg:genmul:multileaving}); we refer to the resulting list $\mathbf{m}$ as the  multileaved result list. In theory a multileaved result list could contain the entire document set, however in practice a length $k$ is chosen beforehand, since users generally only view a restricted number of result pages. This multileaved result list is presented to the user who has the choice to interact with it or not. Any interactions are recorded in $\mathbf{c}$ and returned to the system (Line~\ref{alg:genmul:display}). While $\mathbf{c}$ could contain any interaction information~\cite{kharitonov2015generalized}, in practice multileaved comparison methods only consider clicks. Preferences between the rankers in $\mathcal{R}$ can be inferred from the interactions and the preference matrix $P$ is updated accordingly (Line~\ref{alg:genmul:infer}). The method of inference (Line~\ref{alg:genmul:infer}) is defined by the multileaved comparison method (Line~\ref{alg:genmul:multileaving}). By aggregating the inferred preferences of many interactions a multileaved comparison method can detect preferences of users between the rankers in $\mathcal{R}$. Thus it provides a method of evaluation without requiring a form of explicit annotation.

By instantiating the general pipeline for multileaved comparisons shown in Algorithm~\ref{alg:generalmultileaving}, i.e., the combination method at Line~\ref{alg:genmul:rank} and the inference method at Line~\ref{alg:genmul:infer}, we obtain a specific multileaved comparison method. We detail all known multileaved comparison methods in Section~\ref{sec:existingmethods} below. 

\smallskip\noindent%
What we add on top of the work discussed above is a theoretical framework that allows us to assess and compare multileaved comparison methods. In addition, we propose an accurate and scalable multileaved comparison method that is the only one to satisfy the properties specified in our theoretical framework and that also proves to be the most efficient multileaved comparison method in terms of much reduced data requirements.

%% file: 04-multileave/03-framework.tex
\section{A Framework for Assessing Multileaved Comparison Methods}
\label{sec:multileaving}

Before we introduce a novel multileaved comparison method in Section~\ref{sec:novelmethod}, we propose two theoretical requirements for multileaved comparison methods. These theoretical requirements will allow us to assess and compare existing multileaved comparison methods. Specifically, we introduce two theoretical properties: \emph{considerateness} and \emph{fidelity}. These properties guarantee correct outcomes in \emph{unambigious} cases while always maintaining the user experience. In Section~\ref{sec:existingmethods} we show that no currently available multileaved comparison method satisfies both properties. This motivates the introduction of a method that satisfies both properties in Section~\ref{sec:novelmethod}. 

\subsection{Considerateness}
Firstly, one of the most important properties of a multileaved comparison method is how \textbf{considerate} it is. Since evaluation is done online it is important that the search experience is not substantially altered \citep{Joachims2002, radlinski2013optimized}. In other words, users should not be obstructed to perform their search tasks during evaluation. As maintaining a user base is at the core of any search engine, methods that potentially degrade the user experience are generally avoided. Therefore, we set the following requirement: the displayed multileaved result list should never show a document $d$ at a rank $i$ if every ranker in $\mathcal{R}$ places it at a lower rank. Writing $r(d,\mathbf{l}_j)$ for the rank of $d$ in the ranking $\mathbf{l}_j$ produced by ranker $\mathbf{r}_j$, this boils down to:
\begin{align}
\mathbf{m}_i = d \rightarrow \exists \mathbf{r}_j \in \mathcal{R}, \,  r(d,\mathbf{l}_j) \leq i. \label{eq:obstruction}
\end{align}
Requirement~\ref{eq:obstruction} guarantees that a document can never be displayed higher in a multileaved result list than any ranker would. In addition, it guarantees that if all rankers agree on the top $n$ documents, the resulting multileaved result list $\mathbf{m}$ will display the same top $n$.

\subsection{Fidelity}
Secondly, the preferences inferred by a multileaved comparison method should correspond with those of the user with respect to retrieval quality, and should be robust to user behavior that is unrelated to retrieval quality \cite{Joachims2002}. In other words, the preferences found should be correct in terms of ranker quality. However, in many cases the relative quality of rankers is unclear. For that reason we will use the notion of \textbf{fidelity} \cite{hofmann2013fidelity} to compare the correctness of a multileaved comparison method. \emph{Fidelity} was introduced by \citet{hofmann2013fidelity} and describes two general cases in which the preference between two rankers is unambiguous. To have \emph{fidelity} the expected outcome of a method is required to be correct in all matching cases. However, the original notion of \emph{fidelity} only considers two rankers as it was introduced for interleaved comparison methods, therefore the definition of \emph{fidelity} must be expanded to the multileaved case. First we describe the following concepts:

\paragraph{Uncorrelated clicks} Clicks are considered \emph{uncorrelated} if relevance has no influence on the likelihood that a document is clicked.
We write $r(d_i,\mathbf{m})$ for the rank of document $d_i$ in multileaved result list $\mathbf{m}$ and $P(\mathbf{c}_l\mid q, \mathbf{m}_l=d_i)$ for the probability of a click at the rank $l$ at which $d_i$ is displayed: $l = r(d_i,\mathbf{m})$. Then, for a given query $q$
\begin{equation}
\textit{uncorrelated}(q) \Leftrightarrow 
 \forall l, \forall d_{i,j}, \, P(\mathbf{c}_l\mid q, \mathbf{m}_l=d_i) = P(\mathbf{c}_l\mid q, \mathbf{m}_l=d_j).
\label{eq:uncorrelated} 
\end{equation}

\paragraph{Correlated clicks} We consider clicks correlated if there is a positive correlation between document relevance and clicks. However we differ from \citet{hofmann2013fidelity} by introducing a variable $k$ that denotes at which rank users stop considering documents. Writing $P(\mathbf{c}_i\mid \textit{rel}(\mathbf{m}_i,q))$ for the probability of a click at rank $i$ if a document relevant to query $q$ is displayed at this rank, we set
\begin{equation}
\begin{aligned}
& \textit{correlated}(q, k) \Leftrightarrow {}  \\
& \qquad \forall i \geq k, \, P(\mathbf{c}_i) = 0 \land \forall i < k, \, P(\mathbf{c}_i\mid \textit{rel}(\mathbf{m}_i,q)) > P(\mathbf{c}_i \mid \neg \textit{rel}(\mathbf{m}_i,q)).
\end{aligned}
\label{eq:correlated}
\end{equation}
Thus under correlated clicks a relevant document is more likely to be clicked than a non-relevant one at the same rank, if they appear above rank $k$.

\paragraph{Pareto domination} Ranker $\mathbf{r}_1$ \emph{Pareto dominates} ranker $\mathbf{r}_2$ if all relevant documents are ranked at least as high by $\mathbf{r}_1$ as by $\mathbf{r}_2$ and $\mathbf{r}_1$ ranks at least one relevant document higher. Writing $\mathit{rel}$ for the set of relevant documents that are ranked above $k$ by at least one ranker, i.e., $\mathit{rel} = \{d \mid \mathit{rel}(d,q) \land \exists \mathbf{r}_n \in \mathcal{R}, r(d,\mathbf{l}_n) > k\}$,
we require that the following holds for every query $q$ and any rank $k$:
\begin{equation}
\begin{split}
&\textit{Pareto}( \mathbf{r}_i, \mathbf{r}_j, q, k) \Leftrightarrow \\
& \qquad\qquad \forall d \in \textit{rel}, \, r(d, \mathbf{l}_i) \leq r(d, \mathbf{l}_j) \land \exists d \in \textit{rel}, \, r(d, \mathbf{l}_i) < r(d, \mathbf{l}_j).
\end{split}
\end{equation}

\noindent%
Then, \emph{fidelity} for multileaved comparison methods is defined by the following two requirements:
\begin{enumerate}[nosep,leftmargin=14pt]
\item \label{fidelity:unbias}
Under uncorrelated clicks the expected outcome may find no preferences between any two rankers in $\mathcal{R}$:
\begin{align}
\forall q, \forall (\mathbf{r}_i, \mathbf{r}_j) \in \mathcal{R}, \quad \mathit{uncorrelated}(q) \Rightarrow E[P_{ij}\mid q] = 0.
\end{align}
\item  \label{fidelity:pareto}
Under correlated clicks, a ranker that Pareto dominates all other rankers must win the multileaved comparison in expectation:
\begin{equation}
\begin{split}
& \forall k, \forall q, \forall \mathbf{r}_{i} \in \mathcal{R},\\
& \hspace{1cm} \big(\mathit{correlated}(q,k) \land {}
\forall \mathbf{r}_{j} \in \mathcal{R}, \, i \not = j \rightarrow \mathit{Pareto}( \mathbf{r}_i, \mathbf{r}_j, q,k)\big)\\
& \hspace{4.5cm} \Rightarrow  \left(\forall \mathbf{r}_{j} \in \mathcal{R}, i \not = j \rightarrow  E[P_{ij}\mid q] > 0\right).
\end{split}
\end{equation}
\end{enumerate}
Note that for the case where $|\mathcal{R}| = 2$ and if only $k = |D|$ is considered, these requirements are the same as for interleaved comparison methods~\citep{hofmann2013fidelity}.
The $k$ parameter was added to allow for \emph{fidelity} in \emph{considerate} methods, since it is impossible to detect preferences at ranks that users never consider without breaking the \emph{considerateness} requirement. We argue that differences at ranks that users are not expected to observe should not affect comparison outcomes.
 \emph{Fidelity} is important for a multileaved comparison method as it ensures that an unambiguous winner is expected to be identified. Additionally, the first requirement ensures that in exception no preferences are inferred when clicks are unaffected by relevancy.

\subsection{Additional properties}
\label{sec:non-theory}
In addition to the two theoretical properties listed above, considerateness and fidelity, we also scrutinize multileaved comparison methods to determine whether they accurately find preferences between all rankers in $\mathcal{R}$ and minimize the number of user impressions required do so. This empirical property is commonly known as \textbf{sensitivity}~\citep{Schuth2014a, hofmann2013fidelity}. In Section~\ref{sec:multileave:experiments} we describe experiments that are aimed at comparing the sensitivity of multileaved comparison methods. Here, two aspects of every comparison are considered: the level of error at which a method converges and the number of impressions required to reach that level. Thus, an interleaved comparison method that learns faster initially but does not reach the same final level of error is deemed worse.

%% file: 04-multileave/04-assessment.tex
\section{An Assessment of Existing Multileaved Comparison Methods}
\label{sec:existingmethods}

We briefly examine all existing multileaved comparison methods to determine whether they meet the \emph{considerateness} and \emph{fidelity} requirements. An investigation of the empirical sensitivity requirement is postponed until Section~\ref{sec:multileave:experiments}~and~\ref{sec:multileave:results}. 

\subsection{Team Draft Multileaving}

\ac{TDM} was introduced by \citet{Schuth2014a} and is based on the previously proposed \ac{TDI} \cite{radlinski2008does}. Both methods are inspired by how team assignments are often chosen for friendly sport matches.
The multileaved result list is created by sequentially sampling rankers without replacement; the first sampled ranker places their top document at the first position of the multileaved list.
Subsequently, the next sampled ranker adds their top pick of the remaining documents. When all rankers have been sampled, the process is continued by sampling from the entire set of rankers again. The method is stops when all documents have been added. When a document is clicked, \ac{TDM} assigns the click to the ranker that contributed the document. For each impression binary preferences are inferred by comparing the number of clicks each ranker received.

It is clear that \ac{TDM} is \emph{considerate} since each added document is the top pick of at least one ranker. However, \ac{TDM} does not meet the fidelity requirements. This is unsurprising as previous work has proven that \ac{TDI} does not meet these requirements \cite{radlinski2013optimized,hofmann2011probabilistic,hofmann2013fidelity}. Since \ac{TDI} is identical to \ac{TDM} when the number of rankers is $|\mathcal{R}| = 2$, \ac{TDM} does not have \emph{fidelity} either.

\subsection{Optimized Multileaving}

\ac{OM} was proposed by \citet{Schuth2014a} and serves as an extension of \ac{OI} introduced by \citet{radlinski2013optimized}. The allowed multileaved result lists of \ac{OM} are created by sampling rankers with replacement at each iteration and adding the top document of the sampled ranker. However, the probability that a multileaved result list is shown is not determined by the generative process. Instead, for a chosen credit function \ac{OM} performs an optimization that computes a probability for each multileaved result list so that the expected outcome is unbiased and sensitive to correct preferences.

All of the allowed multileaved result lists of \ac{OM} meet the \emph{considerateness} requirement, and in theory instantiations of \ac{OM} could have \emph{fidelity}. However, in practice \ac{OM} does not meet the \emph{fidelity} requirements. There are two main reasons for this. First, it is not guaranteed that a solution exists for the optimization that \ac{OM} performs. For the interleaving case this was proven empirically when $k=10$ \cite{radlinski2013optimized}. However, this approach does not scale to any number of rankers. Secondly, unlike \ac{OI}, \ac{OM} allows more result lists than can be computed in a feasible amount of time. Consider the top $k$ of all possible multileaved result lists; in the worst case this produces $|\mathcal{R}|^k$ lists. Computing all lists for a large value of $|\mathcal{R}|$ and performing linear constraint optimization over them is simply not feasible. As a solution, \citet{Schuth2014a} propose a method that samples from the allowed multileaved result lists and relaxes constraints when there is no exact solution. Consequently, there is no guarantee that this method does not introduce bias. Together, these two reasons show that the \emph{fidelity} of \ac{OI} does not imply fidelity of \ac{OM}. It also shows that \ac{OM} is  computationally very costly. 

\subsection{Probabilistic Multileaving}
\ac{PM} \citep{schuth2015probabilistic} is an extension of \ac{PI} \cite{hofmann2011probabilistic}, which was designed to solve the flaws of \ac{TDI}. Unlike the previous methods, \ac{PM} considers every ranker as a distribution over documents, which is created by applying a soft-max to each of them. A multileaved result list is created by sampling a ranker with replacement at each iteration and sampling a document from the ranker that was selected. After the sampled document has been added, all rankers are renormalized to account for the removed document. During inference \ac{PM} credits every ranker the expected number of clicked documents that were assigned to them. This is done by marginalizing over the possible ways the list could have been constructed by \ac{PM}. A benefit of this approach is that it allows for comparisons on historical data \cite{hofmann2011probabilistic, hofmann2013fidelity}.

A big disadvantage of \ac{PM} is that it allows any possible ranking to be shown, albeit not with uniform probabilities. This is a big deterrent for the usage of \ac{PM} in operational settings. Furthermore, it also means that \ac{PM} does not meet the \emph{considerateness} requirement. On the other hand, \ac{PM} does meet the \emph{fidelity} requirements, the proof for this follows from the fact that every ranker is equally likely to add a document at each location in the ranking. Moreover, if multiple rankers want to place the same document somewhere they have to share the resulting credits.\footnote{\citet{brost2016improved} proved that if the preferences at each impression are made binary the  \emph{fidelity} of \ac{PM} is lost.} Similar to \ac{OM}, \ac{PM} becomes infeasible to compute for a large number of rankers $|\mathcal{R}|$; the number of assignments in the worst case is $|R|^{k}$. Fortunately, \ac{PM} inference can be estimated by sampling assignments in a way that maintains \emph{fidelity} \cite{schuth2015probabilistic, oosterhuis2016probabilistic}.

\subsection{Sample Only Scored Multileaving}
\ac{SOSM} was introduced by \citet{brost2016improved} in an attempt to create a more scalable multileaved comparison method. It is the only existing multileaved comparison method that does not have an interleaved comparison counterpart.
\ac{SOSM} attempts to increase \emph{sensitivity} by ignoring all non-sampled documents during inference. Thus, at each impression a ranker receives credits according to how it ranks the documents that were sampled for the displayed multileaved result list of size $k$. The preferences at each impression are made binary before being added to the mean. \ac{SOSM} creates multileaved result lists following the same procedure as \ac{TDM}, a choice that seems arbitrary.

\ac{SOSM} meets the \emph{considerateness} requirements for the same reason \ac{TDM} does. However, \ac{SOSM} does not meet the fidelity requirement. We can prove this by providing an example where preferences are found under uncorrelated clicks. Consider the two documents \textit{A} and \textit{B} and the three rankers with the following three rankings:
\begin{align*}
\mathbf{l}_1 = \textit{AB}, \qquad
\mathbf{l}_2 = \mathbf{l}_3 = \textit{BA}. 
\end{align*}
The first requirement of fidelity states that under uncorrelated clicks no preferences may be found in expectation. Uncorrelated clicks are unconditioned on document relevance (Equation~\ref{eq:uncorrelated}); however, it is possible that they display position bias \cite{yue2010beyond}. Thus the probability of a click at the first rank may be greater than at the second:
\begin{align*}
P(\mathbf{c}_1\mid q) > P(\mathbf{c}_2\mid q).
\end{align*}
Under position biased clicks the expected outcome for each possible multileaved result list is not zero. For instance, the following preferences are expected:
\begin{align}
E[P_{12}\mid \mathbf{m} = \textit{AB}] &> 0, \nonumber\\
E[P_{12}\mid \mathbf{m} = \textit{BA}] &< 0, \nonumber\\
E[P_{12}\mid \mathbf{m} = \textit{AB}] &= -E[P_{12}\mid \mathbf{m} = \textit{BA}]. \nonumber
\end{align}
Since \ac{SOSM} creates multileaved result lists following the \ac{TDM} procedure the probability $P(\mathbf{m} = \textit{BA})$ is twice as high as $P(\mathbf{m} = \textit{AB})$. As a consequence, the expected preference is biased against the first ranker:
\begin{align}
E[P_{12}] < 0. \nonumber
\end{align}
Hence, \ac{SOSM} does not have \emph{fidelity}. This outcome seems to stem from a disconnect between how multileaved results lists are created and how preferences are inferred.

\smallskip\noindent%
To conclude this section, Table~\ref{tab:properties} provides an overview of our findings thus far, i.e., the theoretical requirements that each multileaved comparison method satisfies; we have also included \acs{PPM}, the multileaved comparison method that we will introduce below.

\begin{table}[tb]
\caption{Overview of multileaved comparison methods and whether they meet the \emph{considerateness} and \emph{fidelity} requirements.}
\centering
\begin{tabular}{ l  c c c   }
\toprule
 & \textbf{Considerateness} & \textbf{Fidelity} & \textbf{Source} \\
\midrule
\acs{TDM} & \checkmark & & \citep{Schuth2014a} \\
\acs{OM} & \checkmark & & \citep{Schuth2014a} \\
\acs{PM} & & \checkmark & \citep{schuth2015probabilistic} \\
\acs{SOSM} & \checkmark & & \citep{brost2016improved} \\
\acs{PPM} & \checkmark & \checkmark & this chapter \\
\bottomrule 
\end{tabular}
\label{tab:properties}
\end{table}

%% file: 04-multileave/05-method.tex
\section{A Novel Multileaved Comparison Method}
\label{sec:novelmethod}

The previously described multileaved comparison methods are based around direct credit assignment, i.e., credit functions are based on single documents.
In contrast, we introduce a method that estimates differences based on pairwise document preferences.
We prove that this novel method is the only multileaved comparison method that meets the \emph{considerateness} and \emph{fidelity} requirements set out in Section~\ref{sec:multileaving}.

\label{sec:preferencemultileave}
The multileaved comparison method that we introduce is \acf{PPM}. It infers pairwise preferences between documents from clicks and bases comparisons on the agreement of rankers with the inferred preferences.
\ac{PPM} is based on the assumption that a clicked document is preferred to: (i) all of the unclicked documents above it; and (ii) the next unclicked document. These assumptions are long-established \cite{joachims2002unbiased} and form the basis of pairwise \ac{LTR} \cite{Joachims2002}.%

We write $\mathbf{c}_{r(d_i,\mathbf{m})}$ for a click on document $d_i$ displayed in multileaved result list $\mathbf{m}$ at the rank $r(d_i,\mathbf{m})$. For a document pair $(d_i, d_j)$, a click $\mathbf{c}_{r(d_i,\mathbf{m})}$ infers a preference as follows:
\begin{equation}
c_{r(d_i,\mathbf{m})} \land \neg c_{r(d_j,\mathbf{m})}
\land
\big(\exists i, (c_i \land r(d_j,\mathbf{m}) < i) \lor c_{r(d_j,\mathbf{m}) -1}\big)
\Leftrightarrow d_i >_{c} d_j.
\end{equation}
In addition, the preference of a ranker $\mathbf{r}$ is denoted by $d_i >_{\mathbf{r}} d_j$.
Pairwise preferences also form the basis for \ac{PBI} introduced by~\citet{he2009evaluation}. However, previous work has shown that \ac{PBI} does not meet the \emph{fidelity} requirements~\citep{hofmann2013fidelity}. Therefore, we do not use \ac{PBI} as a starting point for \ac{PPM}. Instead, \ac{PPM} is derived directly from the \emph{considerateness} and \emph{fidelity} requirements. Consequently, \ac{PPM} constructs multileaved result lists inherently differently and its inference method has \emph{fidelity}, in contrast with \ac{PBI}.

\begin{algorithm}[t]
\begin{algorithmic}[1]
\STATE \textbf{Input}: set of rankers $\mathcal{R}$, rankings $\{\mathbf{l}\}$, documents $D$.
\STATE $\mathbf{m} \leftarrow []$ \hfill \textit{\small // initialize empty multileaving}
\FOR{$n = 1,\ldots,|D|$}
    \STATE $\hat{\Omega}_n \leftarrow \Omega(n,\mathcal{R},D) \setminus \mathbf{m}$  \hfill \textit{\small // choice set of remaining documents} \label{alg:ppm:docset}
    \STATE $d \leftarrow \text{uniform\_sample}(\hat{\Omega}_n)$   \hfill \textit{\small // uniformly sample next document} \label{alg:ppm:next}
    \STATE $\mathbf{m} \leftarrow \text{append}(\mathbf{m},d)$   \hfill \textit{\small // add sampled document to multileaving} \label{alg:ppm:append}
\ENDFOR
\RETURN $\mathbf{m}$
\end{algorithmic}
\caption{Multileaved result list construction for \ac{PPM}.}
\label{alg:ppm-listconstruction}
\end{algorithm}

\begin{algorithm}[t]
\begin{algorithmic}[1]
\STATE \textbf{Input}: rankers $\mathcal{R}$, rankings $\{\mathbf{l}\}$, documents $D$, multileaved result list $\mathbf{m}$, clicks $\mathbf{c}$.
\STATE  $P \leftarrow \mathbf{0}$ \hfill \textit{\small // preference matrix of $|\mathcal{R}|\times|\mathcal{R}|$ } 
\FOR{$(d_i, d_j) \in \{(d_i,d_j)\mid d_i >_{\mathbf{c}} d_j\}$}
       \IF{$\ubar{r}(i,j, \mathbf{m}) \geq \bar{r}(i,j)$}
           \STATE ${w} \leftarrow 1$ \hfill \textit{\small // variable to store $P(\ubar{r}(i,j, \mathbf{m})  \geq \bar{r}(i,j))$ }
           \STATE $\textit{min\_x} \leftarrow \min_{d \in \{d_i,d_j\}} \min_{\mathbf{r}_n \in \mathcal{R}} r(d,\mathbf{l}_n)$
           \FOR{$x = \textit{min\_x},\ldots,\bar{r}(i,j)-1$}
               \STATE ${w} \leftarrow {w} \cdot (1 - (|\Omega(x,\mathcal{R},D)|-x-1)^{-1})$
           \ENDFOR
       \FOR{$n = 1,\ldots,|R|$}
       	   \FOR{$m = 1,\ldots,|R|$}
                    \IF{$d_i >_{\mathbf{r}_n} d_j \land n \not = m$}
                        \STATE $P_{nm} \leftarrow P_{nm} + {w}^{-1}$\hfill \textit{\small // result of scoring function $\phi$ }
                    \ELSIF{$n \not = m$}
                         \STATE $P_{nm} \leftarrow P_{nm} - {w}^{-1}$
        	       	   \ENDIF
            \ENDFOR
        \ENDFOR
   \ENDIF   
\ENDFOR
\RETURN $P$
\end{algorithmic}
\caption{Preference inference for \ac{PPM}.}
\label{alg:ppm-inference}
\end{algorithm}

When constructing a multileaved result list $\mathbf{m}$ we want to be able to infer unbiased preferences while simultaneously being \emph{considerate}. Thus, with the requirement for \emph{considerateness} in mind we define a choice set as:
\begin{align}
\Omega(i,\mathcal{R},D) = \{ d \mid  d \in D \land \exists \mathbf{r}_j \in \mathcal{R}, r(d,\mathbf{l}_j) \leq i \}.
\end{align}
This definition is chosen so that any document in $\Omega(i,\mathcal{R}, D)$ can be placed at rank $i$ without breaking the \emph{considerateness} requirement (Equation~\ref{eq:obstruction}). The multileaving method of \ac{PPM} is described in Algorithm~\ref{alg:ppm-listconstruction}. The approach is straightforward: at each rank $n$ the set of documents $\hat{\Omega}_n$ is determined (Line~\ref{alg:ppm:docset}). This set of documents is $\Omega(n,\mathcal{R},D)$ with the previously added documents removed to avoid document repetition. Then, the next document is sampled uniformly from $\hat{\Omega}_n$  (Line~\ref{alg:ppm:next}), thus every document in $\hat{\Omega}_n$ has a probability:
\begin{align}
\frac{1}{|\Omega(n,\mathcal{R},D)| - n + 1}
\end{align}
of being placed at position $n$ (Line~\ref{alg:ppm:append}). Since $\hat{\Omega}_n \subseteq \Omega(n,\mathcal{R},D)$ the resulting $\mathbf{m}$ is guaranteed to be \emph{considerate}.

While the multileaved result list creation method used by \ac{PPM} is simple, its preference inference method is more complicated as it has to meet the \emph{fidelity} requirements. First, the preference found between a ranker $\mathbf{r}_n$ and $\mathbf{r}_m$ from a single interaction $\mathbf{c}$ is determined by:
\begin{align}
P_{nm} = \sum_{d_i >_\mathbf{c} d_j} \phi(d_i,d_j,\mathbf{r}_n, \mathbf{m}, \mathcal{R}) - \phi(d_i,d_j,\mathbf{r}_m, \mathbf{m}, \mathcal{R}),
\end{align}
which sums over all document pairs $(d_i,d_j)$ where interaction $\mathbf{c}$ inferred a preference. Before the scoring function $\phi$ can be defined we introduce the following function:
\begin{align}
 \bar{r}(i,j,\mathcal{R}) = \max_{d \in \{d_i,d_j\}} \min_{\mathbf{r}_n \in \mathcal{R}} r(d,\mathbf{l}_n).
\end{align}
For succinctness we will note $\bar{r}(i,j) =  \bar{r}(i,j,\mathcal{R})$.
Here, $\bar{r}(i,j)$ provides the highest rank at which both documents $d_i$ and $d_j$ can appear in $\mathbf{m}$.
Position $\bar{r}(i,j)$ is important to the document pair $(d_i,d_j)$, since if both documents are in the remaining documents $\hat{\Omega}_{\bar{r}(i,j)}$, then the rest of the multileaved result list creation process is identical for both. To keep notation short we introduce:
\begin{align}
\ubar{r}(i,j, \mathbf{m}) = \min_{d \in \{d_i,d_j\}} r(d,\mathbf{m}).
\end{align}
Therefore, if $\ubar{r}(i,j, \mathbf{m})  \geq \bar{r}(i,j)$ then both documents appear below $\bar{r}(i,j)$. This, in turn, means that both documents are equally likely to appear at any rank:
\begin{equation}
\begin{split}
&\forall n,\,
P\mleft(\mathbf{m}_n = d_i \mid \ubar{r}(i,j, \mathbf{m})  \geq \bar{r}(i,j)\mright)
\\
& \hspace{4cm}
= P\mleft(\mathbf{m}_n = d_j \mid  \ubar{r}(i,j, \mathbf{m})  \geq \bar{r}(i,j)\mright).
\end{split}
\label{eq:uniformpos}
\end{equation}
The scoring function $\phi$ is then defined as follows:
\begin{align}
 \phi(d_i,d_j,\mathbf{r}, \mathbf{m}) = 
 \begin{cases}
 0, &\ubar{r}(i,j, \mathbf{m}) < \bar{r}(i,j)
 \\
 \frac{-1}{P(\ubar{r}(i,j, \mathbf{m})  \geq \bar{r}(i,j))}, & d_i <_\mathbf{r} d_j
 \\
 \frac{1}{P(\ubar{r}(i,j, \mathbf{m})  \geq \bar{r}(i,j))}, & d_i >_\mathbf{r} d_j,
 \end{cases}
 \label{eq:scoringfunction}
\end{align}
indicating that a zero score is given if one of the documents appears above $\bar{r}(i,j)$.
Otherwise, the value of $\phi$ is positive or negative depending on whether the ranker $\mathbf{r}$ agrees with the inferred preference between $d_i$ and $d_j$. Furthermore, this score is inversely weighted by the probability $P(\ubar{r}(i,j, \mathbf{m})  \geq \bar{r}(i,j))$. Therefore, pairs that are less likely to appear below their threshold $\bar{r}(i,j)$ result in a higher score than for more commonly occuring pairs. Algorithm~\ref{alg:ppm-inference} displays how the inference of \ac{PPM} can be computed. The scoring function $\phi$ was carefully chosen to guarantee \emph{fidelity}, the remainder of this section will sketch the proof for \ac{PPM} meeting its requirements. 

The two requirements for \emph{fidelity} will be discussed in order:

\subsubsection*{Requirement~\ref{fidelity:unbias}} The first fidelity requirement states that under uncorrelated clicks the expected outcome should be zero. Consider the expected preference:
\begin{equation}
\begin{split}
E[P_{nm}]
= \sum_{d_i,d_j} \sum_{\mathbf{m}}  P(d_i >_c d_j \mid  \mathbf{m}) P(\mathbf{m}) 
 (&\phi(d_i,d_j,\mathbf{r}_n, \mathbf{m}) \\ & \, -  \phi(d_i,d_j,\mathbf{r}_m, \mathbf{m})).
 \end{split}
\label{eq:expectedoutcome}
\end{equation}
To see that $E[P_{nm}] = 0$ under uncorrelated clicks, take any multileaving $\mathbf{m}$ where $P(\mathbf{m}) > 0$ and $\phi(d_i,d_j,\mathbf{r}, \mathbf{m}) \not= 0$ with $\mathbf{m}_x = d_i$ and $\mathbf{m}_y = d_j$.  Then there is always a multileaved result list $\mathbf{m}'$ that is identical expect for swapping the two documents so that $\mathbf{m}'_x = d_j$ and $\mathbf{m}'_y = d_i$.
The scoring function only gives non-zero values if both documents appear below the threshold $\ubar{r}(i,j, \mathbf{m}) < \bar{r}(i,j)$ (Equation~\ref{eq:scoringfunction}). At this point the probability of each document appearing at any position is the same (Equation~\ref{eq:uniformpos}), thus the following holds:
\begin{align}
P(\mathbf{m}) &= P(\mathbf{m}'),\\
\phi(d_i,d_j,\mathbf{r}_n, \mathbf{m})  &= -\phi(d_j,d_i,\mathbf{r}_n, \mathbf{m}').
\end{align}
Finally, from the definition of uncorrelated clicks (Equation~\ref{eq:uncorrelated}) the following holds:
\begin{align}
P(d_i >_c d_j \mid  \mathbf{m}) &= P(d_j >_c d_i \mid  \mathbf{m}').
\end{align}
As a result, any document pair $(d_i,d_j)$ and multileaving $\mathbf{m}$ that affects the expected outcome is cancelled by the multileaving $\mathbf{m}'$. Therefore, we can conclude that $E[P_{nm}] = 0$ under uncorrelated clicks, and that \ac{PPM} meets the first requirement of \emph{fidelity}.

\subsubsection*{Requirement~\ref{fidelity:pareto}} The second \emph{fidelity} requirement states that under correlated clicks a ranker that Pareto dominates all other rankers should win the multileaved comparison. Therefore, the expected value for a Pareto dominating ranker $r_n$ should be:
\begin{align}
\forall m, n \not = m \rightarrow E[P_{nm}] > 0.
\end{align}
Take any other ranker $\mathbf{r}_m$ that is thus Pareto dominated by $\mathbf{r}_n$.
The proof for the first requirement shows that $E[P_{nm}]$ is not affected by any pair of documents $d_i,d_j$  with the same relevance label.
Furthermore, any pair on which $\mathbf{r}_n$ and $\mathbf{r}_m$ agree will not affect the expected outcome since:
\begin{equation}
(d_i >_{\mathbf{r}_n} d_j \leftrightarrow d_i >_{\mathbf{r}_m} d_j) \Rightarrow \phi(d_i,d_j,\mathbf{r}_n, \mathbf{m}) -  \phi(d_i,d_j,\mathbf{r}_m, \mathbf{m}) = 0.
\end{equation}
Then, for any relevant document $d_i$, consider the set of documents that $\mathbf{r}_n$ incorrectly prefers over $d_i$:
\begin{align}
A = \{d_j \mid  \neg \mathit{rel}(d_j) \land d_j >_{\mathbf{r}_n} d_i \}
\end{align}
and the set of documents that $\mathbf{r}_m$ incorrectly prefers over $d_i$ and places higher than where $\mathbf{r}_n$ places $d_i$:
\begin{align}
B = \{d_j \mid  \neg \mathit{rel}(d_j) \land  d_j >_{\mathbf{r}_m} d_i \land r(d_j, \mathbf{l}_m) <  r(d_i, \mathbf{l}_n)\}.
\end{align}
Since $\mathbf{r}_n$ Pareto dominates $\mathbf{r}_m$, it has the same or fewer incorrect preferences: $|A| \leq |B|$. Furthermore, for any document $d_j$ in either $A$ or $B$ the threshold of the pair $d_i,d_j$ is the same:
\begin{align}
\forall d_j \in A \cup B, \bar{r}(i,j) = r(d_i, \mathbf{l}_n).
\end{align}
Therefore, all pairs with documents from $A$ and $B$ will only get a non-zero value from $\phi$ if they both appear at or below $r(d_i, \mathbf{l}_n)$.
Then, using  Equation~\ref{eq:uniformpos} and the Bayes rule we see:
\begin{equation}
\begin{aligned}
\forall (d_j, d_l) \in  A \cup B, \quad
&
\frac{
P(\mathbf{m}_x = d_j, \ubar{r}(i,j, \mathbf{m})  \geq \bar{r}(i,j,\mathcal{R}))
}{
P(\ubar{r}(i,j, \mathbf{m})  \geq \bar{r}(i,j,\mathcal{R}))
}
\\
&{}=
\frac{
P(\mathbf{m}_x = d_l, \ubar{r}(i,l, \mathbf{m})  \geq \bar{r}(i,l,\mathcal{R}))
}{
P(\ubar{r}(i,l, \mathbf{m})  \geq \bar{r}(i,l,\mathcal{R}))
}.
\end{aligned}
\end{equation}
Similarly, the reweighing of $\phi$ ensures that every pair in $A$ and $B$ contributes the same to the expected outcome. Thus, if both rankers rank $d_i$ at the same position the following sum:
\begin{equation}
\begin{aligned}
 &\sum_{d_j \in A \cup B} \sum_{\mathbf{m}} P(\mathbf{m}) 
 \\
 & \hspace{2cm} \cdot \mleft[
 P(d_i >_c d_j \mid  \mathbf{m}) 
 (\phi(d_i,d_j,\mathbf{r}_n, \mathbf{m}) -  \phi(d_i,d_j,\mathbf{r}_m, \mathbf{m}))\mright.
 \\
& \hspace{2cm} \mleft.{}+ P(d_j >_c d_i \mid  \mathbf{m})
 (\phi(d_j,d_i,\mathbf{r}_n, \mathbf{m}) -  \phi(d_j,d_i,\mathbf{r}_m, \mathbf{m}))\mright]
 \end{aligned}
 \end{equation}
 will be zero if $|A| = |B|$ and positive if $|A| < |B|$ under correlated clicks. Moreover, since $\mathbf{r}_n$ Pareto dominates $\mathbf{r}_m$, there will be at least one document $d_j$ where:
 \begin{align}
 \exists d_i, \exists d_j, \mathit{rel}(d_i) \land \neg \mathit{rel}(d_j) \land r(d_i, \mathbf{l}_n) = r(d_j, \mathbf{l}_m).
 \end{align}
This means that the expected outcome (Equation~\ref{eq:expectedoutcome}) will always be positive under correlated clicks, i.e., $E[P_{nm}] > 0$, for a Pareto dominating ranker $\mathbf{r}_n$ and any other ranker $\mathbf{r}_m$.

In summary, we have introduced a new multileaved comparison method, \ac{PPM}.
Furthermore, we answered \ref{rq:theory} in the affirmative since we have shown it to be \emph{considerate} and to have \emph{fidelity}. We further note that \ac{PPM} has polynomial complexity: to calculate $P(\ubar{r}(i,j, \mathbf{m})  \geq \bar{r}(i,j))$ only the size of the choice sets $\Omega$ and the first positions at which $d_i$ and $d_j$ occur in $\Omega$ have to be known.

%% file: 04-multileave/06-experiments.tex
\section{Experiments}
\label{sec:multileave:experiments}

In order to answer Research Question~\ref{rq:sensitive} posed in Section~\ref{sec:multileave:intro} several experiments were performed to evaluate the \emph{sensitivity} of \ac{PPM}. The methodology of evaluation follows previous work on interleaved and multileaved comparison methods \cite{Schuth2014a, hofmann2011probabilistic, schuth2015probabilistic, hofmann2013fidelity, brost2016improved} and is completely reproducible.%

\subsection{Ranker selection and comparisons}
In order to make fair comparisons between rankers, we will use the \ac{LTR} datasets described in Section~\ref{sec:datasets} below. From the feature representations in these datasets a handpicked set of features was taken and used as ranking models. To match the real-world scenario as best as possible this selection consists of features that are known to perform well as relevance signals independently. This selection includes but is not limited to: BM25, LMIR.JM, Sitemap, PageRank, HITS and TF.IDF~\citep{Schuth2014a}.

Then the ground-truth comparisons between the rankers are based on their NDCG scores computed on a held-out test set, resulting in a binary preference matrix $P_{nm}$ for all ranker pairs $(\mathbf{r}_n,\mathbf{r}_m)$:
\begin{align}
P_{nm} = NDCG(\mathbf{r}_n) - NDCG(\mathbf{r}_m).
\end{align}
The metric by which multileaved comparison methods are compared is the \emph{binary error}, $E_{bin}$~ \citep{Schuth2014a, brost2016improved, schuth2015probabilistic}. Let $\hat{P}_{nm}$ be the preference inferred by a multileaved comparison method; then the error is:
\begin{align}
E_{bin} = \frac{\sum_{n,m\in\mathcal{R} \land n \not = m}\sgn(\hat{P}_{nm}) \not = sgn(P_{nm})}{|\mathcal{R}| \times (|\mathcal{R}| - 1)}.
\end{align}

\subsection{Datasets}
\label{sec:datasets}

Our experiments are performed over ten publicly available \acs{LTR} datasets with varying sizes and representing different search tasks. Each dataset consists of a set of queries and a set of corresponding documents for every query. While queries are represented only by their identifiers, feature representations and relevance labels are available for every document-query pair. Relevance labels are graded differently by the datasets depending on the task they model, for instance, navigational datasets have binary labels for not relevant (0), and relevant (1), whereas  most informational tasks have labels ranging from not relevant (0), to perfect relevancy (4).
Every dataset consists of five folds, each dividing the dataset in different training, validation and test partitions.

The first publicly available \ac{LTR} datasets are distributed as LETOR 3.0 and 4.0~\cite{letor}; they use representations of 45, 46, or 64 features encoding ranking models such as TF.IDF, BM25, Language Modelling, PageRank, and HITS on different parts of the documents. The datasets in LETOR are divided by their tasks, most of which come from the TREC Web Tracks between 2003 and 2008 \cite{craswell2003overview,clarke2009overview}.
\emph{HP2003, HP2004, NP2003, NP2004, TD2003} and \emph{TD2004} each contain between 50 and 150 queries and 1,000 judged documents per query and use binary relevance labels. Due to their similarity we report average results over these six datasets noted as \emph{LETOR 3.0}.
The \emph{OH\-SU\-MED} dataset is based on the query log of the search engine on the MedLine abstract database, and contains 106 queries. The last two datasets, \emph{MQ2007} and \emph{MQ2008}, were based on the Million Query Track \cite{allan2007million} and consist of 1,700 and 800 queries, respectively, but have far fewer assessed documents per query.

The \emph{MLSR-WEB10K} dataset \cite{qin2013introducing} consists of 10,000 queries obtained from a retired labelling set of a commercial web search engine. %
The datasets uses 136 features to represent its documents, each query has around 125 assessed documents.

Finally, we note that there are more \ac{LTR} datasets that are publicly available \cite{Chapelle2011, dato2016fast}, but there is no public information about their feature representations. Therefore, they are unfit for our evaluation as no selection of well performing ranking features can be made.

\subsection{Simulating user behavior}
\label{sec:experiments:users}

\begin{table}[tb]
\caption{Instantiations of Cascading Click Models~\cite{guo09:efficient} as used for simulating user behaviour in experiments.}
\label{tab:multileave:clickmodels}
\centering
\begin{tabular}{ @{} l c c c c c c c c c c }
\toprule
& \multicolumn{5}{c}{\small $P(\mathit{click}=1\mid R)$} & \multicolumn{5}{c}{\small $P(\mathit{stop}=1\mid R)$} \\
\cmidrule(r){2-6}\cmidrule(l){7-11}
\small $R$ & \small \emph{$ 0$} & \small \emph{$ 1$}  & \small \emph{$ 2$} & \small \emph{$ 3$} & \small \emph{$ 4$}
 & \small \emph{$0$} & \small \emph{$ 1$} & \small \emph{$ 2$} & \small \emph{$ 3$} & \small \emph{$ 4$} \\
\midrule
\small \emph{perfect} & \small 0.0 & \small 0.2 & \small 0.4 & \small 0.8 & \small 1.0 & \small 0.0 & \small 0.0 & \small 0.0 & \small 0.0 & \small 0.0 \\
\small \emph{navigational} & \small ~~0.05 & \small 0.3 & \small 0.5 & \small 0.7 & \small ~~0.95 & \small 0.2 & \small 0.3 & \small 0.5 & \small 0.7 & \small 0.9 \\
\small \emph{informational} & \small 0.4 & \small 0.6 & \small 0.7 & \small 0.8 & \small 0.9 & \small 0.1 & \small 0.2 & \small 0.3 & \small 0.4 & \small 0.5 \\
\bottomrule
\end{tabular}
\end{table}

While experiments using real users are preferred \cite{chuklin2015comparative, chapelle2012large, kharitonov2015generalized, yue2010learning}, most researchers do not have access to search engines. As a result the most common way of comparing online evaluation methods is by using simulated user behaviour~\citep{Schuth2014a, hofmann2011probabilistic, schuth2015probabilistic, hofmann2013fidelity, brost2016improved}. Such simulated experiments show the performance of multileaved comparison methods when user behaviour adheres to a few simple assumptions. 

Our experiments follow the precedent set by previous work on online evaluation:
First, a user issues a query simulated by uniformly sampling a query from the static dataset. Subsequently, the multileaved comparison method constructs the multileaved result list of documents to display. The behavior of the user after receiving this list is simulated using a \emph{cascade click model}~\cite{chuklin-click-2015,guo09:efficient}. This model assumes a user to examine documents in their displayed order. For each document that is considered the user decides whether it warrants a click, which is modeled as the conditional probability $P(click=1\mid R)$ where $R$ is the relevance label provided by the dataset. Accordingly, \emph{cascade click model} instantiations increase the probability of a click with the degree of the relevance label. After the user has clicked on a document their information need may be satisfied; otherwise they continue considering the remaining documents. The probability of the user not examining more documents after clicking is modeled as $P(stop=1\mid R)$,
where it is more likely that the user is satisfied from a very relevant document. At each impression we display $k=10$ documents to the user.

Table~\ref{tab:multileave:clickmodels} lists the three instantiations of cascade click models that we use for this chapter.
The first models a \emph{perfect} user who considers every document and clicks on all relevant documents and nothing else.
Secondly, the \emph{navigational} instantiation models a user performing a navigational task who is mostly looking for a single highly relevant document. Finally, the \emph{informational} instantiation models a user without a very specific information need who typically clicks on multiple documents.
These three models have increasing levels of noise, as the behavior of each depends less on the relevance labels of the displayed documents.

\subsection{Experimental runs}
\label{sec:experiments:runs}

Each experimental run consists of applying a multileaved comparison method to a sequence of $T=10,000$ simulated user impressions. To see the effect of the number of rankers in a comparison, our runs consider $|\mathcal{R}| = 5$, $|\mathcal{R}| = 15$, and $|\mathcal{R}| = 40$. However only the \emph{MSLR} dataset contains $|\mathcal{R}| = 40$ rankers. Every run is repeated for every click model to see how different behaviours affect performance. For statistical significance every run is repeated 25 times per fold, which means that 125 runs are conducted for every dataset and click model pair. Since our evaluation covers five multileaved comparison methods, we generate over 393 million impressions in total. We test for statistical significant differences using a two tailed t-test. Note that the results reported on the LETOR 3.0 data are averaged over six datasets and thus span 750 runs per datapoint.

The parameters of the baselines are selected based on previous work on the same datasets; for \ac{OM} the sample size $\eta=10$ was chosen as reported by \citet{Schuth2014a}; for \ac{PM} the degree $\tau=3.0$ was chosen according to \citet{hofmann2011probabilistic} and the sample size $\eta=10,000$ in accordance with \citet{schuth2015probabilistic}.

%% file: 04-multileave/07-results.tex
\begin{figure}[tb]
\centering
\includegraphics[width=\textwidth]{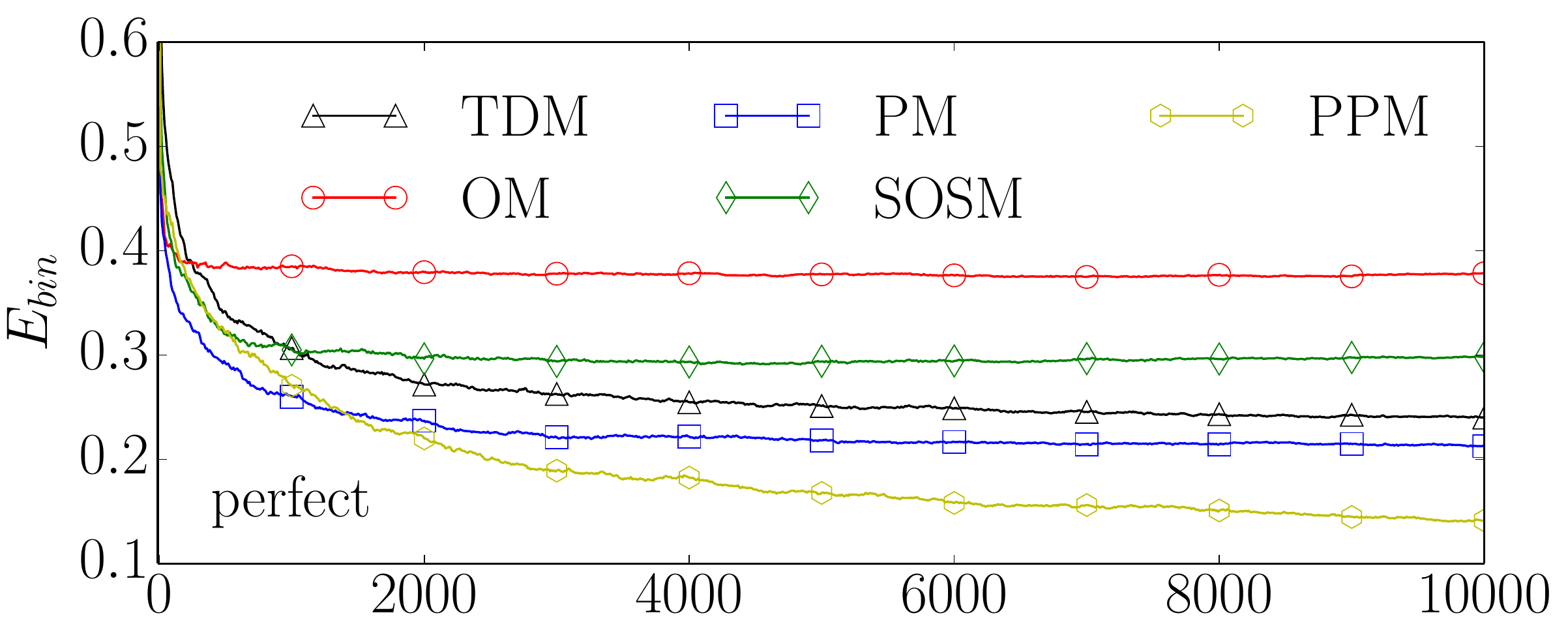}
\includegraphics[width=\textwidth]{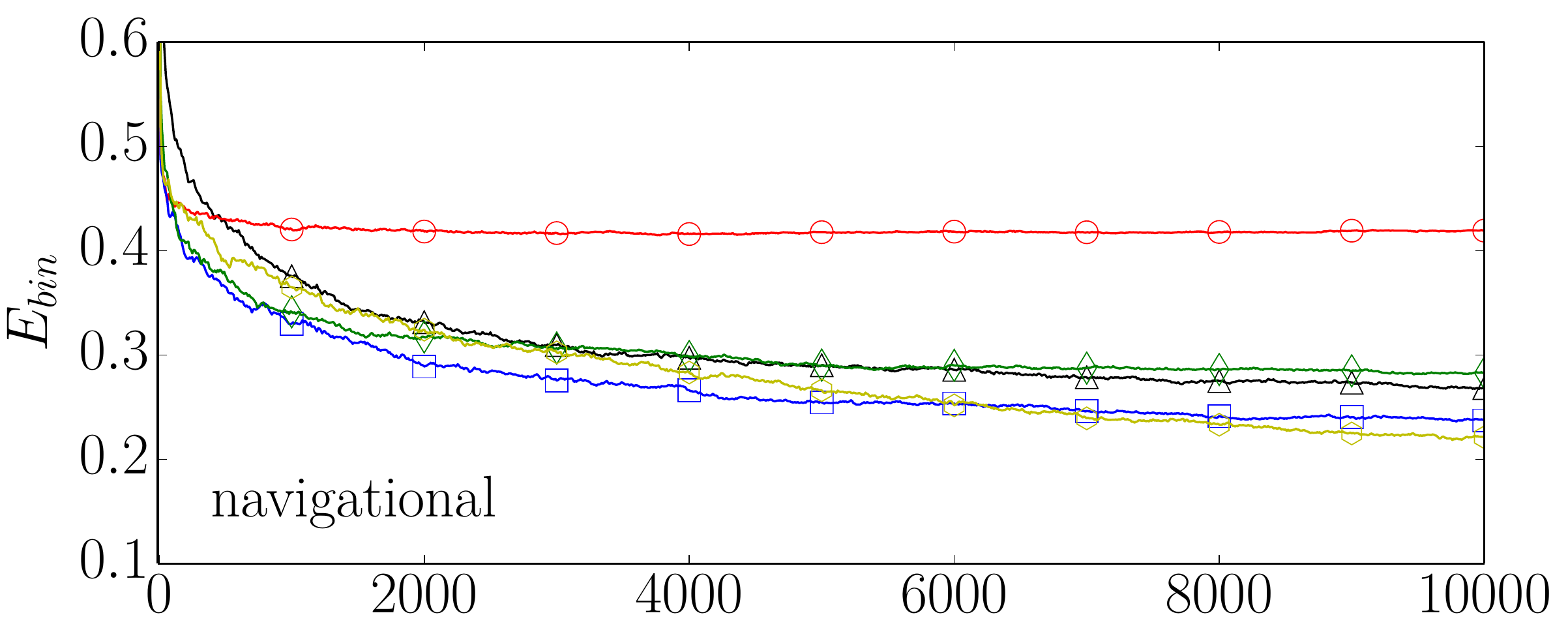}
\includegraphics[width=\textwidth]{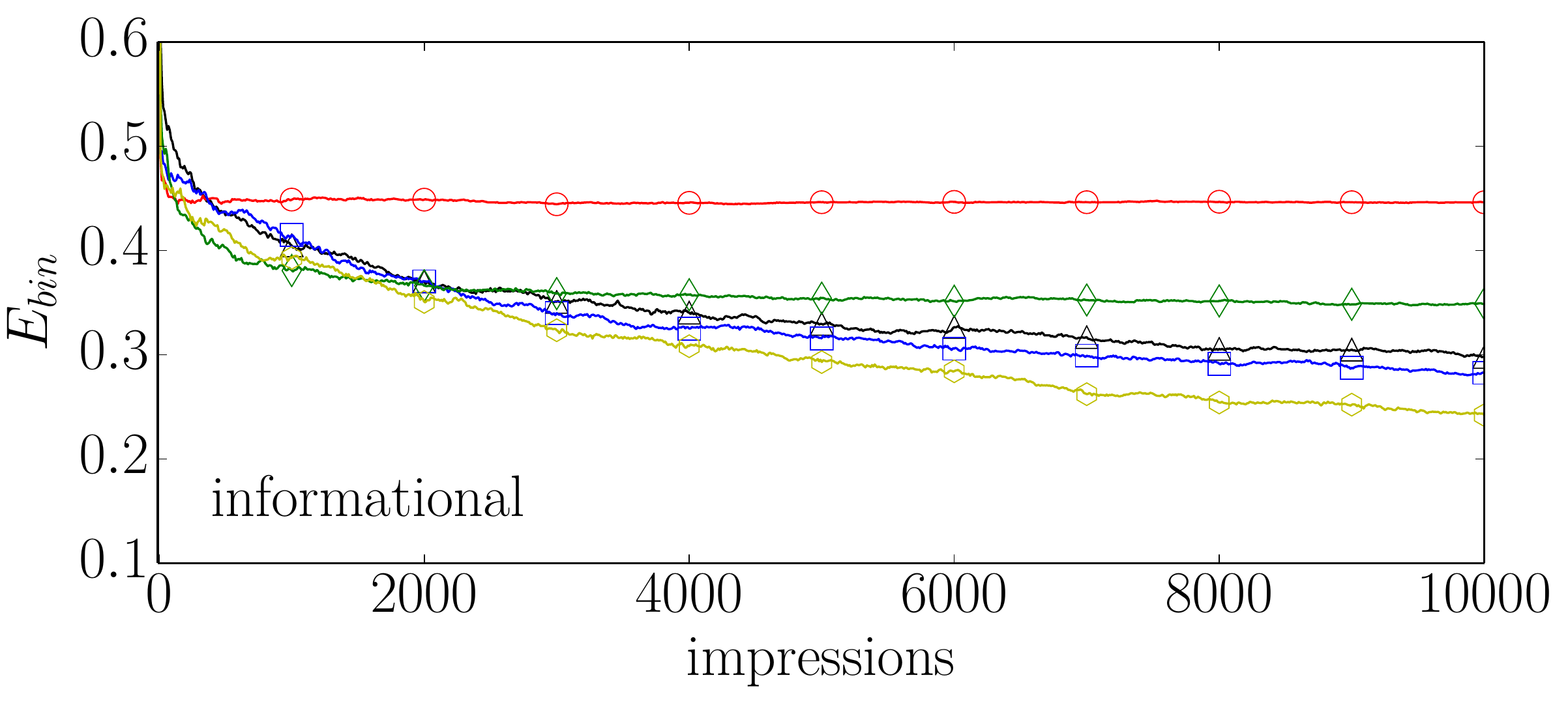}
\caption{The binary error of different multileaved comparison methods on comparisons of $|\mathcal{R}| = 15$ rankers on the \emph{MSLR-WEB10k} dataset.}
\label{fig:binaryerror}
\end{figure}

\begin{table*}[tb]
\centering
\caption{
The binary error $E_{bin}$ of all multileaved comparison methods after 10,000 impressions on comparisons of $|\mathcal{R}| = 5$ rankers. Average per dataset and click model; standard deviation in brackets. The best performance per click model and dataset is noted in bold, statistically significant improvements of \acs{PPM} are noted by \dubbelneer $(p < 0.01)$ and \enkelneer $(p < 0.05)$ and losses by \dubbelop $ $ and \enkelop $ $ respectively or $\circ$ for no difference, per baseline.
}
\input{04-multileave/tables/5rankers-binary_error}
\label{tab:5rankers}
\end{table*}

\begin{table*}[tb]
\centering
\caption{
The binary error $E_{bin}$ after 10,000 impressions on comparisons of $|\mathcal{R}| = 15$ rankers. Notation is identical to Table~\ref{tab:5rankers}.
}
\input{04-multileave/tables/15rankers-binary_error}
\label{tab:15rankers}
\end{table*}

\begin{table*}[t]
\centering
\caption{The binary error $E_{bin}$ of all multileaved comparison methods after 10,000 impressions on comparisons of $|\mathcal{R}| = 40$ rankers.
Averaged over the \emph{MSLR-WEB10k}, notation is identical to Table~\ref{tab:5rankers}.
}
\input{04-multileave/tables/40rankers-binary_error}
\label{tab:40rankers}
\end{table*}

\section{Results and Analysis}
\label{sec:multileave:results}

We answer Research Question~\ref{rq:sensitive} by evaluating the \emph{sensitivity} of \ac{PPM} based on the results of the experiments detailed in Section~\ref{sec:multileave:experiments}.

The results of the experiments with a smaller number of rankers: $|\mathcal{R}| = 5$ are displayed in Table~\ref{tab:5rankers}. Here we see that after 10,000 impressions \ac{PPM} has a significantly lower error on many datasets and at all levels of interaction noise. Furthermore, for $|\mathcal{R}| = 5$ there are no significant losses in performance under any circumstances.

When $|\mathcal{R}| = 15$ as displayed in Table~\ref{tab:15rankers}, we see a single case where \ac{PPM} performs worse than a previous method: on \emph{MQ2007} under the \emph{perfect} click model \ac{SOSM} performs significantly better than \ac{PPM}. However, on the same dataset \ac{PPM} performs significantly better under the \emph{informational} click model. Furthermore, there are more significant improvements for $|\mathcal{R}| = 15$ than when the number of rankers is the smaller $|\mathcal{R}| = 5$.

Finally, when the number of rankers in the comparison is increased to $|\mathcal{R}| = 40$ as displayed in Table~\ref{tab:40rankers}, \ac{PPM} still provides significant improvements.

We conclude that \ac{PPM}, in the experimental conditions that we considered, provides a performance that is at least as good as any existing method. Moreover, \ac{PPM} is robust to noise as we see more significant improvements under click-models with increased noise. Furthermore, since improvements are found with the number of rankers $|\mathcal{R}|$ varying from $5$ to $40$, we conclude that \ac{PPM} is scalable in the comparison size.
Additionally, the dataset type seems to affect the relative performance of the methods. For instance, on \emph{LETOR 3.0} little significant differences are found, whereas the \emph{MSLR} dataset displays the most significant improvements. This suggests that on more artificial data, i.e., the smaller datasets simulating navigational tasks, the differences are fewer, while on the other hand on large commercial data the preference for \ac{PPM} increases further.
Lastly, Figure~\ref{fig:binaryerror} displays the binary error of all multileaved comparison methods on the \emph{MSLR} dataset over 10,000 impressions. Under the \emph{perfect} click model we see that all of the previous methods display converging behavior around 3,000 impressions. In contrast, the error of \ac{PPM} continues to drop throughout the experiment. The fact that the existing methods converge at a certain level of error in the absence of click-noise is indicative that they are lacking in \emph{sensitivity}.

Overall, our results show that \ac{PPM} reaches a lower level of error than previous methods seem to be capable of. This feat can be observed on a diverse set of datasets, various levels of interaction noise and for different comparison sizes. To answer Research Question~\ref{rq:sensitive}: from our results we conclude that \ac{PPM} is more sensitive than any existing multileaved comparison method.

%% file: 04-multileave/tables/5rankers-binary_error.tex
\begin{tabularx}{\textwidth}{ X  c c c c l @{~}c@{}c@{}c@{}c }
\toprule
 & \multicolumn{1}{l}{ \small \textbf{\acs{TDM}}}  & \multicolumn{1}{l}{ \small \textbf{\acs{OM}}}  & \multicolumn{1}{l}{ \small \textbf{\acs{PM}}}  & \multicolumn{1}{l}{ \small \textbf{SOSM}}  & \multicolumn{1}{l}{ \small \textbf{\acs{PPM}}} \\
\midrule
& \multicolumn{5}{c}{\textit{perfect}} \\
\midrule
LETOR 3.0 & {\small 0.16 {\tiny( 0.13)}} & \textbf {\small 0.14 {\tiny( 0.15)}} & {\small 0.15 {\tiny( 0.15)}} & {\small 0.16 {\tiny( 0.15)}} & \textbf {\small 0.14 {\tiny( 0.13)}} &{$\circ$}& {$\circ$} & {$\circ$} & {$\circ$} 
\\
MQ2007 & {\small 0.19 {\tiny( 0.16)}} & {\small 0.22 {\tiny( 0.18)}} & \textbf {\small 0.16 {\tiny( 0.14)}} & {\small 0.18 {\tiny( 0.16)}} & \textbf {\small 0.16 {\tiny( 0.14)}} & {$\circ$} & {\tiny \dubbelneer} & {$\circ$} & {$\circ$} 
\\
MQ2008 & {\small 0.15 {\tiny( 0.12)}} & {\small 0.19 {\tiny( 0.14)}} & {\small 0.16 {\tiny( 0.12)}} & {\small 0.18 {\tiny( 0.15)}} & \textbf {\small 0.14 {\tiny( 0.12)}} & {$\circ$} & {\tiny \dubbelneer} & {$\circ$} & {\tiny \enkelneer} 
\\
MSLR-WEB10k & {\small 0.23 {\tiny( 0.13)}} & {\small 0.27 {\tiny( 0.17)}} & {\small 0.20 {\tiny( 0.14)}} & {\small 0.25 {\tiny( 0.18)}} & \textbf {\small 0.14 {\tiny( 0.13)}} & {\tiny \dubbelneer} & {\tiny \dubbelneer} & {\tiny \dubbelneer} & {\tiny \dubbelneer} 
\\
OHSUMED & {\small 0.14 {\tiny( 0.12)}} & {\small 0.19 {\tiny( 0.15)}} & \textbf {\small 0.11 {\tiny( 0.09)}} & \textbf {\small 0.11 {\tiny( 0.10)}} & \textbf {\small 0.11 {\tiny( 0.10)}} & {\tiny \dubbelneer} & {\tiny \dubbelneer} & {$\circ$} & {$\circ$} 
\\
\midrule
& \multicolumn{5}{c}{\textit{navigational}} \\
\midrule
LETOR 3.0 & {\small 0.16 {\tiny( 0.13)}} & \textbf {\small 0.15 {\tiny( 0.15)}} & \textbf {\small 0.15 {\tiny( 0.14)}} & {\small 0.17 {\tiny( 0.15)}} & {\small 0.16 {\tiny( 0.14)}} & {$\circ$} & {$\circ$} & {$\circ$} & {$\circ$} 
\\
MQ2007 & {\small 0.21 {\tiny( 0.17)}} & {\small 0.33 {\tiny( 0.21)}} & {\small 0.18 {\tiny( 0.12)}} & {\small 0.29 {\tiny( 0.23)}} & \textbf {\small 0.17 {\tiny( 0.14)}} & {$\circ$} & {\tiny \dubbelneer} & {$\circ$} & {\tiny \dubbelneer} 
\\
MQ2008 & {\small 0.17 {\tiny( 0.14)}} & {\small 0.21 {\tiny( 0.20)}} & {\small 0.17 {\tiny( 0.15)}} & {\small 0.23 {\tiny( 0.18)}} & \textbf {\small 0.15 {\tiny( 0.13)}} & {$\circ$} & {\tiny \dubbelneer} & {$\circ$} & {\tiny \dubbelneer} 
\\
MSLR-WEB10k & {\small 0.24 {\tiny( 0.14)}} & {\small 0.32 {\tiny( 0.20)}} & {\small 0.24 {\tiny( 0.17)}} & {\small 0.31 {\tiny( 0.19)}} & \textbf {\small 0.20 {\tiny( 0.15)}} & {\tiny \enkelneer} & {\tiny \dubbelneer} & {\tiny \enkelneer} & {\tiny \dubbelneer} \\
OHSUMED & \textbf {\small 0.12 {\tiny( 0.11)}} & {\small 0.27 {\tiny( 0.19)}} & {\small 0.14 {\tiny( 0.12)}} & {\small 0.23 {\tiny( 0.17)}} & {\small 0.13 {\tiny( 0.12)}} & {$\circ$} & {\tiny \dubbelneer} & {$\circ$} & {\tiny \dubbelneer} 
\\
\midrule
& \multicolumn{5}{c}{\textit{informational}} \\
\midrule
LETOR 3.0 & {\small 0.16 {\tiny( 0.14)}} & {\small 0.22 {\tiny( 0.19)}} & \textbf {\small 0.14 {\tiny( 0.11)}} & {\small 0.17 {\tiny( 0.15)}} & {\small 0.15 {\tiny( 0.13)}} & {$\circ$} & {\tiny \dubbelneer} & {$\circ$} & {$\circ$} 
\\
MQ2007 & {\small 0.23 {\tiny( 0.15)}} & {\small 0.41 {\tiny( 0.26)}} & {\small 0.23 {\tiny( 0.15)}} & {\small 0.37 {\tiny( 0.23)}} & \textbf {\small 0.17 {\tiny( 0.16)}} & {\tiny \dubbelneer} & {\tiny \dubbelneer} & {\tiny \dubbelneer} & {\tiny \dubbelneer} 
\\
MQ2008 & {\small 0.18 {\tiny( 0.13)}} & {\small 0.28 {\tiny( 0.19)}} & {\small 0.18 {\tiny( 0.16)}} & {\small 0.23 {\tiny( 0.18)}} & \textbf {\small 0.17 {\tiny( 0.14)}} & {$\circ$} & {\tiny \dubbelneer} & {$\circ$} & {\tiny \dubbelneer} 
\\
MSLR-WEB10k & {\small 0.27 {\tiny( 0.18)}} & {\small 0.42 {\tiny( 0.23)}} & {\small 0.24 {\tiny( 0.17)}} & {\small 0.36 {\tiny( 0.20)}} & \textbf {\small 0.19 {\tiny( 0.17)}} & {\tiny \dubbelneer} & {\tiny \dubbelneer} & {\tiny \enkelneer} & {\tiny \dubbelneer} 
\\
OHSUMED & {\small 0.13 {\tiny( 0.10)}} & {\small 0.37 {\tiny( 0.24)}} & \textbf {\small 0.12 {\tiny( 0.11)}} & {\small 0.27 {\tiny( 0.21)}} & \textbf {\small 0.12 {\tiny( 0.10)}} & {$\circ$} & {\tiny \dubbelneer} & {$\circ$} & {\tiny \dubbelneer} 
\\
\bottomrule
\end{tabularx}

%% file: 04-multileave/tables/15rankers-binary_error.tex
\begin{tabularx}{\textwidth}{ X  c c c c l @{~}c@{}c@{}c@{}c }
\toprule
 & \multicolumn{1}{l}{ \small \textbf{\acs{TDM}}}  & \multicolumn{1}{l}{ \small \textbf{\acs{OM}}}  & \multicolumn{1}{l}{ \small \textbf{\acs{PM}}}  & \multicolumn{1}{l}{ \small \textbf{SOSM}}  & \multicolumn{1}{l}{ \small \textbf{\acs{PPM}}} \\
\midrule
& \multicolumn{5}{c}{\textit{perfect}} \\
\midrule
LETOR 3.0 & {\small 0.16 {\tiny( 0.07)}} & \textbf {\small 0.14 {\tiny( 0.08)}} & {\small 0.15 {\tiny( 0.07)}} & {\small 0.17 {\tiny( 0.08)}} & {\small 0.16 {\tiny( 0.08)}} & {$\circ$} & {$\circ$} & {$\circ$} & {$\circ$} 
\\
MQ2007 & {\small 0.20 {\tiny( 0.07)}} & {\small 0.25 {\tiny( 0.09)}} & {\small 0.18 {\tiny( 0.06)}} & \textbf {\small 0.15 {\tiny( 0.07)}} & {\small 0.19 {\tiny( 0.07)}} & {$\circ$} & {\tiny \dubbelneer} & {$\circ$} & {\tiny \dubbelop} 
\\
MQ2008 & {\small 0.16 {\tiny( 0.05)}} & {\small 0.17 {\tiny( 0.05)}} & {\small 0.16 {\tiny( 0.05)}} & \textbf {\small 0.15 {\tiny( 0.07)}} & \textbf {\small 0.15 {\tiny( 0.06)}} & {$\circ$} & {\tiny \enkelneer} & {$\circ$} & {$\circ$} 
\\
MSLR-WEB10k & {\small 0.24 {\tiny( 0.07)}} & {\small 0.38 {\tiny( 0.11)}} & {\small 0.21 {\tiny( 0.06)}} & {\small 0.30 {\tiny( 0.08)}} & \textbf {\small 0.14 {\tiny( 0.05)}} & {\tiny \dubbelneer} & {\tiny \dubbelneer} & {\tiny \dubbelneer} & {\tiny \dubbelneer} 
\\
OHSUMED & {\small 0.14 {\tiny( 0.03)}} & {\small 0.18 {\tiny( 0.05)}} & {\small 0.13 {\tiny( 0.03)}} & {\small 0.13 {\tiny( 0.03)}} & \textbf {\small 0.11 {\tiny( 0.03)}} & {\tiny \dubbelneer} & {\tiny \dubbelneer} & {\tiny \dubbelneer} & {\tiny \dubbelneer} 
\\
\midrule
& \multicolumn{5}{c}{\textit{navigational}} \\
\midrule
LETOR 3.0 & {\small 0.16 {\tiny( 0.08)}} & {\small 0.16 {\tiny( 0.09)}} & \textbf {\small 0.15 {\tiny( 0.08)}} & {\small 0.17 {\tiny( 0.08)}} & {\small 0.17 {\tiny( 0.08)}} & {$\circ$} & {$\circ$} & {$\circ$} & {$\circ$} 
\\
MQ2007 & {\small 0.24 {\tiny( 0.07)}} & {\small 0.33 {\tiny( 0.11)}} & \textbf {\small 0.20 {\tiny( 0.07)}} & {\small 0.22 {\tiny( 0.08)}} & {\small 0.21 {\tiny( 0.08)}} & {\tiny \dubbelneer} & {\tiny \dubbelneer} & {$\circ$} & {$\circ$} 
\\
MQ2008 & {\small 0.19 {\tiny( 0.05)}} & {\small 0.21 {\tiny( 0.07)}} & \textbf {\small 0.16 {\tiny( 0.05)}} & {\small 0.18 {\tiny( 0.06)}} & \textbf {\small 0.16 {\tiny( 0.06)}} & {\tiny \dubbelneer} & {\tiny \dubbelneer} & {$\circ$} & {\tiny \dubbelneer} 
\\
MSLR-WEB10k & {\small 0.27 {\tiny( 0.07)}} & {\small 0.42 {\tiny( 0.12)}} & {\small 0.24 {\tiny( 0.06)}} & {\small 0.28 {\tiny( 0.09)}} & \textbf {\small 0.22 {\tiny( 0.08)}} & {\tiny \dubbelneer} & {\tiny \dubbelneer} & {$\circ$} & {\tiny \dubbelneer} 
\\
OHSUMED & {\small 0.14 {\tiny( 0.04)}} & {\small 0.25 {\tiny( 0.07)}} & \textbf {\small 0.13 {\tiny( 0.03)}} & {\small 0.18 {\tiny( 0.06)}} & \textbf {\small 0.13 {\tiny( 0.04)}} & {$\circ$} & {\tiny \dubbelneer} & {$\circ$} & {\tiny \dubbelneer} 
\\
\midrule
& \multicolumn{5}{c}{\textit{informational}} \\
\midrule
LETOR 3.0 & {\small 0.18 {\tiny( 0.07)}} & {\small 0.20 {\tiny( 0.11)}} & {\small 0.17 {\tiny( 0.08)}} & \textbf {\small 0.16 {\tiny( 0.08)}} & {\small 0.18 {\tiny( 0.08)}} & {$\circ$} & {\tiny \enkelneer} & {$\circ$} & {$\circ$} 
\\
MQ2007 & {\small 0.28 {\tiny( 0.07)}} & {\small 0.42 {\tiny( 0.14)}} & {\small 0.26 {\tiny( 0.08)}} & {\small 0.28 {\tiny( 0.11)}} & \textbf {\small 0.21 {\tiny( 0.08)}} & {\tiny \dubbelneer} & {\tiny \dubbelneer} & {\tiny \dubbelneer} & {\tiny \dubbelneer} 
\\
MQ2008 & {\small 0.23 {\tiny( 0.06)}} & {\small 0.26 {\tiny( 0.11)}} & {\small 0.18 {\tiny( 0.06)}} & {\small 0.20 {\tiny( 0.06)}} & \textbf {\small 0.15 {\tiny( 0.06)}} & {\tiny \dubbelneer} & {\tiny \dubbelneer} & {\tiny \dubbelneer} & {\tiny \dubbelneer} 
\\
MSLR-WEB10k & {\small 0.30 {\tiny( 0.09)}} & {\small 0.45 {\tiny( 0.12)}} & {\small 0.28 {\tiny( 0.08)}} & {\small 0.35 {\tiny( 0.11)}} & \textbf {\small 0.24 {\tiny( 0.08)}} & {\tiny \dubbelneer} & {\tiny \dubbelneer} & {\tiny \dubbelneer} & {\tiny \dubbelneer} 
\\
OHSUMED & {\small 0.15 {\tiny( 0.03)}} & {\small 0.42 {\tiny( 0.09)}} & \textbf {\small 0.13 {\tiny( 0.03)}} & {\small 0.25 {\tiny( 0.06)}} & \textbf {\small 0.13 {\tiny( 0.04)}} & {\tiny \dubbelneer} & {\tiny \dubbelneer} & {$\circ$} & {\tiny \dubbelneer} 
\\
\bottomrule
\end{tabularx}

%% file: 04-multileave/tables/40rankers-binary_error.tex
\begin{tabularx}{\textwidth}{ X  c c c c l @{~}c@{}c@{}c@{}c }
\toprule
 & \multicolumn{1}{l}{ \small \textbf{\acs{TDM}}}  & \multicolumn{1}{l}{ \small \textbf{\acs{OM}}}  & \multicolumn{1}{l}{ \small \textbf{\acs{PM}}}  & \multicolumn{1}{l}{ \small \textbf{SOSM}}  & \multicolumn{1}{l}{ \small \textbf{\acs{PPM}}} 
 \\
 \midrule
\textit{perfect} & {\small 0.26 {\tiny( 0.03)}} & {\small 0.43 {\tiny( 0.02)}} & {\small 0.23 {\tiny( 0.02)}} & {\small 0.31 {\tiny( 0.02)}} & \textbf {\small 0.18 {\tiny( 0.04)}} & {\tiny \dubbelneer} & {\tiny \dubbelneer} & {\tiny \dubbelneer} & {\tiny \dubbelneer} 
\\
\textit{navigational} & {\small 0.31 {\tiny( 0.03)}} & {\small 0.44 {\tiny( 0.01)}} & {\small 0.25 {\tiny( 0.03)}} & \textbf {\small 0.23 {\tiny( 0.03)}} & {\small 0.24 {\tiny( 0.05)}} & {\tiny \dubbelneer} & {\tiny \dubbelneer} & {\tiny \dubbelneer} & {$\circ$} 
\\
\textit{informational} & {\small 0.37 {\tiny( 0.04)}} & {\small 0.47 {\tiny( 0.01)}} & {\small 0.30 {\tiny( 0.05)}} & {\small 0.34 {\tiny( 0.05)}} & \textbf {\small 0.27 {\tiny( 0.06)}} & {\tiny \dubbelneer} & {\tiny \dubbelneer} & {\tiny \dubbelneer} & {\tiny \dubbelneer} 
\\
\bottomrule
\end{tabularx}

%% file: 04-multileave/08-conclusion.tex
\section{Conclusion}
\label{sec:multileave:conclusion}

In this chapter we have examined multileaved comparison methods for evaluating ranking models online.

We have presented a new multileaved comparison method, \acf{PPM}, that is more sensitive to user preferences than existing methods. Additionally, we have proposed a theoretical framework for assessing multileaved comparison methods, with \emph{considerateness} and \emph{fidelity} as the two key requirements. We have shown that no method published prior to \ac{PPM} has \emph{fidelity} without lacking \emph{considerateness}. In other words, prior to \ac{PPM} no multileaved comparison method has been able to infer correct preferences without degrading the search experience of the user. In contrast, we prove that \ac{PPM} has both \emph{considerateness} and \emph{fidelity}, thus it is guaranteed to correctly identify a Pareto dominating ranker without altering the search experience considerably. Furthermore, our experimental results spanning ten datasets show that \ac{PPM} is more sensitive than existing methods, meaning that it can reach a lower level of error than any previous method. Moreover, our experiments show that the most significant improvements are obtained on the more complex datasets, i.e., larger datasets with more grades of relevance. Additionally, similar improvements are observed under different levels of noise and numbers of rankers in the comparison, indicating that \ac{PPM} is robust to interaction noise and scalable to large comparisons. As an extra benefit, the computational complexity of \ac{PPM} is polynomial and, unlike previous methods, does not depend on sampling or approximations.

With these findings we can answer the thesis research question \ref{thesisrq:multileaving} positively: with the introduction of our novel \acf{PPM} method the effectiveness of online evaluation scales to large comparisons.

The theoretical framework that we have introduced allows future research into multileaved comparison methods to guarantee improvements that generalize better than empirical results alone. In turn, properties like \emph{considerateness} can further stimulate the adoption of multileaved comparison methods in production environments; future work with real-world users may yield further insights into the effectiveness of the multileaving paradigm.
Rich interaction data enables the introduction of multileaved comparison methods that consider more than just clicks, as has been done for interleaving methods \cite{kharitonov2015generalized}. These methods could be extended to consider other signals such as \emph{dwell-time} or \emph{the order of clicks in an impression}, etc.

Furthermore, the field of \ac{OLTR} has depended on online evaluation from its inception \cite{yue2009interactively}. The introduction of multileaving and subsequent novel multileaved comparison methods brought substantial improvements to both fields \cite{schuth2016mgd, oosterhuis2016probabilistic}. Similarly, \ac{PPM} and any future extensions are likely to  benefit the \ac{OLTR} field too.

Finally, while the theoretical and empirical improvements of \ac{PPM} are convincing, future work should investigate whether the sensitivity can be made even stronger. For instance, it is possible to have clicks from which no preferences between rankers can be inferred. Can we devise a method that avoids such situations as much as possible without introducing any form of bias, thus increasing the sensitivity even further while maintaining theoretical guarantees?

In Chapter~\ref{chapter:06-onlinecountereval} we will take another look at online ranker evaluation and contrast it with counterfactual evaluation.
We will see that existing interleaving methods (and by extension some multileaving methods) are biased w.r.t.\ the definition of position bias common in counterfactual evaluation.
The novel method introduced in Chapter~\ref{chapter:06-onlinecountereval} combines aspects of counterfactual and online ranker evaluation, creating a method with strong theoretical guarantees while also being very effective.

Furthermore, similar to this chapter, Chapter~\ref{chapter:02-pdgd} will look at whether a pairwise \ac{LTR} method is suitable for online \ac{LTR}.
While different from \ac{PPM}, the method introduced in Chapter~\ref{chapter:02-pdgd} also infers pairwise preferences between documents, and weights inferred preferences to account for position bias.

%% file: 04-multileave/notation.tex
\section{Notation Reference for Chapter~\ref{chapter:01-online-evaluation}}
\label{notation:01-online-evaluation}

\begin{center}
\begin{tabular}{l l}
 \toprule
\bf Notation  & \bf Description \\
\midrule
$q$ & a user-issued query \\
$T$ & the total number of interactions \\
$\mathbf{r}_i$ & an individual ranker a.k.a. a single ranking system or ranking model \\
$\mathcal{R}$ & a set of rankers to compare \\
$\mathbf{l}_i$ & a ranking generated by ranker $\mathbf{r}_i$ \\
$\mathbf{m}$ & a multileaved result list \\
$k$ & the length of the multileaved result lists \\
$\mathbf{c}$ & a vector indicating clicks on a displayed multileaved result list \\
$P$ & a preference matrix to store inferred preferences between rankers \\
$r(d, \mathbf{l}_i)$ & the rank at which ranker $\mathbf{r}_i$ places document $d$ \\
\bottomrule
\end{tabular}
\end{center}

%% file: 05-pdgd/main.tex
\chapter{Differentiable Online Learning to Rank}
\label{chapter:02-pdgd}

\newcommand{\OurMethod}{PDGD}
\newcommand{\OtherDocuments}{\{\ldots\}}

\footnote[]{This chapter was published as~\citep{oosterhuis2018differentiable}.
Appendix~\ref{notation:02-pdgd} gives a reference for the notation used in this chapter.
}

\acf{OLTR} methods optimize rankers based on direct interaction with users.
State-of-the-art \acs{OLTR} methods rely on online evaluation and sampling model variants, they were designed specifically for linear models. 
Their approaches do not extend well to non-linear models such as neural networks.

To address this limitation, this chapter will consider the thesis research question:
\begin{itemize}
\item[\ref{thesisrq:pdgd}] \emph{Is online \ac{LTR} possible without relying on model-sampling and online evaluation?}
\end{itemize}
\noindent
We introduce an entirely novel approach to \acs{OLTR} that constructs a weighted differentiable pairwise loss after each interaction: \acf{\OurMethod}.
\ac{\OurMethod} breaks away from the traditional approach that relies on interleaving or multileaving and extensive sampling of models to estimate gradients.
Instead, its gradient is based on inferring preferences between document pairs from user clicks and can optimize any differentiable model.
We prove that the gradient of \acs{\OurMethod} is unbiased w.r.t.\ user document pair preferences.
Our experiments on the largest publicly available \ac{LTR} datasets show considerable and significant improvements under all levels of interaction noise.
\acs{\OurMethod} outperforms existing \ac{OLTR} methods both in terms of learning speed as well as final convergence.
Furthermore, unlike previous \ac{OLTR} methods, \ac{\OurMethod} also allows for non-linear models to be optimized effectively.
Our results show that using a neural network leads to even better performance at convergence than a linear model.
In summary, \acs{\OurMethod} is an efficient and unbiased \ac{OLTR} approach that provides a better user experience than previously possible.

\input{05-pdgd/01-intro}

\input{05-pdgd/02-related}

\input{05-pdgd/03-method}

\input{05-pdgd/04-experiments}

\input{05-pdgd/05-results}

\input{05-pdgd/06-conclusion}
\begin{subappendices}
\input{05-pdgd/notation}
\end{subappendices}

%% file: 05-pdgd/01-intro.tex
\section{Introduction}
\label{sec:pdgd:intro}

In order to benefit from unprecedented volumes of content, users rely on ranking systems to provide them with the content of their liking.
\ac{LTR} in \ac{IR} concerns methods that optimize ranking models so that they order documents according to user preferences.
In web search engines such models combine hundreds of signals to rank web-pages according to their relevance to user queries \cite{liu2009learning}.
Similarly, ranking models are a vital part of recommender systems where there is no explicit search intent \cite{karatzoglou2013learning}.
\ac{LTR} is also prevalent in settings where other content is ranked, e.g., videos~\citep{chelaru2014useful}, products~\citep{karmaker2017application}, conversations~\citep{DBLP:conf/chiir/RadlinskiC17} or personal documents~\citep{wang2016learning}.

Traditionally, \ac{LTR} has been applied in the \emph{offline} setting where a dataset with annotated query-document pairs is available.
Here, the model is optimized to rank documents according to the relevance annotations, which are based on the judgements of human annotators.
Over time the limitations of this supervised approach have become apparent: annotated sets are expensive and time-consuming to create~\citep{letor, Chapelle2011}; when personal documents are involved such a dataset would breach privacy~\citep{wang2016learning}; the relevance of documents to queries can change over time, like in a news search engine~\citep{dumais-web-2010,lefortier-online-2014}; and judgements of raters are not necessarily aligned with the actual users~\citep{sanderson2010}.

In order to overcome the issues with annotated datasets, previous work in \ac{LTR} has looked into learning from user interactions.
Work along these lines can be divided into \emph{approaches that learn from historical interactions}, i.e., in the form of interaction logs~\citep{Joachims2002}, and \emph{approaches that learn in an online setting}~\citep{yue2009interactively}.
The latter regard methods that determine what to display to the user at each impression, and then immediately learn from observed user interactions and update their behavior accordingly.
This online approach has the advantage that it does not require an existing ranker of decent quality, and thus can handle cold-start situations.
Additionally, it is more responsive to the user by updating continuously and instantly, therefore allowing for a better experience.
However, it is important that an online method can handle biases that come with user behavior:
for instance, the observed interactions only take place with the displayed results, i.e., there is item-selection bias, and are more likely to occur with higher ranked items, i.e., there is position bias.
Accordingly, a method should learn user preferences w.r.t.\ document relevance, and be robust to the forms of noise and bias present in the online setting.
Overall, the online \ac{LTR} approach promises to learn ranking models that are in line with user preferences, in a responsive matter, reaching good performance from few interactions, even in cold-start situations.

Despite these highly beneficial properties, previous work in \ac{OLTR} has only considered linear models \cite{schuth2016mgd, hofmann12:balancing, yue2009interactively} or trivial variants thereof~\cite{oosterhuis2017balancing}.
The reason for this is that existing work in \ac{OLTR} has worked with the \ac{DBGD} algorithm \cite{yue2009interactively} as a basis.
While very influential and effective, we identify two main problems with the gradient estimation of the \ac{DBGD} algorithm:
\begin{enumerate}[align=left,leftmargin=*]
\item Gradient estimation is based on sampling model variants from a unit circle around the current model. 
This concept does not extend well to non-linear models. Computing rankings for variants is also computationally costly for larger complex models.
\item It uses online evaluation methods, i.e., interleaving or multileaving, to determine the gradient direction from the resulting set of models.
However, these evaluation methods are designed for finding preferences between ranking systems, not (primarily) for determining how a model should be updated.
\end{enumerate}
As an alternative we introduce \acfi{\OurMethod}\acused{\OurMethod}, the first unbiased \ac{OLTR} method that is applicable to any differentiable ranking model.
\ac{\OurMethod} infers pairwise document preferences from user interactions and constructs an unbiased gradient after each user impression.
In addition, \ac{\OurMethod} does not rely on sampling models for exploration, but instead models rankings as probability distributions over documents.
Therefore, it allows the \ac{OLTR} model to be very certain for specific queries and perform less exploration in those cases, while being much more explorative in other, uncertain cases.
Our results show that, consequently, \ac{\OurMethod} provides significant and considerable improvements over previous \ac{OLTR} methods.
This indicates that its gradient estimation is more in line with the preferences to be learned.

In this chapter, we address the thesis research question \ref{thesisrq:pdgd} by answering the following three specific research questions:
 \begin{enumerate}[align=left, label={\bf RQ3.\arabic*},leftmargin=*]
    \item Does using \ac{\OurMethod} result in significantly better performance than the current state-of-the-art \acl{MGD}?\label{rq:performance}
    \item Is the gradient estimation of \ac{\OurMethod} unbiased? \label{rq:unbiased}
    \item Is \ac{\OurMethod} capable of effectively optimizing different types of ranking models? \label{rq:nonlinear}
\end{enumerate}
To facilitate replicability and repeatability of our findings, we provide open source implementations of \ac{\OurMethod} and our experiments under the permissive MIT open-source license.\footnote{https://github.com/HarrieO/OnlineLearningToRank}

%% file: 05-pdgd/02-related.tex
\section{Related Work}
\label{sec:relatedwork}

\subsection{Learning to rank}
\ac{LTR} can be applied to the offline and online setting.
In the offline setting \ac{LTR} is approached as a supervised problem where the relevance of each query-document pair is known.
Most of the challenges with offline \ac{LTR} come from obtaining annotations.
For instance, gathering annotations is time-consuming and expensive \cite{letor, qin2013introducing, Chapelle2011}.
Furthermore, in privacy sensitive-contexts it would be unethical to annotate items, e.g., for personal emails or documents \cite{wang2016learning}.
Moreover, for personalization problems annotators are unable to judge what specific users would prefer.
Also, (perceived) relevance chances over time, due to cognitive changes on the user's end~\citep{vakkari-changes-2000} or due to changes in document collections~\citep{dumais-web-2010} or the real world~\citep{lefortier-online-2014}.
Finally, annotations are not necessarily aligned with user satisfaction, as judges may interpret queries differently from actual users~\cite{sanderson2010}.
Consequently, the limitations of offline \ac{LTR} have led to an increased interest in alternative approaches to \ac{LTR}.

\subsection{Online learning to rank}
\ac{OLTR} is an attractive alternative to offline \ac{LTR} as it learns directly from interacting with users~\cite{yue2009interactively}.
By doing so it attempts to solve the issues with offline annotations that occur in \ac{LTR}, as user preferences are expected to be better represented by interactions than by offline annotations~\cite{radlinski2008does}.
Unlike methods in the offline setting, \ac{OLTR} algorithms have to simultaneously perform ranking while also optimizing their ranking model.
In other words, an \ac{OLTR} algorithm decides what rankings to display to users, while at the same time learning from the interactions with the presented rankings.
While the potential of learning in the online setting is great, it has its own challenges.
In particular, the main difficulties of the \ac{OLTR} task are \emph{bias} and \emph{noise}.
Any user interaction that does not reflect their true preference is considered noise, this happens frequently e.g., clicks often occur for unexpected reasons~\citep{sanderson2010}.
Bias comes in many forms, for instance, item-selection bias occurs because interactions only involve displayed documents~\citep{wang2016learning}.
Another common bias is position bias, a consequence from the fact documents at the top of a ranking are more likely to be considered~\citep{yue2010beyond}.
An \ac{OLTR} method should thus take into account the biases that affect user behavior while also being robust to noise, in order to learn the \emph{true} user preferences.

\ac{OLTR} methods can be divided into two groups~\citep{zoghi-online-2017}: \emph{tabular} methods that learn the best ranked list under some model of user interaction with the list~\cite{Radlinski2008,Slivkins2013}, such as a click model~\citep{chuklin-click-2015}, and
\emph{feature-based} algorithms that learn the best ranker in a family of rankers~\citep{yue2009interactively,hofmann2013reusing}.
Model-based methods may have greater statistical efficiency but they give up generality, essentially requiring us to learn a separate model for every query.
For the remainder of this chapter, we focus on model-free \ac{OLTR} methods.

\subsection{DBGD and beyond}
State-of-the-art (model-free) \ac{OLTR} approaches learn user preferences by approaching optimization as a dueling bandit problem~\cite{yue2009interactively}.
They estimate the gradient of the model w.r.t.\ user satisfaction by comparing the current model to sampled variations of the model.
The original \ac{DBGD} algorithm~\cite{yue2009interactively} uses interleaving methods to make these comparisons: at each interaction the rankings of two rankers are combined to create a single result list.
From a large number of clicks on such a combined result list a user preference between the two rankers can reliably be inferred~\cite{hofmann2011probabilistic}.
Conversely, \ac{DBGD} compares its current ranking model to a different slight variation at each impression.
Then, if a click is indicative of a preference for the variation, the current model is slightly updated towards it.
Accordingly, the model of \ac{DBGD} will continuously update itself and oscillate towards an inferred optimum.

Other work in \ac{OLTR} has used \ac{DBGD} as a basis and extended upon it.
Notably, \citet{hofmann2013reusing} have proposed a method that guides exploration by only sampling variations that seem promising from historical interaction data.
Unfortunately, while this approach provides faster initial learning, the historical data introduces bias which leads to the quality of the ranking model to steadily decrease over time~\cite{oosterhuis2016probabilistic}.
Alternatively, \citet{schuth2016mgd} introduced \acf{MGD}, this extension replaced the interleaving of \ac{DBGD} with multileaving methods.
In turn the multileaving paradigm is an extension of interleaving where a set of rankers are compared efficiently~\citep{oosterhuis2017sensitive, schuth2015probabilistic, Schuth2014a}.
Conversely, multileaving methods can combine the rankings of more than two rankers and thus infer preferences over a set of rankers from a single click.
\ac{MGD} uses this property to estimate the gradient more effectively by comparing a large number of model variations per user impression~\citep{schuth2016mgd, oosterhuis2016probabilistic}.
As a result, \ac{MGD} requires fewer user interactions to converge on the same level of performance as \ac{DBGD}.
Another alternative approach was considered by \citet{hofmann11:balancing}, who inject the ranking from the current model with randomly sampled documents.
Then, after each user impression, a pairwise loss is constructed from inferred preferences between documents.
This pairwise approach was not found to be more effective than \ac{DBGD}.

Quite remarkably, all existing work in \ac{OLTR} has only considered linear models.
Recently, \citet{oosterhuis2017balancing} recognized that a tradeoff unique to \ac{OLTR} arises when choosing models.
High capacity models such as neural networks~\cite{burges2010ranknet} require more data than simpler models.
On the one hand, this means that high capacity models need more user interactions to reach the same level of performance, thus giving a worse initial user experience.
On the other hand, high capacity models are capable of finding better optima, thus lead to better final convergence and a better long-term user experience.
This dilemma is named the \emph{speed}-\emph{quality} tradeoff, and as a solution a cascade of models can be optimized: combining the initial learning speed of a simple model with the convergence of a complex one.
But there are more reasons why non-linear models have so far been absent from \ac{OLTR}.
Importantly, the \ac{DBGD} algorithm was designed for linear models from the ground up; relying on a unit circle to sample model variants and averaging models to estimate the gradient.
Furthermore, the computational cost of maintaining an extensive set of model variants for large and complex models makes this approach very impractical.

Our contribution over the work listed above is an \ac{OLTR} method that is not an extension of \ac{DBGD}, instead it computes a differentiable pairwise loss to update its model.
Unlike the existing pairwise approach, our loss function is unbiased and our exploration is performed using the model's confidence over documents.
Finally, we also show that this is the first \ac{OLTR} method to effectively optimize neural networks in the online setting.

%% file: 05-pdgd/03-method.tex
\section{Method}
\label{sec:simmgd}

In this section we introduce a novel \ac{OLTR} algorithm: \acf{\OurMethod}.
First, Section~\ref{sec:method:description} describes \ac{\OurMethod} in detail, before Section~\ref{sec:method:unbiased} formalizes and proves the unbiasedness of the method.
Appendix~\ref{notation:02-pdgd} lists the notation we use.

\subsection{\acl{\OurMethod}}
\label{sec:method:description}

\ac{\OurMethod} revolves around optimizing a ranking model $f_{\theta}(\mathbf{d})$ that takes a feature representation of a query-document pair $\mathbf{d}$ as input and outputs a score. 
The aim of the algorithm is to find the parameters $\theta$ so that sorting the documents by their scores in descending order provides the most optimal rankings.
Because this is an online algorithm, the method must first decide what ranking to display to the user, then after the user has interacted with the displayed ranking, it may update $\theta$ accordingly.

Unlike previous \ac{OLTR} approaches, \ac{\OurMethod} does not rely on any online evaluation methods. 
Instead, a \ac{PL} model is applied to the ranking function $f_{\theta}(\cdot)$ resulting in a distribution over the document set $D$:
\begin{equation}
P(d|D) = \frac{e^{f_{\theta}(\mathbf{d})}}{\sum_{d' \in D} e^{f_{\theta}(\mathbf{d'})}}.
\label{eq:pdgd:docprob}
\end{equation}
A ranking $R$ to display to the user is then created by sampling from the distribution $k$ times, where after each placement the distribution is renormalized to prevent duplicate placements.
\ac{PL} models have been used before in \ac{LTR}. 
For instance, the ListNet method~\cite{Cao2007} optimizes such a model in the offline setting.
With $R_i$ denoting the document at position $i$, the probability of the ranking $R$ then becomes:
\begin{equation}
P(R|D) = \prod^k_{i=1} P(R_i | D \setminus \{ R_1, \ldots, R_{i-1} \}).
\end{equation}
After the ranking $R$ has been displayed to the user, they have the option to interact with it.
The user may choose to click on some (or none) of the documents.
Based on these clicks, \ac{\OurMethod} will infer preferences between the displayed documents.
We assume that clicked documents are preferred over observed unclicked documents.
However, to the algorithm it is unknown which unclicked documents the user has considered.
As a solution, \ac{\OurMethod} relies on the assumption that every document preceding a clicked document and the first subsequent unclicked document was observed, as illustrated in Figure~\ref{fig:clickselection}.
This preference assumption has been proven useful in \ac{IR} before, for instance in pairwise \ac{LTR} on click logs \cite{Joachims2002} and recently in online evaluation \cite{oosterhuis2017sensitive}.
We will denote preferences between documents inferred from clicks as: $d_k >_\mathbf{c} d_l$ where $d_k$ is preferred over $d_l$.

\begin{figure}[tb]
\centering

\subfloat[]{
\includegraphics[height=9em]{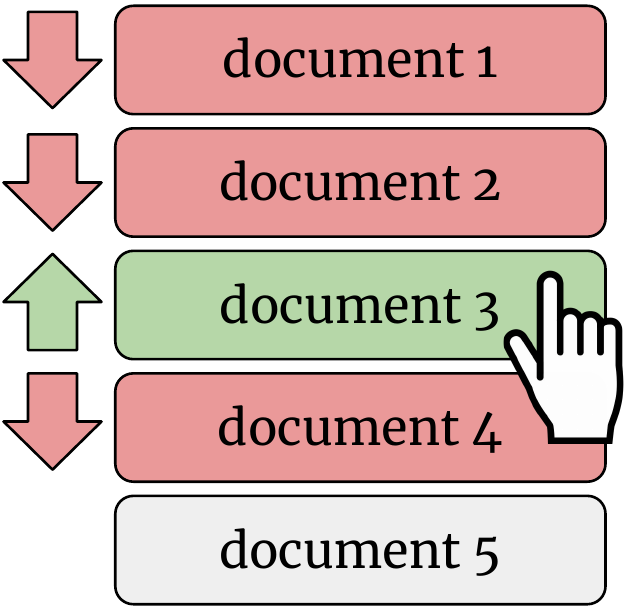}
\label{fig:clickselection}
\hspace{1em} %
}
\subfloat[]{
\includegraphics[height=9em]{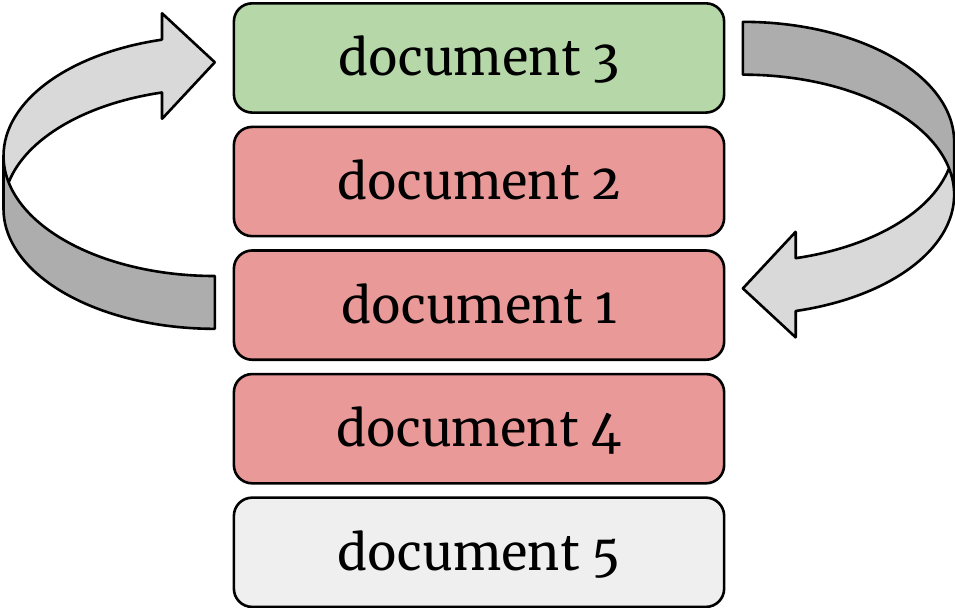}
\label{fig:reversepair}
}
\caption{
Left: a click on a document ranking $R$ and the inferred preferences of $d_3$ over $\{d_1, d_2, d_4\}$.
Right: the reversed pair ranking $R^*(d_1, d_3, R)$ for the document pair $d_1$ and $d_3$.
}
\label{fig:selectionandreverse}
\end{figure}

Then $\theta$ is updated by optimizing pairwise probabilities over the preference pairs;
for each inferred document preference $d_k >_\mathbf{c} d_l$, the probability that the preferred document $d_k$ is sampled before $d_l$ is sampled is increased \citep{szorenyi2015online}:
\begin{equation}
P(d_k \succ d_l)
=
\frac{P(d_k | D)}{P(d_k | D) + P(d_l | D)}
=
\frac{ e^{f(\mathbf{d}_k)} }{
e^{f(\mathbf{d}_k)}
+
e^{f(\mathbf{d}_l)}
}.
\label{eq:pairprob}
\end{equation}
We have chosen for pairwise optimization over listwise optimization because a pairwise method can be made unbiased by reweighing preference pairs.
To do this we introduce the weighting function $\rho(d_k, d_l, R, D)$ and estimate the gradient of the user preferences by the weighted sum:
\begin{equation}
\begin{split}
 \nabla f_\theta (\cdot)
&\approx\!\!
\sum_{ d_k >_\mathbf{c} d_l } \!\!\rho(d_k, d_l, R, D) \left[\nabla P(d_k \succ d_l)\right] \\ 
& =\!\!
\sum_{ d_k >_\mathbf{c} d_l } \!\!\rho(d_k, d_l, R, D)
\frac{ 
e^{f_{\theta}(\mathbf{d}_k)}e^{f_{\theta}(\mathbf{d}_l)}
}{
(e^{f_{\theta}(\mathbf{d}_k)} + e^{f_{\theta}(\mathbf{d}_l)}) ^ 2
}
\!\left(f'_{\theta}(\mathbf{d}_k) - f'_{\theta}(\mathbf{d}_l)\right)
\!.\!\!
\end{split}
\label{eq:novelgradient}
\end{equation}
The $\rho$ function is based on the reversed pair ranking $R^*(d_k, d_l, R)$, which is the same ranking as $R$ with the position of $d_k$ and $d_l$ swapped.
An example of a reversed pair ranking is illustrated in Figure~\ref{fig:reversepair}.
The idea is that if a preference for $d_k >_\mathbf{c} d_l$ is inferred in $R$ and both documents are equally relevant, then the reverse preference $d_l >_\mathbf{c} d_k$ is equally likely to be inferred in $R^*(d_k, d_l, R)$.
The $\rho$ function reweighs the found preferences to the ratio between the probabilities of $R$ or $R^*(d_k, d_l, R)$ occurring:
\begin{equation}
\rho(d_k, d_l, R, D) = \frac{P(R^*(d_k, d_l, R)|D)}{P(R|D) + P(R^*(d_k, d_l, R)|D)}.
\label{eq:pdgd:rho}
\end{equation}
This procedure has similarities with importance sampling~\cite{owen2013monte}; %
however, we found that reweighing according to the ratio between $R$ and $R^*$ provides a more stable performance, since it produces less extreme values. 
Section~\ref{sec:method:unbiased} details exactly how $\rho$ creates an unbiased gradient.

\begin{algorithm}[t]
\caption{\acf{\OurMethod}.} 
\label{alg:novel}
\begin{algorithmic}[1]
\STATE \textbf{Input}: initial weights: $\mathbf{\theta}_1$; scoring function: $f$; learning rate $\eta$.  \label{line:novel:initmodel}
\FOR{$t \leftarrow  1 \ldots \infty$ }
	\STATE $q_t \leftarrow \text{receive\_query}(t)$\hfill \textit{\small // obtain a query from a user} \label{line:novel:query}
	\STATE $D_t \leftarrow \text{preselect\_documents}(q_t)$\hfill \textit{\small // preselect documents for query} \label{line:novel:preselect}
	\STATE $\mathbf{R}_t \leftarrow \text{sample\_list}(f_{\theta_t}, D_t)$ \hfill \textit{\small // sample list according to Eq.~\ref{eq:pdgd:docprob}} \label{line:novel:samplelist}
	\STATE $\mathbf{c}_t \leftarrow \text{receive\_clicks}(\mathbf{R}_t)$ \hfill \textit{\small // show result list to the user} \label{line:novel:clicks}
	\STATE $\nabla f_{\theta_{t}} \leftarrow \mathbf{0}$ \hfill \textit{\small // initialize gradient} \label{line:novel:initgrad}
	\FOR{$d_k >_{\mathbf{c}} d_l \in \mathbf{c}_t$} \label{line:novel:prefinfer}
	\STATE $w \leftarrow \rho(d_k, d_l, R, D)$  \hfill \textit{\small // initialize pair weight (Eq.~\ref{eq:pdgd:rho})} \label{line:novel:initpair}
	\STATE $w \leftarrow w 
            \frac{ 
            e^{f_{\theta_t}(\mathbf{d}_k)}e^{f_{\theta_t}(\mathbf{d}_l)}
            }{
            \mleft(e^{f_{\theta_t}(\mathbf{d}_k)} + e^{f_{\theta_t}(\mathbf{d}_l)}\mright) ^ 2
            }$
             \hfill \textit{\small // pair gradient (Eq.~\ref{eq:novelgradient})} \label{line:novel:pairgrad}
	\STATE  $\nabla f_{\theta_{t}} \leftarrow \nabla \theta_t + w (f'_{\theta_t}(\mathbf{d}_k) - f'_{\theta_t}(\mathbf{d}_l))$
	  \hfill \textit{\small // model gradient (Eq.~\ref{eq:novelgradient})} \label{line:novel:modelgrad}
	\ENDFOR
	\STATE $\theta_{t+1} \leftarrow \theta_{t} + \eta \nabla f_{\theta_{t}}$
	\hfill \textit{\small // update the ranking model} \label{line:novel:update}
\ENDFOR
\end{algorithmic}
\end{algorithm}

Algorithm~\ref{alg:novel} describes the \ac{\OurMethod} method step by step:
Given the initial parameters $\theta_1$ and a differentiable scoring function $f$ (Line~\ref{line:novel:initmodel}), the method waits for a user-issued query $q_t$ to arrive (Line~\ref{line:novel:query}).
Then the preselected set of documents $D_t$ for the query is fetched (Line~\ref{line:novel:preselect}), in our experiments these preselections are given in the \ac{LTR} datasets that we use.
A result list $R$ is sampled from the current model (Line~\ref{line:novel:samplelist} and Equation~\ref{eq:pdgd:docprob}) and displayed to the user.
The clicks from the user are logged (Line~\ref{line:novel:clicks}) and preferences between the displayed documents inferred (Line~\ref{line:novel:prefinfer}).
The gradient is initialized (Line~\ref{line:novel:initgrad}), and for each pair document pair $d_k$, $d_l$ such that $d_k >_{\mathbf{c}} d_l$, the weight $\rho(d_k, d_l, R, D)$ is calculated (Line~\ref{line:novel:initpair} and Equation~\ref{eq:pdgd:rho}), followed by the gradient for the pair probability (Line~\ref{line:novel:pairgrad} and Equation~\ref{eq:novelgradient}).
Finally, the gradient for the scoring function $f$ is weighted and added to the gradient (Line~\ref{line:novel:modelgrad}), resulting in the estimated gradient.
The model is then updated by taking an $\eta$ step in the direction of the gradient (Line~\ref{line:novel:update}).
The algorithm again waits for the next query to arrive and thus the process continues indefinitely.

\ac{\OurMethod} has some notable advantages over \ac{MGD}~\citep{schuth2016mgd}.
Firstly, it explicitly models uncertainty over the documents per query, thus \ac{\OurMethod} is able to have high confidence in its ranking for one query, while being completely uncertain for another query.
As a result, it will vary the amount of exploration per query, allowing it to avoid exploration in cases where it is not required and focussing on areas where it can improve.
In contrast, \ac{MGD} does not explicitly model confidence: its degree of exploration is only affected by the norm of its linear model \citep{oosterhuis2017balancing}.
Consequently, \ac{MGD} is unable to vary exploration per query nor is there a way to directly measure its level of confidence.
Secondly, \ac{\OurMethod} works for any differentiable scoring function $f$ and does not rely on sampling model variants.
Conversely, \ac{MGD} is based around sampling from the unit sphere around a model; this approach is very ineffective for non-linear models.
Additionally, sampling large models and producing rankings for them can be very computationally expensive.
Besides these beneficial properties, our experimental results in Section~\ref{sec:pdgd:results} show that \ac{\OurMethod} achieves significantly higher levels of performance than \ac{MGD} and other previous methods.

\subsection{Unbiased gradient estimation}
\label{sec:method:unbiased}

The previous section introduced \ac{\OurMethod}; this section answers \ref{rq:unbiased}:
 \begin{enumerate}[align=left,leftmargin=*]
    \item[\ref{rq:unbiased}] Is the gradient estimation of \ac{\OurMethod} unbiased?
\end{enumerate}

First, Theorem~\ref{theorem:unbiased} will provide a definition of unbiasedness w.r.t.\ user document pair preferences.
Then we state the assumptions we make about user behavior and use them to prove Theorem~\ref{theorem:unbiased}.
Our notation will use $d_k =_{rel} d_l$ to indicate no user preference between two documents t $d_k$ and $d_l$; and $d_k >_{rel} d_l$ to indicate a preference for $d_k$ over $d_l$; and $d_k <_{rel} d_l$ for the opposite preference.

\begin{theorem}
The expected estimated gradient of \ac{\OurMethod} can be written as a weighted sum, with a unique weight $\alpha_{k,l}$ for each possible document pair $d_k$ and $d_l$ in the document collection $D$:
\begin{equation}
E[\nabla f_\theta(\cdot)] = \sum_{d_k, d_l \in D} \alpha_{k,l} (f'_{\theta_t}(\mathbf{d}_k) - f'_{\theta_t}(\mathbf{d}_l)).
\label{eq:theorem:pair}
\end{equation}
The signs of the weights $\alpha_{k,l}$  adhere to user preferences between documents.  
That is, if there is no preference:
\begin{equation}
d_k =_{rel} d_l \Leftrightarrow \alpha_{k,l} = 0;
\label{eq:theorem:equal}
\end{equation}
if $d_k$ is preferred over $d_l$:
\begin{equation}
d_k >_{rel} d_l \Leftrightarrow \alpha_{k,l} > 0; 
\label{eq:theorem:greater}
\end{equation}
and if $d_l$ is preferred over $d_k$:
\begin{equation}
d_k <_{rel} d_l \Leftrightarrow \alpha_{k,l} < 0.
\label{eq:theorem:lesser}
\end{equation}
Therefore, in expectation \ac{\OurMethod} will perform updates that adhere to the preferences between the documents in every possible document pair.
\label{theorem:unbiased}
\end{theorem}

\paragraph{Assumptions.}
To prove Theorem~\ref{theorem:unbiased} the following assumptions about user behavior will be used:

\smallskip
\subparagraph{Assumption 1.} 
We assume that clicks from a user are position biased and conditioned on the relevance of the current document and the previously considered documents.
For a click on a document in ranking $R$ at position $i$ the probability can be written as:
\begin{equation}
P(\textit{click}(R_i) | \{R_0, \ldots, R_{i-1}, R_{i+1}\}).
\label{eq:userassumption}
\end{equation}
For ease of notation, we will denote the set of ``other documents'' as $\OtherDocuments{}$ from here on.

\subparagraph{Assumption 2.} 
If there is no user preference between two documents $d_k, d_l$, denoted by $d_k =_\textit{rel} d_l$, we assume that each is equally likely to be clicked given the same context:
\begin{equation}
d_k =_\textit{rel} d_l \Rightarrow P(\textit{click}(d_k)| \OtherDocuments{}) = P(\textit{click}(d_l) | \OtherDocuments{} ).
\label{eq:equalsame}
\end{equation}

\subparagraph{Assumption 3.}
If a document in the set of documents being considered is replaced with an equally preferred document the click probability is not affected:
\begin{equation}
d_k =_\textit{rel} d_l \Rightarrow  P(\textit{click}(R_i)| \{\ldots, d_k\} ) = P(\textit{click}(R_i)|  \{\ldots, d_l\} ).
\label{eq:equaladded}
\end{equation}

\subparagraph{Assumption 4.}
Similarly, given the same context if one document is preferred over another, then it is more likely to be clicked:
\begin{equation}
d_k >_\textit{rel} d_l \Rightarrow P(\textit{click}(d_k)| \OtherDocuments{}) > P(\textit{click}(d_l)| \OtherDocuments{} ).
\label{eq:greatersame}
\end{equation}

\subparagraph{Assumption 5.}
Lastly, for any pair $d_k >_\textit{rel} d_l$, the considered document set $\{\ldots, d_k\}$ and the same set with $d_k$ replaced by $d_l$ $\{\ldots, d_l\}$,
we assume that the preferred $d_k$ in the context of $\{\ldots, d_l\}$ is more likely to be clicked than $d_l$ in the context of $\{\ldots, d_k\}$:
\begin{equation}
d_k >_\textit{rel} d_l \Rightarrow  P(\textit{click}(d_k)| \{\ldots, d_k\} ) > P(\textit{click}(d_l)|  \{\ldots, d_l\} ).
\label{eq:greateradded}
\end{equation}
These are all the assumptions we make about the user. 
With these assumptions, we can proceed to prove Theorem~\ref{theorem:unbiased}.

\begin{proof}[Proof of Theorem~\ref{theorem:unbiased}.]
We denote the probability of inferring the preference of $d_k$ over $d_l$ in ranking $R$ as $P(d_k >_\mathbf{c} d_l | R)$. 
Then the expected gradient $\nabla f_\theta(\cdot)$ of \ac{\OurMethod} can be written as:
\begin{equation}
\begin{split}
&E[\nabla f_\theta(\cdot)] = \\
&\qquad \sum_R \sum_{d_k, d_l \in D} \mleft[ \vphantom{\frac{\delta}{\delta \theta}} 
P(d_k >_\mathbf{c} d_l | R) \cdot P(R) \cdot 
 \rho(d_k, d_l, R, D)  
 \mleft[\nabla P(d_k \succ d_l)\mright]\mright] .
\end{split}
\end{equation}
We will rewrite this expectation using the symmetry property of the reversed pair ranking:
\begin{equation}
R^n = R^*(d_k, d_l, R^m) \Leftrightarrow R^m = R^*(d_k, d_l, R^n).
\end{equation}
First, we define a weight $\omega_{k,l}^R$ for every document pair $d_k,d_l$ and ranking $R$ so that:
\begin{equation}
\omega_{k,l}^R = P(R) \rho(d_k, d_l, R, D) 
= \frac{P(R|D)P(R^*(d_k, d_l, R)|D)}{P(R|D) + P(R^*(d_k, d_l, R)|D)}.
\end{equation}
Therefore, the weight for the reversed pair ranking is equal:
\begin{equation}
\begin{split}
\omega_{k,l}^{R^*(d_k, d_l, R)} &= P(R^*(d_k, d_l, R)) \rho(d_k, d_l, R^*(d_k, d_l, R), D) \\
&= \omega_{k,l}^R.
\end{split}
\end{equation}
Then, using the symmetry of Equation~\ref{eq:pairprob} we see that:
\begin{equation}
\nabla P(d_k \succ d_l) = -\nabla P(d_l \succ  d_k).
\end{equation}
Thus, with $R^*$ as a shorthand for $R^*(d_k, d_l, R)$, the expectation can be rewritten as:
\begin{align}
& E[\nabla f_\theta(\cdot)] = \\
& \qquad \sum_{d_k, d_l \in D} \sum_R \frac{\omega_{i,j}^R}{2}  
  \left( \vphantom{\frac{\omega_{i,j}^R}{2}  }P(d_k >_\mathbf{c} d_l \mid R)  - P(d_l >_\mathbf{c} d_k \mid R^*)\right) \left[\vphantom{\frac{\omega_{i,j}^R}{2}  }\nabla P(d_k \succ d_l)\right],
  \nonumber
\end{align}
proving that the expected gradient matches the form of Equation~\ref{eq:theorem:pair}.
Then to prove that Equations~\ref{eq:theorem:equal},~\ref{eq:theorem:greater}, and~\ref{eq:theorem:lesser} are correct we will show that:
\begin{align}
d_k =_\textit{rel} d_l &\Rightarrow  P(d_k >_\mathbf{c} d_l | R) = P(d_l >_\mathbf{c} d_k | R^*),
\label{eq:equalrequirement}
\\
d_k >_\textit{rel} d_l &\Rightarrow  P(d_k >_\mathbf{c} d_l | R) > P(d_l >_\mathbf{c} d_k | R^*),
\label{eq:greaterrequirement}
\\
d_k <_\textit{rel} d_l &\Rightarrow  P(d_k >_\mathbf{c} d_l | R) < P(d_l >_\mathbf{c} d_k | R^*).
\label{eq:lesserrequirement}
\end{align}
If a preference $R_i >_\mathbf{c} R_j$ is inferred then there are only three possible cases based on the positions:
\begin{enumerate}
\item The clicked document succeeds the unclicked document by more than one position: $i + 1 > j$.
\item The clicked document precedes the unclicked document by more than one position: $i - 1 < j$.
\item The clicked document is one position before or after the unclicked document: $i = j +1 \lor i = j-1$.
\end{enumerate}
In the first case the clicked document succeeds the other by more than one position, the probability of an inferred preference is then:
\begin{equation}
i + 1 > j 
\Rightarrow
P(R_i >_\mathbf{c} R_j | R) =  P(\mathbf{c}_i | R_i, \{\ldots, R_j\}\ )  (1- P(\mathbf{c}_j | R_j, \OtherDocuments{} )).
\label{eq:simplecase}
\end{equation}
Combining Assumption~2~and~3 with Equation~\ref{eq:simplecase} proves Equation~\ref{eq:equalrequirement} for this case.
Furthermore, combining Assumption~4~and~5 with Equation~\ref{eq:simplecase} proves Equations~\ref{eq:greaterrequirement}~and~\ref{eq:lesserrequirement} for this case as well.

Then the second case is when the clicked document appears more than one position before the unclicked document, the probability of the inferred preference is then:
\begin{equation}
\begin{split}
& i + 1 < j \Rightarrow 
\\ & \qquad
P(R_i >_\mathbf{c} R_j | R) =  P(\mathbf{c}_i | R_i, \OtherDocuments{} )   (1- P(\mathbf{c}_j | R_j,\{\ldots, R_i\}\ ))  P(\mathbf{c}_\textit{rem}),
\end{split}
\label{eq:complexcase}
\end{equation}
where $P(\mathbf{c}_\textit{rem})$ denotes the probability of an additional click that is required to add $R_j$ to the inferred observed documents.
First, due to Assumption~1 this probability will be the same for $R$ and $R^*$:
\begin{equation}
P(\mathbf{c}_\textit{rem}| R_i, R_j , R) = P(\mathbf{c}_\textit{rem}| R_i, R_j , R^*).
\end{equation}
Combining Assumption~2~and~3 with Equation~\ref{eq:complexcase} also proves Equation~\ref{eq:equalrequirement} for this case.
Furthermore, combining Assumption~4~and~5 with Equation~\ref{eq:complexcase} also proves Equation~\ref{eq:greaterrequirement}~and~\ref{eq:lesserrequirement} for this case as well.

Lastly, in the third case the clicked document is one position before or after the other document, the probability of the inferred preference is then:
\begin{align}
&i = j +1 \lor i = j-1 \\
&\qquad \Rightarrow
P(R_i >_\mathbf{c} R_j | R) =  P(\mathbf{c}_i | R_i, \{\ldots, R_j\}\ )(1- P(\mathbf{c}_j | R_j, \{\ldots, R_i\} )).
\nonumber
\label{eq:specialcase}
\end{align}
Combining Assumption~3 with Equation~\ref{eq:specialcase} proves Equation~\ref{eq:equalrequirement} for this case as well.
Then, combining Assumption~5 with Equation~\ref{eq:specialcase} also proves Equation~\ref{eq:greaterrequirement}~and~\ref{eq:lesserrequirement} for this case.
\end{proof}

\noindent%
This concludes our proof of the unbiasedness of \ac{\OurMethod}. 
Hence, we answer \ref{rq:unbiased} positively: the gradient estimation of \ac{\OurMethod} is unbiased.
We have shown that the expected gradient is in line with user preferences between document pairs.

%% file: 05-pdgd/04-experiments.tex
\begin{table}[tb]
\caption{Instantiations of Cascading Click Models~\citep{guo09:efficient} as used for simulating user behavior in experiments.}
\centering
\begin{tabularx}{\columnwidth}{ X c c c c c c c c c c }
\toprule
& \multicolumn{5}{c}{ $P(\mathit{click}=1\mid R)$} & \multicolumn{5}{c}{ $P(\mathit{stop}=1\mid click=1,  R)$} \\
\cmidrule(lr){2-6} \cmidrule(l){7-11}
$R$ & \emph{$ 0$} & \emph{$ 1$}  &  \emph{$ 2$} & \emph{$ 3$} & \emph{$ 4$}
 & \emph{$0$} & \emph{$ 1$} & \emph{$ 2$} & \emph{$ 3$} & \emph{$ 4$} \\
\midrule
 \emph{perfect} &  0.0 &  0.2 &  0.4 &  0.8 &  1.0 &  0.0 &  0.0 &  0.0 &  0.0 &  0.0 \\
 \emph{navigational} &  \phantom{5}0.05 &  0.3 &  0.5 &  0.7 &  \phantom{5}0.95 &  0.2 &  0.3 &  0.5 &  0.7 &  0.9 \\
 \emph{informational} &  0.4 &  0.6 &  0.7 &  0.8 &  0.9 &  0.1 &  0.2 &  0.3 &  0.4 &  0.5 \\
\bottomrule
\end{tabularx}
\label{tab:clickmodels}
\end{table}

\section{Experiments}
\label{sec:experiments}

In this section we detail the experiments that were performed to answer the research questions in Section~\ref{sec:pdgd:intro}.

\subsection{Datasets}
\label{sec:experiments:datasets}

Our experiments are performed over five publicly available \acs{LTR} datasets; we have selected three large labelled dataset from commercial search engines and two smaller research datasets.
Every dataset consists of a set of queries with each query having a corresponding preselected document set.
The exact content of the queries and documents are unknown, each query is represented only by an identifier, but each query-document pair has a feature representation and relevance label.
Depending on the dataset, the relevance labels are graded differently; we have purposefully chosen datasets that have at least two grades of relevance.
Each dataset is divided in training, validation and test partitions.

The oldest datasets we use are \emph{MQ2007} and \emph{MQ2008}~\cite{qin2013introducing} which are based on the Million Query Track \cite{allan2007million} and consist of \numprint{1700} and 800 queries. They use representations of 46~features that encode ranking models such as TF.IDF, BM25, Language Modeling, Page\-Rank, and HITS on different parts of the documents. They are divided into five folds and the labels are on a three-grade scale from \emph{not relevant} (0) to \emph{very relevant}~(2).

In 2010 Microsoft released the \emph{MSLR-WEB30k} and \emph{MLSR-WEB10K} datasets \cite{qin2013introducing}, which are both created from a retired labelling set of a commercial web search engine (Bing).
The former contains \numprint{30000} queries with each query having 125~assessed documents on average, query-document pairs are encoded in 136~features,
The latter is a subsampling of \numprint{10000} queries from the former dataset.
For practical reasons only \emph{MLSR-WEB10K} was used for this chapter.
Also in 2010 Yahoo!\ released an \ac{LTR} dataset \cite{Chapelle2011}.
It consists of \numprint{29921} queries and \numprint{709877} documents encoded in 700~features, all sampled from query logs of the Yahoo! search engine.
Finally, in 2016 a \ac{LTR} dataset was released by the Istella search engine \cite{dato2016fast}.
It is the largest with \numprint{33118} queries, an average of~315 documents per query and 220~features.
These three commercial datasets all label relevance on a five-grade scale: from \emph{not relevant} (0) to \emph{perfectly relevant}~(4).

\subsection{Simulating user behavior}
\label{sec:experiments:users}

For simulating users we follow the standard setup for \acs{OLTR} simulations \cite{he2009evaluation,hofmann11:balancing,schuth2016mgd,oosterhuis2016probabilistic,zoghi:wsdm14:relative}.
First, queries issued by  users are simulated by uniformly sampling from the static dataset.
Then the algorithm determines the result list of documents to display.
User interactions with the displayed list are then simulated using a \emph{cascade click model}~\cite{chuklin-click-2015,guo09:efficient}.
This models a user who goes through the documents one at a time in the displayed order.
At each document, the user decides whether to click it or not, modelled as a probability conditioned on the relevance label $R$: $P(click=1\mid R)$.
After a click has occurred, the user's information need may be satisfied and they may then stop considering documents.
The probability of a user stopping after a click is modelled as $P(stop=1\mid click=1,  R)$.
For our experiments $\kappa=10$ documents are displayed at each impression.

The three instantiations of cascade click models that we used are listed in Table~\ref{tab:clickmodels}.
First, a \emph{perfect} user is modelled who considers every document and solely clicks on all relevant documents.
The second models a user with a \emph{navigational} task, where a single highly relevant document is searched.
Finally, an \emph{informational} instantiation models a user without a specific information need, and thus typically clicks on many documents.
These models have varying levels of noise, as each behavior depends on the relevance labels of documents with a different degree.

\subsection{Experimental runs}
\label{sec:experiments:runs}

For our experiments three baselines are used. 
First, \ac{MGD} with Probabilistic Multileaving~\citep{oosterhuis2016probabilistic}; this is the highest performing existing \ac{OLTR} method~\citep{oosterhuis2016probabilistic, oosterhuis2017balancing}.
For this chapter $n=49$ candidates were sampled per iteration from the unit sphere with $\delta=1$; updates are performed with $\eta = 0.01$ and zero initialization was used.
Additionally, \ac{DBGD} is used for comparison as it is one of the most influential methods, it was run with the same parameters except that only $n=1$ candidate is sampled per iteration.
Furthermore, we also let \ac{DBGD} optimize a single hidden-layer neural network with 64 hidden nodes and sigmoid activation functions with \emph{Xavier} initialization~\cite{glorot2010understanding}.
These parameters were also found most effective in previous work~\citep{yue2009interactively, hofmann11:balancing, schuth2016mgd, oosterhuis2016probabilistic}.

Additionally, the pairwise method introduced by \citet{hofmann11:balancing} is used as a baseline. Despite not showing significant improvements over \ac{DBGD} in the past~\cite{hofmann11:balancing}, the comparison with \ac{\OurMethod} is interesting because they both estimate gradients from pairwise preferences. 
For this baseline, $\eta = 0.01$ and $\epsilon = 0.8$ is used; these parameters are chosen to maximize the performance at convergence~\cite{hofmann11:balancing}.

Runs with \ac{\OurMethod} are performed with both a linear and neural ranking model.
For the linear ranking model $\eta = 0.1$ and zero initialization was used.
The neural network has the same parameters as the one optimized by \ac{DBGD}, except for $\eta = 0.1$.

\subsection{Metrics and tests}
\label{sec:experiments:evaluation}

Two aspects of performance are evaluated seperately: the final convergence and the ranking quality during training.

Final convergence is addressed in \emph{offline performance} which is the average NDCG@10 of the ranking model over the queries in the held-out test-set.
The offline performance is measured after 10,000 impressions at which point most ranking models have reached convergence.
The user experience during optimization should be considered as well, since deterring users during training would compromise the goal of \ac{OLTR}.
To address this aspect of evaluation \emph{online performance} has been introduced~\citep{Hofmann2013a}; it is the cumulative discounted NDCG@10 of the rankings displayed during training.
For $T$ sequential queries with $R^t$ as the ranking displayed to the user at timestep $t$, this is:
\begin{align}
\mathit{Online\_Performance} = \sum_{t=1}^T \mathit{NDCG}(R^t) \cdot \gamma^{(t-1)}.
\end{align}
This metric models the expected reward a user receives with a $\gamma$ probability that the user stops searching after each query.
We follow previous work~\cite{oosterhuis2017balancing, oosterhuis2016probabilistic} by choosing a discount factor of $\gamma = 0.9995$, consequently queries beyond the horizon of \numprint{10000} queries have a less than $1\%$ impact.

Lastly, all experimental runs are repeated 125 times, spread evenly over the available dataset folds.
Results are averaged and a two-tailed Student's t-test is used for significance testing.
In total, our results are based on more than \numprint{90000000} user impressions.

%% file: 05-pdgd/05-results.tex
\begin{figure}[tb]
\centering
\includegraphics[scale=0.47]{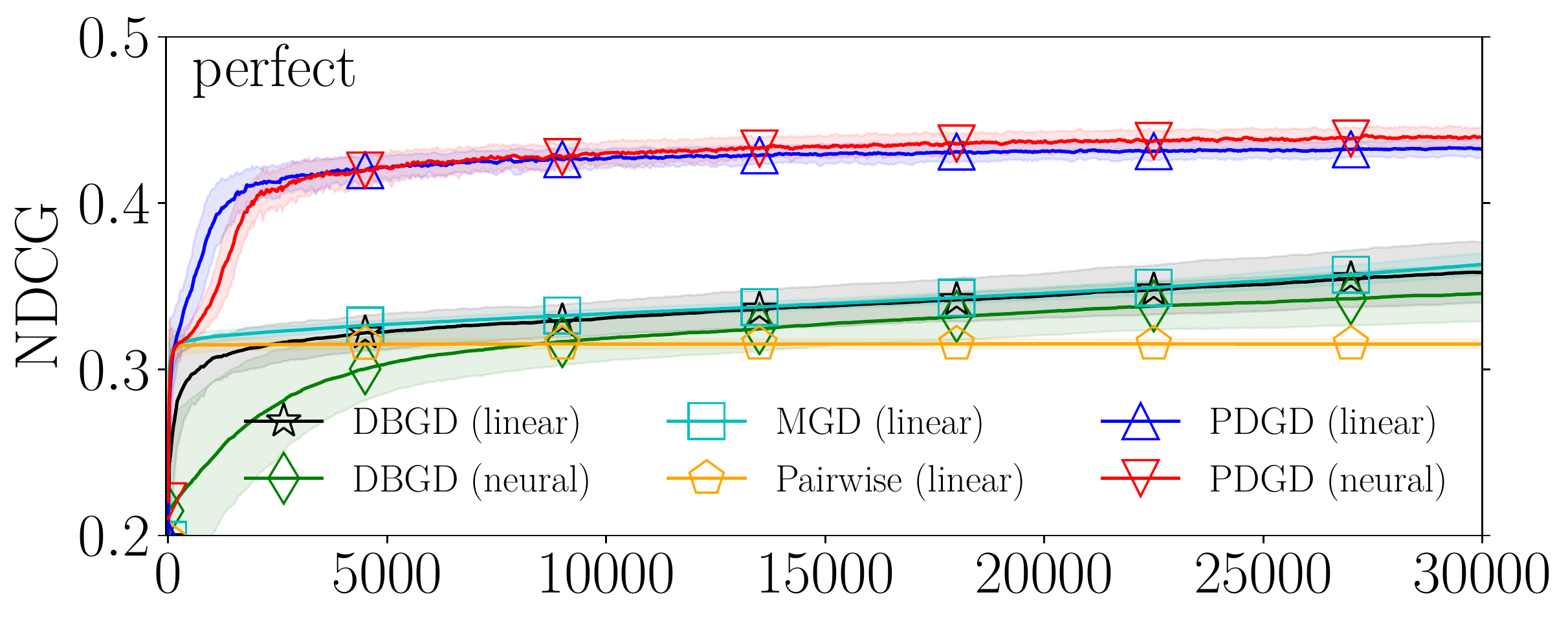}
\includegraphics[scale=0.47]{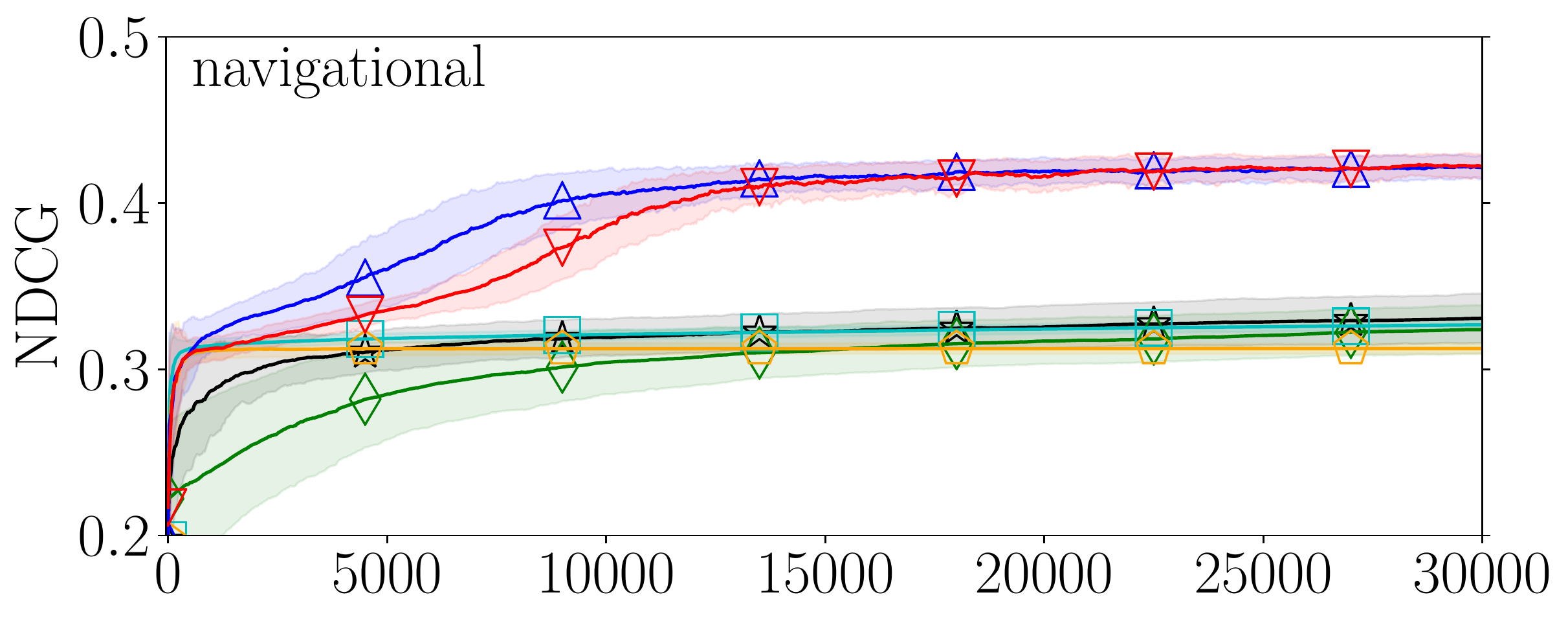}
\includegraphics[scale=0.47]{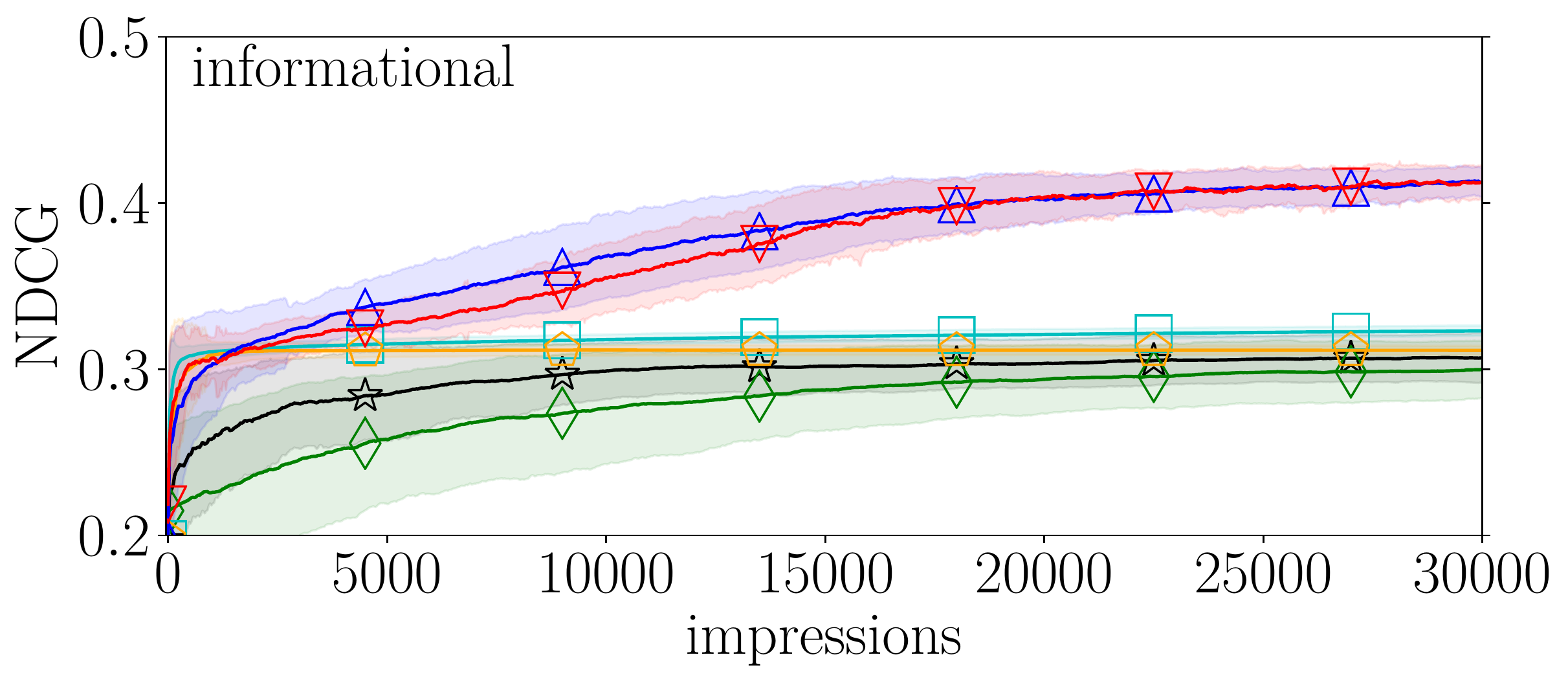}
\caption{Offline performance (NDCG) on the MSLR-WEB10k dataset under three different click models, the shaded areas indicate the standard deviation.}
\label{fig:offline}
\end{figure}

\begin{figure}[tb]
\centering
\includegraphics[scale=0.47]{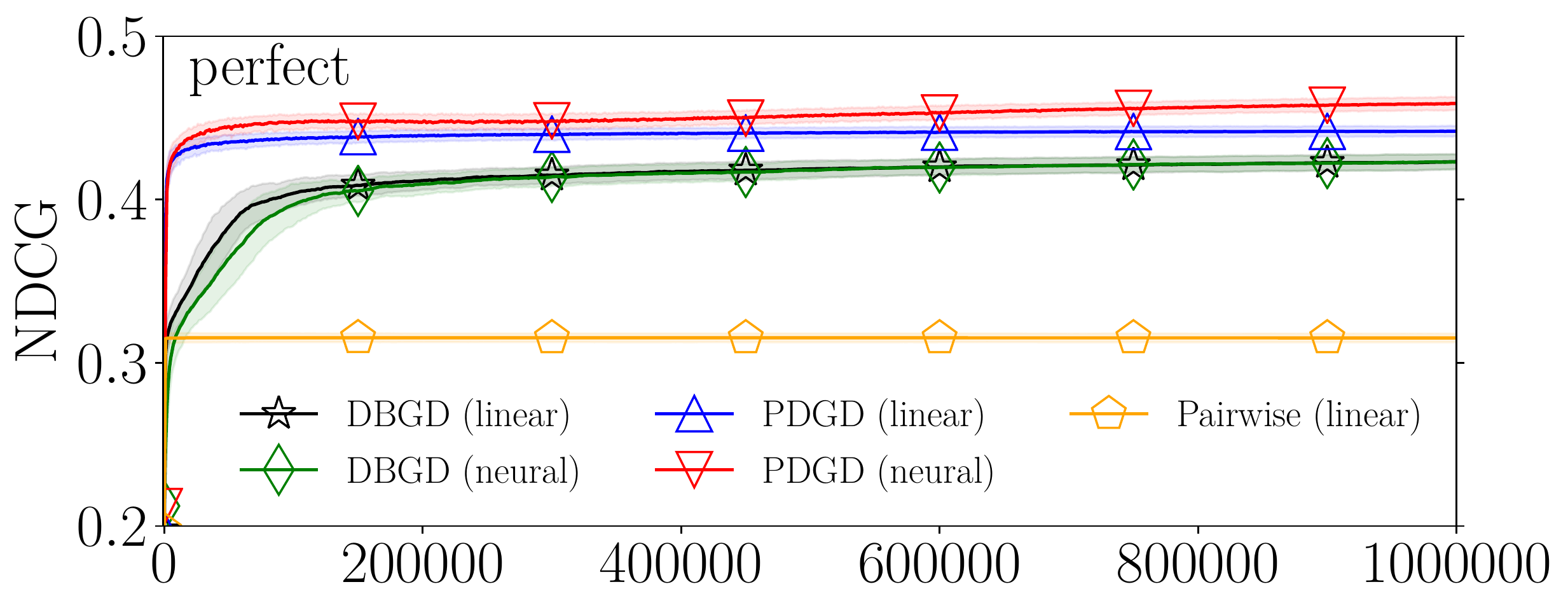}
\includegraphics[scale=0.47]{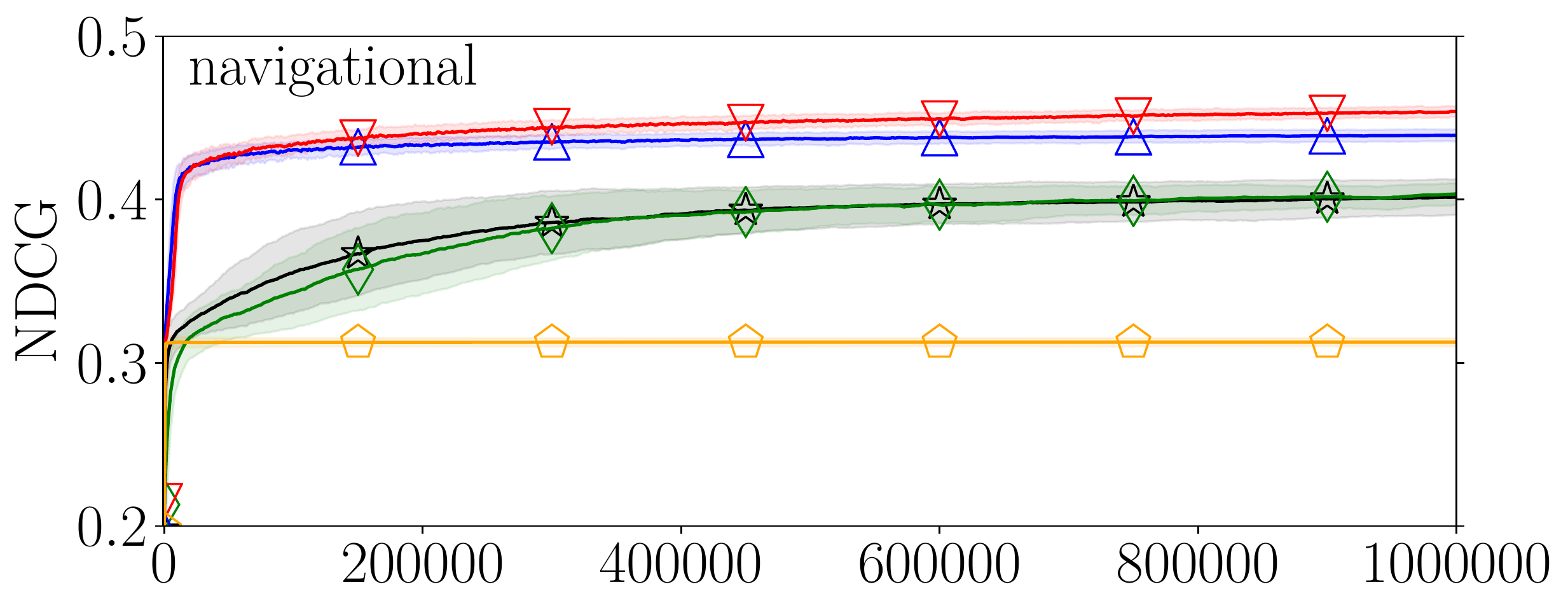}
\includegraphics[scale=0.47]{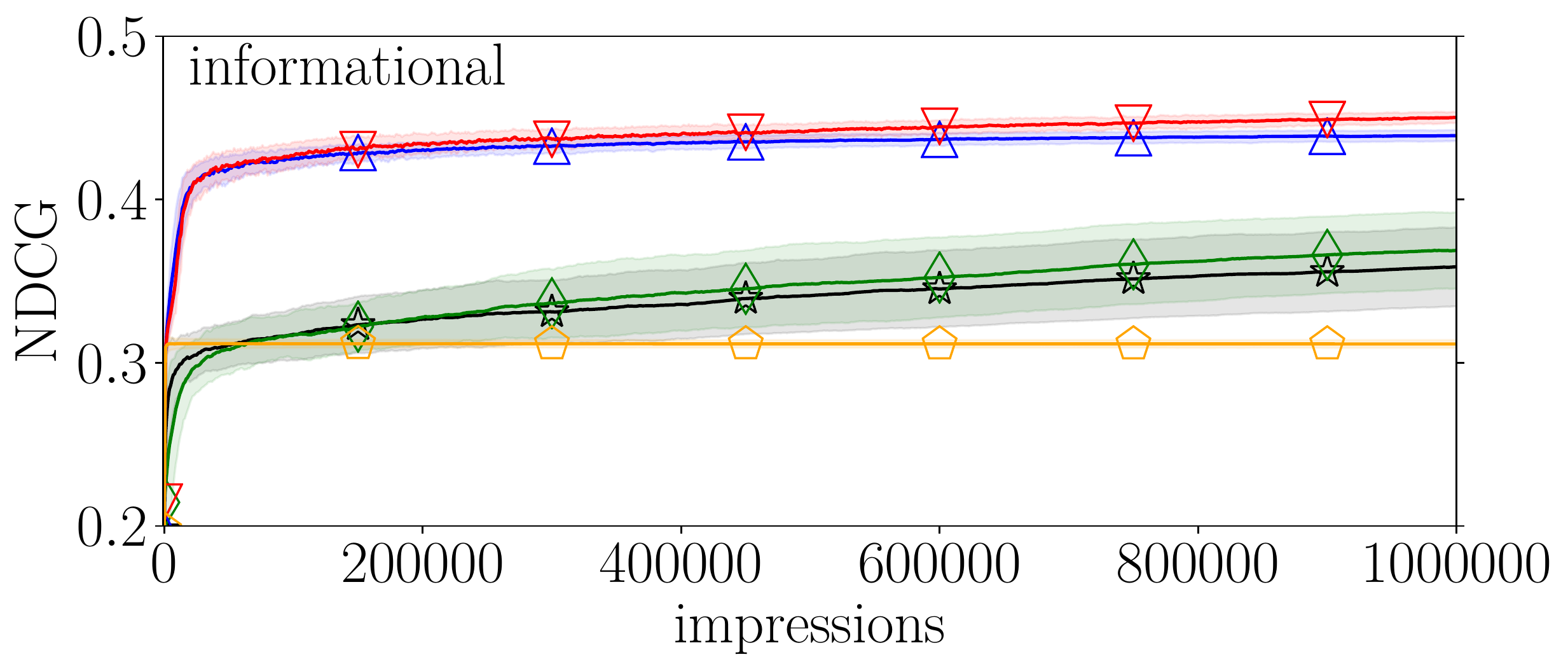}
\caption{Long-term offline performance (NDCG) on the MSLR-WEB10k dataset under three click models, the shaded areas indicate the standard deviation.}
\label{fig:long}
\end{figure}

\begin{sidewaystable}[t]
\centering
\caption{Offline performance (NDCG) for different instantiations of CCM (Table~\ref{tab:clickmodels}). The standard deviation is shown in brackets, bold values indicate the highest performance per dataset and click model, significant improvements over the \acs{DBGD}, \acs{MGD} and pairwise baselines are indicated by  \enkelop\ (p $<$ 0.05) and \dubbelop\ (p $<$ 0.01), no losses were measured.}
\vspace{-0.6\baselineskip}
\input{05-pdgd/tables/offline}
\label{tab:offline}
\end{sidewaystable}

\begin{sidewaystable}[t]
\centering
\caption{Online performance (Discounted Cumulative NDCG, Section~\ref{sec:experiments:evaluation}) for different instantiations of CCM (Table~\ref{tab:clickmodels}). The standard deviation is shown in brackets, bold values indicate the highest performance per dataset and click model, significant improvements and losses over the \acs{DBGD}, \acs{MGD} and pairwise baselines are indicated by  \enkelop\ (p $<$ 0.05) and \dubbelop\ (p $<$ 0.01) and by \enkelneer\ and \dubbelneer, respectively.}
\vspace{-0.6\baselineskip}
\input{05-pdgd/tables/online}
\label{tab:online}
\end{sidewaystable}

\section{Results and Analysis}
\label{sec:pdgd:results}

Our main results are displayed in Table~\ref{tab:offline} and Table~\ref{tab:online}, showing the offline and online performance of all methods, respectively.
Additionally, Figure~\ref{fig:offline} displays the offline performance on the MSLR-WEB10k dataset over \numprint{30000} impressions and Figure~\ref{fig:long} over \numprint{1000000} impressions.
We use these results to answer \ref{rq:performance} -- whether \ac{\OurMethod} provides significant improvements over existing \ac{OLTR} methods -- and \ref{rq:nonlinear} -- whether \ac{\OurMethod} is successful at optimizing different types of ranking models.

\subsection{Convergence of ranking models}
First, we consider the offline performance after \numprint{10000} impressions as reported in Table~\ref{tab:offline}.
We see that the \ac{DBGD} and \ac{MGD} baselines reach similar levels of performance, with marginal differences at low levels of noise.
Our results seem to suggest that \ac{MGD} provides an efficient alternative to \ac{DBGD} that requires fewer user interactions and is more robust to noise.
However, \ac{MGD} does not appear to have an improved point of convergence over \ac{DBGD}, Figure~\ref{fig:offline} further confirms this conclusion.
Additionally, Table~\ref{tab:offline} and Figure~\ref{fig:long} reveal thats \ac{DBGD} is incapable of training its neural network so that it improves over the linear model, even after \numprint{1000000} impressions.

Alternatively, the pairwise baseline displays different behavior, providing improvements over \ac{DBGD} and \ac{MGD} on most datasets under all levels of noise.
However, on the istella dataset large decreases in performance are observed.
Thus it is unclear if this method provides a reliable alternative to  \ac{DBGD} or \ac{MGD} in terms of convergence.
Figure~\ref{fig:offline} also reveals that it converges within several hundred impressions, while  \ac{DBGD} or \ac{MGD} continue to learn and considerably improve over the total  \numprint{30000} impressions.
Because the pairwise baseline also converges sub-optimally under the perfect click model, we do not attribute its suboptimal convergence to noise but to the method being biased.

Conversely, Table~\ref{tab:offline} shows that \ac{\OurMethod} reaches significantly higher performance than all the baselines within \numprint{10000} impressions.
Improvements are observed on all datasets under all levels of noise, especially on the commercial datasets where increases up to $0.17$ NDCG are observed.
Our results also show that \ac{\OurMethod} learns faster than the baselines; at all time-steps the offline performance of \ac{\OurMethod} is at least as good or better than all other methods,  across all datasets.
This increased learning speed can also be observed in Figure~\ref{fig:offline}.
Besides the faster learning it also appears as if \ac{\OurMethod} converges at a better optimum than \ac{DBGD} or \ac{MGD}.
However, Figure~\ref{fig:offline} reveals that \ac{DBGD} does not fully converge within \numprint{30000} iterations.
Therefore, we performed an additional experiment where \ac{\OurMethod} and \ac{DBGD} optimize models over \numprint{1000000} impressions on the MSLR-WEB10k dataset, as displayed in Figure~\ref{fig:long}.
Clearly the performance of \ac{DBGD} plateaus at a considerably lower level than that of \ac{\OurMethod}.
Therefore, we conclude that \ac{\OurMethod} indeed has an improved point of final convergence compared to \ac{DBGD} and \ac{MGD}.

Finally, Figure~\ref{fig:offline}~and~\ref{fig:long} also shows the behavior predicted by the speed-quality tradeoff~\cite{oosterhuis2017balancing}: a more complex model will have a worse initial performance but a better final convergence.
Here, we see that depending on the level of interaction noise the neural model requires \numprint{3000} to \numprint{20000} iterations to match the performance of a linear model.
However, in the long run the neural model does converge at a significantly  better point of convergence.
Thus, we conclude that \ac{\OurMethod} is capable of effectively optimizing different kinds of models in terms of offline performance.

In conclusion, our results show that \ac{\OurMethod} learns faster than existing \ac{OLTR} methods while also converging at significantly better levels of performance.

\subsection{User experience during training}
\label{sec:results:onlineperformance}

Besides the ranking models learned by the \ac{OLTR} methods, we also consider the user experience during optimization.
Table~\ref{tab:online} shows that the online performance of \ac{DBGD} and \ac{MGD} are close to each other; \ac{MGD} has a higher online performance due to its faster learning speed~\cite{oosterhuis2016probabilistic, schuth2016mgd}.
In contrast, the pairwise baseline has a substantially lower online performance in all cases.
Because Figure~\ref{fig:offline} shows that the learning speed of the pairwise baseline sometimes matches that of \ac{DBGD} and \ac{MGD}, we attribute this difference to the exploration strategy it uses.
Namely, the random insertion of uniformly sampled documents by this baseline appears to have a strong negative effect on the user experience.

The linear model optimized by \ac{\OurMethod} has significant improvements over all baseline methods on all datasets and under all click models.
This improvement indicates that the exploration of \ac{\OurMethod}, which uses a distribution over documents, does not lead to a worse user experience.
In conclusion, \ac{\OurMethod} provides a considerably better user experience than all existing methods.

Finally, we also discuss the performance of the neural models optimized by \ac{\OurMethod} and \ac{DBGD}.
This model has both significant increases and decreases in online performance varying per dataset and amount of interaction noise.
The decrease in user experience is predicted by the speed-quality tradeoff~\cite{oosterhuis2017balancing}, as Figure~\ref{fig:offline} also shows, the neural model has a slower learning speed leading to a worse initial user experience.
A solution to this tradeoff has been proposed by \citet{oosterhuis2017balancing}, which optimizes a cascade of models.
In this case, the cascade could combine the user experience of the linear model with the final convergence of the neural model, providing the best of both worlds.

\subsection{Improvements of \acs{\OurMethod}}

After having discussed the offline and online performance of \ac{\OurMethod}, we will now answer \ref{rq:performance} and \ref{rq:nonlinear}.

First, concerning \ref{rq:performance} (whether \ac{\OurMethod} performs significantly better than \ac{MGD}), the results of our experiments show that models optimized with \ac{\OurMethod} learn faster and converge at better optima than \ac{MGD}, \ac{DBGD}, and the pairwise baseline, regardless of dataset or level of interaction noise.
Moreover, the level of performance reached with \ac{\OurMethod} is significantly higher than the final convergence of any other method.
Thus, even in the long run \ac{DBGD} and \ac{MGD} are incapable of reaching the offline performance of \ac{\OurMethod}.
Additionally, the online performance of a linear model optimized with \ac{\OurMethod} is significantly better across all datasets and user models.
Therefore, we answer \ref{rq:performance} positively: \ac{\OurMethod} outperforms existing methods both in terms of model convergence and user experience during learning.

Then, with regards to \ref{rq:nonlinear} (whether \ac{\OurMethod} can effectively optimize different types of models), in our experiments we have successfully optimized models from two families: linear models and neural networks.
Both models reach a significantly higher level of performance of model convergence than previous \ac{OLTR} methods, across all datasets and degrees of interaction noise.
As expected, the simpler linear model has a better initial user experience, while the more complex neural model has a better point of convergence.
In conclusion, we answer \ref{rq:nonlinear} positively: \ac{\OurMethod} is applicable to different ranking models and effective for both linear and non-linear models.

%% file: 05-pdgd/tables/offline.tex
\begin{tabular*}{\textwidth}{@{\extracolsep{\fill} } l  l l l l l  }
\toprule
 & { \small \textbf{MQ2007}}  & { \small \textbf{MQ2008}}  & { \small \textbf{MSLR-WEB10k}}  & { \small \textbf{Yahoo}}  & { \small \textbf{istella}} \\
\midrule
& \multicolumn{5}{c}{\textit{perfect}} \\
\midrule
DBGD (linear) & 0.483 {\tiny (0.023)} & 0.683 {\tiny (0.024)} & 0.331 {\tiny (0.010)} & 0.684 {\tiny (0.010)} & 0.448 {\tiny (0.014)} \\
DBGD (neural) & 0.463 {\tiny (0.025)} & 0.670 {\tiny (0.026)} & 0.319 {\tiny (0.014)} & 0.676 {\tiny (0.016)} & 0.429 {\tiny (0.017)} \\
MGD (linear) & 0.494 {\tiny (0.022)} & 0.690 {\tiny (0.019)} & 0.333 {\tiny (0.003)} & 0.714 {\tiny (0.002)} & 0.496 {\tiny (0.004)} \\
Pairwise (linear) & 0.479 {\tiny (0.022)} & 0.674 {\tiny (0.017)} & 0.315 {\tiny (0.003)} & 0.709 {\tiny (0.001)} & 0.252 {\tiny (0.002)} \\
PDGD (linear) & \bf 0.511 {\tiny (0.017)} {\tiny \dubbelop} {\tiny \dubbelop} {\tiny \dubbelop} {\tiny \dubbelop} & \bf 0.699 {\tiny (0.024)} {\tiny \dubbelop} {\tiny \dubbelop} {\tiny \dubbelop} {\tiny \dubbelop} & 0.427 {\tiny (0.005)} {\tiny \dubbelop} {\tiny \dubbelop} {\tiny \dubbelop} {\tiny \dubbelop} & \bf 0.736 {\tiny (0.004)} {\tiny \dubbelop} {\tiny \dubbelop} {\tiny \dubbelop} {\tiny \dubbelop} & 0.573 {\tiny (0.004)} {\tiny \dubbelop} {\tiny \dubbelop} {\tiny \dubbelop} {\tiny \dubbelop} \\
PDGD (neural) & 0.509 {\tiny (0.020)} {\tiny \dubbelop} {\tiny \dubbelop} {\tiny \dubbelop} {\tiny \dubbelop} & 0.698 {\tiny (0.024)} {\tiny \dubbelop} {\tiny \dubbelop} {\tiny \dubbelop} {\tiny \dubbelop} & \bf 0.430 {\tiny (0.006)} {\tiny \dubbelop} {\tiny \dubbelop} {\tiny \dubbelop} {\tiny \dubbelop} & 0.733 {\tiny (0.005)} {\tiny \dubbelop} {\tiny \dubbelop} {\tiny \dubbelop} {\tiny \dubbelop} & \bf 0.575 {\tiny (0.006)} {\tiny \dubbelop} {\tiny \dubbelop} {\tiny \dubbelop} {\tiny \dubbelop} \\
\midrule
& \multicolumn{5}{c}{\textit{navigational}} \\
\midrule
DBGD (linear) & 0.461 {\tiny (0.025)} & 0.670 {\tiny (0.025)} & 0.319 {\tiny (0.011)} & 0.661 {\tiny (0.023)} & 0.401 {\tiny (0.015)} \\
DBGD (neural) & 0.430 {\tiny (0.033)} & 0.646 {\tiny (0.031)} & 0.304 {\tiny (0.019)} & 0.649 {\tiny (0.029)} & 0.382 {\tiny (0.024)} \\
MGD (linear) & 0.426 {\tiny (0.020)} & 0.662 {\tiny (0.015)} & 0.321 {\tiny (0.003)} & 0.706 {\tiny (0.009)} & 0.405 {\tiny (0.004)} \\
Pairwise (linear) & 0.476 {\tiny (0.022)} & 0.677 {\tiny (0.018)} & 0.312 {\tiny (0.003)} & 0.696 {\tiny (0.004)} & 0.209 {\tiny (0.002)} \\
PDGD (linear) & \bf 0.496 {\tiny (0.019)} {\tiny \dubbelop} {\tiny \dubbelop} {\tiny \dubbelop} {\tiny \dubbelop} & \bf 0.695 {\tiny (0.021)} {\tiny \dubbelop} {\tiny \dubbelop} {\tiny \dubbelop} {\tiny \dubbelop} & \bf 0.406 {\tiny (0.015)} {\tiny \dubbelop} {\tiny \dubbelop} {\tiny \dubbelop} {\tiny \dubbelop} & \bf 0.725 {\tiny (0.005)} {\tiny \dubbelop} {\tiny \dubbelop} {\tiny \dubbelop} {\tiny \dubbelop} & \bf 0.540 {\tiny (0.008)} {\tiny \dubbelop} {\tiny \dubbelop} {\tiny \dubbelop} {\tiny \dubbelop} \\
PDGD (neural) & 0.493 {\tiny (0.020)} {\tiny \dubbelop} {\tiny \dubbelop} {\tiny \dubbelop} {\tiny \dubbelop} & 0.692 {\tiny (0.019)} {\tiny \dubbelop} {\tiny \dubbelop} {\tiny \dubbelop} {\tiny \dubbelop} & 0.386 {\tiny (0.019)} {\tiny \dubbelop} {\tiny \dubbelop} {\tiny \dubbelop} {\tiny \dubbelop} & 0.722 {\tiny (0.006)} {\tiny \dubbelop} {\tiny \dubbelop} {\tiny \dubbelop} {\tiny \dubbelop} & 0.532 {\tiny (0.011)} {\tiny \dubbelop} {\tiny \dubbelop} {\tiny \dubbelop} {\tiny \dubbelop} \\
\midrule
& \multicolumn{5}{c}{\textit{informational}} \\
\midrule
DBGD (linear) & 0.411 {\tiny (0.036)} & 0.631 {\tiny (0.036)} & 0.299 {\tiny (0.017)} & 0.620 {\tiny (0.035)} & 0.360 {\tiny (0.028)} \\
DBGD (neural) & 0.383 {\tiny (0.047)} & 0.595 {\tiny (0.053)} & 0.276 {\tiny (0.033)} & 0.603 {\tiny (0.040)} & 0.316 {\tiny (0.057)} \\
MGD (linear) & 0.406 {\tiny (0.021)} & 0.647 {\tiny (0.036)} & 0.318 {\tiny (0.003)} & 0.676 {\tiny (0.043)} & 0.387 {\tiny (0.005)} \\
Pairwise (linear) & 0.478 {\tiny (0.022)} & 0.677 {\tiny (0.018)} & 0.311 {\tiny (0.003)} & 0.690 {\tiny (0.006)} & 0.183 {\tiny (0.001)} \\
PDGD (linear) & \bf 0.487 {\tiny (0.021)} {\tiny \dubbelop} {\tiny \dubbelop} {\tiny \dubbelop} {\tiny \dubbelop} & \bf 0.690 {\tiny (0.022)} {\tiny \dubbelop} {\tiny \dubbelop} {\tiny \dubbelop} {\tiny \dubbelop} & \bf 0.368 {\tiny (0.025)} {\tiny \dubbelop} {\tiny \dubbelop} {\tiny \dubbelop} {\tiny \dubbelop} & \bf 0.713 {\tiny (0.008)} {\tiny \dubbelop} {\tiny \dubbelop} {\tiny \dubbelop} {\tiny \dubbelop} & \bf 0.532 {\tiny (0.010)} {\tiny \dubbelop} {\tiny \dubbelop} {\tiny \dubbelop} {\tiny \dubbelop} \\
PDGD (neural) & 0.483 {\tiny (0.022)} {\tiny \dubbelop} {\tiny \dubbelop} {\tiny \dubbelop} \hphantom{\tiny \dubbelneer} & 0.686 {\tiny (0.022)} {\tiny \dubbelop} {\tiny \dubbelop} {\tiny \dubbelop} {\tiny \dubbelop} & 0.355 {\tiny (0.021)} {\tiny \dubbelop} {\tiny \dubbelop} {\tiny \dubbelop} {\tiny \dubbelop} & 0.709 {\tiny (0.009)} {\tiny \dubbelop} {\tiny \dubbelop} {\tiny \dubbelop} {\tiny \dubbelop} & 0.525 {\tiny (0.012)} {\tiny \dubbelop} {\tiny \dubbelop} {\tiny \dubbelop} {\tiny \dubbelop} \\
\bottomrule
\end{tabular*}

%% file: 05-pdgd/tables/online.tex
\begin{tabular*}{\textwidth}{@{\extracolsep{\fill} } l  l l l l l  }
\toprule
 & { \small \textbf{MQ2007}}  & { \small \textbf{MQ2008}}  & { \small \textbf{MSLR-WEB10k}}  & { \small \textbf{Yahoo}}  & { \small \textbf{istella}} \\
\midrule
& \multicolumn{5}{c}{\textit{perfect}} \\
\midrule
DBGD (linear) & 675.7 {\tiny (21.8)} & 843.6 {\tiny (40.8)} & 533.6 {\tiny (15.6)} & 1159.3 {\tiny (31.6)} & 589.9 {\tiny (19.2)} \\
DBGD (neural) & 602.7 {\tiny (58.1)} & 776.9 {\tiny (67.4)} & 481.2 {\tiny (53.0)} & 1135.7 {\tiny (41.3)} & 494.3 {\tiny (60.5)} \\
MGD (linear) & 689.6 {\tiny (15.3)} & 858.6 {\tiny (40.6)} & 558.7 {\tiny (6.4)} & 1203.9 {\tiny (9.9)} & 670.9 {\tiny (8.6)} \\
Pairwise (linear) & 458.4 {\tiny (13.3)} & 616.6 {\tiny (25.8)} & 345.3 {\tiny (4.6)} & 1027.2 {\tiny (9.2)} & 64.5 {\tiny (2.1)} \\
PDGD (linear) & \bf 797.3 {\tiny (17.3)} {\tiny \dubbelop} {\tiny \dubbelop} {\tiny \dubbelop} {\tiny \dubbelop} & \bf 959.7 {\tiny (43.4)} {\tiny \dubbelop} {\tiny \dubbelop} {\tiny \dubbelop} {\tiny \dubbelop} & \bf 691.4 {\tiny (12.3)} {\tiny \dubbelop} {\tiny \dubbelop} {\tiny \dubbelop} {\tiny \dubbelop} & \bf 1360.3 {\tiny (10.8)} {\tiny \dubbelop} {\tiny \dubbelop} {\tiny \dubbelop} {\tiny \dubbelop} & \bf 957.5 {\tiny (9.4)} {\tiny \dubbelop} {\tiny \dubbelop} {\tiny \dubbelop} {\tiny \dubbelop} \\
PDGD (neural) & 743.7 {\tiny (18.8)} {\tiny \dubbelop} {\tiny \dubbelop} {\tiny \dubbelop} {\tiny \dubbelop} & 925.4 {\tiny (43.3)} {\tiny \dubbelop} {\tiny \dubbelop} {\tiny \dubbelop} {\tiny \dubbelop} & 619.2 {\tiny (13.6)} {\tiny \dubbelop} {\tiny \dubbelop} {\tiny \dubbelop} {\tiny \dubbelop} & 1319.6 {\tiny (10.1)} {\tiny \dubbelop} {\tiny \dubbelop} {\tiny \dubbelop} {\tiny \dubbelop} & 834.0 {\tiny (22.2)} {\tiny \dubbelop} {\tiny \dubbelop} {\tiny \dubbelop} {\tiny \dubbelop} \\
\midrule
& \multicolumn{5}{c}{\textit{navigational}} \\
\midrule
DBGD (linear) & 638.6 {\tiny (29.7)} & 816.9 {\tiny (42.0)} & 508.2 {\tiny (21.6)} & 1129.9 {\tiny (32.2)} & 538.2 {\tiny (29.0)} \\
DBGD (neural) & 573.7 {\tiny (68.4)} & 740.3 {\tiny (69.7)} & 465.8 {\tiny (52.0)} & 1116.0 {\tiny (45.7)} & 414.3 {\tiny (96.2)} \\
MGD (linear) & 635.9 {\tiny (14.7)} & 824.5 {\tiny (34.0)} & 538.1 {\tiny (7.6)} & 1181.7 {\tiny (20.0)} & 593.2 {\tiny (9.7)} \\
Pairwise (linear) & 459.9 {\tiny (12.9)} & 618.6 {\tiny (25.2)} & 347.3 {\tiny (5.4)} & 1031.2 {\tiny (9.0)} & 72.6 {\tiny (2.2)} \\
PDGD (linear) & \bf 703.0 {\tiny (17.9)} {\tiny \dubbelop} {\tiny \dubbelop} {\tiny \dubbelop} {\tiny \dubbelop} & \bf 903.1 {\tiny (40.7)} {\tiny \dubbelop} {\tiny \dubbelop} {\tiny \dubbelop} {\tiny \dubbelop} & \bf 578.1 {\tiny (16.0)} {\tiny \dubbelop} {\tiny \dubbelop} {\tiny \dubbelop} {\tiny \dubbelop} & \bf 1298.4 {\tiny (33.4)} {\tiny \dubbelop} {\tiny \dubbelop} {\tiny \dubbelop} {\tiny \dubbelop} & \bf 704.1 {\tiny (33.5)} {\tiny \dubbelop} {\tiny \dubbelop} {\tiny \dubbelop} {\tiny \dubbelop} \\
PDGD (neural) & 560.9 {\tiny (14.6)} {\tiny \dubbelneer} {\tiny \enkelneer} {\tiny \dubbelneer} {\tiny \dubbelop} & 788.7 {\tiny (38.5)} {\tiny \dubbelneer} {\tiny \dubbelop} {\tiny \dubbelneer} {\tiny \dubbelop} & 448.1 {\tiny (12.3)} {\tiny \dubbelneer} {\tiny \dubbelneer} {\tiny \dubbelneer} {\tiny \dubbelop} & 1176.1 {\tiny (17.0)} {\tiny \dubbelop} {\tiny \dubbelop} {\tiny \enkelneer} {\tiny \dubbelop} & 390.2 {\tiny (35.1)} {\tiny \dubbelneer} {\tiny \dubbelneer} {\tiny \dubbelneer} {\tiny \dubbelop} \\
\midrule
& \multicolumn{5}{c}{\textit{informational}} \\
\midrule
DBGD (linear) & 584.2 {\tiny (41.1)} & 757.4 {\tiny (56.9)} & 477.2 {\tiny (32.2)} & 1110.0 {\tiny (37.0)} & 436.8 {\tiny (57.4)} \\
DBGD (neural) & 550.8 {\tiny (75.7)} & 720.9 {\tiny (79.0)} & 444.7 {\tiny (60.9)} & 1091.2 {\tiny (48.6)} & 322.9 {\tiny (121.0)} \\
MGD (linear) & 618.8 {\tiny (21.7)} & 815.1 {\tiny (44.5)} & 540.0 {\tiny (7.7)} & 1159.1 {\tiny (40.0)} & 581.8 {\tiny (10.7)} \\
Pairwise (linear) & 462.6 {\tiny (14.4)} & 619.6 {\tiny (25.0)} & 349.7 {\tiny (6.6)} & 1034.1 {\tiny (9.0)} & 77.0 {\tiny (2.4)} \\
PDGD (linear) & \bf 704.8 {\tiny (30.5)} {\tiny \dubbelop} {\tiny \dubbelop} {\tiny \dubbelop} {\tiny \dubbelop} & \bf 907.9 {\tiny (42.0)} {\tiny \dubbelop} {\tiny \dubbelop} {\tiny \dubbelop} {\tiny \dubbelop} & \bf 567.3 {\tiny (36.5)} {\tiny \dubbelop} {\tiny \dubbelop} {\tiny \dubbelop} {\tiny \dubbelop} & \bf 1266.7 {\tiny (50.0)} {\tiny \dubbelop} {\tiny \dubbelop} {\tiny \dubbelop} {\tiny \dubbelop} & \bf 731.5 {\tiny (80.0)} {\tiny \dubbelop} {\tiny \dubbelop} {\tiny \dubbelop} {\tiny \dubbelop} \\
PDGD (neural) & 594.6 {\tiny (23.0)} {\tiny \enkelop} {\tiny \dubbelop} {\tiny \dubbelneer} {\tiny \dubbelop} & 818.3 {\tiny (39.6)} {\tiny \dubbelop} {\tiny \dubbelop} \hphantom{\tiny \dubbelneer} {\tiny \dubbelop} & 470.1 {\tiny (19.4)} {\tiny \enkelneer} {\tiny \dubbelop} {\tiny \dubbelneer} {\tiny \dubbelop} & 1178.1 {\tiny (22.8)} {\tiny \dubbelop} {\tiny \dubbelop} {\tiny \dubbelop} {\tiny \dubbelop} & 484.3 {\tiny (64.8)} {\tiny \dubbelop} {\tiny \dubbelop} {\tiny \dubbelneer} {\tiny \dubbelop} \\
\bottomrule
\end{tabular*}

%% file: 05-pdgd/06-conclusion.tex
\section{Conclusion}
\label{sec:conclusion}

In this chapter, we have introduced a novel \ac{OLTR} method: \ac{\OurMethod} that estimates its gradient using inferred pairwise document preferences.
In contrast with previous \ac{OLTR} approaches \ac{\OurMethod} does not rely on online evaluation to update its model.
Instead after each user interaction it infers preferences between document pairs.
Subsequently, it constructs a pairwise gradient that updates the ranking model according to these preferences.

We have proven that this gradient is unbiased w.r.t.\ user preferences, that is, if there is a preference between a document pair, then in expectation the gradient will update the model to meet this preference.
Furthermore, our experimental results show that \ac{\OurMethod} learns faster and converges at a higher performance level than existing \ac{OLTR} methods.
Thus, it provides better performance in the short and long term, leading to an improved user experience during training as well.
On top of that, \ac{\OurMethod} is also applicable to any differentiable ranking model, in our experiments a linear and a neural network were optimized effectively.
Both reached significant improvements over \ac{DBGD} and \ac{MGD} in performance at convergence.
In conclusion, the novel unbiased \ac{\OurMethod} algorithm provides better performance than existing methods in terms of convergence and user experience.
Unlike the previous state-of-the-art, it can be applied to any differentiable ranking model.

We can now answer thesis research question \ref{thesisrq:pdgd} positively: \ac{OLTR} is possible without relying on model-sampling and online evaluation.
Moreover, our results shows that using \ac{\OurMethod} instead leads to much higher performance, and is much more effective at optimizing non-linear models.

Future research could consider the regret bounds of \ac{\OurMethod}; these could give further insights into why it outperforms \ac{DBGD} based methods.
Furthermore, while we proved the unbiasedness of our method w.r.t.\ document pair preferences, the expected gradient weighs document pairs differently.
Offline \ac{LTR} methods like Lambda\-MART~\cite{burges2010ranknet} use a weighted pairwise loss to create a listwise method that directly optimizes \ac{IR} metrics.
However, in the online setting there is no metric that is directly optimized. 
Instead, future work could see if different weighing approaches are more in line with user preferences.
Another obvious avenue for future research is to explore the effectiveness of different ranking models in the online setting.
There is a large collection of research in ranking models in offline \ac{LTR}, with the introduction of \ac{\OurMethod} such an extensive exploration in models is now also possible in \ac{OLTR}.

Based on the big difference in observed performance between \ac{PDGD} and \ac{DBGD}, Chapter~\ref{chapter:03-oltr-comparison} will further extend this comparison to more extreme experimental conditions.
Furthermore, Chapter~\ref{chapter:06-onlinecounterltr} will also consider the performance of \ac{PDGD} and compare it with methods inspired by counterfactual \ac{LTR}.
Additionally, Chapter~\ref{chapter:06-onlinecounterltr} will consider applying \ac{PDGD} as a counterfactual method and without debiasing weights, and finds that in both these scenarios this leads to biased convergence.

%% file: 05-pdgd/notation.tex
\section{Notation Reference for Chapter~\ref{chapter:02-pdgd}}
\label{notation:02-pdgd}

\begin{center}
\begin{tabular}{l l}
 \toprule
\bf Notation  & \bf Description \\
\midrule
$q$ & a user-issued query \\
$d$, $d_k$, $d_l$ & document\\
$\mathbf{d}$ & feature representation of a query-document pair \\
$D$ & set of documents\\
$R$ & ranked list \\
$R^*$ & the reversed pair ranking $R^*(d_k, d_l, R)$ \\
$R_i$ & document placed at rank $i$ \\
$\rho$ & preference pair weighting function \\
$\theta$ & parameters of the ranking model\\
$f_\theta(\cdot)$ & ranking model with parameters $\theta$ \\
$f(\mathbf{d}_k)$ & ranking score for a document from model \\
$\mathit{click}(d)$ & a click on document $d$ \\
$d_k =_\mathit{rel} d_l$ & two documents equally preferred by users \\
$d_k >_\mathit{rel} d_l$ & a user preference between two documents \\
$d_k >_\mathbf{c} d_l$ & document preference inferred from clicks \\
\bottomrule
\end{tabular}
\end{center}

%% file: 06-oltrcomparison/main.tex
\chapter{A Critical Comparison of Online Learning to Rank Methods}
\label{chapter:03-oltr-comparison}

\footnote[]{This chapter was published as~\citep{oosterhuis2019optimizing}.
Appendix~\ref{notation:03-oltr-comparison} gives a reference for the notation used in this chapter.
}

\ac{OLTR} methods optimize ranking models by directly interacting with users, which allows them to be very efficient and responsive.
All \ac{OLTR} methods introduced during the past decade have extended on the original \ac{OLTR} method: \ac{DBGD}.
In Chapter~\ref{chapter:02-pdgd}, a fundamentally different approach was introduced with the \ac{PDGD} algorithm.
The empirical comparisons in Chapter~\ref{chapter:02-pdgd} suggested that \ac{PDGD} converges at much higher levels of performance and learns considerably faster than \ac{DBGD}-based methods.
In contrast, \ac{DBGD} appeared unable to converge on the optimal model in scenarios with little noise or bias.
Furthermore, it seemed \ac{DBGD} is not effective at optimizing non-linear models.
These observations are quite surprising and prompted us to further investigate \ac{DBGD}.
As a result, this Chapter will address the thesis research question:
\begin{itemize}
\item[\ref{thesisrq:dbgd}] \emph{Are \acs{DBGD} \acl{LTR} methods reliable in terms of theoretical soundness and empirical performance?}
\end{itemize}
In this chapter, we investigate whether the previous conclusions about the \ac{PDGD} and \ac{DBGD} comparison generalize from ideal to worst-case circumstances.
We do so in two ways.
First, we compare the theoretical properties of \ac{PDGD} and \ac{DBGD}, by taking a critical look at previously proven properties in the context of ranking.
Second, we estimate an upper and lower bound on the performance of methods by simulating both \emph{ideal} user behavior and extremely \emph{difficult} behavior, i.e., almost-random non-cascading user models.
Our findings show that the theoretical bounds of \ac{DBGD} do not apply to any common ranking model and, furthermore, that the performance of \ac{DBGD} is substantially worse than \ac{PDGD} in both ideal and worst-case circumstances.
These results reproduce previously published findings about the relative performance of \ac{PDGD} vs.\ \ac{DBGD} and generalize them to extremely noisy and non-cascading circumstances.
Overall they show that \ac{DBGD} is a very flawed method for \ac{OLTR} both in terms of theoretical guarantees and performance.

\input{06-oltrcomparison/01-intro}

\input{06-oltrcomparison/02-related}

\input{06-oltrcomparison/03-ranking-models}
\input{06-oltrcomparison/04-DBGD}

\input{06-oltrcomparison/05-PDGD}

\input{06-oltrcomparison/06-experiments}

\input{06-oltrcomparison/07-results}
\input{06-oltrcomparison/08-conclusion}
\begin{subappendices}
\input{06-oltrcomparison/notation}
\end{subappendices}

%% file: 06-oltrcomparison/01-intro.tex
\section{Introduction}
\label{sec:ecir:intro}
\ac{LTR} plays a vital role in information retrieval. It allows us to optimize models that combine hundreds of signals to produce rankings, thereby making large collections of documents accessible to users through effective search and recommendation.
Traditionally, \ac{LTR} has been approached as a supervised learning problem, where annotated datasets provide human judgements indicating relevance.
Over the years, many limitations of such datasets have become apparent: they are costly to produce~\cite{Chapelle2011,qin2013introducing} and actual users often disagree with the relevance annotations~\cite{sanderson2010}.
As an alternative, research into \ac{LTR} approaches that learn from user behavior has increased.
By learning from the implicit feedback in user behavior, users' true preferences can potentially be learned.
However, such methods must deal with the noise and biases that are abundant in user interactions~\cite{yue2010beyond}.
Roughly speaking, there are two approaches to \ac{LTR} from user interactions: learning from historical interactions and \acf{OLTR}\acused{OLTR}.
Learning from historical data allows for optimization without gathering new data~\cite{joachims2017unbiased}, though it does require good models of the biases in logged user interactions~\cite{chuklin-click-2015}.
In contrast, \ac{OLTR} methods learn by interacting with the user, thus they gather their own learning data.
As a result, these methods can adapt instantly and are potentially much more responsive than methods that use historical data.

\acf{DBGD}~\cite{yue2009interactively} is the most prevalent \ac{OLTR} method; it has served as the basis of the field for the past decade.
\ac{DBGD} samples variants of its ranking model, and compares them using interleaving to find improvements~\cite{hofmann2011probabilistic,radlinski2013optimized}.
Subsequent work in \ac{OLTR} has extended on this approach~\cite{hofmann2013reusing,schuth2016mgd,wang2018efficient}.
In Chapter~\ref{chapter:02-pdgd}, the first alternative approach to \ac{DBGD} was introduced with \acfi{PDGD}\acused{PDGD}~\cite{oosterhuis2018differentiable}.
\ac{PDGD} estimates a pairwise gradient that is reweighed to be unbiased w.r.t.\ users' document pair preferences.
Chapter~\ref{chapter:02-pdgd} showed considerable improvements over \ac{DBGD} under simulated user behavior~\cite{oosterhuis2019optimizing}: a substantially higher point of performance at convergence and a much faster learning speed.
The results in Chapter~\ref{chapter:02-pdgd} are based on simulations using low-noise cascading click models.
The pairwise assumption that \ac{PDGD} makes, namely, that all documents preceding a clicked document were observed by the user, is always correct in these circumstances, thus potentially giving it an unfair advantage over \ac{DBGD}.
Furthermore, the low level of noise presents a close-to-ideal situation, and it is unclear whether the findings in Chapter~\ref{chapter:02-pdgd} generalize to less perfect circumstances.

In this chapter, we contrast \ac{PDGD} and \ac{DBGD}.
Prior to an experimental comparison, we determine whether there is a theoretical advantage of \ac{DBGD} over \ac{PDGD} and examine the regret bounds of \ac{DBGD} for ranking problems.
We then investigate whether the benefits of \ac{PDGD} over \ac{DBGD} reported in Chapter~\ref{chapter:02-pdgd} generalize to circumstances ranging from ideal to worst-case.
We simulate circumstances that are perfect for both methods -- behavior without noise or position-bias --  and circumstances that are the worst possible scenario -- almost-random, extremely-biased, non-cascading behavior.
These settings provide estimates of upper and lower bounds on performance, and indicate how well previous comparisons generalize to different circumstances.
Additionally, we introduce a version of \ac{DBGD} that is provided with an oracle interleaving method; its performance shows us the maximum performance \ac{DBGD} could reach from hypothetical extensions.

In summary, we map thesis research question \ref{thesisrq:dbgd} into the following more fine-grained research questions:
\begin{enumerate}[align=left, label={\bf RQ4.\arabic*}, leftmargin=*]
    \item Do the regret bounds of \ac{DBGD} provide a benefit over \ac{PDGD}?\label{rq:regret}
    \item Do the advantages of \ac{PDGD} over \ac{DBGD} observed in Chapter~\ref{chapter:02-pdgd} generalize to extreme levels of noise and bias? \label{rq:noisebias}
    \item Is the performance of \ac{PDGD} reproducible under non-cascading user behavior? \label{rq:cascading}
\end{enumerate}

%% file: 06-oltrcomparison/02-related.tex
\section{Related Work}
\label{sec:relatedwork}

This section provides a brief overview of traditional \ac{LTR} (Section~\ref{sec:ltr:annotated}), of \ac{LTR} from historical interactions (Section~\ref{sec:ltr:history}), and \ac{OLTR} (Section~\ref{sec:ltr:online}).

\subsection{Learning to rank from annotated datasets}
\label{sec:ltr:annotated}

Traditionally, \ac{LTR} has been approached as a supervised problem; in the context of \ac{OLTR} this approach is often referred to as \emph{offline} \ac{LTR}.
It requires a dataset containing relevance annotations of query-document pairs, after which a variety of methods can be applied \cite{liu2009learning}.
The limitations of offline \ac{LTR} mainly come from obtaining such annotations.
The costs of gathering annotations are high as it is both time-consuming and expensive \cite{Chapelle2011,qin2013introducing}.
Furthermore, annotators cannot judge for very specific users, i.e., gathering data for personalization problems is infeasible.
Moreover, for certain applications it would be unethical to annotate items, e.g., for search in personal emails or documents \cite{wang2016learning}.
Additionally, annotations are stationary and cannot account for (perceived) relevance changes~\cite{dumais-web-2010,lefortier-online-2014,vakkari-changes-2000}.
Most importantly, though, annotations are not necessarily aligned with user preferences; judges often interpret queries differently from actual users~\cite{sanderson2010}.
As a result, there has been a shift of interest towards \ac{LTR} approaches that do not require annotated data.

\subsection{Learning to rank from historical interactions}
\label{sec:ltr:history}

The idea of \ac{LTR} from user interactions is long-established; one of the earliest examples is the original pairwise \ac{LTR} approach \cite{Joachims2002}.
This approach uses historical click-through interactions from a search engine and considers clicks as indications of relevance.
Though very influential and quite effective, this approach ignores the \emph{noise} and \emph{biases} inherent in user interactions.
Noise, i.e., any user interaction that does not reflect the user's true preference, occurs frequently, since many clicks happen for unexpected reasons~\cite{sanderson2010}.
Biases are systematic forms of noise that occur due to factors other than relevance.
For instance, interactions will only involve displayed documents, resulting in selection bias~\cite{wang2016learning}.
Another important form of bias in \ac{LTR} is position bias, which occurs because users are less likely to consider documents that are ranked lower~\cite{yue2010beyond}.
Thus, to effectively learn true preferences from user interactions, a \ac{LTR} method should be robust to noise and handle biases correctly.

In recent years counterfactual \ac{LTR} methods have been introduced that correct for some of the bias in user interactions.
Such methods use inverse propensity scoring to account for the probability that a user observed a ranking position~\cite{joachims2017unbiased}.
Thus, clicks on positions that are observed less often due to position bias will have greater weight to account for that difference.
However, the position bias must be learned and estimated somewhat accurately~\cite{ai2018unbiased}.
On the other side of the spectrum are click models, which attempt to model user behavior completely~\cite{chuklin-click-2015}.
By predicting behavior accurately, the effect of relevance on user behavior can also be estimated~\cite{borisov2016neural,wang2016learning}.

An advantage of these approaches over \ac{OLTR} is that they only require historical data and thus no new data has to be gathered. However, unlike \ac{OLTR}, they do require a fairly accurate user model, and thus they cannot be applied in cold-start situations.

\subsection{Online learning to rank}
\label{sec:ltr:online}

\ac{OLTR} differs from the approaches listed above because its methods intervene in the search experience.
They have control over what results are displayed, and can learn from their interactions instantly.
Thus, the online approach performs \ac{LTR} by interacting with users directly~\cite{yue2009interactively}.
Similar to \ac{LTR} methods that learn from historical interaction data, \ac{OLTR} methods have the potential to learn the true user preferences.
However, they also have to deal with the noise and biases that come with user interactions.
Another advantage of \ac{OLTR} is that the methods are very responsive, as they can apply their learned behavior instantly.
Conversely, this also brings a danger as an online method that learns incorrect preferences can also worsen the experience immediately.
Thus, it is important that \ac{OLTR} methods are able to learn reliably in spite of noise and biases.
Thus, \ac{OLTR} methods have a two-fold task: they have to simultaneously present rankings that provide a good user experience \emph{and} learn from user interactions with the presented rankings.

The original \ac{OLTR} method is \acf{DBGD}; it approaches optimization as a dueling bandit problem~\cite{yue2009interactively}.
This approach requires an online comparison method that can compare two rankers w.r.t. user preferences; traditionally, \ac{DBGD} methods use interleaving.
Interleaving methods take the rankings produced by two rankers and combine them in a single result list, which is then displayed to users.
From a large number of clicks on the presented list the interleaving methods can reliably infer a preference between the two rankers~\cite{hofmann2011probabilistic,radlinski2013optimized}.
At each timestep, \ac{DBGD} samples a candidate model, i.e., a slight variation of its current model, and compares the current and candidate models using interleaving.
If a preference for the candidate is inferred, the current model is updated towards the candidate slightly.
By doing so, \ac{DBGD} will update its model continuously and should oscillate towards an inferred optimum.
Section~\ref{sec:DBGD} provides a complete description of the \ac{DBGD} algorithm.

Virtually all work in \ac{OLTR} in the decade since the introduction of \ac{DBGD} has used \ac{DBGD} as a basis.
A straightforward extension comes in the form of Multileave Gradient Descent~\cite{schuth2016mgd} which compares a large number of candidates per interaction~\cite{oosterhuis2017sensitive,schuth2015probabilistic,Schuth2014a}.
This leads to a much faster learning process, though in the long run this method does not seem to improve the point of convergence.

One of the earliest extensions of \ac{DBGD} proposed a method for reusing historical interactions to guide exploration for faster learning \cite{hofmann2013reusing}.
While the initial results showed great improvements~\cite{hofmann2013reusing}, later work showed performance drastically decreasing in the long term due to bias introduced by the historical data~\cite{oosterhuis2016probabilistic}.
Unfortunately, \ac{OLTR} work that continued this historical approach~\cite{wang2018efficient} also only considered short term results; moreover, the results of some work~\cite{zhao2016constructing} are not based on held-out data.
As a result, we do not know whether these extensions provide decent long-term performance and it is unclear whether the findings of these studies generalize to more realistic settings.

In Chapter~\ref{chapter:02-pdgd}, an inherently different approach to \ac{OLTR} was introduced with \ac{PDGD}~\cite{oosterhuis2018differentiable}.
\ac{PDGD} interprets its ranking model as a distribution over documents; it estimates a pairwise gradient from user interactions with sampled rankings.
This gradient is differentiable, allowing for non-linear models like neural networks to be optimized, something \ac{DBGD} is ineffective at~\cite{oosterhuis2017balancing,oosterhuis2018differentiable}.
Section~\ref{sec:PDGD} provides a detailed description of \ac{PDGD}.
In the chapter in which we introduced \ac{PDGD} (Chapter~\ref{chapter:02-pdgd}), we claim that it provides substantial improvements over \ac{DBGD}.
However, those claims are based on cascading click models with low levels of noise.
This is problematic because \ac{PDGD} assumes a cascading user, and could thus have an unfair advantage in this setting.
Furthermore, it is unclear whether \ac{DBGD} with a perfect interleaving method could still improve over \ac{PDGD}.
Lastly, \ac{DBGD} has proven regret bounds while \ac{PDGD} has no such guarantees.

In this chapter, we clear up these questions about the relative strengths of \ac{DBGD} and \ac{PDGD} by comparing the two methods under non-cascading, high-noise click models.
Additionally, by providing \ac{DBGD} with an oracle comparison method, its hypothetical maximum performance can be measured; thus, we can study whether an improvement over \ac{PDGD} is hypothetically possible.
Finally, a brief analysis of the theoretical regret bounds of \ac{DBGD} shows that they do not apply to any common ranking model, therefore hardly providing a guaranteed advantage over \ac{PDGD}.

%% file: 06-oltrcomparison/04-DBGD.tex
\section{Dueling Bandit Gradient Descent}
\label{sec:DBGD}

This section describes the \ac{DBGD} algorithm in detail, before discussing the regret bounds of the algorithm.

\begin{algorithm}[t]
\caption{\acf{DBGD}.} 
\label{alg:dbgd}
\begin{algorithmic}[1]
\STATE \textbf{Input}: initial weights: $\theta_1$; unit: $u$; learning rate $\eta$.  \label{line:dbgd:initmodel}
\FOR{$t \leftarrow  1 \ldots \infty$ }
	\STATE $q_t \leftarrow \text{receive\_query}(t)$\hfill \textit{\small //  obtain a query from a user} \label{line:dbgd:query}
	\STATE $\theta_t^{c} \gets \theta_t +  \text{sample\_from\_unit\_sphere}(u)$  \hfill \textit{\small //  create candidate ranker} \label{line:dbgd:candidate}
	\STATE $R_t \leftarrow \text{get\_ranking}(\theta_t, D_{q_t})$  \hfill \textit{\small //  get current ranker ranking} \label{line:dbgd:ranking}
	\STATE $R_t^c \leftarrow \text{get\_ranking}(\theta_t^c, D_{q_t})$  \hfill \textit{\small //  get candidate ranker ranking} \label{line:dbgd:candidateranking}
	\STATE $I_t \leftarrow \text{interleave}(R_t, R_t^c)$ \hfill \textit{\small //  interleave both rankings} \label{line:dbgd:interleave}
	\STATE $\mathbf{c}_t \leftarrow \text{display\_to\_user}(I_t)$ \hfill \textit{\small //  displayed interleaved list, record clicks} \label{line:dbgd:display}
	\IF{$\text{preference\_for\_candidate}(I_t, \mathbf{c}_t, R_t, R_t^c)$}
    		\STATE $\theta_{t+1} \leftarrow \theta_t + \eta (\theta^c_t - \theta_t)$ \hfill \textit{\small //  update model towards candidate} \label{line:dbgd:update}
	\ELSE
		\STATE $\theta_{t+1} \leftarrow \theta_t$ \hfill \textit{\small //  no update} \label{line:dbgd:noupdate}
	\ENDIF
\ENDFOR
\end{algorithmic}
\end{algorithm}

\subsection{The \acl{DBGD} method}

The \ac{DBGD} algorithm~\cite{yue2009interactively} describes an indefinite loop that aims to improve a ranking model at each step; Algorithm~\ref{alg:dbgd} provides a formal description.
The algorithm starts a given model with weights $\theta_1$ (Line~\ref{line:dbgd:initmodel}); then it waits for a user-submitted query (Line~\ref{line:dbgd:query}).
At this point a candidate ranker is sampled from the unit sphere around the current model (Line~\ref{line:dbgd:candidate}), and the current and candidate model both produce a ranking for the current query (Line~\ref{line:dbgd:ranking}~and~\ref{line:dbgd:candidateranking}).
These rankings are interleaved (Line~\ref{line:dbgd:interleave}) and displayed to the user (Line~\ref{line:dbgd:display}).
If the interleaving method infers a preference for the candidate ranker from subsequent user interactions the current model is updated towards the candidate (Line~\ref{line:dbgd:update}), otherwise no update is performed (Line~\ref{line:dbgd:noupdate}).
Thus, the model optimized by \ac{DBGD} should converge and oscillate towards an optimum.

\subsection{Regret bounds of \acl{DBGD}}
\label{subsection:regretbounds}
 
Unlike \ac{PDGD}, \ac{DBGD} has proven regret bounds~\cite{yue2009interactively}, potentially providing an advantage in the form of theoretical guarantees.
In this section we answer \ref{rq:regret} by critically looking at the assumptions which form the basis of \ac{DBGD}'s proven regret bounds.

The original \ac{DBGD} paper~\cite{yue2009interactively} proved a sublinear regret under several assumptions.
\ac{DBGD} works with the parameterized space of ranking functions $\mathcal{W}$, that is, every $\theta \in \mathcal{W}$ is a different set of parameters for a ranking function.
For this chapter we will only consider deterministic linear models because all existing \ac{OLTR} work has dealt with them \cite{hofmann2013reusing,hofmann11:balancing,oosterhuis2018differentiable,oosterhuis2016probabilistic,schuth2016mgd,wang2018efficient,yue2009interactively,zhao2016constructing}.
But we note that the proof is easily extendable to neural networks where the output is a monotonic function applied to a linear combination of the last layer.
Then there is assumed to be a concave utility function $u : \mathcal{W} \rightarrow \mathbb{R}$; since this function is concave, there should only be a single instance of weights that are optimal $\theta^*$.
Furthermore, this utility function is assumed to be L-Lipschitz smooth:
\begin{align}
\exists L \in \mathbb{R},\quad \forall (\theta_a, \theta_b) \in \mathcal{W},\quad |u(\theta_a) - u(\theta_b)| < L \|\theta_a - \theta_b \|.
\end{align}
We will show that these assumptions are \emph{incorrect}: there is an infinite number of optimal weights, and the utility function $u$ cannot be L-Lipschitz smooth.
Our proof relies on two assumptions that avoid cases where the ranking problem is trivial.
First, the zero ranker is not the optimal model:
\begin{align}
\theta^* \not= \mathbf{0}.
\end{align}
Second, there should be at least two models with different utility values:
\begin{align}
\exists (\theta, \theta') \in \mathcal{W},\quad u(\theta) \not= u(\theta').
\end{align}
We will start by defining the set of rankings a model $f(\cdot, \theta)$ will produce as:
\begin{align}
\mathcal{R}_D(f(\cdot, \theta))
=
\{ R \mid \forall (d, d') \in D, [f(d, \theta) > f(d', \theta) \rightarrow d \succ_R d']\} .
\end{align}
It is easy to see that multiplying a model with a positive scalar $\alpha > 0$ will not affect this set:
\begin{align}
\forall \alpha \in \mathbb{R}_{>0},\quad  \mathcal{R}_D(f(\cdot, \theta))
= \mathcal{R}_D(\alpha f(\cdot, \theta)).
\end{align}
Consequently, the utility of both functions will be equal:
\begin{align}
\forall \alpha \in \mathbb{R}_{>0},\quad  u(f(\cdot, \theta))
= u(\alpha f(\cdot, \theta)).
\end{align}
For linear models scaling weights has the same effect: $\alpha f(\cdot, \theta) = f(\cdot, \alpha \theta)$.
Thus, the first assumption cannot be true since for any optimal model $f(\cdot, \theta^*)$ there is an infinite set of equally optimal models: $\{f(\cdot, \alpha \theta^*) \mid \alpha \in \mathbb{R}_{>0}\}$.

Then, regarding L-Lipschitz smoothness, using any positive scaling factor:
\begin{align}
\forall \alpha \in \mathbb{R}_{>0},\quad& |u(\theta_a) - u(\theta_b)| = |u(\alpha\theta_a) - u(\alpha\theta_b)|, \\
\forall \alpha \in \mathbb{R}_{>0},\quad& \| \alpha\theta_a - \alpha\theta_b \| = \alpha \| \theta_a - \theta_b \|.
\end{align}
Thus the smoothness assumption can be rewritten as:
\begin{align}
\exists L \in \mathbb{R},\quad \forall \alpha \in \mathbb{R}_{>0},\quad \forall (\theta_a, \theta_b) \in \mathcal{W},\quad |u(\theta_a) - u(\theta_b)| < \alpha L\| \theta_a - \theta_b \|.
\end{align}
However, there is always an infinite number of values for $\alpha$ small enough to break the assumption.
Therefore, we conclude that a concave L-Lipschitz smooth utility function can never exist for a deterministic linear ranking model, thus the proof for the regret bounds is not applicable when using deterministic linear models.

Consequently, the regret bounds of \ac{DBGD} do not apply to the ranking problems in previous work.
One may consider other models (e.g., spherical coordinate based models or stochastic ranking models), however this still means that for the simplest and most common ranking problems there are no proven regret bounds.
As a result, we answer \ref{rq:regret} negatively, the regret bounds of \ac{DBGD} do not provide a benefit over \ac{PDGD} for the ranking problems in \ac{LTR}.

%% file: 06-oltrcomparison/05-PDGD.tex
\section{Pairwise Differentiable Gradient Descent}
\label{sec:PDGD}

The \acf{PDGD}~\cite{oosterhuis2018differentiable} algorithm is formally described in Algorithm~\ref{alg:pdgd}.
\ac{PDGD} interprets a ranking function $f(\cdot, \theta)$ as a probability distribution over documents by applying a Plackett-Luce model:
\begin{align}
P(d | D, \theta) = \frac{e^{f(d,\theta)}}{\sum_{d' \in D} e^{f(d',\theta)}}. \label{eq:docprob}
\end{align}
First, the algorithm waits for a user query (Line~\ref{line:pdgd:query}), then a ranking $R$ is created by sampling documents without replacement (Line~\ref{line:pdgd:samplelist}).
Then \ac{PDGD} observes clicks from the user and infers pairwise document preferences from them.
All documents preceding  a clicked document and the first succeeding one are assumed to be observed by the user.
Preferences between clicked and unclicked observed documents are inferred by \ac{PDGD}; this is a long-standing assumption in pairwise \ac{LTR} \cite{Joachims2002}.
We denote an \emph{inferred} preference between documents as $d_i \succ_{\mathbf{c}} d_j$, and the probability of the model placing $d_i$ earlier than $d_j$ is denoted and calculated by:
\begin{align}
P(d_i \succ d_j \mid \theta) = \frac{e^{f(d_i,\theta)}}{e^{f(d_i,\theta)} + e^{f(d_j,\theta)}}.
\end{align}
The gradient is estimated as a sum over inferred preferences with a weight $\rho$ per pair:
\begin{align}
\nabla &f(\cdot, \theta) \nonumber \\
&\approx \sum_{d_i \succ_{\mathbf{c}} d_j} \rho(d_i, d_j, R, D) [\Delta P(d_i \succ d_j \mid \theta)] \label{eq:comparison:novelgradient}  \\
&= \sum_{d_i \succ_{\mathbf{c}} d_j} \rho(d_i, d_j, R, D) P(d_i \succ d_j \mid \theta)P(d_j \succ d_i \mid \theta)(f'({d}_i, \theta) - f'({d}_j, \theta)).
\nonumber
\end{align}
After computing the gradient (Line~\ref{line:pdgd:modelgrad}), the model is updated accordingly (Line \ref{line:pdgd:update}).
This will change the distribution (Equation~\ref{eq:docprob}) towards the inferred preferences.
This distribution models the confidence over which documents should be placed first; the exploration of \ac{PDGD} is naturally guided by this confidence and can vary per query.

\begin{algorithm}[t]
\caption{\acf{PDGD}.} 
\label{alg:pdgd}
\begin{algorithmic}[1]
\STATE \textbf{Input}: initial weights: $\mathbf{\theta}_1$; scoring function: $f$; learning rate $\eta$.  \label{line:pdgd:initmodel}
\FOR{$t \leftarrow  1 \ldots \infty$ }
	\STATE $q_t \leftarrow \text{receive\_query}(t)$\hfill \textit{\small // obtain a query from a user} \label{line:pdgd:query}
	\STATE $\mathbf{R}_t \leftarrow \text{sample\_list}(f_{\theta_t}, D_{q_t})$ \hfill \textit{\small // sample list according to Eq.~\ref{eq:docprob}} \label{line:pdgd:samplelist}
	\STATE $\mathbf{c}_t \leftarrow \text{receive\_clicks}(\mathbf{R}_t)$ \hfill \textit{\small // show result list to the user} \label{line:pdgd:clicks}
	\STATE $\nabla  f(\cdot,\theta_t) \leftarrow \mathbf{0}$ \hfill \textit{\small // initialize gradient} \label{line:pdgd:initgrad}
	\FOR{$d_i \succ_{\mathbf{c}} d_j \in \mathbf{c}_t$} \label{line:pdgd:prefinfer}
	\STATE $w \leftarrow \rho(d_i, d_j, R, D)$  \hfill \textit{\small // initialize pair weight (Eq.~\ref{eq:rho})} \label{line:pdgd:initpair}
	\STATE $w \leftarrow w \times  P(d_i \succ d_j \mid \theta_t)P(d_j \succ d_i \mid \theta_t)$
             \hfill \textit{\small // pair gradient (Eq.~\ref{eq:novelgradient})} \label{line:pdgd:pairgrad}
	\STATE  $\nabla  f(\cdot,\theta_t) \leftarrow \nabla f_{\theta_t} + w \times (f'({d}_i, \theta_t) - f'({d}_j, \theta_t))$
	  \hfill \textit{\small // model gradient (Eq.~\ref{eq:novelgradient})} \label{line:pdgd:modelgrad}
	\ENDFOR
	\STATE $\theta_{t+1} \leftarrow \theta_{t} + \eta \nabla  f(\cdot,\theta_t)$
	\hfill \textit{\small // update the ranking model} \label{line:pdgd:update}
\ENDFOR
\end{algorithmic}
\end{algorithm}

The weighting function $\rho$ is used to make the gradient of \ac{PDGD} unbiased w.r.t. document pair preferences.
It uses the reverse pair ranking: $R^*(d_i, d_j, R)$, which is the same ranking as $R$ but with the document positions of $d_i$ and $d_j$ swapped.
Then $\rho$ is the ratio between the probability of $R$ and $R^*$:
\begin{align}
\rho(d_i, d_j, R, D) &= \frac{P(R^*(d_i, d_j, R) \mid D)}{P(R \mid D) + P(R^*(d_i, d_j, R) \mid D)}. \label{eq:rho}
\end{align}
In Chapter~\ref{chapter:02-pdgd}, the weighted gradient is proven to be unbiased w.r.t. document pair preferences under certain assumptions about the user.
Here, this unbiasedness is defined by being able to rewrite the gradient as:
\begin{align}
E[\Delta f(\cdot, \theta)] = \sum_{(d_i, d_j) \in D} \alpha_{ij}(f'(\mathbf{d}_i, \theta) - f'(\mathbf{d}_j, \theta)), \label{eq:unbias}
\end{align}
and the sign of $\alpha_{ij}$ agreeing with the preference of the user:
\begin{align}
\mathit{sign}(\alpha_{ij}) = \textit{sign}(\textit{relevance}(d_i) - \textit{relevance}(d_j)). \label{eq:signunbias}
\end{align}
The proof in Chapter~\ref{chapter:02-pdgd} only relies on the difference in the probabilities of inferring a preference: $d_i \succ_{\mathbf{c}} d_j$ in $R$ and the opposite preference $d_j \succ_{\mathbf{c}} d_i$ in $R^*(d_i, d_j, R)$. 
The proof relies on the sign of this difference to match the user's preference:
\begin{align}
\begin{split}
& \textit{sign}(P(d_i \succ_{\mathbf{c}} d_j \mid R) - P(d_j \succ_{\mathbf{c}} d_i \mid R^*))
\\&\hspace{4cm} = \textit{sign}(\textit{relevance}(d_i) - \textit{relevance}(d_j)). \label{eq:signclick}
 \end{split}
\end{align}
As long as Equation~\ref{eq:signclick} is true, Equation~\ref{eq:unbias}~and~\ref{eq:signunbias} hold as well.
Interestingly, this means that other assumptions about the user can be made than in Chapter~\ref{chapter:02-pdgd}, and other variations of \ac{PDGD} are possible, e.g., the algorithm could assume that all documents are observed and the proof still holds.

Chapter~\ref{chapter:02-pdgd} reports large improvements over \ac{DBGD}, however these improvements were observed under simulated cascading user models.
This means that the assumption that \ac{PDGD} makes about which documents are observed are always true.
As a result, it is currently unclear whether the method is really better in cases where the assumption does not hold.

%% file: 06-oltrcomparison/06-experiments.tex
\begin{table}[tb]
\caption{Click probabilities for simulated \emph{perfect} or \emph{almost random} behavior.}
\centering
\begin{tabularx}{\columnwidth}{ l X X X X X  }
\toprule
& \multicolumn{5}{c}{ $P(\mathit{click}(d)\mid \mathit{relevance}(d), \mathit{observed}(d))$} \\
\cmidrule(lr){2-6} 
$\mathit{relevance}(d)$ & \emph{$ 0$} & \emph{$ 1$}  &  \emph{$ 2$} & \emph{$ 3$} & \emph{$ 4$} \\
\midrule
 \emph{perfect}                           &  0.00 &  0.20 &  0.40 &  0.80 &  1.00  \\
  \emph{almost random}  \quad\quad\quad\quad &  0.40 &  0.45 &  0.50 &  0.55 &  0.60  \\
\bottomrule
\end{tabularx}
\label{tab:oltrcomp:clickmodels}
\end{table}

\section{Experiments}
\label{sec:experiments}

In this section we detail the experiments that were performed to answer the research questions in Section~\ref{sec:ecir:intro}.\footnote{The resources for reproducing the experiments in this chapter are available at \url{https://github.com/HarrieO/OnlineLearningToRank}}

\subsection{Datasets}
\label{sec:experiments:datasets}

Our experiments are performed over three large labelled datasets from commercial search engines, the largest publicly available \ac{LTR} datasets.
These datasets are the \emph{MLSR-WEB10K}~\cite{qin2013introducing}, \emph{Yahoo!\ Webscope}~\cite{Chapelle2011}, and \emph{Istella}~\cite{dato2016fast} datasets.
Each contains a set of queries with corresponding preselected document sets.
Query-document pairs are represented by feature vectors and five-grade relevance annotations ranging from \emph{not relevant} (0) to \emph{perfectly relevant}~(4).
Together, the datasets contain over \numprint{29900} queries and between 136 and 700~features per representation.

\subsection{Simulating user behavior}
\label{sec:experiments:users}

In order to simulate user behavior we partly follow the standard setup for \ac{OLTR}~\cite{he2009evaluation,hofmann11:balancing,oosterhuis2016probabilistic,schuth2016mgd,zoghi:wsdm14:relative}.
At each step a user issued query is simulated by uniformly sampling from the datasets.
The algorithm then decides what result list to display to the user, the result list is limited to $k = 10$ documents.
Then user interactions are simulated using click models~\cite{chuklin-click-2015}.
Past \ac{OLTR} work has only considered \emph{cascading click models}~\cite{guo09:efficient}; in contrast, we also use \emph{non-cascading click models}.
The probability of a click is conditioned on relevance and observance:
\begin{align}
P(\mathit{click}(d)\mid \mathit{relevance}(d), \mathit{observed}(d)).
\end{align}
We use two levels of noise to simulate \emph{perfect} user behavior and \emph{almost random} behavior~\cite{Hofmann2013a}, Table~\ref{tab:oltrcomp:clickmodels} lists the probabilities of both.
The \emph{perfect} user observes all documents, never clicks on anything non-relevant, and always clicks on the most relevant documents.
Two variants of \emph{almost random} behavior are used.
The first is based on cascading behavior, here the user first observes the top document, then decides to click according to Table~\ref{tab:oltrcomp:clickmodels}.
If a click occurs, then, with probability $P(stop \mid click) = 0.5$ the user stops looking at more documents, otherwise the process continues on the next document.
The second \emph{almost random} behavior is simulated in a non-cascading way; here we follow~\cite{joachims2017unbiased} and model the observing probabilities as:
\begin{align}
P(\mathit{observed}(d) \mid \mathit{rank}(d)) = \frac{1}{rank(d)}.
\end{align}
The important distinction is that it is safe to assume that the cascading user has observed all documents ranked before a click, while this is not necessarily true for the non-cascading user.
Since \ac{PDGD} makes this assumption, testing under both models can show us how much of its performance relies on this assumption.
Furthermore, the \emph{almost random} model has an extreme level of noise and position bias compared to the click models used in previous \ac{OLTR} work~\cite{hofmann11:balancing,oosterhuis2016probabilistic,schuth2016mgd}, and we argue it simulates an (almost) worst-case scenario.

\subsection{Experimental runs}
\label{sec:experiments:runs}

In our experiments we simulate runs consisting of \numprint{1000000} impressions; each run was repeated 125 times under each of the three click models.
\ac{PDGD} was run with $\eta = 0.1$ and zero initialization, \ac{DBGD} was run using Probabilistic Interleaving~\cite{oosterhuis2016probabilistic} with zero initialization, $\eta = 0.001$, and the unit sphere with $\delta=1$.
Other variants like Multileave Gradient Descent~\cite{schuth2016mgd} are not included; previous work has shown that their performance matches that of regular \ac{DBGD} after around \numprint{30000} impressions~\cite{oosterhuis2018differentiable,oosterhuis2016probabilistic,schuth2016mgd}.
The initial boost in performance comes at a large computational cost, though, as the fastest approaches keep track of at least 50 ranking models~\cite{oosterhuis2016probabilistic}, which makes running long experiments extremely impractical.
Instead, we introduce a novel oracle version of \ac{DBGD}, where, instead of interleaving, the NDCG values on the current query are calculated and the highest scoring model is selected.
This simulates a hypothetical perfect interleaving method, and we argue that the performance of this oracle run indicates what the upper bound on \ac{DBGD} performance is.

Performance is measured by NDCG@10 on a held-out test set, a two-sided t-test is performed for significance testing.
We do not consider the user experience during training, because Chapter~\ref{chapter:02-pdgd} has already investigated this aspect thoroughly.

%% file: 06-oltrcomparison/07-results.tex
\section{Experimental Results and Analysis}
\label{sec:results}

Recall that in Section~\ref{subsection:regretbounds} we have already provided a negative answer to \ref{rq:regret}: the regret bounds of \ac{DBGD} do not provide a benefit over \ac{PDGD} for the common ranking problem in \ac{LTR}.
In this section we present our experimental results and answer \ref{rq:noisebias} (whether the advantages of \ac{PDGD} over \ac{DBGD} of previous work generalize to extreme levels of noise and bias) and \ref{rq:cascading} (whether the performance of \ac{PDGD} is reproducible under non-cascading user behavior). 

\begin{figure}[t]
\caption{Performance (NDCG@10) on held-out data from Yahoo (top), MSLR (center), Istella (bottom) datasets, under the \emph{perfect}, and \emph{almost random} user models: cascading (casc.) and non-cascading (non-casc.).
The shaded areas display the standard deviation.}

\begin{tabular}{c}
\includegraphics[scale=.4]{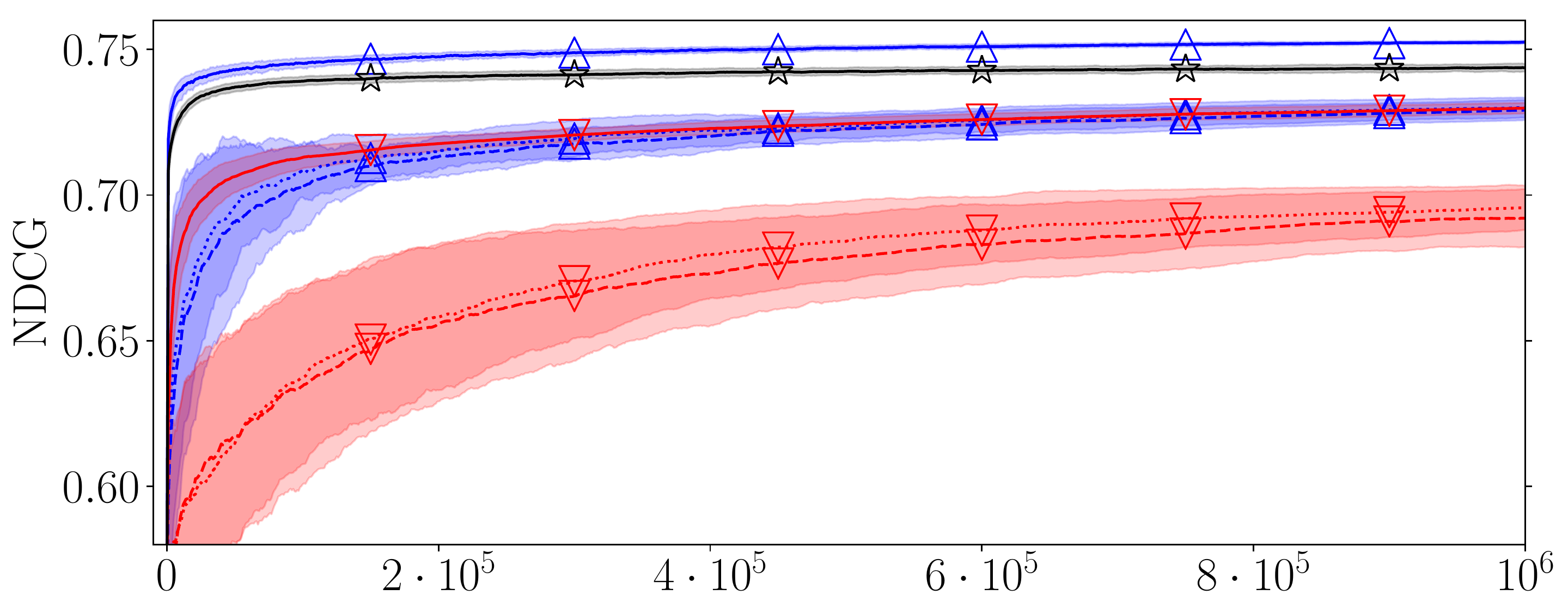}   \\
\includegraphics[scale=.4]{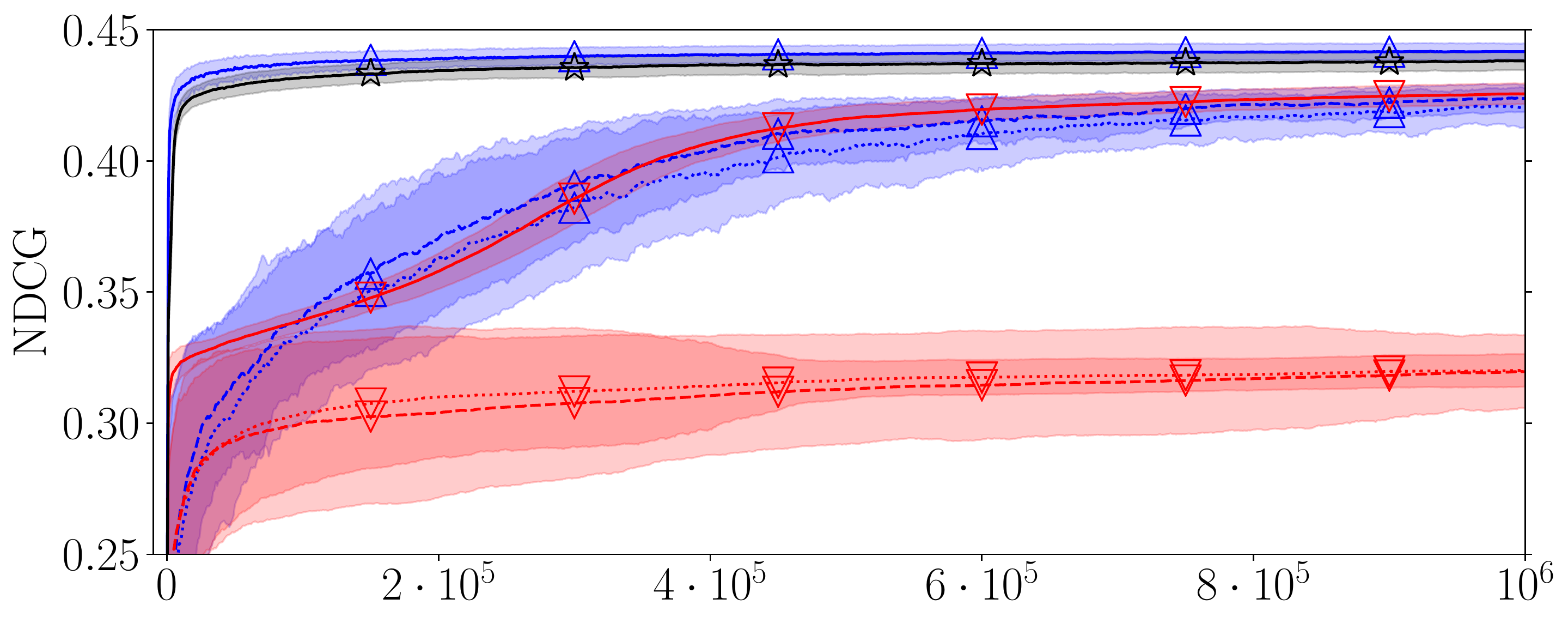}  \\
\hspace*{0.5em}\includegraphics[scale=.4]{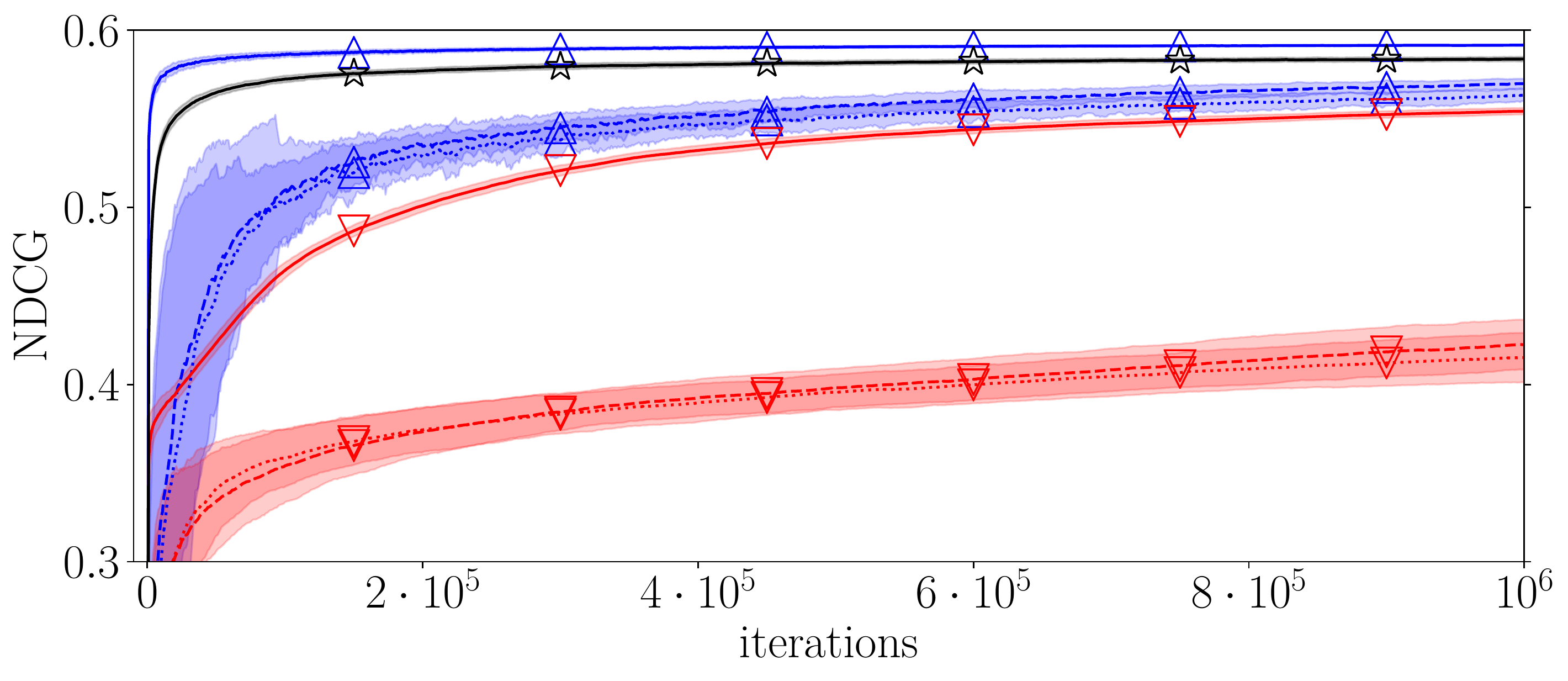} \\
\includegraphics[width=\textwidth]{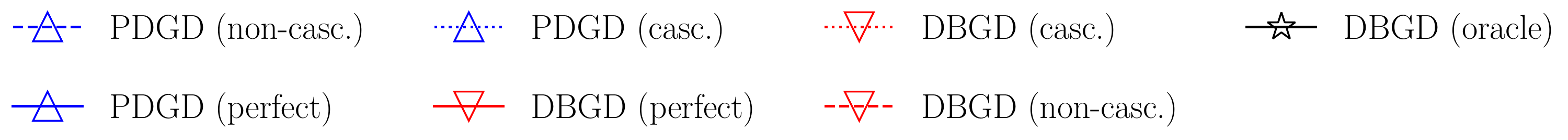}
\end{tabular}

\label{fig:main}
\end{figure}
Our main results are presented in Table~\ref{tab:main}. Additionally, Figure~\ref{fig:main} displays the average performance over \numprint{1000000} impressions.
First, we consider the performance of \ac{DBGD}; there is a substantial difference between its performance under the \emph{perfect} and \emph{almost random} user models on all datasets.
Thus, it seems that \ac{DBGD} is strongly affected by noise and bias in interactions; interestingly, there is little difference between performance under the cascading and non-cascading behavior.
On all datasets the \emph{oracle} version of \ac{DBGD} performs significantly better than \ac{DBGD} under \emph{perfect} user behavior.
This means there is still room for improvement and hypothetical improvements in, e.g., interleaving could lead to significant increases in long-term \ac{DBGD} performance.

Next, we look at the performance of \ac{PDGD}; here, there is also a significant difference between performance under the \emph{perfect} and \emph{almost random} user models on all datasets.
However, the effect of noise and bias is very limited compared to \ac{DBGD}, and this difference at \numprint{1000000} impressions is always less than $0.03$ NDCG on any dataset.

To answer \ref{rq:noisebias}, we compare the performance of \ac{DBGD} and \ac{PDGD}.
Across all datasets, when comparing \ac{DBGD} and \ac{PDGD} under the same levels of interaction noise and bias, the performance of \ac{PDGD} is significantly better in every case.
Furthermore, \ac{PDGD} under the \emph{perfect} user model significantly outperforms the \emph{oracle} run of \ac{DBGD}, despite the latter being able to directly observe the NDCG of rankers on the current query.
Moreover, when comparing \acs{PDGD}'s performance under the \emph{almost random} user model with \ac{DBGD} under the \emph{perfect} user model, we see the differences are limited and in both directions.
Thus, even under ideal circumstances \ac{DBGD} does not consistently outperform \ac{PDGD} under extremely difficult circumstances.
As a result, we answer \ref{rq:noisebias} positively: our results strongly indicate that the performance of \ac{PDGD} is considerably better than \ac{DBGD} and that these findings generalize from ideal circumstances to settings with extreme levels of noise and bias.

Finally, to answer \ref{rq:cascading}, we look at the performance under the two \emph{almost random} user models.
Surprisingly, there is no clear difference between the performance of \ac{PDGD} under \emph{cascading} and \emph{non-cascading} user behavior.
The differences are small and per dataset it differs which circumstances are slightly preferred.
Therefore, we answer \ref{rq:cascading} positively: the performance of \ac{PDGD} is reproducible under \emph{non-cascading} user behavior.

\begin{table*}[t]
\centering
\caption{Performance (NDCG@10) after \numprint{1000000} impressions for \ac{DBGD} and \ac{PDGD} under a \emph{perfect} click model and two almost-random click models: \emph{cascading} and \emph{non-cascading}, and \ac{DBGD} with an \emph{oracle} comparator.
Significant improvements and losses (p $<$ 0.01) between \ac{DBGD} and \ac{PDGD} are indicated by \dubbelop, \dubbelneer, and $\circ$ (no significant difference). Indications are in order of: \emph{oracle}, \emph{perfect}, \emph{cascading}, and \emph{non-cascading}.
}
\input{06-oltrcomparison/tables/comparison}
\label{tab:main}
\end{table*}

%% file: 06-oltrcomparison/tables/comparison.tex
\begin{tabular*}{\textwidth}{@{\extracolsep{\fill} } l  l l l  }
\toprule
 & { \small \textbf{Yahoo}}  & { \small \textbf{MSLR}}  & { \small \textbf{Istella}} \\
\midrule
& \multicolumn{3}{c}{\textit{\acl{DBGD}}} \\
\midrule
\textit{oracle} & 0.744 {\tiny (0.001)} {\tiny \dubbelneer} {\tiny \dubbelop} {\tiny \dubbelop} & 0.438 {\tiny (0.004)} {\tiny \dubbelneer} {\tiny \dubbelop} {\tiny \dubbelop} & 0.584 {\tiny (0.001)} {\tiny \dubbelneer} {\tiny \dubbelop} {\tiny \dubbelop} \\
\textit{perfect} & 0.730 {\tiny (0.002)} {\tiny \dubbelneer} {\small $\circ$} {\small $\circ$} & 0.426 {\tiny (0.004)} {\tiny \dubbelneer} {\tiny \dubbelop} {\tiny \dubbelop} & 0.554 {\tiny (0.002)} {\tiny \dubbelneer} {\tiny \dubbelneer} {\tiny \dubbelneer} \\
\textit{cascading} & 0.696 {\tiny (0.008)} {\tiny \dubbelneer} {\tiny \dubbelneer} {\tiny \dubbelneer} & 0.320 {\tiny (0.006)} {\tiny \dubbelneer} {\tiny \dubbelneer} {\tiny \dubbelneer} & 0.415 {\tiny (0.014)} {\tiny \dubbelneer} {\tiny \dubbelneer} {\tiny \dubbelneer} \\
\textit{non-cascading} & 0.692 {\tiny (0.010)} {\tiny \dubbelneer} {\tiny \dubbelneer} {\tiny \dubbelneer} & 0.320 {\tiny (0.014)} {\tiny \dubbelneer} {\tiny \dubbelneer} {\tiny \dubbelneer} & 0.422 {\tiny (0.014)} {\tiny \dubbelneer} {\tiny \dubbelneer} {\tiny \dubbelneer} \\
\midrule
& \multicolumn{3}{c}{\textit{\acl{PDGD}}} \\
\midrule
\textit{perfect} & 0.752 {\tiny (0.001)} {\tiny \dubbelop} {\tiny \dubbelop} {\tiny \dubbelop} {\tiny \dubbelop} & 0.442 {\tiny (0.003)} {\tiny \dubbelop} {\tiny \dubbelop} {\tiny \dubbelop} {\tiny \dubbelop} & 0.592 {\tiny (0.000)} {\tiny \dubbelop} {\tiny \dubbelop} {\tiny \dubbelop} {\tiny \dubbelop} \\
\textit{cascading} & 0.730 {\tiny (0.003)} {\tiny \dubbelneer} {\small $\circ$} {\tiny \dubbelop} {\tiny \dubbelop} & 0.420 {\tiny (0.007)} {\tiny \dubbelneer} {\tiny \dubbelneer} {\tiny \dubbelop} {\tiny \dubbelop} & 0.563 {\tiny (0.003)} {\tiny \dubbelneer} {\tiny \dubbelop} {\tiny \dubbelop} {\tiny \dubbelop} \\
\textit{non-cascading} & 0.729 {\tiny (0.003)} {\tiny \dubbelneer} {\small $\circ$} {\tiny \dubbelop} {\tiny \dubbelop} & 0.424 {\tiny (0.005)} {\tiny \dubbelneer} {\tiny \dubbelneer} {\tiny \dubbelop} {\tiny \dubbelop} & 0.570 {\tiny (0.003)} {\tiny \dubbelneer} {\tiny \dubbelop} {\tiny \dubbelop} {\tiny \dubbelop} \\
\bottomrule
\end{tabular*}

%% file: 06-oltrcomparison/08-conclusion.tex
\section{Conclusion}
\label{sec:conclusion}

In this chapter, we have reproduced and generalized findings about the relative performance of \acf{DBGD} and \acf{PDGD}.
Our results show that the performance of \ac{PDGD} is reproducible under non-cascading user behavior.
Furthermore, \ac{PDGD} outperforms \ac{DBGD} in both \emph{ideal} and extremely \emph{difficult} circumstances with high levels of noise and bias.
Moreover, the performance of \ac{PDGD} in extremely \emph{difficult} circumstances is comparable to that of \ac{DBGD} in \emph{ideal} circumstances.
Additionally, we have shown that the regret bounds of \ac{DBGD} are not applicable to the common ranking problem in \ac{LTR}.
In summary, our results strongly confirm the previous finding that \ac{PDGD} consistently outperforms \ac{DBGD}, and generalizes this conclusion to circumstances with extreme levels of noise and bias.

With these findings we can answer \ref{thesisrq:dbgd} mostly negatively: the theory behind \ac{DBGD} is not sound for the common deterministic ranking problem, moreover, \ac{DBGD} has extremely poor performance when compared to the \ac{PDGD} method under varying conditions.
Consequently, there appears to be no advantage to using \ac{DBGD} over \ac{PDGD} in either theoretical or empirical terms.
In addition, a decade of \ac{OLTR} work has attempted to extend \ac{DBGD} in numerous ways without leading to any measurable long-term improvements.
Together, this suggests that the general approach of \ac{DBGD} based methods, i.e., sampling models and comparing with online evaluation, is not an effective way of optimizing ranking models.
Although the \ac{PDGD} method considerably outperforms the \ac{DBGD} approach, we currently do not have a theoretical explanation for this difference.
Thus it seems plausible that a more effective \ac{OLTR} method could be derived, if the theory behind the effectiveness of \ac{OLTR} methods is better understood.
Due to this potential and the current lack of regret bounds applicable to \ac{OLTR}, we argue that a theoretical analysis of \ac{OLTR} would make a very valuable future contribution to the field.

Finally, we consider the limitations of the comparison in this chapter.
As is standard in \ac{OLTR} our results are based on simulated user behavior. 
These simulations provide valuable insights: they enable direct control over biases and noise, and evaluation can be performed at each time step.
In this chapter, the generalizability of this setup was pushed the furthest by varying the conditions to the extremely difficult.
It appears unlikely that more reliable conclusions can be reached from simulated behavior.
Thus we argue that the most valuable future comparisons would be in experimental settings with real users.
Furthermore, with the performance improvements of \ac{PDGD} the time seems right for evaluating the effectiveness of \ac{OLTR} in real-world applications.

The limited theoretical guarantees regarding \ac{OLTR} methods, prompted the second part of this thesis where we consider counterfactual \ac{LTR}.
In contrast with \ac{OLTR}, counterfactual \ac{LTR} methods are founded on assumed models of user behavior and are proven to unbiasedly optimize ranking metrics if the assumed models are correct.
Despite these theoretical strengths, empirical comparisons in previous work show that \ac{PDGD} is more robust than existing counterfactual \ac{LTR} methods.
In Chapter~\ref{chapter:06-onlinecounterltr} we introduce a counterfactual \ac{LTR} method that can reach the same levels of performance as \ac{PDGD} when applied online.

%% file: 06-oltrcomparison/notation.tex
\section{Notation Reference for Chapter~\ref{chapter:03-oltr-comparison}}
\label{notation:03-oltr-comparison}

\begin{center}
\begin{tabular}{l l}
 \toprule
\bf Notation  & \bf Description \\
\midrule
$t$ & a timestep \\
$q$ & a user-issued query \\
$d$, $d_k$, $d_l$ & document\\
$\mathbf{d}$ & feature representation of a query-document pair \\
$D$ & set of documents\\
$R$ & ranked list \\
$I_t$ & an interleaved result list \\
$R^*$ & the reversed pair ranking $R^*(d_k, d_l, R)$ \\
$\rho$ & preference pair weighting function \\
$\theta$ & parameters of the ranking model\\
$f_\theta(\cdot)$ & ranking model with parameters $\theta$ \\
$f(\mathbf{d}_k)$ & ranking score for a document from model \\
$\mathbf{c}_t$ & a binary vector representing the clicks at timestep $t$ \\
\bottomrule
\end{tabular}
\end{center}

%% file: 07-topk/main.tex
\chapter{Policy-Aware Counterfactual Learning to Rank for Top-$k$ Rankings}
\label{chapter:04-topk}

\newcommand{\ranking}{R}
\newcommand{\displayranking}{\bar{R}}
\newcommand{\doc}{d}
\renewcommand{\rank}[2]{\text{rank}(#1\mid#2)}

\footnote[]{This chapter was published as~\citep{oosterhuis2020topkrankings}.
Appendix~\ref{notation:04-topk} gives a reference for the notation used in this chapter.
}

Counterfactual \ac{LTR} methods optimize ranking systems using logged user interactions that contain interaction biases.
Existing methods are only unbiased if users are presented with all relevant items in every ranking.
However, in prevalent top-$k$ ranking settings not all items can be displayed at once.
Therefore, there is currently no existing counterfactual unbiased \ac{LTR} method for top-$k$ rankings.
In this chapter we address this limitation by asking the thesis research question:
\begin{itemize}
\item[\ref{thesisrq:topk}] Can counterfactual \ac{LTR} be extended to top-$k$ ranking settings?
\end{itemize}
We introduce a novel policy-aware counterfactual estimator for \ac{LTR} metrics that can account for the effect of a stochastic logging policy.
We prove that the policy-aware estimator is unbiased if every relevant item has a non-zero probability to appear in the top-$k$ ranking.
Our experimental results show that the performance of our estimator is not affected by the size of $k$: for any $k$, the policy-aware estimator reaches the same retrieval performance while learning from top-$k$ feedback as when learning from feedback on the full ranking.

While the policy-aware estimator allows us to learn from top-$k$ feedback, there is no theoretically-grounded way to optimize for top-$k$ ranking metrics.
Furthermore, existing counterfactual \ac{LTR} work has mostly used novel loss functions for optimization, which are quite different from those used in supervised \ac{LTR}.
This lead us to ask the following thesis research question:
\begin{itemize}
\item[\ref{thesisrq:lambdaloss}] Is it possible to apply state-of-the-art supervised \ac{LTR} to the counterfactual \ac{LTR} problem?
\end{itemize}
In this chapter, we also introduce novel extensions of supervised \ac{LTR} methods to perform counterfactual \ac{LTR} and to optimize top-$k$ metrics.
Together, our contributions introduce the first policy-aware unbiased \ac{LTR} approach that learns from top-$k$ feedback and optimizes top-$k$ metrics.
As a result, counterfactual \ac{LTR} is now applicable to the very prevalent top-$k$ ranking setting in search and recommendation.

\input{07-topk/sections/01-introduction}

\input{07-topk/sections/02-background}
\input{07-topk/sections/03-top-k-feedback}

\input{07-topk/sections/04-top-k-metrics}

\input{07-topk/sections/05-experiments}

\input{07-topk/sections/06-results}

\input{07-topk/sections/07-related-work}
\input{07-topk/sections/08-conclusion}
\begin{subappendices}
\input{07-topk/notation}
\end{subappendices}

%% file: 07-topk/sections/01-introduction.tex
\section{Introduction}
\label{sec:intro}

\ac{LTR} optimizes ranking systems to provide high quality rankings.
Interest in \ac{LTR} from user interactions has greatly increased in recent years with the introduction of unbiased \ac{LTR} methods~\citep{joachims2017unbiased, wang2016learning}.
The potential for learning from logged user interactions is great:
user interactions provide valuable implicit feedback while also being cheap and relatively easy to acquire at scale~\citep{joachims2017accurately}.
However, interaction logs also contain large amounts of bias, which is the result of both user behavior and the ranker used during logging.
For instance, users are more likely to examine items at the top of rankings, consequently the display position of an item heavily affects the number of interactions it receives~\citep{wang2018position}.
This effect is called \emph{position bias} and it is very dominant when learning from interactions with rankings.
Naively ignoring it during learning can be detrimental to ranking performance, as the learning process is strongly impacted by what rankings were displayed during logging instead of \emph{true} user preferences.
The goal of unbiased \ac{LTR} methods is to optimize a ranker w.r.t.\ the \emph{true} user preferences, consequently, they have to account and correct for such forms of bias.

Previous work on unbiased \ac{LTR} has mainly focussed on accounting for position bias through counterfactual learning~\citep{joachims2017unbiased, wang2016learning, ai2018unbiased}.
The prevalent approach models the probability of a user examining an item in a displayed ranking.
This probability can be inferred from user interactions~\citep{joachims2017unbiased, wang2016learning, ai2018unbiased, wang2018position, agarwal2019estimating} and corrected for using \emph{inverse propensity scoring}.
As a result, these methods optimize a loss that in expectation is unaffected by the examination probabilities during logging, hence it is unbiased w.r.t.\ position bias.

This approach has been applied effectively in various ranking settings, including search for scientific articles~\cite{joachims2017unbiased}, email~\cite{wang2016learning} or other personal documents~\cite{wang2018position}.
However, a limitation of existing approaches is that in every logged ranking they require every relevant item to have a non-zero chance of being examined~\cite{carterette2018offline, joachims2017unbiased}.
In this chapter, we focus on top-$k$ rankings where the number of displayed items is systematically limited.
These rankings can display at most $k$ items, making it practically unavoidable that relevant items are missing.
Consequently, existing counterfactual \ac{LTR} methods are not unbiased in these settings.
We recognize this problem as \emph{item-selection bias} introduced by the selection of (only) $k$ items to display.
This is especially concerning since top-$k$ rankings are quite prevalent, e.g., in recommendation~\citep{cremonesi2010performance, hurley2011novelty}, mobile search~\citep{balke2002real, vlachou2011monitoring}, query autocompletion~\citep{cai-survey-2016,wang2016learning, wang2018position}, and digital assistants~\citep{shalyminov-2018-neural}.

Our main contribution is a novel policy-aware estimator for counterfactual \ac{LTR} that accounts for both a stochastic logging policy and the users' examination behavior.
Our policy-aware approach can be viewed as a generalization of the existing counterfactual \ac{LTR} framework~\citep{joachims2017unbiased, agarwal2019counterfactual}.
We prove that our policy-aware approach performs unbiased \ac{LTR} and evaluation while learning from top-$k$ feedback.
Our experimental results show that while our policy-aware estimator is unaffected by the choice of $k$, the existing policy-oblivious approach is strongly affected even under large values of $k$.
For instance, optimization with the policy-aware estimator on top-5 feedback reaches the same performance as when receiving feedback on all results.
Furthermore, because top-$k$ metrics are the only relevant metrics in top-$k$ rankings, we also propose extensions to traditional \ac{LTR} approaches that are proven to optimize top-$k$ metrics unbiasedly and introduce a pragmatic way to choose optimally between available loss functions.

This chapter is based around two main contributions:
\begin{enumerate}[align=left, leftmargin=*]
	\item A novel estimator for unbiased \ac{LTR} from top-$k$ feedback.
	\item Unbiased losses that optimize bounds on top-$k$ \ac{LTR} metrics.
\end{enumerate}
To the best of our knowledge, our policy-aware estimator is the first estimator that is unbiased in top-$k$ ranking settings.

%% file: 07-topk/sections/02-background.tex
\section{Background}
\label{sec:background}

In this section we discuss supervised \ac{LTR} and counterfactual \ac{LTR}~\citep{joachims2017unbiased}.

\subsection{Supervised learning to rank}
\label{section:supervisedLTR}
The goal of \ac{LTR} is to optimize ranking systems w.r.t.\ specific ranking metrics.
Ranking metrics generally involve items $\doc$, their relevance $r$ w.r.t.\ a query $q$, and their position in the ranking  $\ranking$ produced by the system.
We will optimize the \emph{Empirical Risk}~\citep{vapnik2013nature} over the set of queries $Q$, with a loss $\Delta(\ranking_i \mid q_i, r)$ for a single query $q_i$:
\begin{align}
\mathcal{L} = \frac{1}{|Q|} \sum_{q_i \in Q} \Delta(\ranking_i \mid q_i, r).
\end{align}
For simplicity we assume that relevance is binary:
$r(q,\doc) \in \{0, 1\}$; for brevity we write: $r(q, \doc) = r(\doc)$.
Then, ranking metrics commonly take the form of a sum over items:
\begin{align}
\Delta(\ranking \mid q, r) =
\sum_{\doc \in \ranking} \lambda\left(\doc \mid \ranking \right) \cdot r(\doc),
\end{align}
where $\lambda$ can be chosen for a specific metric, e.g., for \ac{ARP} or \ac{DCG}:
\begin{align}
\lambda^\textit{ARP}(\doc \mid \ranking) &= \text{rank}(\doc \mid \ranking) \label{eq:ARP},\\
\lambda^\textit{DCG}(\doc \mid \ranking) &= -\log_2\big(1 + \text{rank}(\doc \mid \ranking)\big)^{-1}.  \label{eq:DCG}
\end{align}
In a so-called \emph{full-information} setting, where the relevance values $r$ are known, optimization can be done through traditional \ac{LTR} methods~\citep{wang2018lambdaloss, burges2010ranknet, Joachims2002, liu2009learning}.

\subsection{Counterfactual learning to rank}
\label{section:counterfactualLTR}
Optimizing a ranking loss from the implicit feedback in interaction logs requires a different approach from supervised \ac{LTR}. %
We will assume that clicks are gathered using a logging policy $\pi$ with the probability of displaying ranking $\displayranking$ for query $q$ denoted as $\pi(\displayranking \mid q)$.
Let $o_i(\doc) \in \{0, 1\}$ indicate whether $\doc$ was examined by a user at interaction $i$ and $o_i(\doc) \sim P( o(\doc) \mid q_i, r, \displayranking_i)$.
Furthermore, we assume that users click on all relevant items they observe and nothing else: $c_i(\doc) = \mathds{1}[r(\doc) \land o_i(\doc)]$.
Our goal is to find an estimator $\hat{\Delta}$ that provides an unbiased estimate of the actual loss; for $N$ interactions this estimate is:
\begin{align}
\hat{\mathcal{L}} =
\frac{1}{N} \sum^N_{i=1} \hat{\Delta}(\ranking_i\mid q_i, \displayranking_i, \pi, c_i). \label{eq:highlevelloss}
\end{align}
We write $\ranking_i$ for the ranking produced by the system for which the loss is being computed, while $\displayranking_i$ is the ranking that was displayed when logging interaction $i$.
For brevity we will drop $i$ from our notation when only a single interaction is involved.
A naive estimator could simply consider every click to indicate relevance:
\begin{align}
\hat{\Delta}_\textit{naive}\left(\ranking \mid q, c \right)
= \sum_{\doc : c(\doc) = 1} \lambda\left(\doc \mid \ranking \right)
.
\label{eq:naiveestimator}
\end{align}
Taking the expectation over the displayed ranking and observance variables results in the following expected loss:
\begin{align}
\mathbb{E}_{o,\displayranking}\left[ \hat{\Delta}_\textit{naive}\left(\ranking \mid  q, c\right) \right] \hspace{-3em} &
\nonumber \\
& =
\mathbb{E}_{o,\displayranking}\left[ \sum_{\doc: c(\doc) = 1} \lambda( \doc \mid \ranking) \right]
=
\mathbb{E}_{o,\displayranking}\left[ \sum_{\doc \in \ranking} \lambda( \doc \mid \ranking) \cdot c(\doc) \right]
\nonumber \\
&=
\mathbb{E}_{o,\displayranking}\left[ \sum_{\doc \in \ranking} o(\doc) \cdot \lambda\left(\doc \mid \ranking\right) \cdot r(\doc)\right] 
 \\
&=
\mathbb{E}_{\displayranking}\left[ \sum_{\doc \in \ranking} P\left( o(\doc) = 1 \mid q, r, \displayranking \right) \cdot \lambda\left(\doc \mid \ranking \right) \cdot r(\doc)\right] \nonumber \\
& =
\sum_{\displayranking \in \pi(\cdot \mid q)} \pi(\displayranking \mid q) \cdot 
\sum_{\doc \in \ranking} P\left( o(\doc) = 1 \mid q, r, \displayranking\right) \cdot \lambda\left(\doc\mid \ranking \right) \cdot r(\doc).
\nonumber
\end{align}
Here, the effect of position bias is very clear; in expectation, items are weighted according to their probability of being examined.
Furthermore, it shows that examination probabilities are determined by both the logging policy $\pi$ and user behavior $P( o(\doc) \mid q, r, \displayranking )$.

In order to avoid the effect of position bias, \citeauthor{joachims2017unbiased}~\citep{joachims2017unbiased} introduced an inverse-propensity-scoring estimator in the same vain as previous work by \citeauthor{wang2016learning}~\citep{wang2016learning}.
The main idea behind this estimator is that if the examination probabilities are known, then they can be corrected for per click:
\begin{align}
\hat{\Delta}_{\textit{oblivious}}\left(\ranking \mid q, c, \displayranking \right)
&= \sum_{\doc :c(\doc) = 1} \frac{\lambda\left(\doc \mid \ranking \right)}{P\left( o(\doc) = 1 \mid q, r, \displayranking \right) }.
\label{eq:obliviousestimator}
\end{align}
In contrast to the naive estimator (Eq.~\ref{eq:naiveestimator}), this policy-oblivious estimator (Eq.~\ref{eq:obliviousestimator}) can provide an unbiased estimate of the loss:
\begin{align}
\begin{split}
\mathbb{E}_{o,\displayranking}\Big[\hat{\Delta}_\textit{oblivious}\left(\ranking \mid q, c, \displayranking \right) \Big] 
\hspace{-2cm}&
\\&= \mathbb{E}_{o,\displayranking} \Bigg[  \sum_{\doc:c(\doc) = 1} \frac{\lambda\left(\doc\mid \ranking \right)}{P\left( o(\doc) = 1 \mid q, r, \displayranking \right) } \Bigg] \\
&= \sum_{\doc \in \ranking} \mathbb{E}_{o,\displayranking} \Bigg[ \frac{o(\doc) \cdot \lambda\left(\doc\mid \ranking \right) \cdot r(\doc)}{P\left( o(\doc) = 1 \mid q, r, \displayranking \right)} \Bigg] \\
&= \sum_{\doc \in \ranking} \mathbb{E}_{\displayranking} \Bigg[ \frac{P\left( o(\doc) = 1 \mid q, r, \displayranking \right) \cdot \lambda\left(\doc\mid \ranking \right) \cdot r(\doc)}{P\left( o(\doc) = 1 \mid q, r, \displayranking \right)} \Bigg] \\
&= \sum_{\doc \in \ranking} \lambda\left(\doc\mid \ranking \right) \cdot r(\doc) 
= \Delta(\ranking\mid q, r).
\end{split}
\end{align}
We note that the last step assumes $P\left( o(\doc) = 1 \mid q, r, \displayranking \right) > 0$, and that only relevant items $r(\doc) = 1$ contribute to the estimate~\citep{joachims2017unbiased}.
Therefore, this estimator is unbiased as long as the examination probabilities are positive for every relevant item:
\begin{align}
\forall \doc , \,
\forall \displayranking \in \pi(\cdot \mid q) \,
\left[
r(\doc) = 1 \to P\left( o(\doc) = 1 \mid q, r, \displayranking \right) > 0
\right]
. 
\label{eq:agnosticcond}
\end{align}
Intuitively, this condition exists because propensity weighting is applied to items clicked in the displayed ranking and items that cannot be observed can never receive clicks.
Thus, there are no clicks that can be weighted more heavily to adjust for the zero observance probability of an item. 

An advantageous property of the policy-oblivious estimator $\hat{\Delta}_{\textit{oblivious}}$ is that the logging policy $\pi$ does not have to be known.
That is, as long as Condition~\ref{eq:agnosticcond} is met, it works regardless of how interactions were logged.
Additionally, \citeauthor{joachims2017unbiased}~\citep{joachims2017unbiased} proved that it is still unbiased under click noise. %
Virtually all recent counterfactual \ac{LTR} methods use the policy-oblivious estimator for \ac{LTR} optimization~\citep{wang2016learning, joachims2017unbiased, agarwal2019addressing, ai2018unbiased, wang2018position, agarwal2019estimating}.

%% file: 07-topk/sections/03-top-k-feedback.tex
\section{Learning from Top-$k$ Feedback}
\label{sec:topkfeedback}

In this section we explain why the existing policy-oblivious counterfactual \ac{LTR} framework is not applicable to top-$k$ rankings.
Subsequently, we propose a novel solution through policy-aware propensity scoring that takes the logging policy into account.

\subsection{The problem with top-$k$ feedback}
An advantage of the existing policy-oblivious estimator for counterfactual \ac{LTR} described in Section~\ref{section:counterfactualLTR} is that the logging policy does not need to be known, making its application easier.
However, the policy-oblivious estimator is only unbiased when Condition~\ref{eq:agnosticcond} is met: every relevant item has a non-zero probability of being observed in every ranking displayed during logging.

We recognize that in top-$k$ rankings, where only $k$ items can be displayed, relevant items may systematically \emph{lack} non-zero examination probabilities.
This happens because items outside the top-$k$ cannot be examined by the user:
\begin{align}
\forall \doc, \forall \displayranking \,
\left[
\text{rank}\big(\doc \mid \displayranking \big) > k \to P\big( o(\doc) = 1 \mid q, r, \displayranking \big) = 0
\right]
.
\end{align}
In most top-$k$ ranking settings it is very unlikely that Condition~\ref{eq:agnosticcond} is satisfied; If $k$ is very small, the number of relevant items is large, or if the logging policy $\pi$ is ineffective at retrieving relevant items, it is unlikely that all relevant items will be displayed in the top-$k$ positions.
Moreover, for a small value of $k$ the performance of the logging policy $\pi$ has to be near ideal for all relevant items to be displayed.
We call this effect \emph{item-selection bias}, because in this setting the logging ranker makes a selection of which $k$ items to display, in addition to the order in which to display them (position bias).
The existing policy-oblivious estimator for counterfactual \ac{LTR} (as described in Section~\ref{section:counterfactualLTR}) cannot correct for item-selection bias when it occurs, and can thus be affected by this bias when applied to top-$k$ rankings.

\subsection{Policy-aware propensity scoring}
Item-selection bias is inevitable in a single top-$k$ ranking, due to the limited number of items that can be displayed.
However, across multiple top-$k$ rankings more than $k$ items could be displayed if the displayed rankings differ enough.
Thus, a stochastic logging-policy could provide every item with a non-zero probability to appear in the top-$k$ ranking.
Then, the probability of examination can be calculated as an expectation over the displayed ranking:
\begin{align}
P\left(o(\doc) = 1 \mid q, r, \pi \right)
&= \mathbb{E}_{\displayranking}\left[P\big(o(\doc) = 1 \mid q, r, \displayranking \big) \right] 
\label{eq:expexam} \\
&= \sum_{\displayranking \in \pi(\cdot \mid q)} \pi\big(\displayranking \mid q\big) \cdot  P\big(o(\doc) = 1 \mid q, r, \displayranking \big).
\nonumber
\end{align}
This policy-dependent examination probability can be non-zero for all items, even if all items cannot be displayed in a single top-$k$ ranking.
Naturally, this leads to a \emph{policy-aware} estimator:
\begin{align}
\hat{\Delta}_{\textit{aware}}\left(\ranking \mid q, c, \pi \right)
&= \sum_{\doc :c(\doc) = 1} \frac{\lambda(\doc \mid \ranking )}{P\big( o(\doc) = 1 \mid q, r, \pi \big) }. \label{eq:policyaware}
\end{align}
By basing the propensity on the policy instead of the individual rankings, the policy-aware estimator can correct for zero observance probabilities in some displayed rankings by more heavily weighting clicks on other displayed rankings with non-zero observance probabilities.
Thus, if a click occurs on an item that the logging policy rarely displays in a top-$k$ ranking, this click may be weighted more heavily than a click on an item that is displayed in the top-$k$ very often.
In contrast, the policy-oblivious approach only corrects for the observation probability for the displayed ranking in which the click occurred, thus it does not correct for the fact that an item may be missing from the top-$k$ in other displayed rankings.

In expectation, the policy-aware estimator provides an unbiased estimate of the ranking loss:
\begin{equation}
\begin{split}
\mathbb{E}_{o,\displayranking}\Big[\hat{\Delta}_\textit{aware}\left(\ranking \mid q, c, \pi \right) \Big]
\hspace{-3cm}&
 \\
&
= \mathbb{E}_{o,\displayranking} \Bigg[  \sum_{\doc:c(\doc) = 1} \frac{\lambda\big(\doc \mid \ranking\big)}{P\left( o(\doc) = 1 \mid q, r, \pi\right) } \Bigg]  \\
&
= \sum_{\doc \in \ranking} \mathbb{E}_{o,\displayranking} \Bigg[ \frac{o(\doc) \cdot \lambda\big(\doc \mid \ranking\big) \cdot r(\doc)}{\sum_{\displayranking' \in \pi(\cdot \mid q)} \pi\big(\displayranking' \mid q\big) \cdot  P\big(o(\doc) = 1 \mid q, r, \displayranking' \big)} \Bigg]\\
&
= \sum_{\doc \in \ranking} \mathbb{E}_{\displayranking} \Bigg[ \frac{P\big( o(\doc) = 1 \mid q, r, \displayranking \big) \cdot \lambda\big(\doc \mid \ranking\big) \cdot r(\doc)}{\sum_{\displayranking' \in \pi(\cdot \mid q)} \pi\big(\displayranking' \mid q\big) \cdot  P\big(o(\doc) = 1 \mid q, r, \displayranking' \big)} \Bigg]   \\
&
= \sum_{\doc \in \ranking}  \frac{\sum_{\displayranking \in \pi(\cdot \mid q)} \pi\big(\displayranking \mid q\big) \cdot  P\big( o(\doc) = 1 \mid q, r, \displayranking \big) \cdot \lambda\big(\doc \mid \ranking\big) \cdot r(\doc)}{\sum_{\displayranking'\in \pi(\cdot \mid q)} \pi\big(\displayranking' \mid q\big) \cdot  P\big(o(\doc) = 1 \mid q, r, \displayranking' \big)}  \\
&
= \sum_{\doc \in \ranking} \lambda\big(\doc \mid \ranking\big) \cdot r(\doc)
 \\ &
= \Delta\big(\ranking\mid q, r\big).
\end{split}
\end{equation}
In contrast to the policy-oblivious approach (Section~\ref{section:counterfactualLTR}), this proof is sound as long as every relevant item has a non-zero probability of being examined under the logging policy $\pi$: 
\begin{align}
\forall \doc \,
\left[
 r(\doc) = 1 \to \sum_{\displayranking \in \pi(\cdot \mid q)} \pi\big(\displayranking \,|\, q\big) \cdot  P\big(o(\doc) = 1 \,|\, q, r, \displayranking \big) > 0
 \right]
 . 
 \label{eq:awarecond}
\end{align}
It is easy to see that Condition~\ref{eq:agnosticcond} implies Condition~\ref{eq:awarecond}, in other words, for all settings where the policy-oblivious estimator (Eq.~\ref{eq:obliviousestimator}) is unbiased, the policy-aware estimator (Eq.~\ref{eq:policyaware}) is also unbiased.
Conversely, Condition~\ref{eq:awarecond} does not imply Condition~\ref{eq:agnosticcond}, thus there are cases where the policy-aware estimator is unbiased but the policy-oblivious estimator is not guaranteed to be.

To better understand for which policies Condition~\ref{eq:awarecond} is satisfied, we introduce a substitute Condition~\ref{eq:awarecond2}:
\begin{align}
\mbox{}
\hspace*{-2mm}
\forall \doc 
\Big[
r(\doc) = 1 \to  \exists \displayranking \left[\pi\big(\displayranking \mid q\big)  > 0  \land  P\big(o(\doc) = 1 \mid q, r, \displayranking \big) > 0 \right]
\!
\Big]
. 
\hspace*{-2mm}
\mbox{}
\label{eq:awarecond2}
\end{align}
Since Condition~\ref{eq:awarecond2} is equivalent to Condition~\ref{eq:awarecond}, we see that 
the policy-aware estimator is unbiased for any logging-policy that provides a non-zero probability for every relevant item to appear in a position with a non-zero examination probability.
Thus to satisfy Condition~\ref{eq:awarecond2} in a top-$k$ ranking setting, every relevant item requires a non-zero probability of being displayed in the top-$k$.

As long as Condition~\ref{eq:awarecond2} is met, a wide variety of policies can be chosen according to different criteria.
Moreover, the policy can be deterministic if $k$ is large enough to display every relevant item.
Similarly, the policy-oblivious estimator can be seen as a special case of the policy-aware estimator where the policy is deterministic (or assumed to be).
The big advantage of our policy-aware estimator is that it is applicable to a much larger number of settings than the existing policy-oblivious estimator, including those were feedback is only received on the top-$k$ ranked items.

\subsection{Illustrative example}

To better understand the difference between the policy-oblivious and policy-aware estimators, we introduce an illustrative example that contrasts the two.
We consider a single query $q$ and a logging policy $\pi$ that chooses between two rankings to display: $\displayranking_1$ and $\displayranking_2$, with: $\pi(\displayranking_1 \mid q) > 0$; $\pi(\displayranking_2 \mid q) > 0$; and $\pi(\displayranking_1 \mid q) + \pi(\displayranking_2 \mid q) = 1$.
Then for a generic estimator we consider how it treats a single relevant item $\doc_n$ with $r(\doc_n) \not= 0$ using the expectation:
\begin{equation}
\begin{split}
&\mathbb{E}_{o, \displayranking}\Big[
\frac{c(\doc_n) \cdot \lambda\big(\doc_n \mid \ranking\big)}{ \rho\big( o(\doc_n) = 1 \mid q, \doc_n, \displayranking, \pi\big)}
\Big]
= \lambda\big(\doc_n \mid \ranking\big) \cdot r(\doc_n) \cdot \phantom{x} \\
& \bigg( \frac{\pi(\displayranking_1 | q) \cdot P\big( o(\doc_n) = 1 | q, r, \displayranking_1 \big)}{ \rho\big( o(\doc_n) = 1 \mid q, \doc_n, \displayranking_1, \pi \big)}
+ \frac{\pi(\displayranking_2 | q) \cdot P\big( o(\doc_n) = 1 | q, r, \displayranking_2 \big)}{ \rho\big( o(\doc_n) = 1 \mid q, \doc_n, \displayranking_2,  \pi \big)}
\bigg),
\end{split}
\end{equation}
where the propensity function $\rho$ can be chosen to match either the policy-oblivious (Eq.~\ref{eq:obliviousestimator}) or policy-aware (Eq.~\ref{eq:policyaware}) estimator.

First, we examine the situation where $\doc_n$ appears in the top-$k$ of both rankings $\displayranking_1$ and $\displayranking_2$, thus it has a positive observance probability in both cases: $P\big( o(\doc_n) = 1 \mid q, r, \displayranking_1 \big) > 0$ and $P\big( o(\doc_n) = 1 \mid q, r, \displayranking_2 \big) > 0$.
Here, the policy-oblivious estimator $\hat{\Delta}_{\textit{oblivious}}$~(Eq.~\ref{eq:obliviousestimator}) removes the effect of observation bias by adjusting for the observance probability per displayed ranking:
\begin{equation}
\begin{split}
\bigg(& \frac{\pi(\displayranking_1 | q) \cdot P\big( o(\doc_n) = 1 | q, r, \displayranking_1 \big)}{ P\big( o(\doc_n) = 1 \mid q, r, \displayranking_1 \big)}
+ \frac{\pi(\displayranking_2 | q) \cdot P\big( o(\doc_n) = 1 | q, r, \displayranking_2 \big)}{ P\big( o(\doc_n) = 1 \mid q, r, \displayranking_2 \big)}
\bigg)  \\
 &\cdot \lambda\big(\doc_n \mid \ranking\big) \cdot r(\doc_n) = \lambda\big(\doc_n \mid \ranking\big) \cdot r(\doc_n).
\end{split}
\end{equation}
The policy-aware estimator $\hat{\Delta}_{\textit{aware}}$~(Eq.~\ref{eq:policyaware}) also corrects for the examination bias, but because its propensity scores are based on the policy instead of the individual rankings~(Eq.~\ref{eq:expexam}), it uses the same score for both rankings:
\begin{equation}
\begin{split}
& \frac{\pi(\displayranking_1 \mid q) \cdot P\big( o(\doc_n) = 1 |\, q, r, \displayranking_1 \big) + \pi(\displayranking_2 \mid q) \cdot P\big( o(\doc_n) = 1 |\, q, r, \displayranking_2 \big)}{
\pi(\displayranking_1 \mid q) \cdot P\big( o(\doc_n) = 1 |\, q, r, \displayranking_1 \big) + \pi(\displayranking_2 \mid q) \cdot P\big( o(\doc_n) = 1 |\, q, r, \displayranking_2 \big)
}
   \\
 &\cdot \lambda\big(\doc_n \mid \ranking\big) \cdot r(\doc_n) = \lambda\big(\doc_n \mid \ranking\big) \cdot r(\doc_n).
\end{split}
\end{equation}
Then, we consider a different relevant item $\doc_m$ with $r(\doc_m) = r(\doc_n)$ that unlike the previous situation only appears in the top-$k$ of $\displayranking_1$.
Thus it only has a positive observance probability in $\displayranking_1$: $P\big( o(\doc_m) = 1 \mid q, r, \displayranking_1 \big) > 0$ and $P\big( o(\doc_m) = 1 \mid q, r, \displayranking_2 \big) = 0$.
Consequently, no clicks will ever be received in $\displayranking_2$ , i.e., $\displayranking = \displayranking_2 \rightarrow c(\doc_m) = 0$, thus the expectation for $\doc_m$ only has to consider $\displayranking_1$:
\begin{equation}
\begin{split}
\mbox{}\hspace*{-1mm}
\mathbb{E}_{o, \displayranking}&\Big[
\frac{c(\doc_m) \cdot \lambda\big(\doc_m \mid \ranking \big)}{ \rho\big( o(\doc_m) = 1 \mid q, \doc_m, \displayranking, \pi \big)}
\Big] \\
&\qquad =\frac{\pi(\displayranking_1 \mid q) \cdot P\big( o(\doc_m) = 1 \mid q, r, \displayranking_1 \big)}{ \rho\big( o(\doc_m) = 1 \mid q, \doc_m, \displayranking_1, \pi \big)} \cdot \lambda\big(\doc_m \,|\, \ranking \big) \cdot r(\doc_m).
\end{split}
\end{equation}
In this situation, Condition~\ref{eq:agnosticcond} is not satisfied, and correspondingly, the policy-oblivious estimator~(Eq.~\ref{eq:obliviousestimator}) does not give an unbiased estimate:
\begin{equation}
\frac{\pi(\displayranking_1 \mid q) \cdot P\big( o(\doc_m) = 1 \mid  q, r, \displayranking_1 \big)}{
P\big( o(\doc_m) = 1 \mid  q, r, \displayranking_1 \big)
} \cdot \lambda\big(\doc_m \mid  \ranking \big) \cdot r(\doc_m) 
 < \lambda\big(\doc_m \mid  \ranking \big) \cdot r(\doc_m).
\end{equation}
Since the policy-oblivious estimator $\hat{\Delta}_{\textit{oblivious}}$ only corrects for the observance probability per displayed ranking, it is unable to correct for the zero probability in $R_2$ as no clicks on $\doc_m$ can occur here.
As a result, the estimate is affected by the logging policy $\pi$: the more item-selection bias $\pi$ introduces (determined by $\pi(\displayranking_1 \mid q)$) the further the estimate will deviate.
Consequently, in expectation $\hat{\Delta}_{\textit{oblivious}}$ will biasedly estimate that $\doc_n$ should be ranked higher than $\doc_m$, which is incorrect since both items are actually equally relevant.

In contrast, the policy-aware estimator $\hat{\Delta}_{\textit{aware}}$~(Eq.~\ref{eq:policyaware}) avoids this issue because its propensities are based on the logging policy $\pi$.
When calculating the probability of observance conditioned on $\pi$, $P\big(o(\doc_m) = 1 \mid q, r, \pi\big)$~(Eq.~\ref{eq:expexam}), it takes into account that there is a $\pi(\displayranking_2 \mid q)$ chance that $\doc_m$ is not displayed to the user:
\begin{equation}
\frac{\pi(\displayranking_1 |\,  q) \cdot P\big( o(\doc_m) = 1 |\,  q, r, \displayranking_1 \big)}{
\pi(\displayranking_1 |\,  q) \cdot P\big( o(\doc_m) = 1 |\,  q, r, \displayranking_1 \big)
} 
\cdot  \lambda\big(\doc_m |\,  \ranking \big) \cdot r(\doc_m) 
= \lambda\big(\doc_m |\,  \ranking \big) \cdot r(\doc_m).
\end{equation}
Since in this situation Condition~\ref{eq:awarecond2} is true (and therefore also Condition~\ref{eq:awarecond}), we know beforehand that in expectation the policy-aware estimator is unaffected by position and item-selection bias.

This concludes our illustrative example.
It was meant to contrast the behavior of the policy-aware and policy-oblivious estimators in two different situations.
When there is no item-selection bias, i.e.,  an item is displayed in the top-$k$ of all rankings the logging policy may display, both estimators provide unbiased estimates albeit using different propensity scores.
However, when there is item-selection bias. i.e., an item is not always present in the top-$k$, the policy-oblivious estimator $\hat{\Delta}_{\textit{oblivious}}$ no longer provides an unbiased estimate, while the policy-aware estimator $\hat{\Delta}_{\textit{aware}}$ is still unbiased w.r.t.\ both position bias and item-selection bias.

%% file: 07-topk/sections/04-top-k-metrics.tex
\section{Learning for Top-$k$ Metrics}
\label{sec:topkmetrics}
This section details how counterfactual \ac{LTR} can be used to optimize top-$k$ metrics, since these are the relevant metrics in top-$k$ rankings.

\subsection{Top-$k$ metrics}

Since top-$k$ rankings only display the $k$ highest ranked items to the user, the performance of a ranker in this setting is only determined by those items.
Correspondingly, only top-$k$ metrics matter here, where items beyond rank $k$ have no effect:
\begin{equation}
\lambda^\textit{metric@k}\big(\doc \mid \ranking \big) =
\begin{cases}
\lambda^\textit{metric}\big(\doc \mid \ranking \big), & \text{if rank}( \doc \mid \ranking) \leq k, \\
0, & \text{if rank}\big(\doc\mid \ranking \big) > k. \\
\end{cases}
\end{equation}
These metrics are commonly used in \ac{LTR} since, usually, performance gains in the top of a ranking are the most important for the user experience.
For instance, NDCG@$k$, which is the normalized version of DCG@$k$, is often used:
\begin{equation}
\lambda^\textit{DCG@k}\big(\doc\,|\, \ranking \big) =
\begin{cases}
-\log_2\big(1 \,{+}\, \text{rank}(\doc\,|\, \ranking)\big)^{-1},\hspace*{-2mm}\mbox{} & \text{if rank}(\doc\,|\,\ranking) \,{\leq}\, k, \\
0, & \text{if rank}\big(\doc\,|\, \ranking \big) \,{>}\, k.
\end{cases}
\end{equation}
Generally in \ac{LTR}, DCG is optimized in order to maximize NDCG~\citep{wang2018lambdaloss, burges2010ranknet}.
In unbiased \ac{LTR} it is not trivial to estimate the normalization factor for NDCG, further motivating the optimization of DCG instead of NDCG~\citep{agarwal2019counterfactual, carterette2018offline}.

Importantly, top-$k$ metrics bring two main challenges for \ac{LTR}.
First, the \text{rank} function is not differentiable, a problem for almost every \ac{LTR} metric~\citep{wang2018lambdaloss, liu2009learning}.
Second,
changes in a ranking beyond position $k$ do not affect the metric's value thus resulting in zero-gradients.
The first problem has been addressed in existing \ac{LTR} methods, we will now propose adaptations of these methods that address the second issue as well.

\subsection{Monotonic upper bounding}
\label{section:monotonic upper bounding}
A common approach for enabling optimization of ranking methods, is by finding lower or upper bounds that can be minimized or maximized, respectively.
For instance, similar to a hinge loss, the \text{rank} function can be upper bounded by a maximum over score differences~\citep{Joachims2002, joachims2017unbiased}.
Let $s$ be the scoring function used to rank (in descending order), then:
\begin{equation}
\text{rank}\big(\doc\mid \ranking \big) \leq \sum_{\doc' \in \ranking} \max\Big(1 - \big(s(\doc) - s(\doc')\big), 0\Big). \label{eq:linearupperbound}
\end{equation}
Alternatively, the logistic function is also a popular choice~\citep{wang2018lambdaloss}:
\begin{equation}
\text{rank}\big(\doc\mid \ranking \big) \leq \sum_{\doc' \in \ranking} \log_2\Big(1 + e^{s(\doc') - s(\doc)}\Big).\label{eq:logupperbound}
\end{equation}
Minimizing one of these differentiable upper bounds will directly minimize an upper bound on the ARP metric (Eq.~\ref{eq:ARP}).

Furthermore, \citeauthor{agarwal2019counterfactual}~\citep{agarwal2019counterfactual} showed that this approach can be extended to any metric based on a monotonically decreasing function.
For instance, if $\overline{\text{rank}}\big(\doc\mid \ranking \big)$ is an upper bound on the $\text{rank}\big(\doc\mid \ranking \big)$ function, then the following is an upper bound on the DCG loss (Eq.~\ref{eq:DCG}):
\begin{equation}
\lambda^\textit{DCG}\big(\doc\mid \ranking \big) \leq -\log_2\big(1 + \overline{\text{rank}}(\doc\mid \ranking )\big)^{-1} = \hat{\lambda}^\textit{DCG}\big(\doc\mid \ranking \big). 
\end{equation}
More generally, let $\alpha$ be a monotonically decreasing function. 
A loss based on $\alpha$ is always upper bounded by:
\begin{equation}
\lambda^\alpha\big(\doc\mid \ranking \big) 
= -\alpha\big(\text{rank}(\doc\mid \ranking)\big) 
\leq -\alpha\big(\overline{\text{rank}}(\doc\mid \ranking )\big)
= \hat{\lambda}^\alpha\big(\doc\mid \ranking \big).
\end{equation}
Though appropriate for many standard ranking metrics, $\hat{\lambda}^\alpha$ is not an upper bound for top-$k$ metric losses.
To understand this, consider that an item beyond rank $k$ may still receive a negative score from $\hat{\lambda}^\alpha$,
for instance, for the DCG upper bound: $\text{rank}\big(\doc\mid \ranking \big) > k \rightarrow \hat{\lambda}^\textit{DCG}\big(\doc\mid \ranking \big) < 0$.
As a result, this is not an upper bound for a DCG@$k$ based loss.

We propose a modification of the $\hat{\lambda}^\alpha$ function to provide an upper bound for top-$k$ metric losses, by simply giving a positive penalty to items beyond rank $k$:
\begin{equation}
\hat{\lambda}^{\alpha\textit{@k}}\big(\doc\mid \ranking \big) = -\alpha\big(\overline{\text{rank}}(\doc\mid \ranking)\big) + \mathds{1}\big[\text{rank}(\doc\mid \ranking) > k\big] \cdot \alpha(k).
 \label{eq:kmonotonic}
\end{equation}
The resulting function is an upper bound on top-$k$ metric losses based on a monotonic function: $\lambda^{\alpha\textit{@k}}\big(\doc\mid \ranking \big) \leq \hat{\lambda}^{\alpha\textit{@k}}\big(\doc\mid \ranking \big)$.
The main difference with $\hat{\lambda}^\alpha$ is that items beyond rank $k$ acquire a positive score from $\lambda^{\alpha\textit{@k}}$, thus providing an upper bound on the actual metric loss.
Interestingly, the gradient of $\hat{\lambda}^{\alpha\textit{@k}}$ w.r.t.\ the scoring function $s$ is the same as that of $\hat{\lambda}^{\alpha}$.\footnote{We consider the indicator function to never have a non-zero gradient.}
Therefore, the gradient of either function optimizes an upper bound on  $\lambda^{\alpha\textit{@k}}$ top-$k$ metric losses, while only $\hat{\lambda}^{\alpha\textit{@k}}$ provides an actual upper bound.

While this monotonic function-based approach is simple, it is unclear how coarse these upper bounds are.
In particular, some upper bounds on the rank function (e.g., Eq.~\ref{eq:linearupperbound}) can provide gross overestimations. %
As a result, these upper bounds on ranking metric losses may be very far removed from their actual values.

\subsection{Lambda-based losses for counterfactual top-$k$ \acs{LTR}}
\label{section:lambda-based losses}
Many supervised \ac{LTR} approaches, such as  the well-known LambdaRank and subsequent LambdaMART methods~\citep{burges2010ranknet}, are based on \ac{EM} procedures~\citep{dempster1977maximum}.
Recently, \citeauthor{wang2018lambdaloss}~\citep{wang2018lambdaloss} introduced the Lambda\-Loss framework, which provides a theoretical way to prove that a method optimizes a lower bound on a ranking metric.
Subsequently, it was used to prove that Lambda\-MART optimizes such a bound on \ac{DCG}, similarly it was also used to introduce the novel Lambda\-Loss method which provides an even tighter bound on \ac{DCG}.
In this section, we will show that the Lambda\-Loss framework can be used to find proven bounds on counterfactual \ac{LTR} losses and top-$k$ metrics.
Since Lambda\-Loss is considered state-of-the-art in supervised \ac{LTR}, making its framework applicable to counterfactual \ac{LTR} could potentially provide competitive performance.
Additionally, adapting the Lambda\-Loss framework to top-$k$ metrics further expands its applicability.

The Lambda\-Loss framework and its \ac{EM}-optimization approach work for metrics that can be expressed in item-based gains, $G(\doc_n \mid q,r)$, and discounts based on position, $D\big(\text{rank}(\doc_n \mid \ranking)\big)$; for brevity we use the shorter $G_n$ and $D_n$, respectively, resulting in:
\begin{equation}
\Delta\big(\ranking \,|\, q, r \big)
=
\sum_{\doc_n \in \ranking} G(\doc_n | \, q, r) \cdot D\big(\text{rank}(\doc_n | \, \ranking)\big)
=
\sum_{n=1}^{|\ranking|} G_n \cdot D_n.
\end{equation}
For simplicity of notation, we choose indexes so that: $n = \text{rank}(\doc_n \mid \ranking)$, thus $D_n$ is always the discount for the rank $n$.
Then, we differ from the existing Lambda\-Loss framework by allowing the discounts to be zero ($\forall n \, D_n \geq 0$), thus also accounting for top-$k$ metrics.
Furthermore, items at the first rank are not discounted or the metric can be scaled so that $D_1 = 1$.
Additionally, higher ranked items should be discounted less or equally: $n > m \rightarrow D_n \leq D_m$.
Most ranking metrics meet these criteria; for instance, $G_n$ and $D_n$ can be chosen to match \ac{ARP} or \ac{DCG}.
Importantly, our adaption also allows $\Delta$ to match top-$k$ metrics such as \ac{DCG}$@k$ or Precision$@k$.

In order to apply the Lambda\-Loss framework to counterfactual \ac{LTR}, we consider a general inverse-propensity-scored estimator:
\begin{equation}
\hat{\Delta}_{\textit{IPS}}(\ranking \mid q, c, \cdot)
= \sum_{\doc_n:c(\doc_n) = 1} \frac{\lambda(\doc_n\mid \ranking)}{\rho\big( o(\doc_n) = 1 \mid q, r, \displayranking, \pi \big) },
\label{eq:generalestimator}
\end{equation}
where the propensity function $\rho$ can match either the policy-obli\-vious (Eq.~\ref{eq:obliviousestimator}) or the policy-aware (Eq.~\ref{eq:policyaware}) estimator.
By choosing 
\begin{equation}
G_n = \frac{1}{\rho\big( o(\doc_n) = 1 \mid q, r, \displayranking, \pi  \big) }\text{ and }D_n = \lambda(\doc_n\mid \ranking),
\label{eq:gaindiscountchoice}
\end{equation}
the estimator can be described in terms of gains and discounts.
In contrast, in the existing Lambda\-Loss framework~\citep{wang2018lambdaloss} gains are based on item relevance. %
For counterfactual top-$k$ \ac{LTR}, we have designed Eq.~\ref{eq:gaindiscountchoice} so that gains are based on the propensity scores of observed clicks, and the discounts can have zero values.

The \ac{EM}-optimization procedure alternates between an expectation step and a maximization step.
In our case, the expectation step sets the discount values $D_n$ according to the current ranking $\ranking$ of the scoring function $s$.
Then the maximization step updates $s$ to optimize the ranking model.
Following the Lambda\-Loss framework~\citep{wang2018lambdaloss}, we derive a slightly different loss.
With the delta function:
\begin{equation}
\delta_{nm} = D_{|n - m|} - D_{|n - m| + 1},
\end{equation}
our differentiable \emph{counterfactual loss} becomes:
\begin{equation}
\sum_{G_n > G_m} -\log_2 \left( \left (\frac{1}{1 + e^{s(\doc_m) - s(\doc_n)}} \right )^{\delta_{nm} \cdot |G_n - G_m|} \right). \label{eq:embound}
\end{equation}
The changes we made do not change the validity of the proof provided in the original Lambda\-Loss paper~\citep{wang2018lambdaloss}.
Therefore, the counterfactual loss (Eq.~\ref{eq:embound}) can be proven to optimize a lower bound on counterfactual estimates of top-$k$ metrics.

Finally, in the same way the LambdaLoss framework can also be used to derive counterfactual variants of other supervised \ac{LTR} losses/methods such as LambdaRank or LamdbaMART.
Unlike previous work that also attempted to find a counterfactual lambda-based method by introducing a pairwise-based estimator~\cite{hu2019unbiased}, our approach is compatible with the prevalent counterfactual approach since it uses the same estimator based on single-document propensities~\citep{wang2016learning, joachims2017unbiased, agarwal2019addressing, ai2018unbiased, wang2018position, agarwal2019estimating}.
Our approach suggests that the divide between supervised and counterfactual \ac{LTR} methods may disappear in the future, as a state-of-the-art supervised \ac{LTR} method can now be applied to the state-of-the-art counterfactual \ac{LTR} estimators.

\subsection{Unbiased loss selection}
So far we have introduced two counterfactual \ac{LTR} approaches that are proven to optimize lower bounds on top-$k$ metrics: with monotonic functions~(Section~\ref{section:monotonic upper bounding}) and through the LambdaLoss framework~(Section~\ref{section:lambda-based losses}).
To the best of our knowledge, we are the first to introduce theoretically proven lower bounds for top-$k$ \ac{LTR} metrics.
Nevertheless, previous work has also attempted to optimize top-$k$ metrics, albeit through heuristic methods.
Notably, \citet{wang2018lambdaloss} used a truncated version of the LambdaLoss loss to optimize \ac{DCG}$@k$.
Their loss uses the discounts $D_n$ based on full-ranking \ac{DCG} but ignores item pairs outside of the top-$k$: 
\begin{equation}
\sum_{G_n > G_m} -\mathds{1}\big[ n \leq k \lor m \leq k\big] \cdot 
 \log_2  \left( \left ( \frac{1}{1 + e^{s(\doc_m) - s(\doc_n)}} \right )^{\delta_{nm} \cdot |G_n - G_m|}  
 \right ).
 \label{eq:truncembound}
\end{equation}
While empirical results motivate its usage, there is no known theoretical justification for this loss, and thus it is considered a heuristic.

This leaves us with a choice between two theoretically-motivated counterfactual \ac{LTR} approaches for optimizing top-$k$ metrics (Eq.~\ref{eq:kmonotonic}~and~\ref{eq:embound}) and an empirically-motivated heuristic (Eq.~\ref{eq:truncembound}).
We propose a pragmatic solution by recognizing that counterfactual estimators can unbiasedly evaluate top-$k$ metrics.
Therefore, in practice one can optimize several ranking models using various approaches, and subsequently, estimate which resulting model provides the best performance.
Thus, using counterfactual evaluation to select from resulting models is an unbiased method to choose between the available counterfactual \ac{LTR} approaches.

%% file: 07-topk/sections/05-experiments.tex
\section{Experimental Setup}
We follow the standard setup in unbiased \ac{LTR}~\citep{joachims2017unbiased, ai2018unbiased, carterette2018offline, jagerman2019comparison} and perform semi-synthetic experiments: queries and items are based on datasets of commercial search engines and interactions are simulated using probabilistic click models.

\subsection{Datasets}
We use the queries and documents from two of the largest publicly available \ac{LTR} datasets: MLSR-WEB30K~\citep{qin2013introducing} and Yahoo!\ Webscope~\citep{Chapelle2011}.
 Each was created by a commercial search engine and contains a set of queries with corresponding preselected document sets.
 Query-document pairs are represented by feature vectors and five-grade relevance annotations ranging from not relevant (0) to perfectly relevant (4).
 In order to binarize the relevancy, we only consider the two highest relevance grades as relevant. %
 The MSLR dataset contains \numprint{30000} queries with on average 125 preselected documents per query, and encodes query-document pairs in 136 features.
 The Yahoo dataset has \numprint{29921} queries and on average 24 documents per query encoded in 700 features. 
Presumably, learning from top-$k$ feedback is harder as $k$ becomes a smaller percentage of the number of items.
Thus, we expect the MSLR dataset with more documents per query to pose a more difficult problem.
 
\subsection{Simulating top-$k$ settings}

The setting we simulate is one where interactions are gathered using a non-optimal but decent production ranker.
We follow existing work~\citep{joachims2017unbiased, ai2018unbiased, jagerman2019comparison} and use supervised optimization for the \ac{ARP} metric on 1\% of the training data.
The resulting model simulates a real-world production ranker since it is much better than a random initialization but leaves enough room for improvement~\citep{joachims2017unbiased}.

We then simulate user-issued queries by uniformly sampling from the training partition of the dataset.
Subsequently, for each query the production ranker ranks the documents preselected by the dataset.
Depending on the experimental run that we consider, randomization is performed on the resulting rankings.
In order for the policy-aware estimator to be unbiased, every relevant document needs a chance of appearing in the top-$k$ (Condition~\ref{eq:awarecond2}).
Since in a realistic setting relevancy is unknown, we choose to give every document a non-zero probability of appearing in the top-$k$.
Our randomization policy takes the ranking of the production ranker and leaves the first $k-1$ documents unchanged
but the document at position $k$ is selected by sampling uniformly from the remaining documents.
The result is a minimally invasive randomized top-$k$ ranking since most of the ranking is unchanged and the placement of the sampled documents is limited to the least important position.

We note that many other logging policies could be applied (see Condition~\ref{eq:awarecond2}), e.g., an alternative policy could insert sampled documents at random ranks for less obvious randomization.
Unfortunately, a full exploration of the effect of using different logging policies is beyond the scope of this chapter.

Clicks are simulated on the resulting ranking $\displayranking$ according to position bias and document relevance.
Top-$k$ position bias is modelled through the probability of observance, as follows:
\begin{equation}
P\big(o(\doc) = 1 \mid q, r, \displayranking \big) =
\begin{cases}
\text{rank}(\doc\mid \displayranking)^{-1},
& \text{if rank}(\doc\mid \displayranking) \leq k,
\\
0,
& \text{if rank}(\doc\mid \displayranking) > k.
\end{cases}
\end{equation}
The randomization policy results in the following examination probabilities w.r.t.\ the logging policy (cf.\ Eq.~\ref{eq:expexam}):
\begin{equation}
\begin{split}
P\big(o(\doc) = 1 \mid q, r, \pi \big) 
\hspace{-1cm}& \\ & =
\begin{cases}
\text{rank}(\doc\mid \displayranking)^{-1}, & \text{if rank}(\doc\mid \displayranking) < k, \\
\big(\text{rank}(\doc\mid \displayranking) \cdot (|\displayranking| - k + 1)\big)^{-1},  & \text{if rank}(\doc\mid\displayranking) \geq k. 
\end{cases}
\end{split}
\end{equation}
The probability of a click is conditioned on the relevance of the document according to the dataset:
\begin{equation}
P\big(c(\doc) = 1 \mid q, r, \displayranking, o\big) =
\begin{cases}
1,
& \text{if } r(\doc) = 1 \land o(\doc) = 1,
\\
0.1,
& \text{if } r(\doc) = 0 \land o(\doc) = 1,
\\
0,
& \text{if } o(\doc) = 0.
\end{cases}
\end{equation}
Note that our previous assumption that clicks only take place on relevant items (Section~\ref{section:counterfactualLTR}) is not true in our experiments.

Optimization is performed on training clicks simulated on the training partition of the dataset.
Hyperparameter tuning is done by estimating performance on (unclipped) validation clicks simulated on the validation partition; the number of validation clicks is always 15\% of the number of training clicks.
Lastly, evaluation metrics are calculated on the test partition using the dataset labels.

\subsection{Experimental runs}

In order to evaluate the performance of the policy-aware estimator (Eq.~\ref{eq:policyaware}) and the effect of item-selection bias, we compare with the following baselines: 
\begin{enumerate*}[label=(\roman*)]
\item The policy-oblivious estimator (Eq.~\ref{eq:obliviousestimator}).
In our setting, where the examination probabilities are known beforehand, the policy-oblivious estimator also represents methods that jointly estimate these probabilities while performing \ac{LTR}, i.e., the following methods reduce to this estimator if the examination probabilities are given:~\citep{wang2016learning, joachims2017unbiased, agarwal2019addressing, ai2018unbiased}.
\item A rerank estimator, an adaption of the policy-oblivious estimator.
During optimization the rerank estimator applies the policy-oblivious estimator but limits the document set of an interaction $i$ to the $k$ displayed items $\ranking_i = \{\doc \mid \text{rank}(\doc|\displayranking_i) \leq k\}$ (cf.\ Eq.~\ref{eq:obliviousestimator}).
Thus, it is optimized to rerank the top-$k$ of the production ranker only, but during inference it is applied to the entire document set.
\item Additionally, we evaluate performance without any cutoff $k$ or randomization; in these circumstances all three estimators (Policy-Aware, Policy-Oblivious, Rerank) are equivalent.
\item Lastly, we use supervised \ac{LTR} on the dataset labels to get a \emph{full-information skyline}, which shows the hypothetical optimal performance.
\end{enumerate*}

To evaluate the effectiveness of our proposed loss functions for optimizing top-$k$ metrics, we apply the monotonic lower bound (Eq.~\ref{eq:kmonotonic}) with a linear (Eq.~\ref{eq:linearupperbound}) and a logistic upper bound (Eq.~\ref{eq:logupperbound}).
Additionally, we apply several versions of the LamdbaLoss loss function (Eq.~\ref{eq:embound}): one that optimizes full \ac{DCG}, another that optimizes \ac{DCG}$@5$, and the heuristic truncated loss also optimizing \ac{DCG}$@5$ (Eq.~\ref{eq:truncembound}).
Lastly, we apply unbiased loss selection where we select the best-performing model based on the estimated performance on the (unclipped) validation clicks.

Optimization is done with stochastic gradient descent; to maximize computational efficiency we rewrite the loss (Eq.~\ref{eq:highlevelloss}) for a propensity scoring function $\rho$ in the following manner:
\begin{equation}
\begin{split}
\hat{\mathcal{L}} &=
\frac{1}{N} \sum^N_{i=1} \hat{\Delta}\big(\ranking_i | q_i, \displayranking_i, \pi, c_i\big)
\\
&=
\frac{1}{N} \sum^N_{i=1} \sum_{\doc : c_i(\doc) = 1} \frac{\lambda\big(\doc\mid \ranking_i \big)}{\rho\big( o_i(\doc) = 1 | q_i, r, \cdot \big)}
\\
&= \frac{1}{N}  \sum_{q \in \mathcal{Q}} \sum_{\doc \in \ranking_q} \left ( \sum^N_{i = 1} \frac{ \mathds{1}[q_i = q] \cdot c_i(\doc)}{\rho\big( o_i(\doc) = 1 \mid q, r, \cdot\big)} \right ) \cdot  \lambda\big(\doc \mid \ranking_q\big) 
\\
&= \frac{1}{N}  \sum_{q \in \mathcal{Q}} \sum_{\doc \in \ranking_q}  \omega_{\doc} \cdot \lambda\big(\doc \mid \ranking_q \big).
\end{split}
\end{equation}
After precomputing the document weights $\omega_{\doc}$,
the complexity of computing the loss is only determined by the dataset size.
This allows us to optimize over very large numbers of clicks with very limited increases in computational costs.

We optimize linear models, but our approach can be applied to any differentiable model~\citep{agarwal2019counterfactual}.
Propensity clipping~\citep{joachims2017unbiased} is applied to training clicks and never applied to the validation clicks; we also use self-normalization~\citep{swaminathan2015self}.

%% file: 07-topk/sections/06-results.tex
\begin{figure*}[t]
\centering
\begin{tabular}{r }
\multicolumn{1}{c}{ \hspace{2em} {Yahoo!\ Webscope}} \\
\includegraphics[scale=0.5]{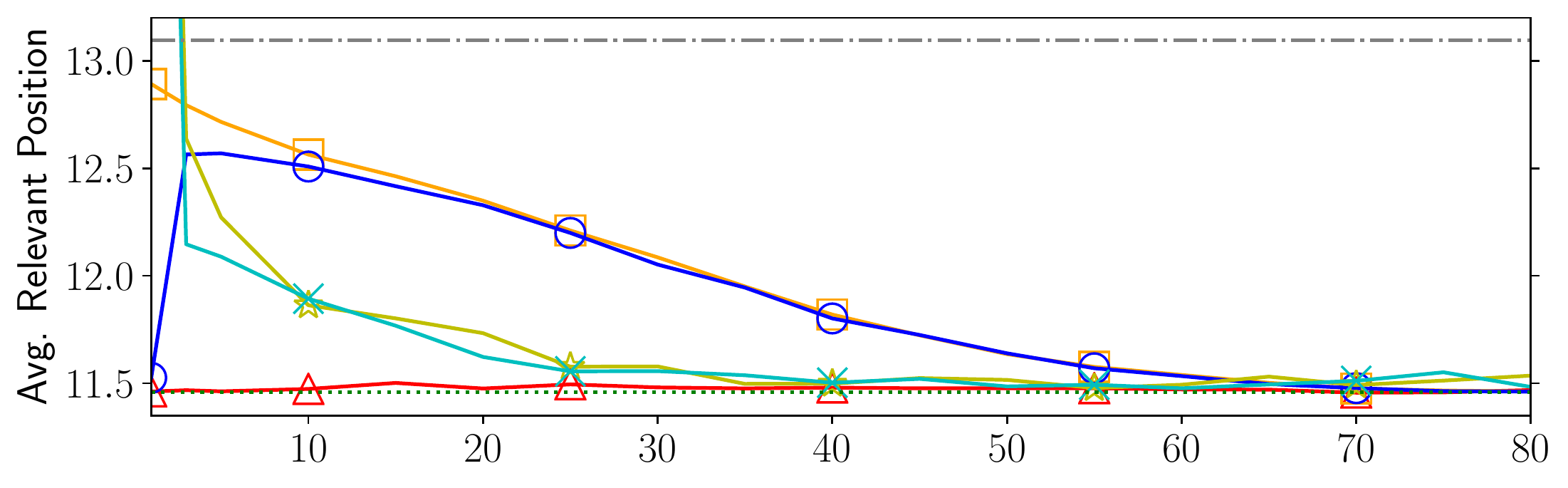} \\
\includegraphics[scale=0.5]{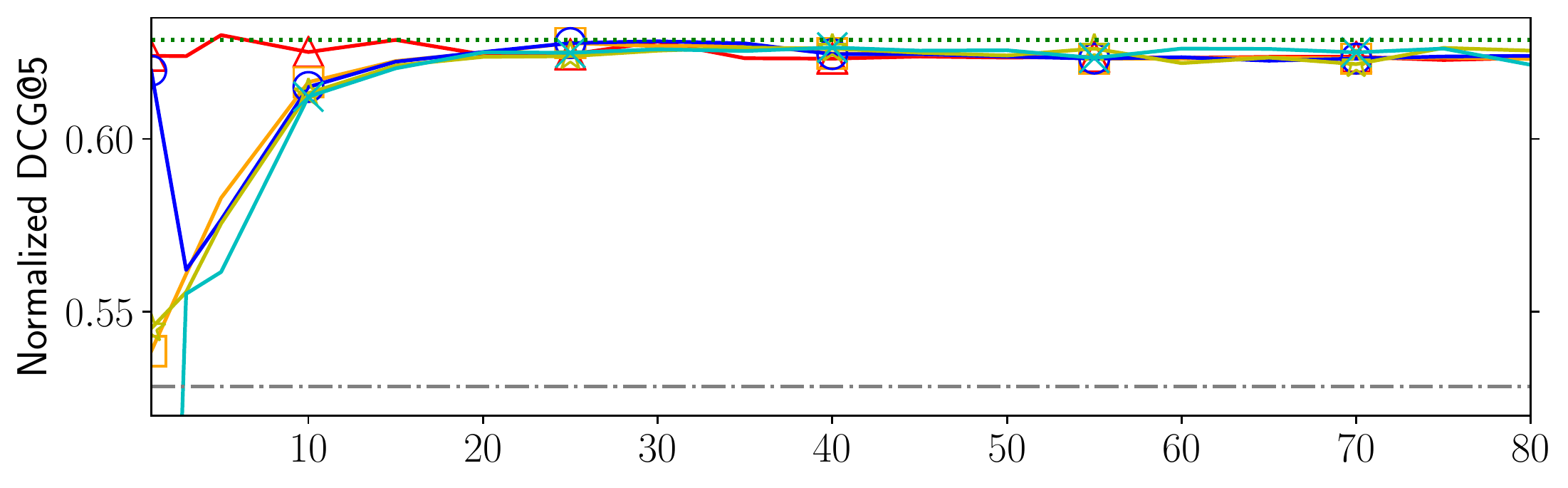} \\
\multicolumn{1}{c}{ \hspace{2em} {MSLR-WEB30k}} \\
\includegraphics[scale=0.5]{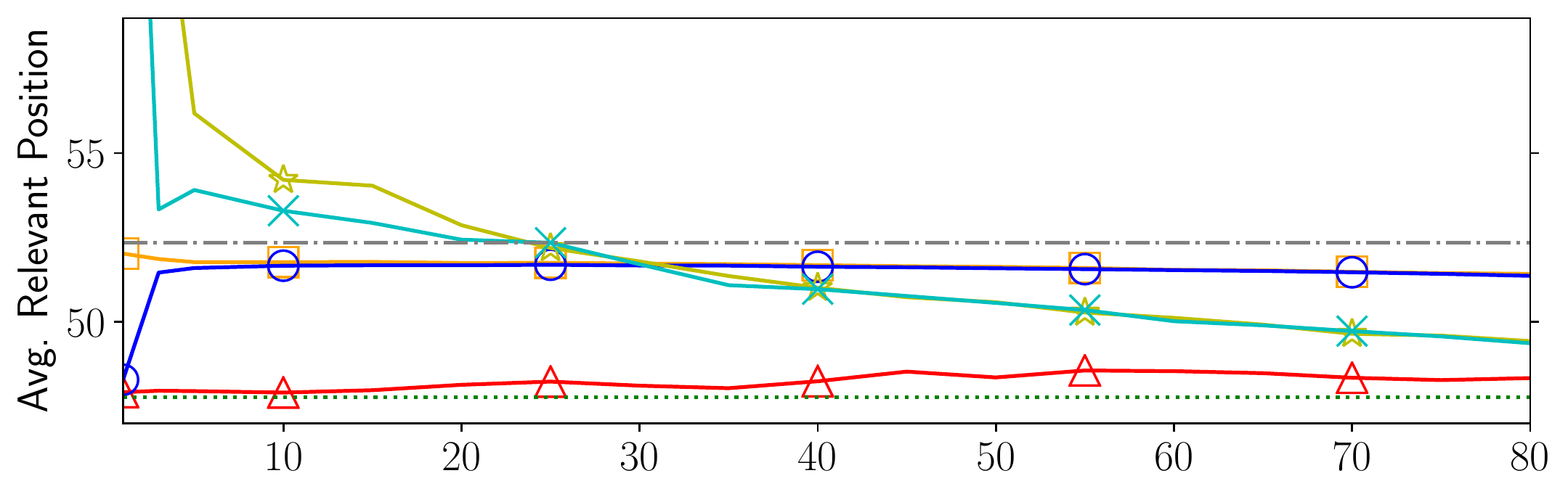} \\
\includegraphics[scale=0.5]{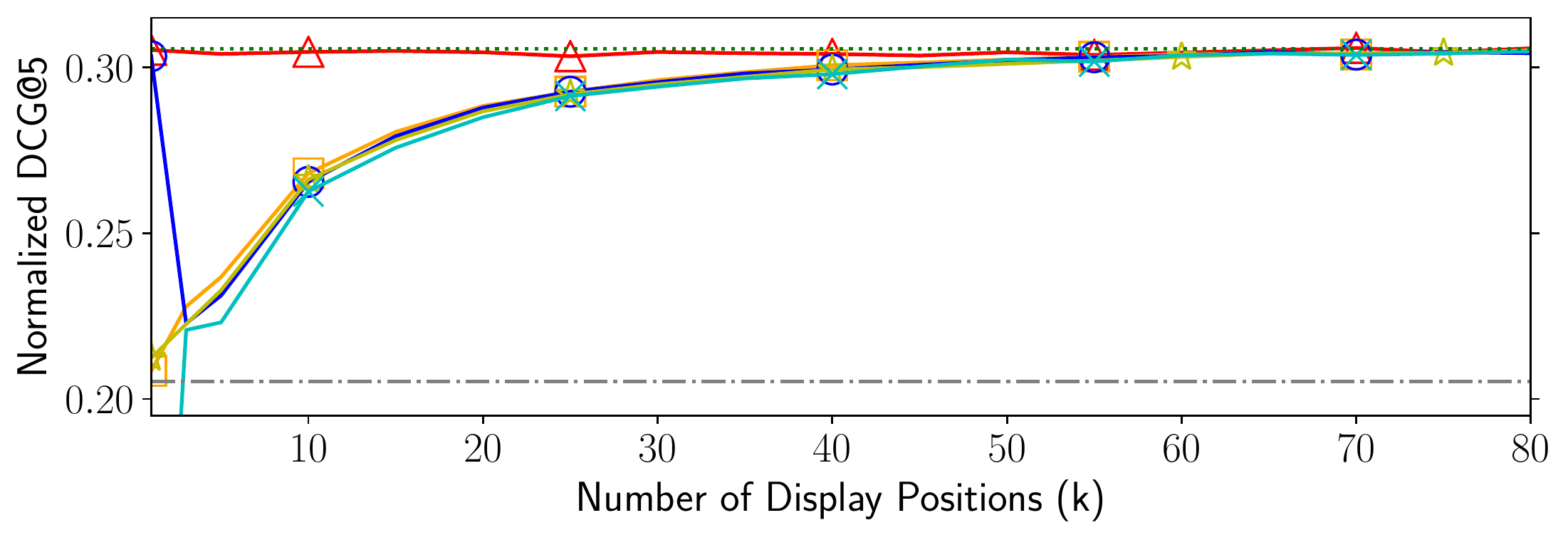} \\
 \includegraphics[scale=.29]{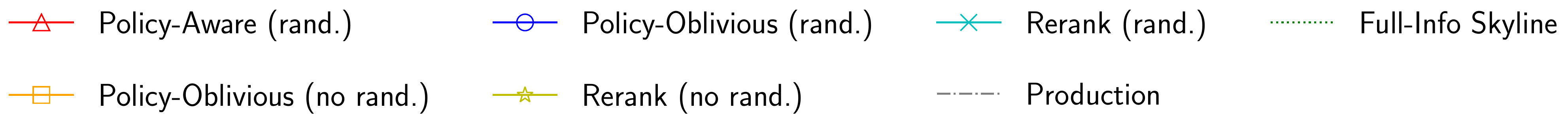} 
\end{tabular}
\caption{
The effect of item-selection bias on different estimators.
Optimization on $10^8$ clicks simulated on top-$k$ rankings with varying number of display positions ($k$), with and without randomization (for each datapoint $10^8$ clicks were simulated independently).
Results on the \emph{Yahoo} dataset and \emph{MSLR} dataset. 
The top graph per dataset optimizes the average relevance position through the linear upper bound (Eq.~\ref{eq:linearupperbound});
the bottom graph per dataset optimizes DCG$@5$ using the truncated LambdaLoss (Eq.~\ref{eq:truncembound}).
}
\label{fig:effectk}
\end{figure*}

\begin{figure*}[t]
\centering
\begin{tabular}{r}
\multicolumn{1}{c}{ \hspace{2em} {Yahoo!\ Webscope}}\\
\includegraphics[scale=0.5]{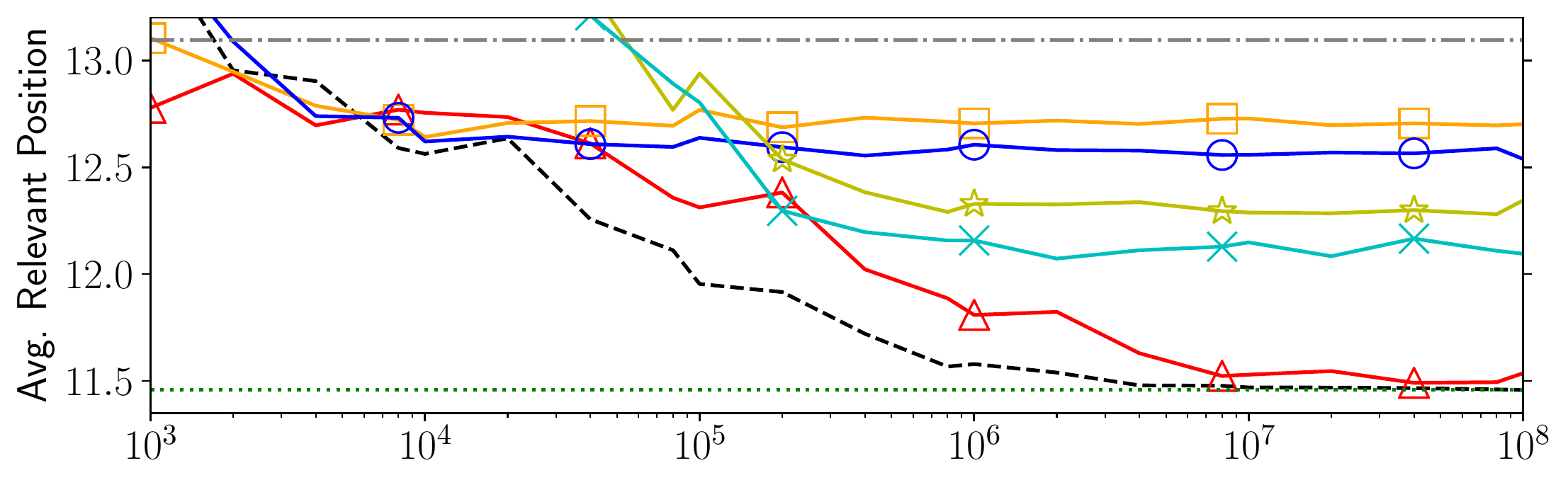} \\
\includegraphics[scale=0.5]{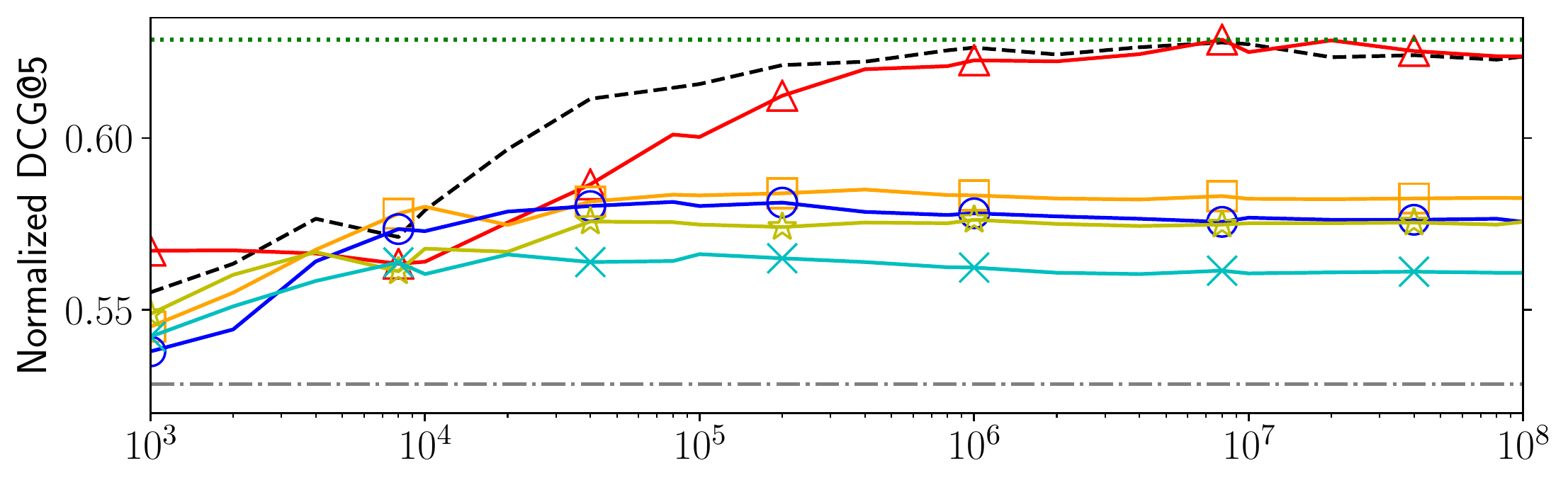} \\
\multicolumn{1}{c}{ \hspace{2em} {MSLR-WEB30k}} \\
\includegraphics[scale=0.5]{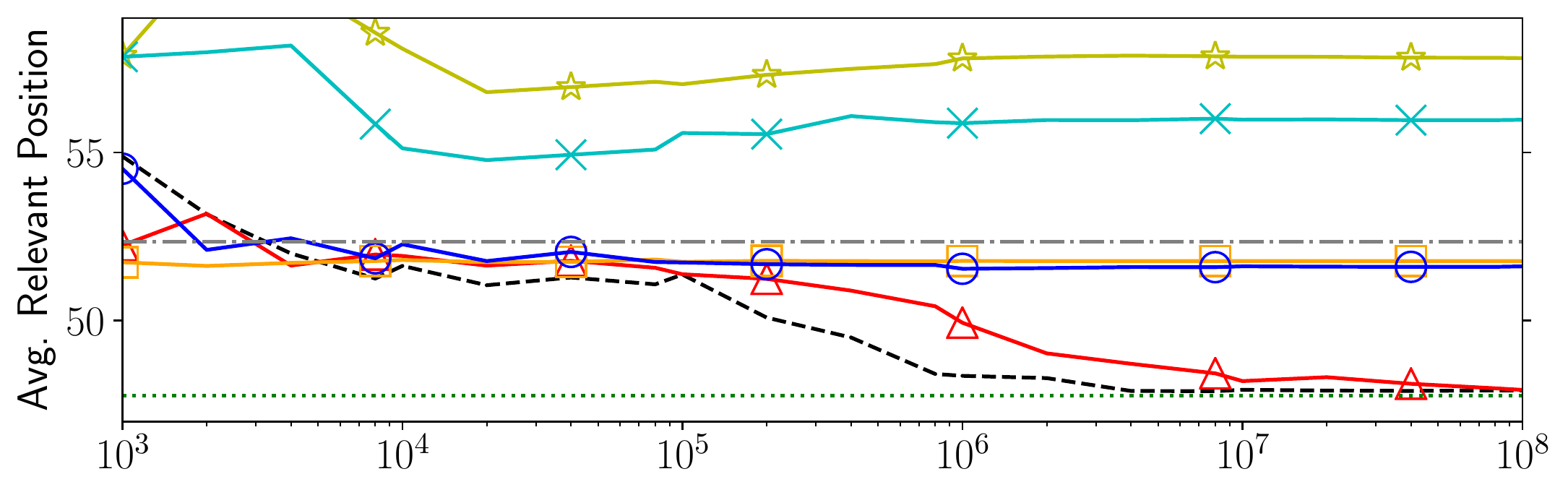} \\
\includegraphics[scale=0.5]{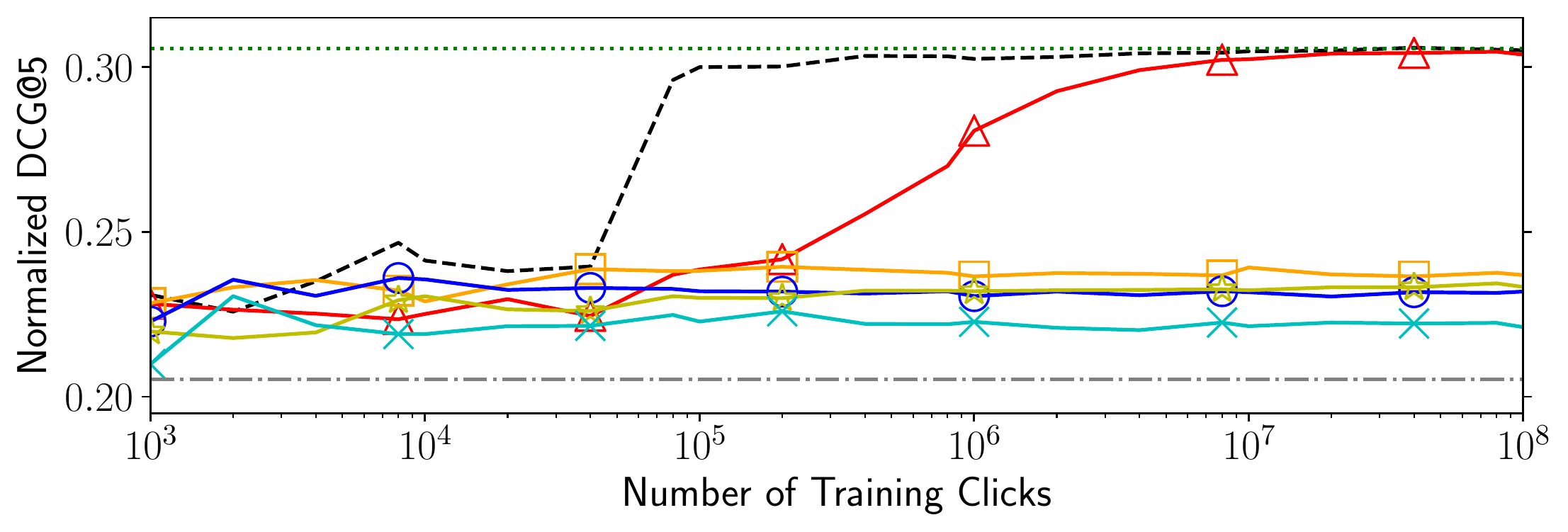} \\
 \includegraphics[scale=.29]{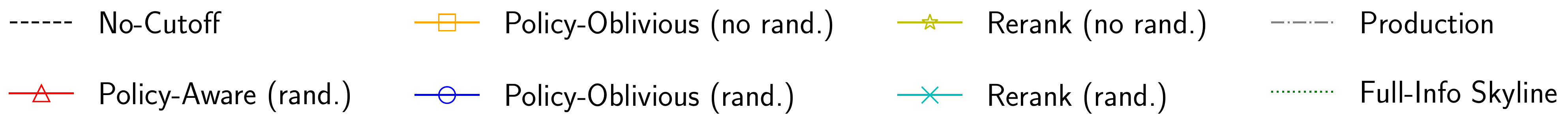}
\end{tabular}
\caption{
Performance of different estimators learning from different numbers of clicks simulated on top-5 rankings, with and without randomization.
Results on the \emph{Yahoo} dataset and \emph{MSLR} dataset. 
The top graph per dataset optimizes the average relevance position through the linear upper bound (Eq.~\ref{eq:linearupperbound});
the bottom graph per dataset optimizes DCG$@5$ using the truncated LambdaLoss (Eq.~\ref{eq:truncembound}).
}
\label{fig:cutoffbaseline}
\end{figure*}

\begin{figure*}[t]
\centering
\begin{tabular}{r}
\multicolumn{1}{c}{ \hspace{2em}  {Yahoo!\ Webscope}} \\
\includegraphics[scale=0.5]{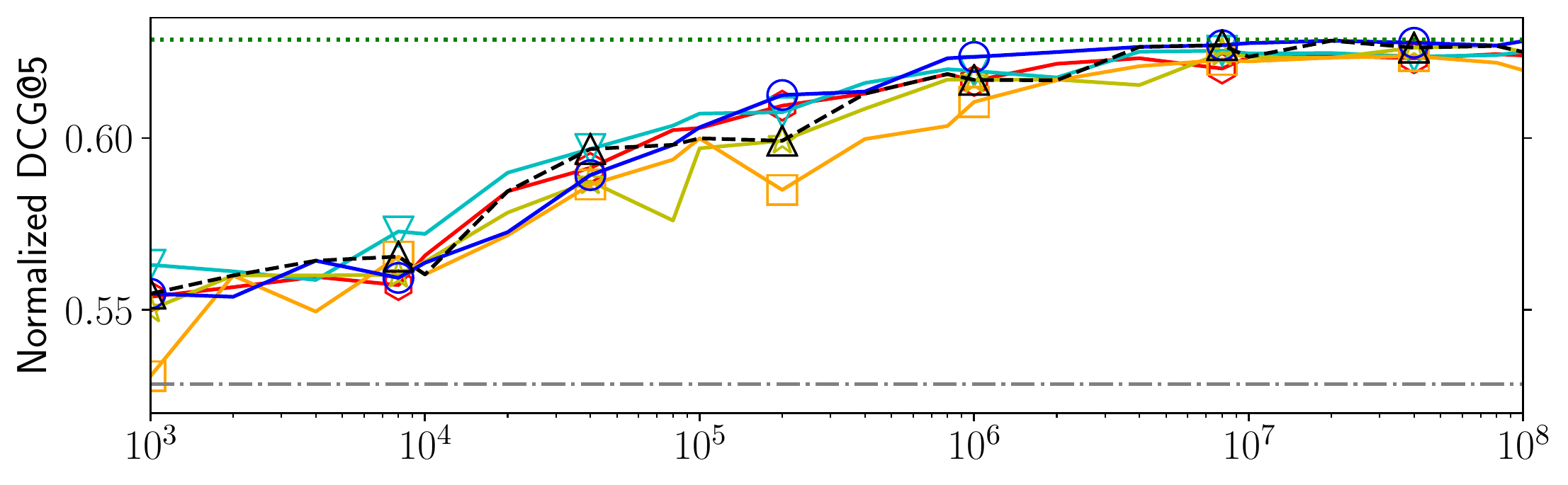} \\
\includegraphics[scale=0.5]{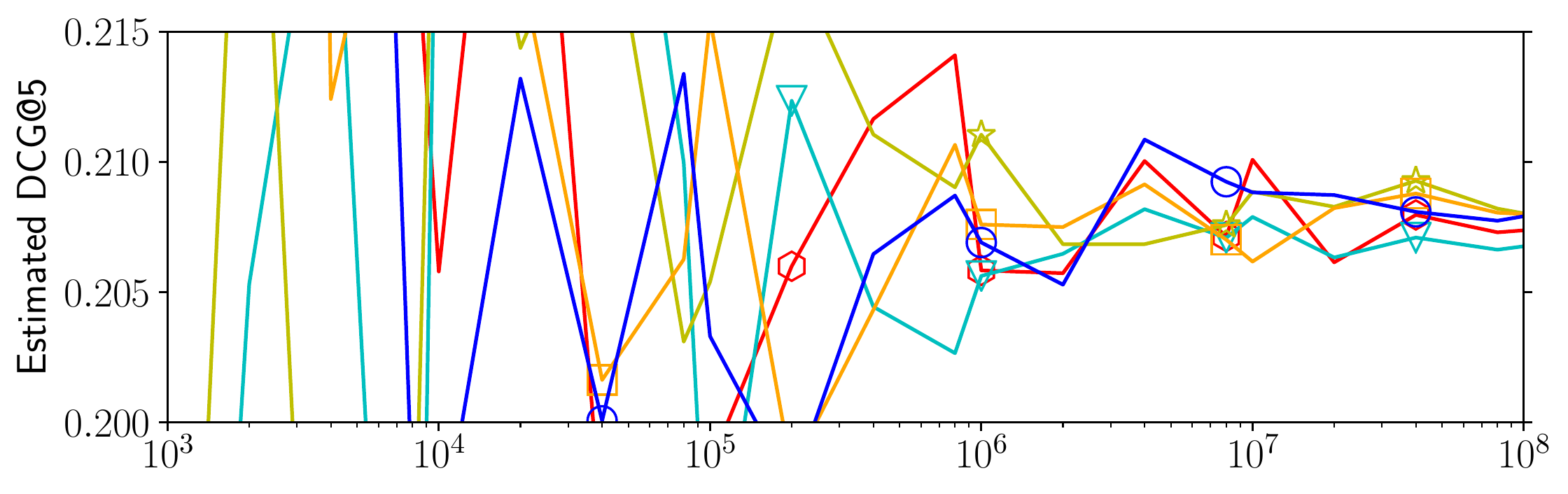} \\
\multicolumn{1}{c}{ \hspace{2em}  {MSLR-WEB30k}} \\
\includegraphics[scale=0.5]{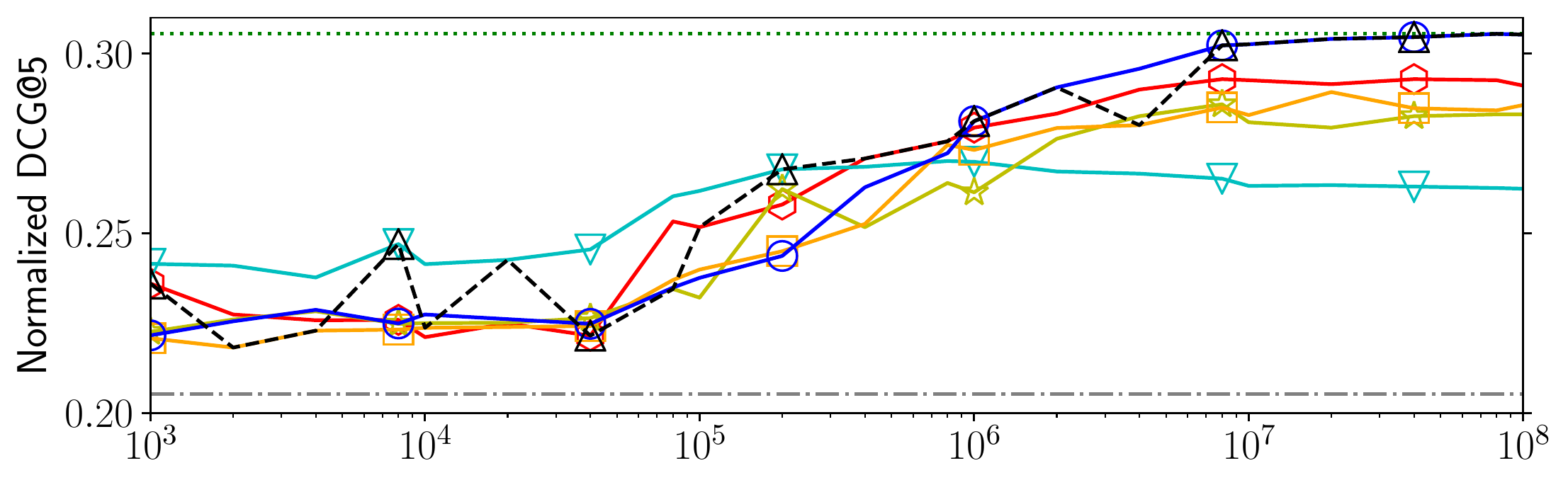} \\
\includegraphics[scale=0.5]{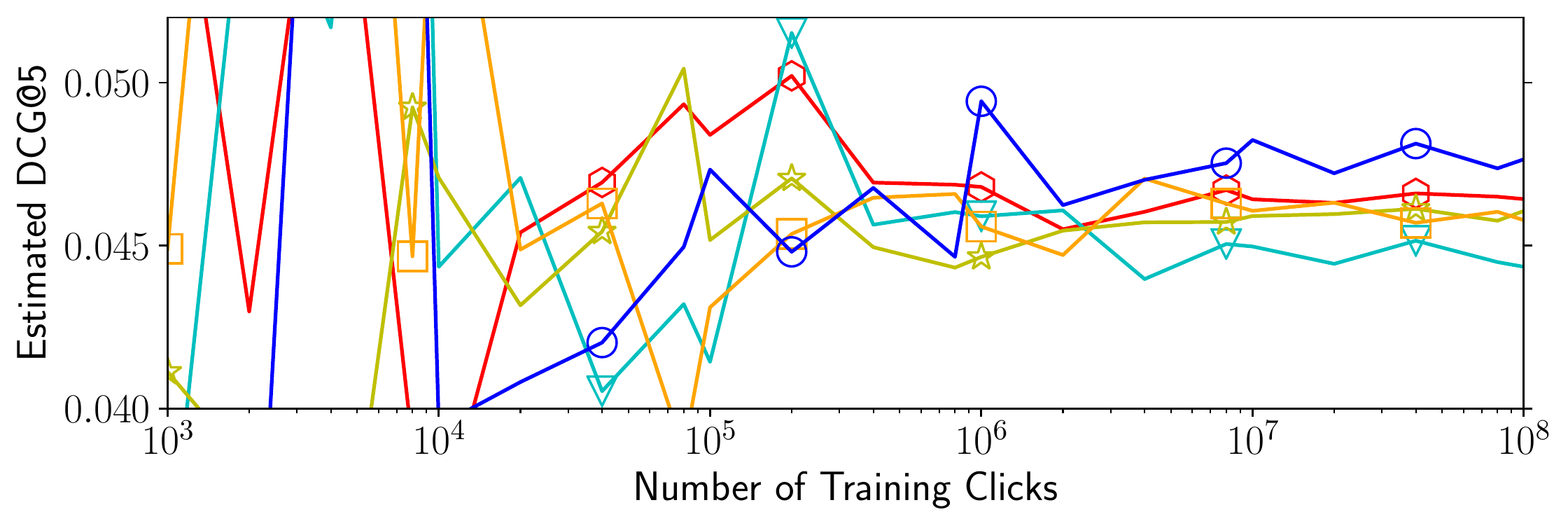} \\
 \includegraphics[scale=.29]{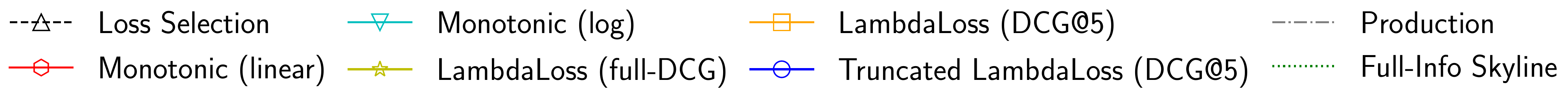}
\end{tabular}
\caption{
Performance of the policy-aware estimator (Eq.~\ref{eq:policyaware}) optimizing \ac{DCG}$@5$ using different loss functions.
The loss selection method selects the estimated optimal model based on clicks gathered on separate validation queries.
Varying numbers of clicks on top-5 rankings with randomization, the number of validation clicks is 15\% of the number of training clicks.
}
\label{fig:dcglosses}
\end{figure*}

\section{Results and Discussion}
In this section we discuss the results of our experiments and evaluate our policy-aware estimator and the methods for top-$k$ \ac{LTR} metric optimization empirically.

\subsection{Learning under item-selection bias}

First we consider the question:
\begin{enumerate}[align=left, label={\bf RQ5.\arabic*}, leftmargin=*]
\item
Is the policy-aware estimator effective for unbiased counterfactual \ac{LTR} from top-$k$ feedback?
\label{rq:topk:feedback}
\end{enumerate}
Figure~\ref{fig:effectk} displays the performance of different approaches after optimization on $10^8$ clicks under varying values for $k$.
Both the policy-oblivious and rerank estimators are greatly affected by the item-selection bias introduced by the cutoff at $k$.
On the MSLR dataset neither approach is able to get close to optimal \ac{ARP} performance, optimal \ac{DCG}$@5$ is only reached when $k >50$.
On the Yahoo dataset, the policy-oblivous approach can only approximate optimal \ac{ARP} when $k > 60$; for \ac{DCG}$@5$ it requires $k > 25$.
The rerank approach reaches optimal \ac{ARP} when $k>50$ and optimal \ac{DCG}$@5$ when $k > 20$.
Considering that on average a query in the Yahoo dataset only has 24 preselected documents, it appears that even a little item-selection bias has a substantial effect on both estimators.
Furthermore, randomization appears to have a very limited positive effect on the policy-oblivious and rerank approaches.
The one exception is the policy-oblivious approach when $k=1$ where it reaches optimal performance under randomization.
Here, the randomization policy gives every item an equal probability of being presented, thus trivially removing item-selection bias; additionally, there is no position bias as there is only a single position.
However, besides this trivial exception, the baseline estimators are strongly affected by item-selection bias and simply logging with randomization is unable to remove the effect of item-selection bias.

In contrast, the policy-aware approach is hardly affected by the choice of $k$.
It consistently approximates optimal performance in terms of \ac{ARP} and \ac{DCG}$@5$ on both datasets.
On the MSLR dataset, the policy-aware approach provides near optimal \ac{ARP} performance; however, for $k > 15$ there is a small but noticeable gap.
We suspect that this is the result of variance from click-noise and can be closed by gathering more clicks.
Across all settings, the policy-aware approach appears unaffected by the choice of $k$ and thus the effect of item-selection bias.
Moreover, it consistently provides performance at least as good as the baselines; and on the Yahoo dataset it outperforms them for $k < 20$ and on the MSLR dataset outperforms them for all tested values of $k$.
We note that the randomization policy is the same for all methods; in other words, under randomization the clicks for the policy-oblivious, policy-aware and rerank approaches are acquired in the exact same way.
Thus, our results show that in order to benefit from randomization, a counterfactual \ac{LTR} method has to take its effect into account, hence only the policy-aware approach has improved performance.

Figure~\ref{fig:cutoffbaseline} displays the performance when learning from top-5 feedback while varying the number of clicks.
Here we see that the policy-oblivious approach performance is stable after $10^5$ clicks have been gathered.
The rerank approach has stable performance after $10^6$ clicks when optimized for \ac{ARP} and $10^5$ for \ac{DCG}$@5$.
Both baseline approaches show biased behavior where adding additional data does not lead to improved performance.
This confirms that their estimators are unable to deal with item-selection bias.
In contrast, the policy-aware approach reaches optimal performance in all settings.
However, it appears that the policy-aware approach requires more clicks than the no-cutoff baseline; we suspect that this difference is due to variance added by the randomization and smaller propensity scores.

In conclusion, we answer our \ref{rq:topk:feedback} positively: our results show that the policy-aware approach is unbiased w.r.t.\ item-selection bias and position bias.
Where all baseline approaches are affected by item-selection bias even in small amounts, the policy-aware approach approximates optimal performance regardless of the cutoff value $k$.

\subsection{Optimizing top-$k$ metrics}

Next, we consider the question:
\begin{enumerate}[align=left, label={\bf RQ5.\arabic*}, leftmargin=*,resume]
\item
Are our novel counterfactual \ac{LTR} loss functions effective for top-$k$ \ac{LTR} metric optimization?
\label{rq:topk:loss}
\end{enumerate}
Figure~\ref{fig:dcglosses} shows the performance of the policy-aware approach after optimizing different loss functions under top-$5$ feedback.
While on the Yahoo dataset small differences are observed, on the MSLR dataset substantial differences are found.
Interestingly, there seems to be no advantage in optimizing for \ac{DCG}$@5$ instead of full \ac{DCG} with the LambdaLoss.
Furthermore, the monotonic loss function works very well with a linear upper bound, yet poorly when using the log upper bound.
On both datasets the heuristic truncated LambdaLoss loss function provides the best performance, despite being the only method without a theoretical basis.
When few clicks are available, the differences change; e.g., the monotonic loss function with a log upper bound outperforms the other losses on the MSLR dataset when fewer than $10^5$ clicks are available.

Finally, we consider unbiased loss selection; Figure~\ref{fig:dcglosses} displays both the performance of the selected models and the estimated performance on which the selections are based.
For the most part the optimal models are selected, but variance does cause mistakes in selection when few clicks are available.
Thus, unbiased optimal loss selection seems effective as long as enough clicks are available.

In conclusion, we answer \ref{rq:topk:loss} positively: our results indicate that the truncated counterfactual LambdaLoss loss function is most effective at optimizing \ac{DCG}$@5$.
Using this loss, our counterfactual \ac{LTR} method reaches state-of-the-art performance comparable to supervised \ac{LTR} on both datasets.
Alternatively, our proposed unbiased loss selection method can choose optimally between models that are optimized by different loss functions.

%% file: 07-topk/sections/07-related-work.tex
\section{Related Work}

Section~\ref{section:supervisedLTR} has discussed supervised \ac{LTR} and Section~\ref{section:counterfactualLTR} has described the existing counterfactual \ac{LTR} framework; this section contrasts additional related work with our policy-aware approach.

Interestingly, some existing work in unbiased \ac{LTR} was performed in top-$k$ rankings settings \citep{agarwal2019addressing, agarwal2019estimating, wang2018position, wang2016learning}.
Our findings suggest that the results of that work are affected by item-selection bias and that there is the potential for considerable improvements by applying the policy-aware method.

\citeauthor{carterette2018offline}~\citep{carterette2018offline} recognized that counterfactual evaluation cannot evaluate rankers that retrieve items that are unseen in the interaction logs, essentially due to a form of item-selection bias.
 Their proposed solution is to gather new interactions on rankings where previously unseen items are randomly injected.
Accordingly, they adapt propensity scoring to account for the random injection strategy.
In retrospect, this approach can be seen as a specific instance of our policy-aware approach.
In contrast, we have focused on settings where item-selection bias takes place systematically and propose that logs should be gathered by any policy that meets Condition~\ref{eq:awarecond2}.
Instead of expanding the logs to correct for missing items, our approach avoids systematic item-selection bias altogether.

Other previous work has also used propensity scores based on a logging policy and examination probabilities.
\citet{Komiyama2015} and subsequently \citet{lagree2016multiple} use such propensities to find the optimal ranking for a single query by casting the ranking problem as a multiple-play bandit.
\citet{li2018offline} use similar propensities to counterfactually evaluate ranking policies where they estimate the number of clicks a ranking policy will receive.
Our policy-aware approach contrasts with these existing methods by providing an unbiased estimate of  \ac{LTR}-metric-based losses, and thus it can be used to optimize \ac{LTR} models similar to supervised \ac{LTR}.

Lastly, online \ac{LTR} methods where interactive processes learn from the user~\citep{yue2009interactively} also make use of stochastic ranking policies.
They correct for biases through randomization in rankings but do not use an explicit model of examination probabilities.
To contrast with counterfactual \ac{LTR}, while online \ac{LTR} methods appear to provide robust performance~\citep{jagerman2019comparison}, they are not proven to unbiasedly optimize \ac{LTR} metrics~\citep{oosterhuis2019optimizing, oosterhuis2018differentiable}.
Unlike counterfactual \ac{LTR}, they are not effective when applied to historical interaction logs~\citep{hofmann2013reusing}.

%% file: 07-topk/sections/08-conclusion.tex
\section{Conclusion}

In this chapter, we have proposed a policy-aware estimator for \ac{LTR}, the first counterfactual method that is unbiased w.r.t. both position bias and item-selection bias.
Our experimental results show that existing policy-oblivious approaches are greatly affected by item-selection bias, even when only small amounts are present.
In contrast, the proposed policy-aware \ac{LTR} method can learn from top-$k$ feedback without being affected by the choice of $k$.
Furthermore, we proposed three counterfactual \ac{LTR} approaches for optimizing top-$k$ metrics: two theoretically proven lower bounds on \ac{DCG}$@k$ based on monotonic functions and the LambdaLoss framework, respectively, and another heuristic truncated loss.
Additionally, we introduced unbiased loss selection that can choose optimally between models optimized with different loss functions.
Together, our contributions provide a method for learning from top-$k$ feedback and for top-$k$ metrics.

With these contributions, we can answer the thesis research questions \ref{thesisrq:topk} and \ref{thesisrq:lambdaloss} positively:
with the policy-aware estimator counterfactual \ac{LTR} is applicable to top-$k$ ranking settings; furthermore, we have shown that the state-of-the-art supervised \ac{LTR} LambdaLoss method can be used for counterfactual \ac{LTR}.
To the best of our knowledge, this is the first counterfactual \ac{LTR} method that is unbiased in top-$k$ ranking settings.
Additionally, this chapter also serves to further bridge the gap between supervised and counterfactual \ac{LTR} methods, as we have shown that state-of-the-art lambda-based supervised \ac{LTR} methods can be applied to the state-of-the-art counterfactual \ac{LTR} estimators.
Therefore, the contributions of this chapter have greatly extended the capability of the counterfactual \ac{LTR} approach and further connected it with the supervised \ac{LTR} field.

Future work in supervised \ac{LTR} could verify whether potential novel supervised methods can be applied to counterfactual losses.
A limitation of the policy-aware \ac{LTR} approach is that the logging policy needs to be known; future work could investigate whether a policy estimated from logs also suffices~\citep{liu2018breaking, li2018offline}.
Finally, existing work on bias in recommendation~\citep{Schnabel2016} has not considered position bias, thus we anticipate further opportunities for counterfactual \ac{LTR} methods for top-$k$ recommendations.

The remaining chapters of this thesis will continue to build on the policy-aware estimator.
Chapter~\ref{chapter:05-genspec} introduces a counterfactual \ac{LTR} algorithm that uses the policy-aware estimator to combine properties of tabular models and feature-based models.
Furthermore, Chapter~\ref{chapter:06-onlinecountereval} looks at how the policy-aware estimator can be used for ranker evaluation.
It introduces an algorithm that optimizes the logging policy to reduce variance when using the policy-aware estimator for evaluation.
Lastly, Chapter~\ref{chapter:06-onlinecounterltr} introduces a novel intervention-aware estimator inspired by the policy-aware estimator.
This novel estimator takes the policy-aware approach even further by considering the effect of all logging policies used during data gathering. 
The intervention-aware approach thus also considers the case where the logging policy is updated during the gathering of data.
Besides the policy-aware estimator, Chapter~\ref{chapter:05-genspec}, Chapter~\ref{chapter:06-onlinecountereval}, and Chapter~\ref{chapter:06-onlinecounterltr} all use the adaption of LambdaLoss for counterfactual \ac{LTR} derived in this chapter.

%% file: 07-topk/notation.tex
\section{Notation Reference for Chapter~\ref{chapter:04-topk}}
\label{notation:04-topk}

\begin{center}
\begin{tabular}{l l}
 \toprule
\bf Notation  & \bf Description \\
\midrule
$k$ & the number of items that can be displayed in a single ranking \\
$i$ & an iteration number \\
$Q$ & set of queries\\
$q$ & a user-issued query \\
$d$ & an item to be ranked\\
$r(d,q)$, $r(d)$ & the relevance of item $d$ w.r.t.\ query $q$\\
$R$ & a ranked list \\
$\bar{R}$ & a ranked list that was displayed to the user \\
$\lambda\left(\doc \mid \ranking \right)$ & a metric that weights items depending on their display rank \\
$c_i(d)$ & a function indicating item $d$ was clicked at iteration $i$ \\
$o_i(d)$ & a function indicating item $d$ was observed at iteration $i$ \\
$\pi$ & a logging policy\\
$\pi(\displayranking \mid q)$ & the probability that policy $\pi$ displays ranking $\displayranking$ for query $q$ \\
$\text{rank}\big(\doc \mid \displayranking \big)$ & the rank of item $\doc$ in displayed ranking $\displayranking$\\
$\rho$ & a propensity function used to represent any \acs{IPS} estimator \\
$s(d)$ & the score given to item $d$ by ranking model $s$, used to sort items by\\
\bottomrule
\end{tabular}
\end{center}

%% file: 08-genspec/main.tex
\chapter[Combining Generalized and Specialized Models in Counterfactual \acs{LTR}]{Combining Generalized and Specialized Models in Counterfactual Learning to Rank}
\label{chapter:05-genspec}

\footnote[]{This chapter was submitted as~\citep{oosterhuis2021genspec}.
Appendix~\ref{notation:05-genspec} gives a reference for the notation used in this chapter.
}

So far, this thesis has only addressed feature-based \ac{LTR} -- the optimization of models that rank items based on their features -- as opposed to tabular online \ac{LTR} -- which optimizes a ranking directly, thus not using any scoring models.
A big advantage of feature-based \ac{LTR} is that its model can be applied to previously unseen queries and items.
As a result, it provides very robust performance in previously unseen circumstances.
However, their behavior is often limited by the available features: in practice they do not provide enough information to determine the optimal ranking.
In stark contrast, tabular \ac{LTR} memorizes rankings, instead of using a features to predict them.
Consequently, tabular \ac{LTR} is not limited by which features are available and can potentially always find the optimal ranking.
Despite this potential, tabular \ac{LTR} does not generalize: it cannot transfer learned behavior to previously unseen queries or items.
In other words, tabular \ac{LTR} has the potential to specialize -- perform very well in circumstances encountered often -- whereas feature-based \ac{LTR} is good at generalization -- performing well overall, including previously unseen circumstances.
In this chapter we investigate whether the advantageous properties of these two areas can be combined in the  counterfactual \ac{LTR} framework, and thus we address the thesis research question:
\begin{itemize}
\item[\ref{thesisrq:genspec}] Can the specialization ability of tabular online \ac{LTR} be combined with the robust feature-based approach of counterfactual \ac{LTR}?
\end{itemize}

In this chapter we introduce a framework for \acf{GENSPEC} for counterfactual learning from logged bandit feedback.
\ac{GENSPEC} is designed for problems that can be divided in many non-overlapping contexts.
It simultaneously learns a generalized policy -- optimized for high performance across all contexts -- and many specialized policies -- each optimized for high performance in a single context.
Using high-confidence bounds on the relative performance of policies, per context \ac{GENSPEC} decides whether to deploy a specialized policy, the general policy, or the current logging policy.
By doing so \ac{GENSPEC} combines the high performance of successfully specialized policies with the safety and robustness of a generalized policy. 

While \ac{GENSPEC} is applicable to many different bandit problems, we focus on query-specialization for counterfactual learning to rank, where a context consists of a query submitted by a user.
Here we learn both a single general feature-based model for robust performance across queries, and many memory-based models, each of which is highly specialized for a single query, \ac{GENSPEC} then chooses which model to deploy on a per query basis.
Our results show that \ac{GENSPEC} leads to massive performance gains on queries with sufficient click data, while still having safe and robust behavior on queries with little or noisy data.

\input{08-genspec/sections/01-introduction}

\input{08-genspec/sections/02-background}

\input{08-genspec/sections/03-counterfactual-ltr}

\input{08-genspec/sections/05-experimental-setup}

\input{08-genspec/sections/06-results}

\input{08-genspec/sections/04-related-work}
\input{08-genspec/sections/07-discussion}
\input{08-genspec/sections/08-conclusion}
\begin{subappendices}
\input{08-genspec/sections/09-appendix}

\input{08-genspec/notation}
\end{subappendices}

%% file: 08-genspec/sections/01-introduction.tex
\section{Introduction}
\label{sec:intro}

\emph{Generalization} is an important goal for most machine learning algorithms: models should perform well across a large range of contexts, especially previously unseen contexts~\citep{bishop2006pattern}.
\emph{Specialization}, the ability to perform well in a single context, is often disfavored over generalization because the latter is more robust~\citep{hawkins2004problem}.
Generally, the same trade-off pertains to contextual bandit problems~\citep[][Chapter 18]{lattimore2019bandit}. There, the goal is to find a policy that maximizes performance over the full distribution of contextual information. 
While a specialized policy, i.e., a policy optimized on a subset of possible contexts, could outperform a generalized policy on that subset, most likely it compromises performance on other contexts to do so since specialization comes with a risk of overfitting: applying a policy that is specialized in a specific set of contexts to different contexts~\citep{hawkins2004problem, claeskens2008model}.
As a consequence, generalization is often preferred as it avoids this issue.

In this chapter, we argue that, depending on the circumstances, specialization may be preferable over generalization, specifically, if with high-confidence it can be guaranteed that a specialized policy is only deployed in contexts where it outperforms policies optimized for generalization.
We focus on counterfactual learning for contextual bandit problems where contexts can be split into non-overlapping sets.
We simultaneously train 
\begin{enumerate*}[label=(\roman*)]
\item a generalized policy that performs well across all contexts, and 
\item many specialized policies, one for each of a specific set of contexts.
\end{enumerate*}
Thus, per context there is a choice between three policies:
\begin{enumerate*}[label=(\roman*)]
\item the logging policy used to gather data, 
\item the generalized policy, and 
\item the specialized policy.
\end{enumerate*}
Depending on the circumstances, e.g., the amount of data available, noise in the data, or the difficulty of the task, a different policy will perform best in a specific context~\citep{claeskens2008model}.
To reliably choose between policies, we estimate high-confidence bounds~\citep{thomas2015high} on the relative performance differences between policies and then choose conservatively: we only apply a specialized policy instead of the generalized policy or logging policy if the lower bounds on their differences in performance are positive in a specific context.
Otherwise, the generalized policy is only applied if with high-confidence it outperforms the logging policy across all contexts.
We call this approach the \acfi{GENSPEC}\acused{GENSPEC} framework: it trains both generalized and specialized policies and results in a meta-policy that chooses between them using high-confidence bounds.
The \ac{GENSPEC} meta-policy is particularly powerful because it can combine the properties of different models: for instance, a generalized policy using a feature-based model can be overruled by a specialized policy using a tabular model that has memorized the best actions.
\ac{GENSPEC} promises the best of two worlds: the safe robustness of a generalized policy with the potentially high performance of a specialized policy.

To evaluate the \ac{GENSPEC} approach, we apply it to query-spe\-cialization in the setting of \emph{Counterfactual} \acfi{LTR}.
Existing approaches in this field either generalize -- by learning a ranking model that ranks items based on their features and generalizes well across all queries~\citep{joachims2017unbiased} -- or they specialize -- by learning tabular ranking models that are specific to a single query and cannot be applied to any other query~\citep{zoghi2016click, lattimore2019bandit}. %
By viewing each query as a different context, \ac{GENSPEC} learns both a generalized ranker and many specialized tabular rankers, and subsequently chooses which ranker to apply per query.
Our empirical results show that \ac{GENSPEC} combines the advantages of both approaches: very high performance on queries where sufficiently many interactions were observed for successful specialization, and safe robust performance on queries where interaction data is limited or noisy.

Our main contributions are:
\begin{enumerate}%
\item an adaptation of existing counterfactual high-confidence bounds for relative performance between ranking policies; 
\item the \ac{GENSPEC} framework that simultaneously learns generalized and specialized ranking policies plus a meta-policy that decides which to deploy per context.
\end{enumerate}

\noindent
To the best of our knowledge, \ac{GENSPEC} is the first counterfactual \ac{LTR} method to simultaneously train generalized and specialized models, and reliably choose between them using high-confidence bounds.

%% file: 08-genspec/sections/02-background.tex
\section{Background: Learning to Rank}
\label{sec:background}

This section covers the basics of counterfactual \ac{LTR}.

\subsection{Supervised learning to rank}
\label{sec:supervisedltr}

The \ac{LTR} task has been approached as a contextual bandit problem before~\citep{ yue2009interactively, kveton2015cascading, swaminathan2017off, lattimore2019bandit}.
The differentiating characteristic of the \ac{LTR} task is that actions are \emph{rankings}, thus they consist of an ordered set of $K$ items:
$
a = (d_{1}, d_{2}, \ldots, d_{K}).
$
The contextual information often contains a user-issued search query, features based on the items available for ranking and item-query combinations, information about the user, among other miscellaneous information.
Since our focus is query specialization, we record the query separately; thus, at each time step $i$ contextual information $x_i$ and a single query $q_i \in \{1, 2, 3, \ldots \}$ are active:
$
x_i, q_i \sim P(x, q).
$
Let $\Delta$ indicate the reward for a ranking $a$.
A policy $\pi$ should maximize the expected reward~\cite{liu2009learning, joachims2017unbiased}:
\begin{equation}
\mathcal{R}(\pi) =  \iint \Big( \sum_{a \in \pi} \Delta(a |\, x, q, r) \cdot \pi(a |\, x, q) \Big) P(x, q) \, dx \,dq.
\label{eq:genspec:truereward}
\end{equation}
Commonly, in \ac{LTR} the reward for a ranking $a$ is a linear combination of the relevance scores of the items in $a$, weighted according to their rank.
We use $r(d \mid x, q)$ to denote the relevance score of item $d$ and $\lambda\big(\textit{rank}(d \mid a)\big)$ for the weight per rank, resulting in:
\begin{equation}
\Delta(a \mid x, q, r) = \sum_{d \in a }\lambda\big(\textit{rank}(d \mid a)\big) \cdot r(d \mid x, q).
\label{eq:ltrtrueest}
\end{equation}
A common choice is to optimize the \acfi{DCG} metric; $\lambda$ can be chosen accordingly:
\begin{equation}
\lambda^{\textit{DCG}}\big(\textit{rank}(d \mid a)\big) = \log_2\big(1+ \textit{rank}(d \mid a)\big)^{-1}. \label{eq:dcg}
\end{equation}
When the relevance function $r$ is given, maximizing $\mathcal{R}$ can be done through traditional \ac{LTR} in a supervised manner~\cite{liu2009learning, burges2010ranknet, wang2018lambdaloss}.

\subsection{Counterfactual learning to rank}
\label{sec:counterfactualltr}
In practice, the relevance score $r$ is often unknown or requires expensive annotation~\citep{sanderson2010,Chapelle2011, qin2013introducing, dato2016fast}.
An attractive alternative comes from \ac{LTR} based on historical interaction logs, which takes a counterfactual approach~\citep{wang2016learning,joachims2017unbiased}.
Let $\pi_0$ be the logging policy that was used when interactions were logged:
\begin{equation}
a_i \sim \pi_0(a \mid x_i, q_i).
\end{equation}
\ac{LTR} focusses mainly on clicks in interactions; clicks are strongly affected by \emph{position bias}~\citep{craswell2008experimental}. %
This bias arises because users often do not examine all items presented to them, and only click on examined items.
As a result, items that are displayed in positions that are more often examined are also more likely to be clicked, without necessarily being more relevant.
Let $o_i(d) \in \{0,1\}$ indicate whether item $d$ was examined by the user or not:
\begin{equation}
o_i(d) \sim P\big(o(d) \mid a_i\big).
\end{equation}
We use $c_i(d) \in \{0,1\}$ to indicate whether $d$ was clicked at time step $i$:
\begin{equation}
c_i(d) \sim P\big(c(d) \mid o_i(d), r( d \mid x, q)\big).
\end{equation}
We assume that click probabilities are only dependent on whether an item was examined,  $o_i(d)$, and its relevance, $r( d \mid x, q)$.
Furthermore, we make the common assumption that clicks only occur on examined items~\citep{wang2016learning,joachims2017unbiased}, thus:
\begin{equation}
P\big(c(d) = 1 \mid o(d) = 0, r( d \mid x, q)\big) = 0.
\label{eq:unobserved}
\end{equation}
Moreover, we assume that, given examination, more relevant documents are more likely to be clicked.
Specifically, click probability is proportional to relevance with an offset $\mu \in \mathbb{R}_{>0}$:
\begin{equation}
P\big(c(d) = 1 \mid o(d) = 1, r(d \mid x, q)\big) \propto r(d \mid x,q) + \mu.
\label{eq:clickprop}
\end{equation}
The data used for counterfactual \ac{LTR} consists of observed clicks $c_i$, propensity scores $\rho_i$, contextual information $x_i$ and query $q_i$ for $N$ interactions:
\begin{equation}
\mathcal{D} = \big\{(c_i, a_i, \rho_i, x_i, q_i)\big\}_{i=1}^N.
\end{equation}
We apply the policy-aware approach~\citep{oosterhuis2020topkrankings} (Chapter~\ref{chapter:04-topk}) and base $\rho$ both on the examination probability of the user and the behavior of the policy:
\begin{equation}
\rho_i(d) = \sum_{a \in \pi_0} P\big(o_i(d)=1 \mid a\big) \cdot \pi_0(a \mid x_i, q_i).
\end{equation}
The estimated reward based on $\mathcal{D}$ is now:
\begin{equation}
\hat{\mathcal{R}}(\pi \mid \mathcal{D}) = \frac{1}{|\mathcal{D}|} \sum_{i \in \mathcal{D}} \sum_{a \in \pi}  \hat{\Delta}(a \mid c_i, \rho_i) \cdot \pi(a \mid x_i, q_i),
\label{eq:rewardestimate}
\end{equation}
where $\hat{\Delta}$ is an \ac{IPS} estimator:
\begin{equation}
\hat{\Delta}(a \mid c_i, \rho_i) = \sum_{d \in a}\frac{\lambda\big(\textit{rank}(d \mid a)\big) \cdot c_i(d)}{\rho_i(d)}. \label{eq:rankips}
\end{equation}
Since the reward $r$ is not observed directly, clicks are used as implicit feedback, which is a biased and noisy indicator of relevance.
The unbiased estimate $\hat{\mathcal{R}}$ can be used for unbiased evaluation and optimization since:\footnote{See Appendix~\ref{sec:counterfactualproof} for a proof.}
\begin{equation}
 \argmax_{\pi} \mathbb{E}_{o_i,a_i}\big[\hat{\mathcal{R}}(\pi \mid \mathcal{D})\big] = \argmax_{\pi} \mathcal{R}(\pi).
\end{equation}
Previous work has introduced several methods for maximizing $\hat{\mathcal{R}}$ so as to optimize different \ac{LTR} metrics~\citep{joachims2017unbiased, agarwal2019counterfactual}. %

This concludes our description of the counterfactual \ac{LTR} basics; importantly, ranking policies can be optimized from clicks without being affected by the logging policy or the users' position bias.

%% file: 08-genspec/sections/03-counterfactual-ltr.tex
\section{GENSPEC for Query Specialization}
\label{sec:genspecltr}

This section introduces the \ac{GENSPEC} framework and applies it to query specialization for \ac{LTR}.
Section~\ref{sec:genspeccontextualbandit} details how it can be applied to the general contextual bandit problem.

\subsection{Generalization and query specialization}
\label{sec:queryspec}

We will now propose the first part of the \ac{GENSPEC} framework, which produces a general policy $\pi_g$ and, for each query $q$, a specialized policy $\pi_q$.

\ac{GENSPEC} uses the logged data $\data$ both to train policies and to evaluate relative performance;
to avoid overfitting we split $\data$ in a training partition $\traindata$ and a policy-selection partition $\bounddata$ so that 
$
\data = \traindata \cup \bounddata
$
and
$
\traindata \cap \bounddata = \emptyset
$.

A policy has optimal generalization performance if it maximizes performance across \emph{all queries}. 
Thus, given the generalization policy space $\Pi_g$, the optimal general policy is:
\begin{equation}
\pi_g = \argmax_{\pi \in \Pi_g} \, \hat{\mathcal{R}}(\pi \mid \traindata).
\end{equation}
Alternatively, we can also choose to optimize performance for a single query $q$.
First, we only select the datapoints in $\mathcal{D}$ where query $q$ was issued:
\begin{equation}
\mathcal{D}_q = \big\{(c_i, a_i, \rho_i, x_i, q_i)\in \mathcal{D} \mid q_i = q \big\}.
\end{equation}
Then the policy $\pi_q$ that is specialized for query $q$ is the policy in the specialization policy space $\Pi_q$ that maximizes the performance when query $q$ is issued:
\begin{equation}
\pi_q = \argmax_{\pi \in \Pi_q} \, \hat{\mathcal{R}}(\pi \mid \traindata_q).
\end{equation}
The motivation for $\pi_q$ is that it has the potential to provide better performance than $\pi_g$ when $q$ is issued.
We may expect $\pi_q$ to outperform $\pi_g$ because $\pi_g$ may compromise performance on the query $q$ for better performance across all queries, whereas $\pi_q$ never makes such compromises.
Furthermore, $\Pi_q$ could contain more optimal policies than $\Pi_g$, because the policies in $\Pi_g$ have to be applicable to all queries whereas $\Pi_q$ can make use of specific properties of $q$.
However, it is also possible that $\pi_g$ and $\pi_q$ provide the same performance. %
Moreover, since $\mathcal{D}_q$ is a subset of $\mathcal{D}$, the optimization of $\pi_q$ is more vulnerable to noise in the data.
As a result, the true performance of $\pi_q$ for query $q$ could be worse than that of $\pi_g$, especially when $\mathcal{D}_q$ is substantially smaller than $\mathcal{D}$.

In other words, a priori it is unclear whether $\pi_g$ or $\pi_q$ are preferred.
We thus need a method to estimate the optimal choice with a reasonable amount of confidence.

\subsection{Safely choosing between policies}

We will now propose the other part of our \ac{GENSPEC} framework: a meta-policy that safely chooses between deploying $\pi_g$ and $\pi_q$ per query $q$.
We wish to avoid deploying $\pi_q$ when it performs worse than $\pi_g$, and similarly, avoid deploying $\pi_g$ when it is outperformed by the logging policy $\pi_0$.
Recently, a method for safe policy deployment was introduced by \citet{jagerman2020safety} based on high-confidence bounds~\citep{thomas2015high}.
The intuition behind their method is that a learned policy $\pi$ should not be deployed before we can be highly confident that it outperforms the logging policy $\pi_0$, otherwise it is safer to keep the logging policy in deployment.

While previous work has bounded the performance of individual policies~\citep{ thomas2015high, jagerman2020safety}, we instead bound the \emph{difference} in performance between two policies directly.
Let $\delta(\pi_1, \pi_2)$ indicate the true difference in performance between a policy $\pi_1$ and policy $\pi_2$:
\begin{equation}
\delta(\pi_1, \pi_2) = \mathcal{R}(\pi_1) - \mathcal{R}(\pi_2).
\end{equation}
Knowing $\delta(\pi_1, \pi_2)$ allows us to optimally choose which of the two policies to deploy.
However, we can only estimate its value from historical data $\mathcal{D}$:
\begin{equation}
\hat{\delta}(\pi_1, \pi_2 \mid \mathcal{D}) = \hat{\mathcal{R}}(\pi_1 \mid \mathcal{D}) - \hat{\mathcal{R}}(\pi_2 \mid \mathcal{D}).
\end{equation}
For brevity, let $R_{i,d}$ indicate the inverse-propensity-scored difference for a single document $d$ at interaction $i$:
\begin{equation}
R_{i,d} =  
\frac{c_i(d)}{\rho_i(d)} \sum_{a \in \pi_1 \cup \pi_2} \big(\pi_1(a | x_i, q_i) - \pi_2(a | x_i, q_i)\big) 
 \cdot  \lambda\big(\textit{rank}(d \,|\, a)\big).
\end{equation}
Then, for computational efficiency we rewrite:
\begin{equation}
\hat{\delta}(\pi_1, \pi_2 \mid \mathcal{D}) = 
 \frac{1}{|\data|} \sum_{i \in \mathcal{D}} \sum_{d \in a_i} R_{i,d} = \frac{1}{|\data|  K} \sum_{(i, d) \in \mathcal{D}} K \cdot R_{i,d}.
\end{equation}
For notational purposes, we let $\sum_{(i,d) \in \mathcal{D}}$ iterate over all actions $a_i$ and $K$ documents $d$ per action $a_i$.
With the confidence parameter $\epsilon \in [0, 1]$,
setting $b$ to be the maximum possible absolute value for $R_{i,d}$, i.e., $b = \frac{\max \lambda(\cdot)}{\min \rho}$, and
\begin{align}
\nu =  \frac{ 2  |\mathcal{D}|  K  \ln\big(\frac{2}{1-\epsilon}\big)}{|\mathcal{D}| K-1}  \sum_{(i,d) \in \mathcal{D}} \big(K \cdot R_{i,d} - \hat{\delta}(\pi_1, \pi_2 \mid \mathcal{D}) \big)^2,
\nonumber
\end{align}
we follow~\citet{thomas2015high} to get the high-confidence bound:
\begin{align}
\textit{CB}
(\pi_1, \pi_2 \mid \mathcal{D})
= \frac{7 K b\ln\big(\frac{2}{1-\epsilon}\big)}{3(|\mathcal{D}|  K-1)} + \frac{1}{|\mathcal{D}|  K}  
 \cdot \sqrt{\nu}. 
\label{eq:CB}
\end{align}
In turn, this provides us with the following upper and lower confidence bounds on $\delta$:
\begin{equation}
\begin{split}
\textit{LCB}(\pi_1, \pi_2 \mid \mathcal{D}) &= \hat{\delta}(\pi_1, \pi_2 \mid \mathcal{D}) - \textit{CB}(\pi_1, \pi_2 \mid \mathcal{D})  \\
\textit{UCB}(\pi_1, \pi_2 \mid \mathcal{D}) &= \hat{\delta}(\pi_1, \pi_2 \mid \mathcal{D}) + \textit{CB}(\pi_1, \pi_2 \mid \mathcal{D}). 
\end{split}
\label{eq:ucb}
\end{equation}
As proven by~\citet{thomas2015high}, with at least a probability of $\epsilon$ they bound the true value of $\delta(\pi_1, \pi_2)$:
\begin{equation}
P\Big(\delta(\pi_1, \pi_2) \in \big[\textit{LCB}(\pi_1, \pi_2 \mid \mathcal{D}), \textit{UCB}(\pi_1, \pi_2 \mid \mathcal{D})\big]
 \Big) > \epsilon. \label{eq:safetygaurantee}
\end{equation}
These guarantees allow us to safely choose between policies per query $q$.
We apply a doubly conservative strategy:
$\pi_g$ is not deployed before we are confident that it outperforms $\pi_0$ across all queries;
and $\pi_q$ is not deployed before we are confident that it outperforms both $\pi_g$ and $\pi_0$ on query $q$.
This strategy results in the \ac{GENSPEC} meta-policy $\pi_{GS}$:
\begin{equation}
\pi_{GS}(a \mid x, q)
= 
\begin{cases}
\pi_q(a \mid x, q),&\text{if } \big( \textit{LCB}(\pi_q, \pi_g \mid \bounddata_q) > 0\\
&\hspace{5mm} {} \land \, \textit{LCB}(\pi_q, \pi_0 \mid \bounddata_q) > 0 \big), \\
\pi_g(a \mid x, q),& \text{if } \big( \textit{LCB}(\pi_q, \pi_g \mid  \bounddata_q) \leq 0 \\
&\hspace{5mm} {} \land \textit{LCB}(\pi_g, \pi_0 \mid \bounddata) > 0 \big), \\
\pi_0(a \mid x, q),& \text{otherwise}. \\
\end{cases}
\hspace*{-2mm}\mbox{}
\label{eq:genspec}
\end{equation}
In theory, this approach can make use of the potential gains of specialization while avoiding its risks.
For instance, if the policy-selection partition $\bounddata_q$ is very small, it may be heavily affected by noise, so that the confidence bound $\textit{CB}$ will be wide and $\pi_q$ will not be deployed.
Simultaneously, $\bounddata$ may be large enough so that with high-confidence $\pi_g$ is deployed.

We expect that in practice the relative bounding of \ac{GENSPEC} is much more data-efficient than the \ac{SEA} approach by \citet{jagerman2020safety}.
 \ac{SEA} computes an upper bound on the trusted policy and a lower bound on a learned policy, and only deploys the learned policy if its lower bound is greater than the other's upper bound.
When the learned policy has higher performance than the other, we expect the relative bounds of \ac{GENSPEC} to require less data to be certain about this difference than the \ac{SEA} bounds.
In Appendix~\ref{sec:theoryrelativebounds} we theoretically analyze the difference between these approaches and conclude that the relative bounding of \ac{GENSPEC} is more efficient if there is a positive covariance between $ \hat{\mathcal{R}}(\pi_1 \mid \mathcal{D})$ and $ \hat{\mathcal{R}}(\pi_2 \mid \mathcal{D})$.
Because both estimates are based on the same interaction data $\mathcal{D}$, a high covariance is extremely likely.

Previous work has described safety constraints for policy deployment~\citep{wu2016conservative, kazerouni2017conservative, jagerman2020safety}.
The authors assume that a baseline policy exists whose behavior is considered safe; other policies are considered unsafe if their performance is worse than the baseline policy by a certain margin.
If the logging policy is taken to be the baseline policy, then \ac{GENSPEC} can meet such constraints~\citep{jagerman2020safety}.
We note that while the safety guarantee is strong for a single bound (Eq.~\ref{eq:safetygaurantee}), when applied to a large number of queries the probability of at least one incorrect bound greatly increases.
{This \emph{problem of multiple comparisons} may cause some non-optimal policies to be deployed for some queries.}
Since we mainly care about overall performance this is not expected to be an issue; however, in cases where safety constraints are very important, $\epsilon$ can be chosen to account for the number of comparisons.

\subsection{Summary of the \acs{GENSPEC} method}
This completes our introduction of the \ac{GENSPEC} framework for query specialization.
Figure~\ref{fig:genspec}  visualizes our approach to query specialization. 
We learn from historical interactions gathered using a logging policy $\pi_0$; the interactions are divided into a training and policy-selection partition per query.
Subsequently, a policy is optimized for generalization --- to perform well across all queries --- and for each query a policy is optimized for specialization ---  to perform well for a single query.
While specialization can potentially maximize performance on a specific query, it brings more risks than generalization, since a general policy is optimized on more data and may provide better performance on previously unseen queries.
As a solution to this dilemma, we propose a strategy that uses high-confidence bounds on the differences in performance between policies.
These bounds are then used to choose safely between the deployment of the logging, general and specialized policies.
In theory, \ac{GENSPEC} combines the best of both worlds: the high potential of specialization and the broad safety of generalization.

\begin{figure}[t]
\centering
\includegraphics[width=0.7\columnwidth]{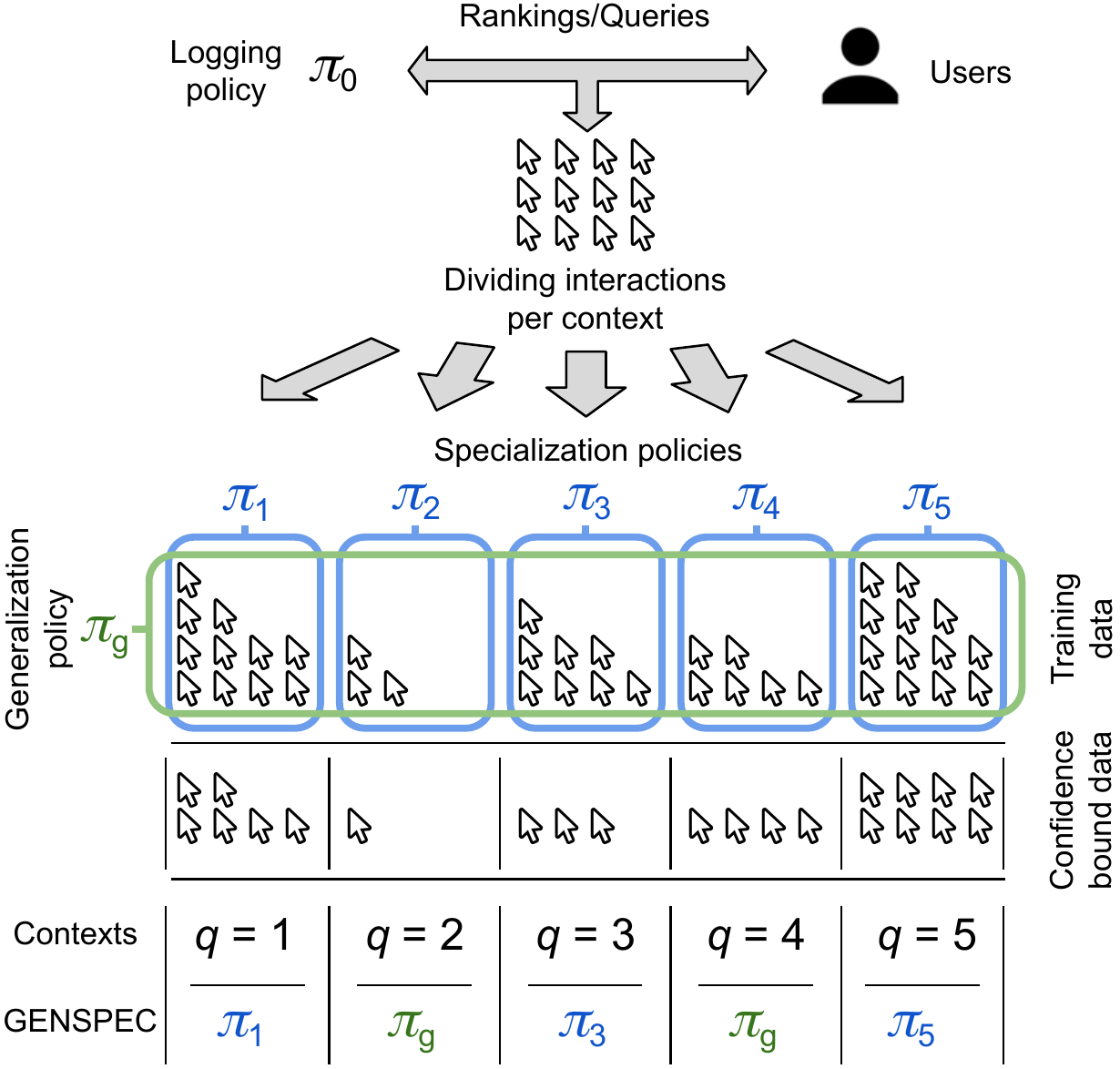} \\
\vspace{0.3\baselineskip}
\caption{
Visualization of the \ac{GENSPEC} framework applied to query specialization for counterfactual \ac{LTR}.
The data $\mathcal{D}$ is divided per query $q$,
many specialized policies $\pi_1, \pi_2, \ldots$ are each optimized for a single query $q\in\{1,2,\ldots\}$,
and single a general policy $\pi_g$ is learned on the data across all queries.
Finally, \ac{GENSPEC} decides which policy to deploy per context, based on high-confidence bounds.
}
\vspace{-0.9\baselineskip}
\label{fig:genspec}
\end{figure}

%% file: 08-genspec/sections/05-experimental-setup.tex
\section{Experimental Setup}
\label{sec:experimentalsetup}

This section discusses our experimental setup and the policies used to evaluate the \ac{GENSPEC} framework.

\subsection{Setup and evaluation}

To evaluate the \ac{GENSPEC} framework, we make use of a semi-synthetic experimental setup: queries, relevance judgements, and documents come from industry datasets, while biased and noisy user interactions are simulated using probabilistic user models.
This setup is very common in the counterfactual \ac{LTR} and online \ac{LTR} literature~\citep{joachims2017unbiased, agarwal2019counterfactual, oosterhuis2019optimizing}.
We make use of the three largest \ac{LTR} industry datasets: \emph{Yahoo! Webscope}~\citep{Chapelle2011}, \emph{MSLR-WEB30k}~\citep{qin2013introducing}, and \emph{Istella}~\citep{dato2016fast}.
Each consists of a set of queries, with for each query a preselected set of documents, document-query combinations are only represented by feature vectors and a label indicating relevance according to expert annotators.
Labels range from $0$ (not relevant) to $4$ (perfectly relevant): $r(d \mid x, q) \in \{0,1,2,3,4\}$.
User issued queries are simulated by uniformly sampling from the training and validation partitions of the datasets.
Displayed rankings are generated by a logging ranker using a linear model optimized on $1\%$ of the training partition using supervised \ac{LTR}~\citep{joachims2017unbiased}.
Then, user examination is simulated with probabilities inverse to the displayed rank of a document:
$
P\big(o(d) = 1 \mid a\big) = \frac{1}{\textit{rank}(d\mid a)}.
$
Finally, user clicks are generated according to the following formula using a single parameter $\alpha \in \mathbb{R}$:
\begin{equation}
P\big(c(d) = 1 \mid o(d) = 1, r(d \mid x, q) \big) = 0.2 + \alpha \cdot r(d \mid x, q).
\end{equation}
In our experiments, we use $\alpha = 0.2$ and $\alpha = 0.025$; the former represents a near-ideal setting where relevant documents receive a very large number of clicks, the latter represents a more noisy and harder setting where the large majority of clicks are on non-relevant documents.
Clicks are only generated on the training and validation partitions, $30\%$ of training clicks are separated for policy selection~($\bounddata$), hyperparameter optimization is done using counterfactual evaluation with clicks on the validation partition~\citep{joachims2017unbiased}.

Some of our baselines are online bandit algorithms, for these baselines no clicks are separated for $\bounddata$, and the algorithms are run online: this means clicks are not gathered using the logging policy but by applying the algorithms in an online interactive setting.

The evaluation metric we use is normalized \ac{DCG}~(Eq.~\ref{eq:dcg})~\citep{jarvelin2017ir} using the ground-truth labels from the datasets.
Unlike most \ac{LTR} work we do not apply a rank-cutoff when computing the metric, thus, an NDCG of $1.0$ indicates that \emph{all} documents are ranked perfectly (not just the top-$k$).
We separately calculate performance on the test set (Test-NDCG), to evaluate performance on previously unseen queries, and the training set (Train-NDCG).
The total number of clicks is varied up to $10^9$ in total, uniformly spread over all queries, the differences in Train-NDCG when more clicks are added allows us to evaluate performance on queries with different levels of popularity. 

\subsection{Choice of policies}

For the generalization policy space $\Pi_g$ we use feature-based ranking models.
This is a natural choice as they can be applied to any query, including previously unseen ones. 
However, the available features could limit the possible behavior of the policies.
We use linear models for $\Pi_g$; optimization is done on $\traindata$ following previous counterfactual \ac{LTR} work ~\citep{agarwal2019counterfactual}.
This results in a learned scoring function $f(d, x, q) \in \mathbb{R}$ according to which items are ranked; due to score-ties there can be multiple valid rankings:
\begin{align}
\mathcal{A}_g(x, q)  =  \mleft \{ a  \mid  
  \forall (d_n, d_m) \in x, 
 \mleft(f(d_n, x, q) > f(d_m, x, q)
 \rightarrow d_n \succ_a d_m \mright)  \mright \}.
\end{align}
The general policy $\pi_g$ samples uniformly random from the set of valid rankings:
\begin{align}
\pi_g(a \mid x, q) =
\begin{cases}
\frac{1}{|\mathcal{A}_g(x, q)|} & \text{if } a \in \mathcal{A}_g(x, q),
\\
0 & \text{otherwise.}
\end{cases}
\end{align}
For the specialization policy space $\Pi_q$, we follow bandit-style online \ac{LTR} work and take the tabular approach~\citep{lagree2016multiple}.
Documents are scored according to an unbiased estimate of \ac{CTR} on query~$q$:
\begin{equation}
\hat{\textit{CTR}}(d, q) = \frac{1}{|\traindata_q|} \sum_{i \in \traindata_q} \frac{c_i(d)}{\rho_i},
\end{equation}
which maximizes the estimated performance (Eq.~\ref{eq:rankips}).
Due to ties there can be multiple valid rankings:
\begin{equation}
\mathcal{A}_q(x, q) = \mleft\{a \,|\,  \forall (d_n, d_m) \in x, 
 \mleft(\hat{\textit{CTR}}(d_n)  > \hat{\textit{CTR}}(d_m)
\rightarrow d_n \succ_a d_m \mright)  \mright\}, 
\end{equation}
The specialized policy $\pi_q$ also chooses uniformly random from the set of valid rankings:
\begin{equation}
\pi_q(a \mid x, q) =
\begin{cases}
\frac{1}{|\mathcal{A}_q(x,q)|} & \text{if } a \in \mathcal{A}_q(x,q),
\\
0 & \text{otherwise.}
\end{cases}
\end{equation}
The tabular approach is not restrained by the available features and can produce any possible ranking~\citep{zoghi2016click}.
Consequently, given enough interactions the tabular approach can perfectly rank items according to relevance.
However, \ac{CTR} cannot be estimated for previously unseen queries and there $\pi_q$ chooses uniformly randomly between all possible rankings.
On a query with a single click, $\pi_q$ will place the \emph{once-clicked} item at the front of the ranking.
Since clicks are very noisy, this behavior is very risky and hence \ac{GENSPEC} uses confidence bounds to avoid the deployment of such unsafe behavior.

%% file: 08-genspec/sections/06-results.tex
\section{Experimental Results and Discussion}

\begin{figure}[t]
\centering
\begin{tabular}{c r r}
&
\multicolumn{1}{c}{\hspace{1.3em} Train-NDCG}
&
\multicolumn{1}{c}{\hspace{1.7em} Test-NDCG}
\\
\rotatebox[origin=lt]{90}{\hspace{0.0em} \footnotesize Yahoo!\ Webscope} &
\includegraphics[scale=0.4]{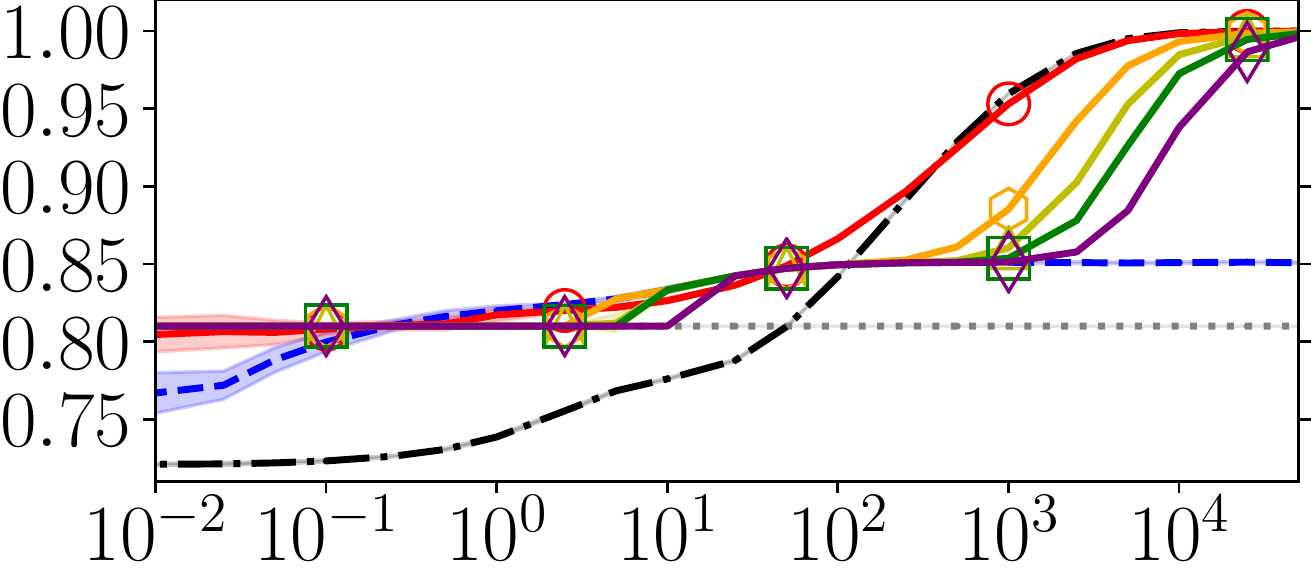} &
\includegraphics[scale=0.4]{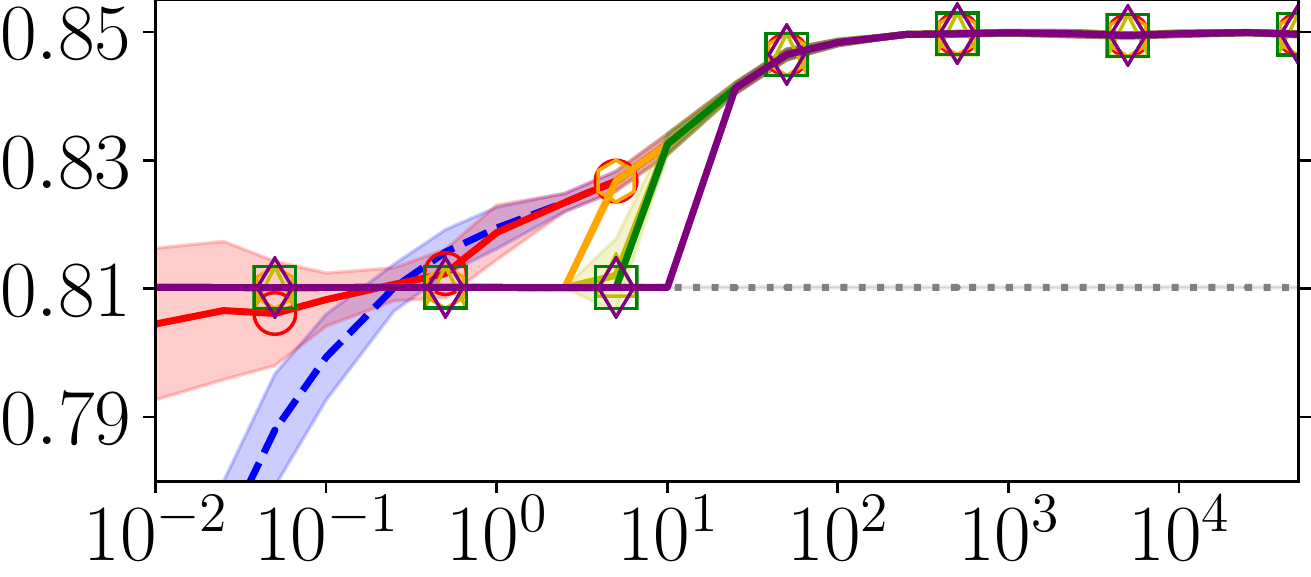}
\\
\rotatebox[origin=lt]{90}{\hspace{0.4em} \footnotesize MSLR-WEB30k} &
\includegraphics[scale=0.4]{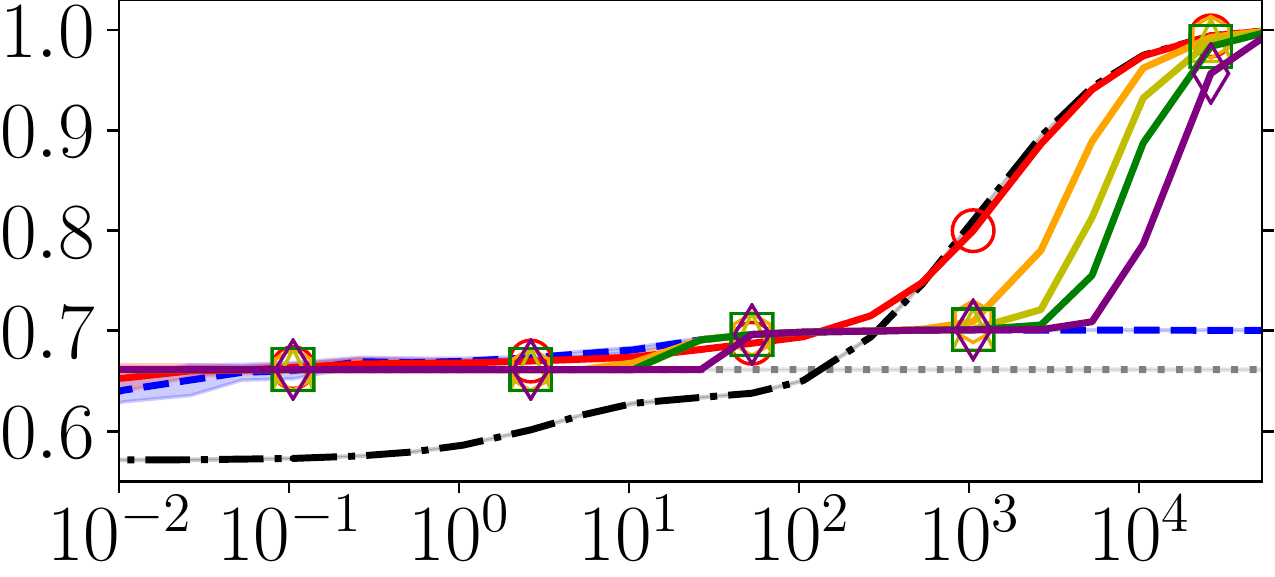} &
\includegraphics[scale=0.4]{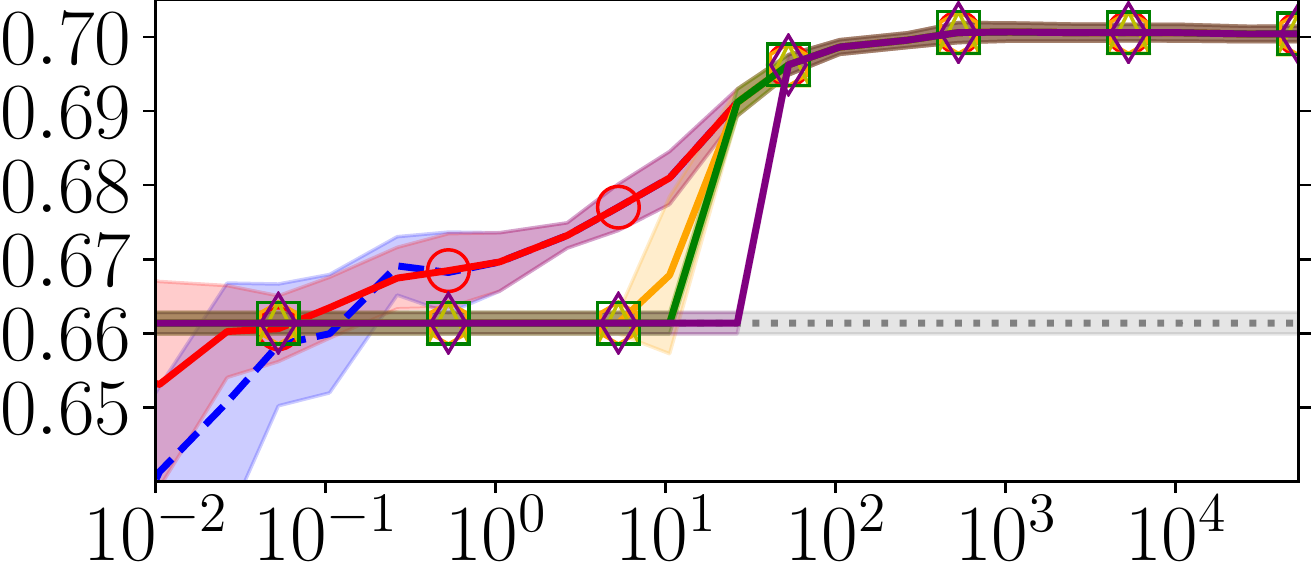}
\\
\rotatebox[origin=lt]{90}{\hspace{2.3em} \footnotesize Istella} &
\includegraphics[scale=0.4]{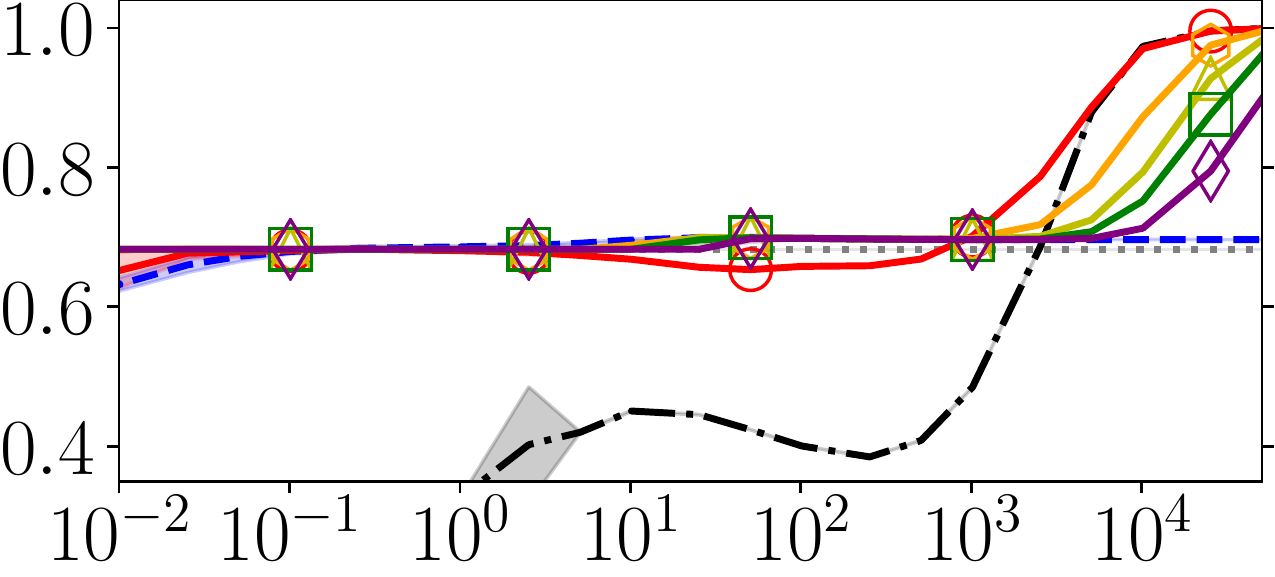} &
\includegraphics[scale=0.4]{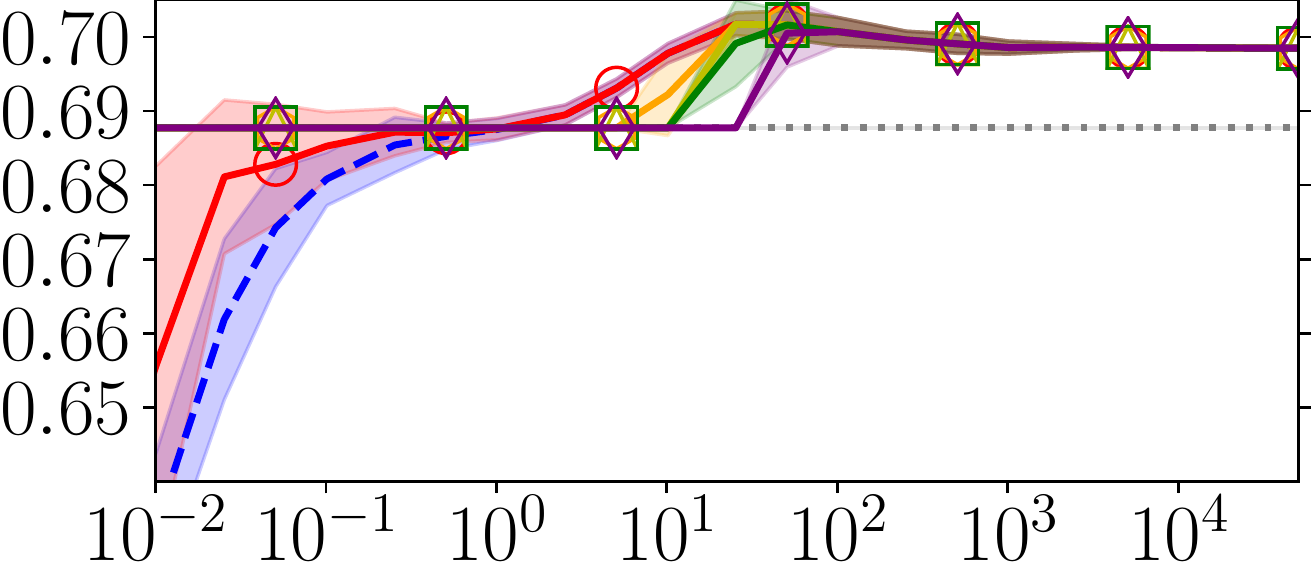}
\\
& \multicolumn{1}{c}{\small \hspace{0.5em} Mean Number of Clicks per Query}
& \multicolumn{1}{c}{\small \hspace{0.5em} Mean Number of Clicks per Query}
\\
\\
 \multicolumn{3}{c}{
 \includegraphics[scale=0.4]{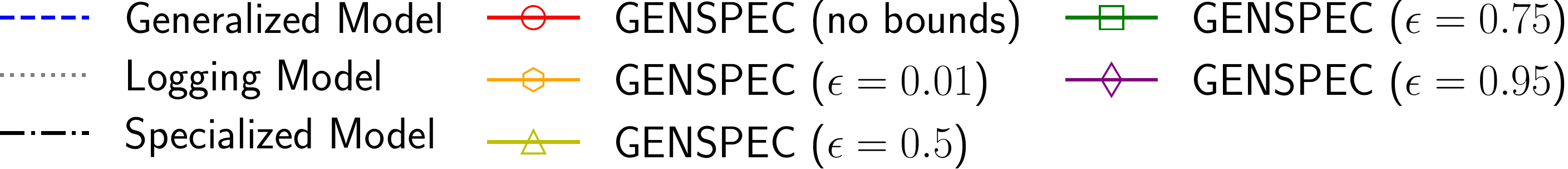}
 }
\end{tabular}
\caption{
Performance of \ac{GENSPEC} with varying levels of confidence, compared to pure generalization and pure specialization, on clicks generated with $\alpha=0.2$.
We separate queries on the training set (Train-NDCG) that have received clicks, and queries on the test set (Test-NDCG) that do not receive any clicks.
Clicks are spread uniformly over the training set, the x-axis indicates the total number of clicks divided by the number of training queries.
Results are an average of 10 runs; shaded area indicates the standard deviation.
}
\label{fig:genspec:main:linear10}
\end{figure}

\begin{figure}[t]
\centering
\begin{tabular}{c r r}
&
\multicolumn{1}{c}{\hspace{1.3em} Train-NDCG}
&
\multicolumn{1}{c}{\hspace{1.7em} Test-NDCG}
\\
\rotatebox[origin=lt]{90}{\hspace{0.0em} \footnotesize Yahoo!\ Webscope} &
\includegraphics[scale=0.4]{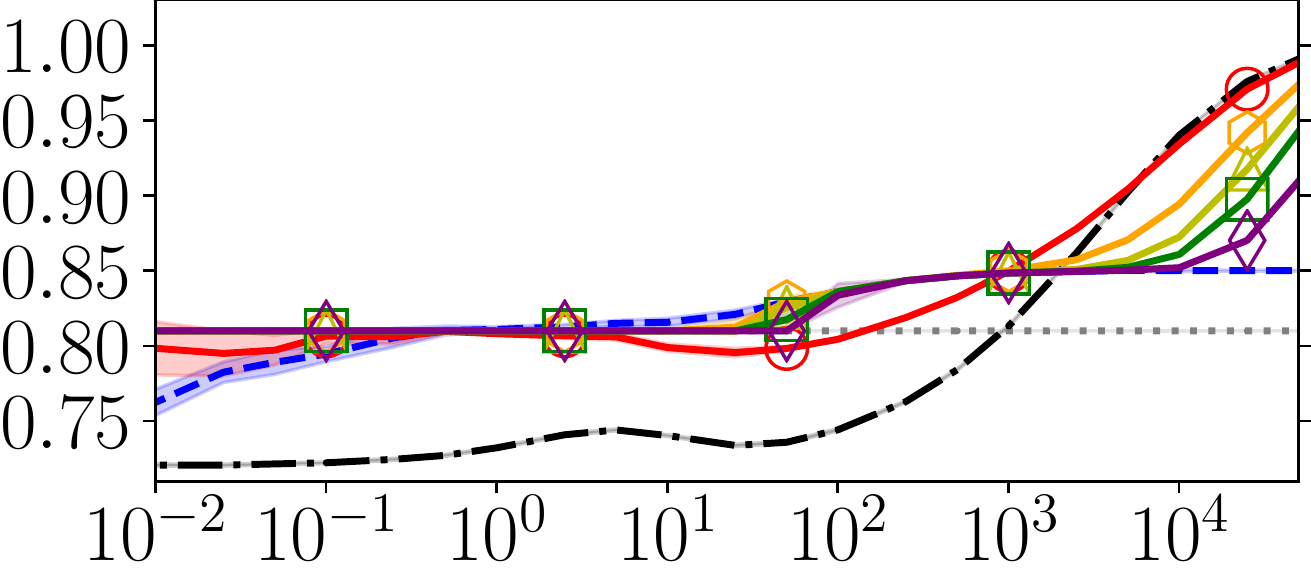} &
\includegraphics[scale=0.4]{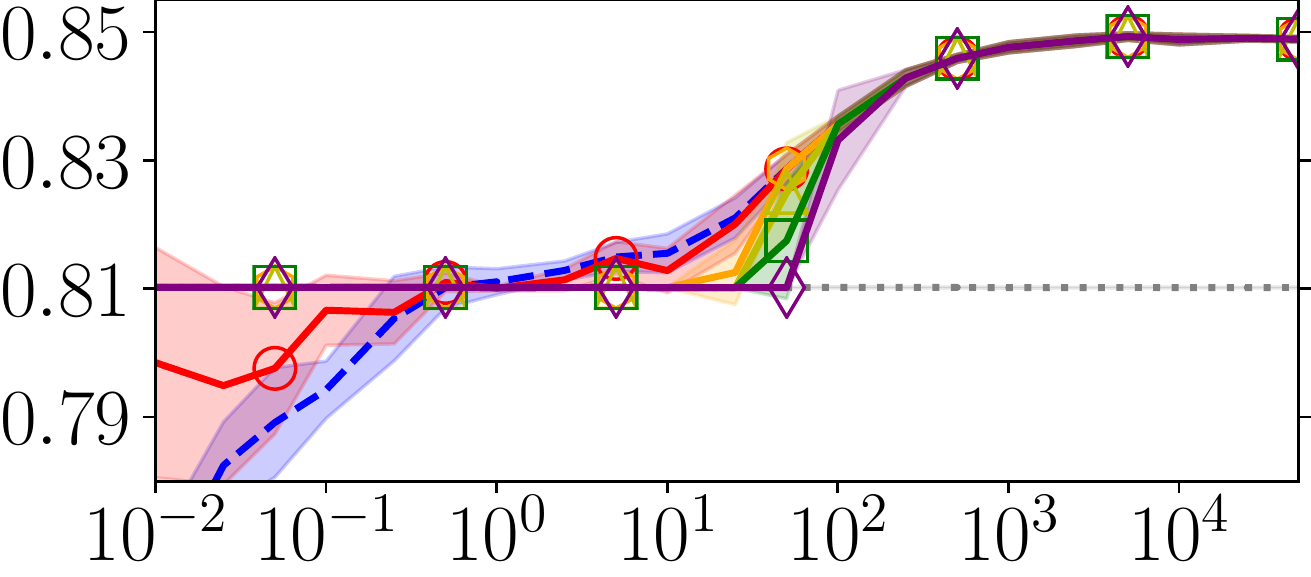}
\\
\rotatebox[origin=lt]{90}{\hspace{0.4em} \footnotesize MSLR-WEB30k} &
\includegraphics[scale=0.4]{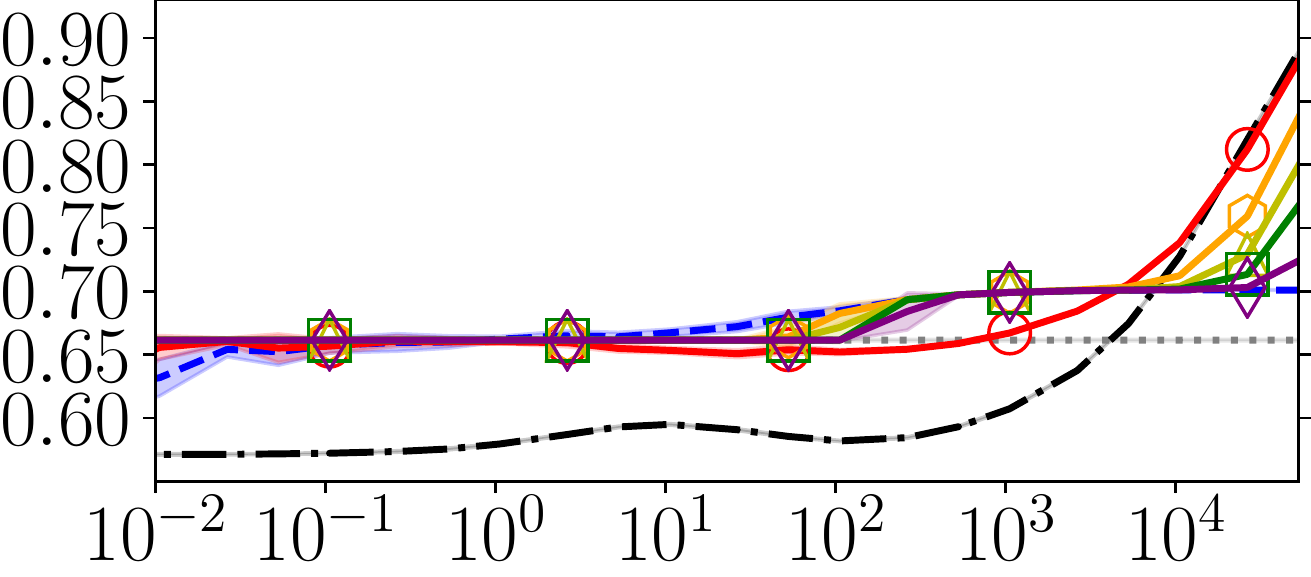} &
\includegraphics[scale=0.4]{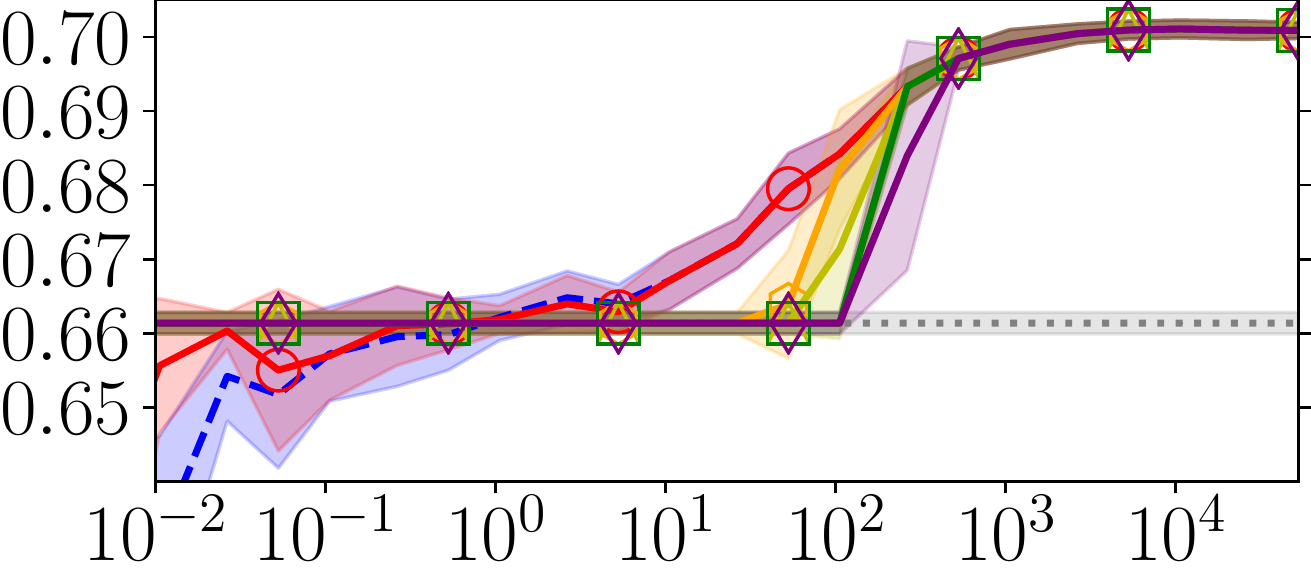}
\\
\rotatebox[origin=lt]{90}{\hspace{2.3em} \footnotesize Istella} &
\includegraphics[scale=0.4]{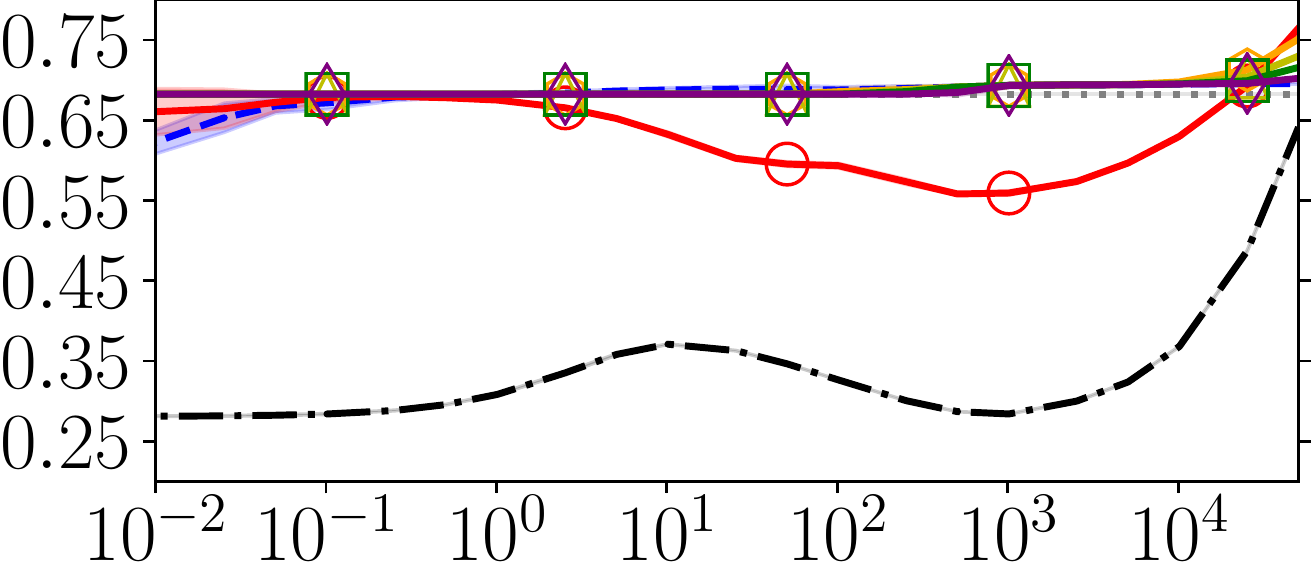} &
\includegraphics[scale=0.4]{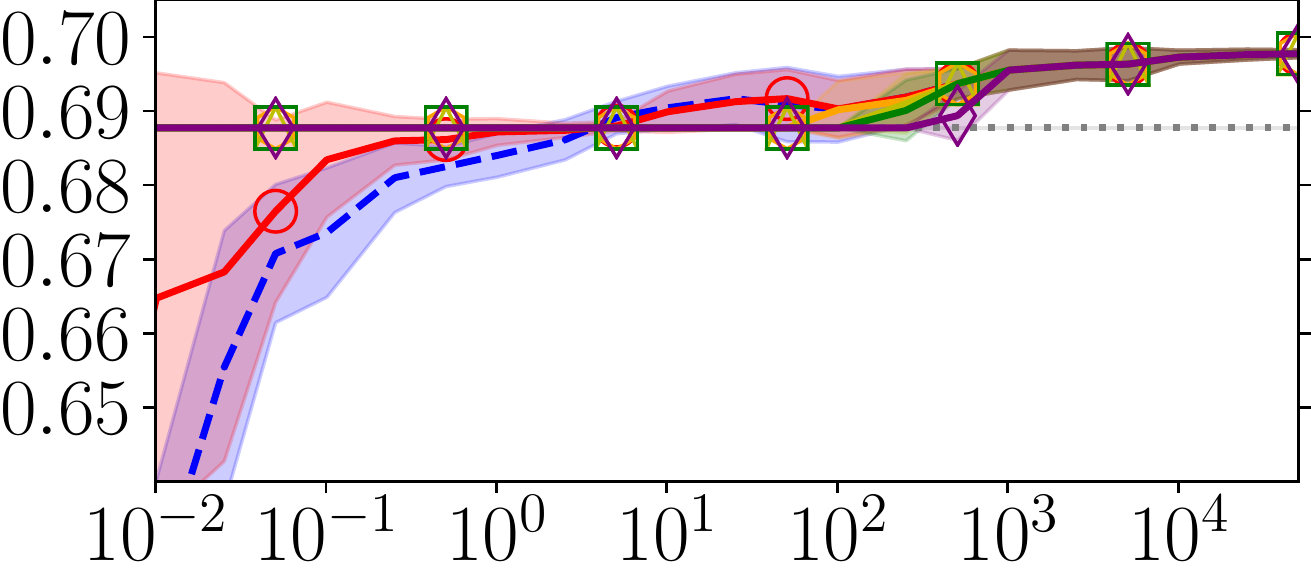}
\\
& \multicolumn{1}{c}{\small \hspace{0.5em} Mean Number of Clicks per Query}
& \multicolumn{1}{c}{\small \hspace{0.5em} Mean Number of Clicks per Query}
\\
\\
 \multicolumn{3}{c}{
 \includegraphics[scale=0.4]{08-genspec/figures/legend_train_ndcg}
 }
\end{tabular}
\caption{
Performance of \ac{GENSPEC} with varying levels of confidence, compared to pure generalization and pure specialization, on clicks generated with $\alpha=0.025$.
Notation is the same as in Figure~\ref{fig:genspec:main:linear10}.
}
\label{fig:genspec:main:linear03}
\end{figure}

\begin{figure}[t]
\centering
\begin{tabular}{c r r}
&
\multicolumn{1}{c}{\hspace{1.3em} Train-NDCG}
&
\multicolumn{1}{c}{\hspace{1.7em} Test-NDCG}
\\
\rotatebox[origin=lt]{90}{\hspace{0.0em} \footnotesize Yahoo!\ Webscope} &
\includegraphics[scale=0.4]{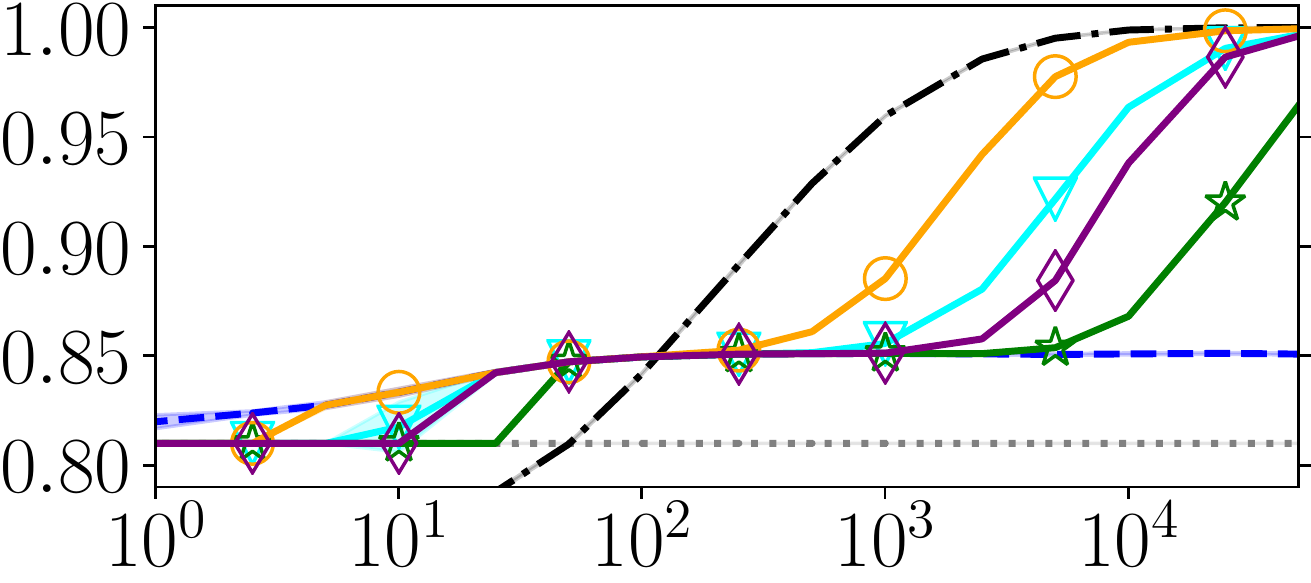} &
\includegraphics[scale=0.4]{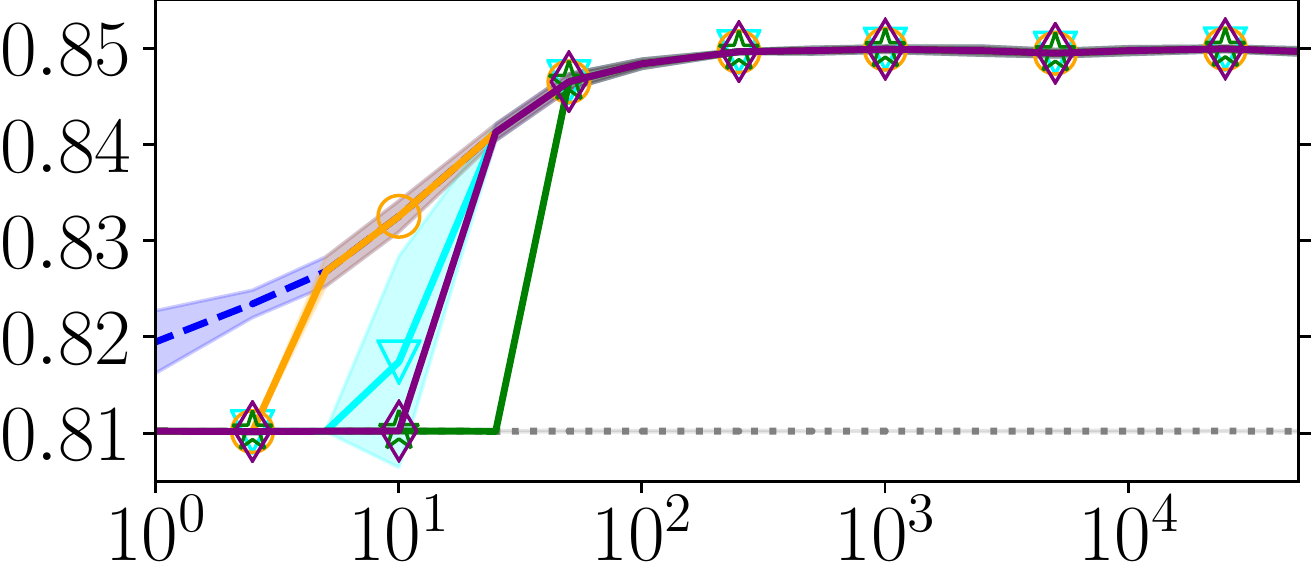}
\\
\rotatebox[origin=lt]{90}{\hspace{0.4em} \footnotesize MSLR-WEB30k} &
\includegraphics[scale=0.4]{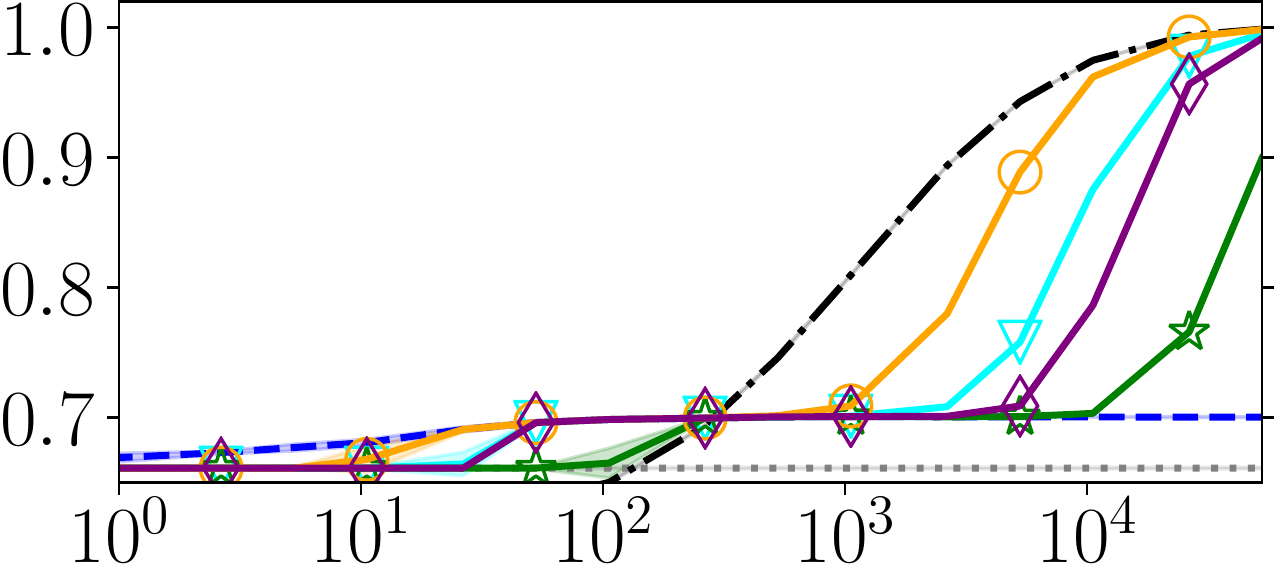} &
\includegraphics[scale=0.4]{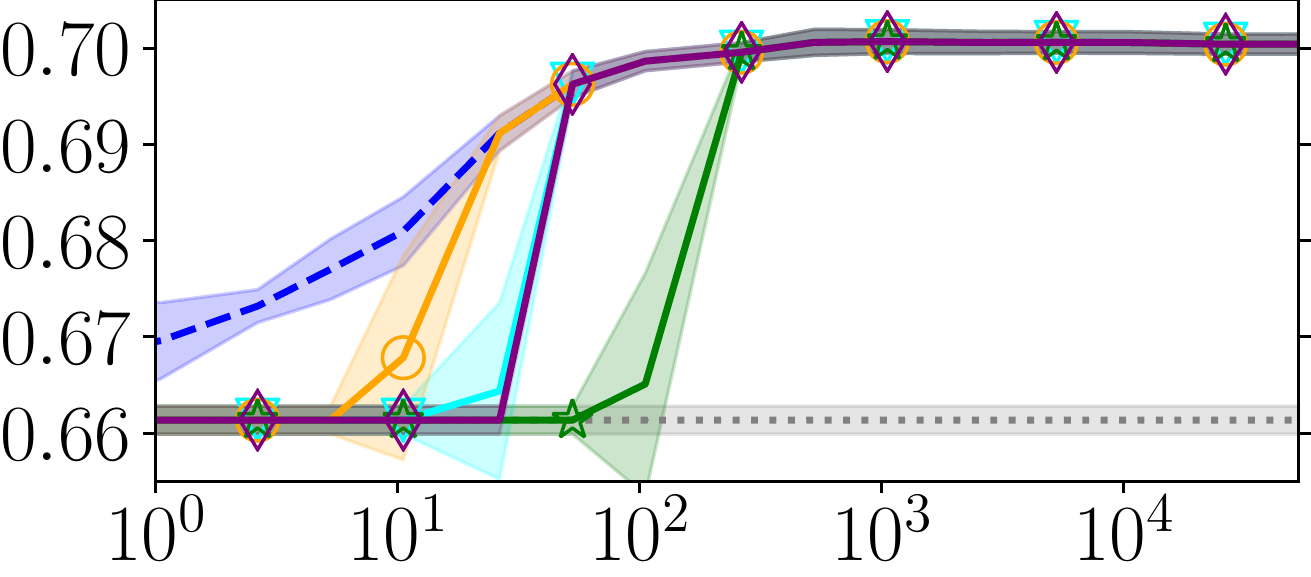}
\\
\rotatebox[origin=lt]{90}{\hspace{2.3em} \footnotesize Istella} &
\includegraphics[scale=0.4]{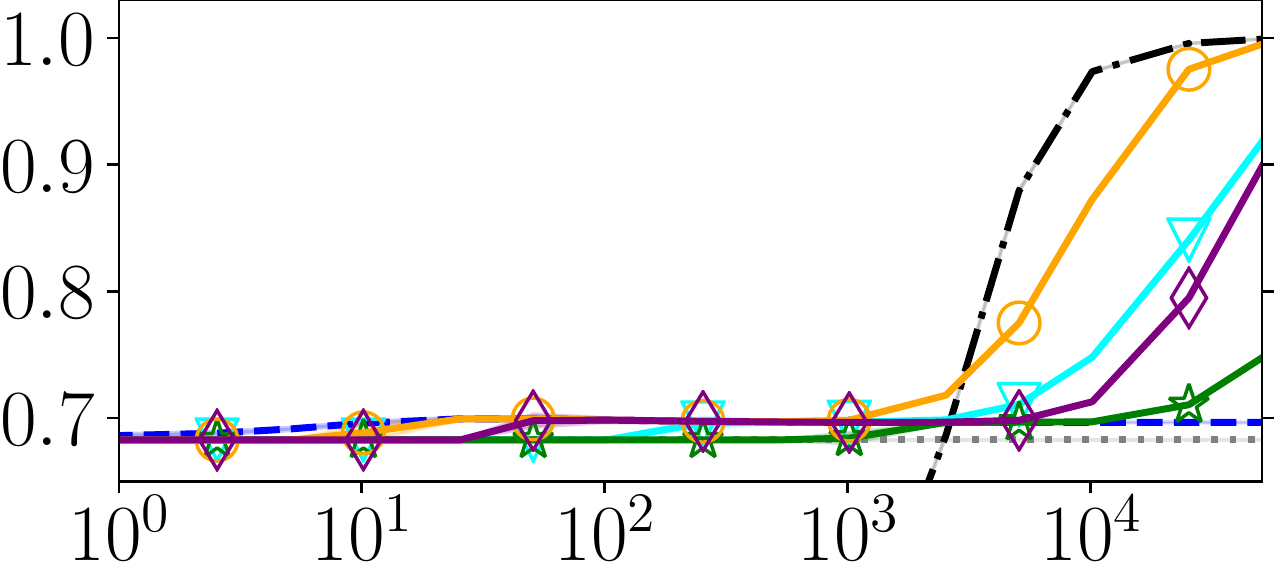} &
\includegraphics[scale=0.4]{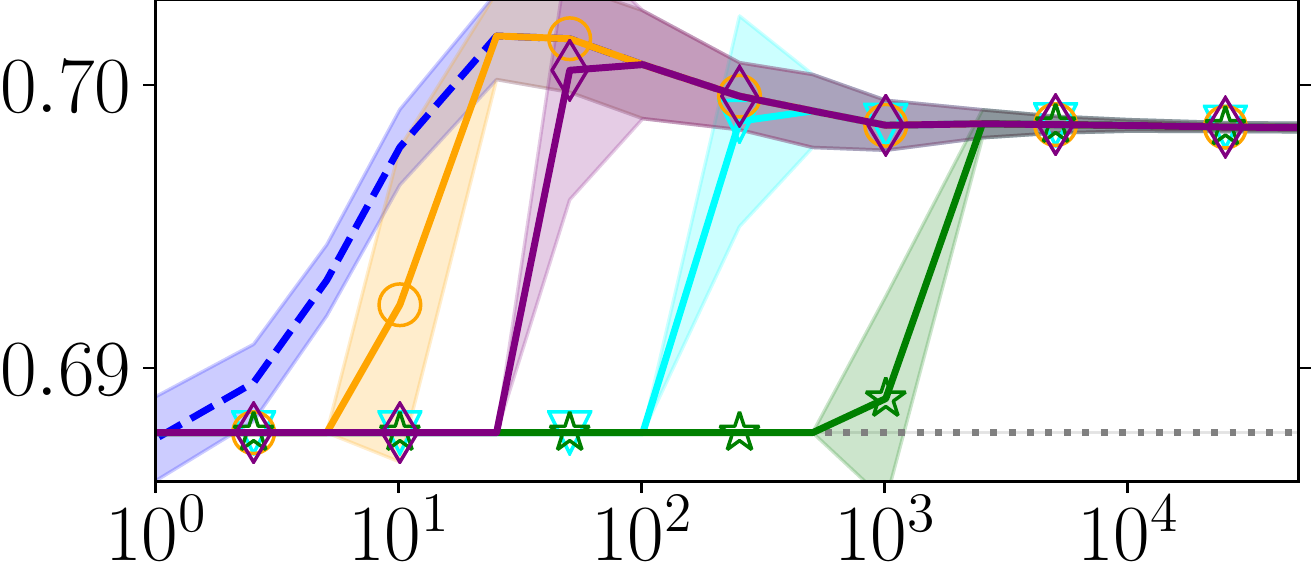}
\\
& \multicolumn{1}{c}{\small \hspace{0.5em} Mean Number of Clicks per Query}
& \multicolumn{1}{c}{\small \hspace{0.5em} Mean Number of Clicks per Query}
\\
\\
 \multicolumn{3}{c}{
 \includegraphics[scale=0.4]{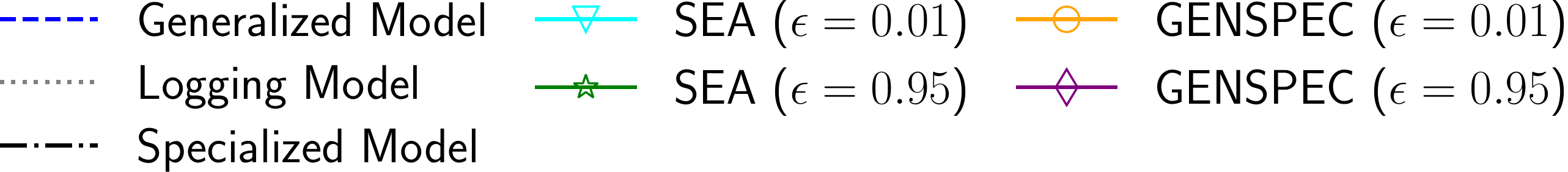}
 }
\end{tabular}
\caption{
\ac{GENSPEC} compared to a meta-policy using the \acs{SEA} bounds~(see Section~\ref{sec:SEAresults}), on clicks generated with $\alpha=0.2$.
Notation is the same as in Figure~\ref{fig:genspec:main:linear10}.
}
\label{fig:sea:linear10}
\end{figure}

\begin{figure}[t]
\centering
\begin{tabular}{c r r}
&
\multicolumn{1}{c}{\hspace{1.3em} Train-NDCG}
&
\multicolumn{1}{c}{\hspace{1.7em} Test-NDCG}
\\
\rotatebox[origin=lt]{90}{\hspace{0.0em} \footnotesize Yahoo!\ Webscope} &
\includegraphics[scale=0.4]{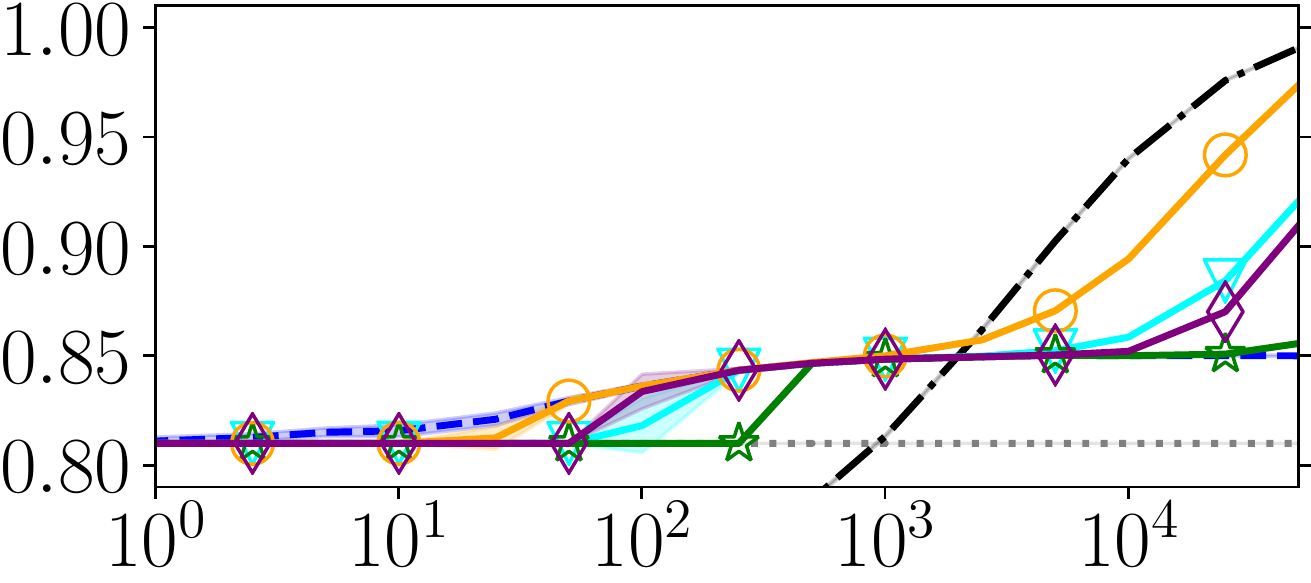} &
\includegraphics[scale=0.4]{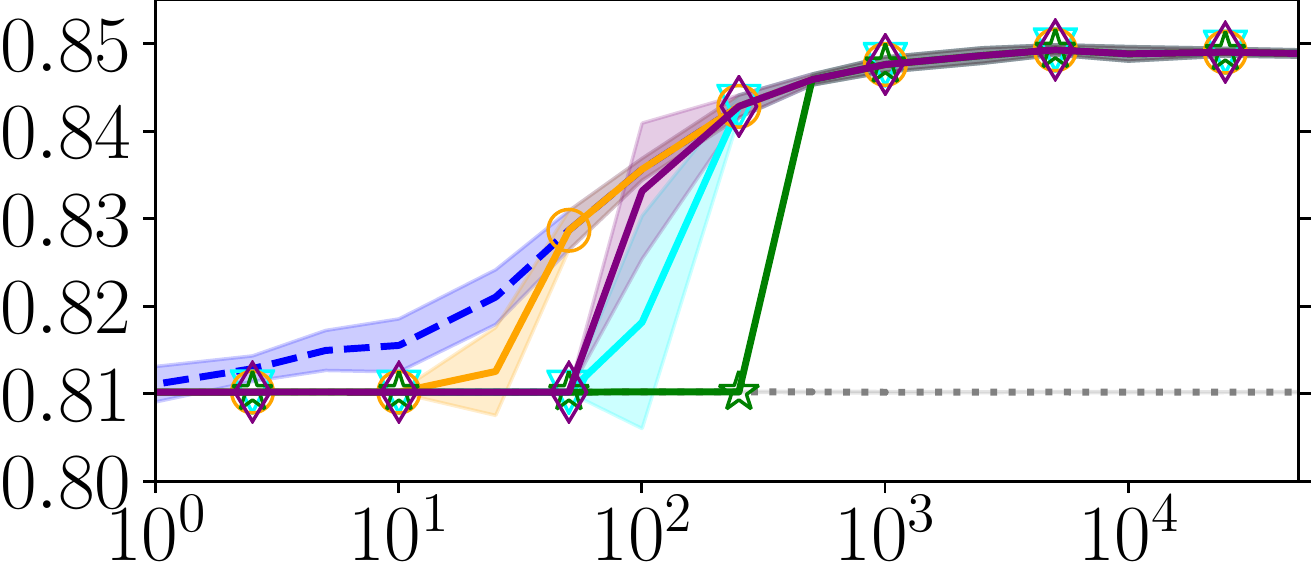}
\\
\rotatebox[origin=lt]{90}{\hspace{0.4em} \footnotesize MSLR-WEB30k} &
\includegraphics[scale=0.4]{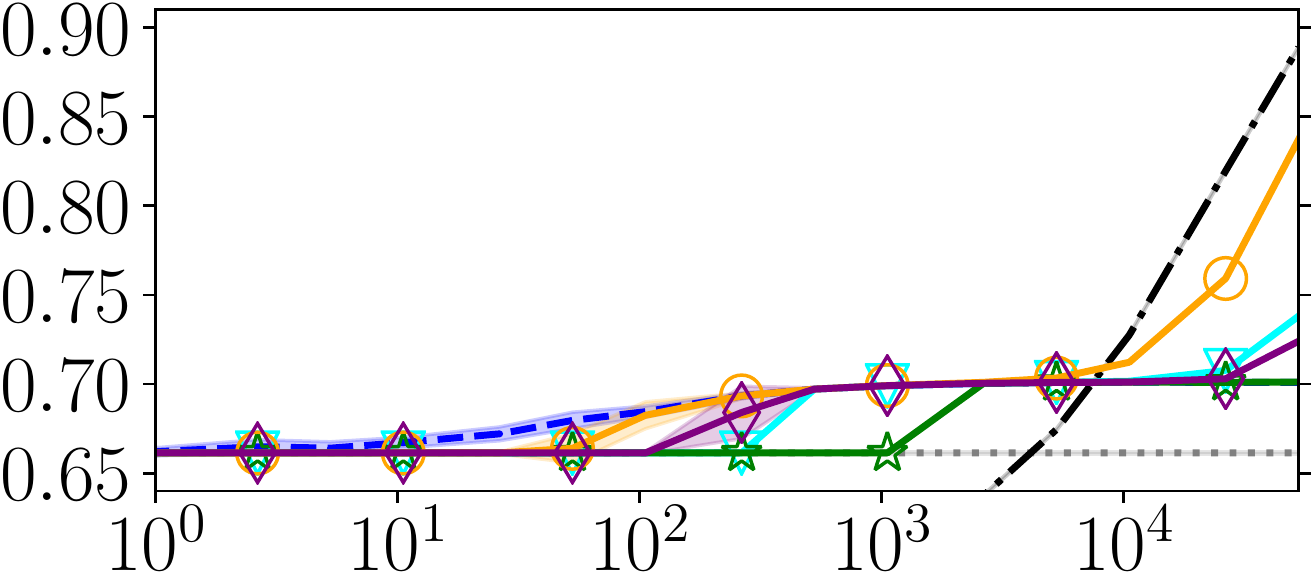} &
\includegraphics[scale=0.4]{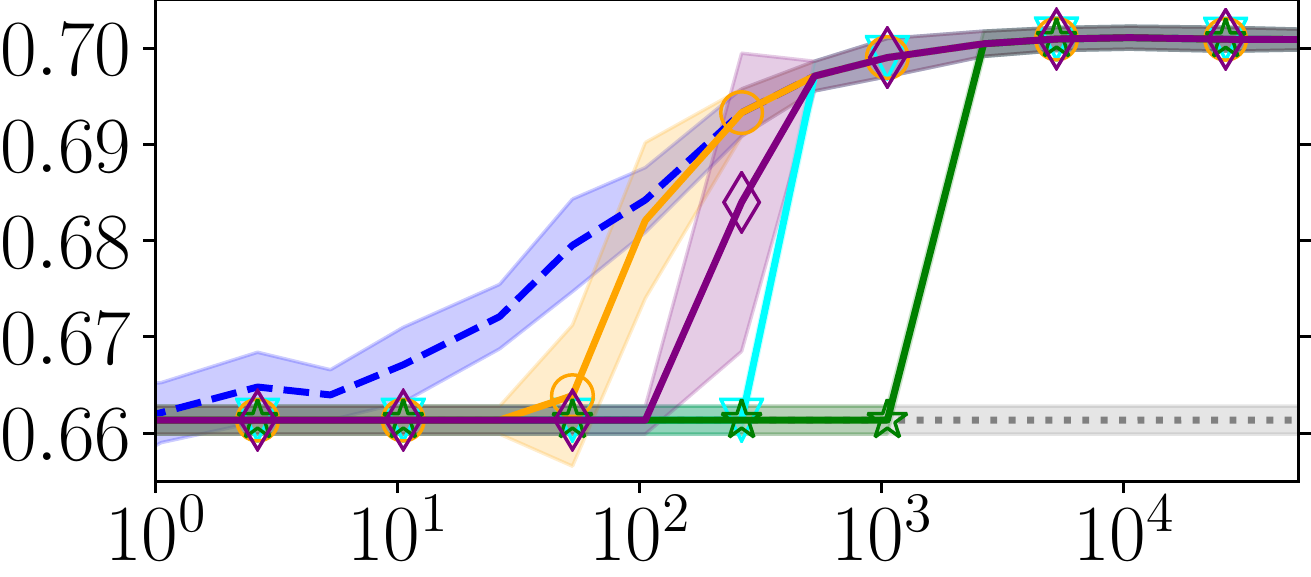}
\\
\rotatebox[origin=lt]{90}{\hspace{2.3em} \footnotesize Istella} &
\includegraphics[scale=0.4]{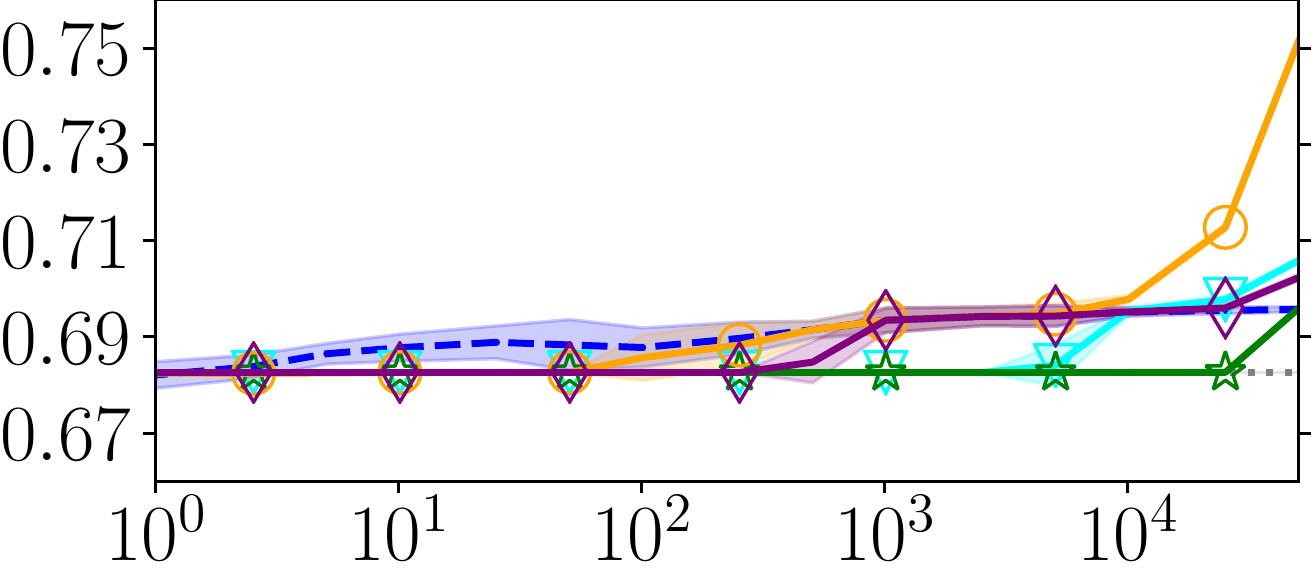} &
\includegraphics[scale=0.4]{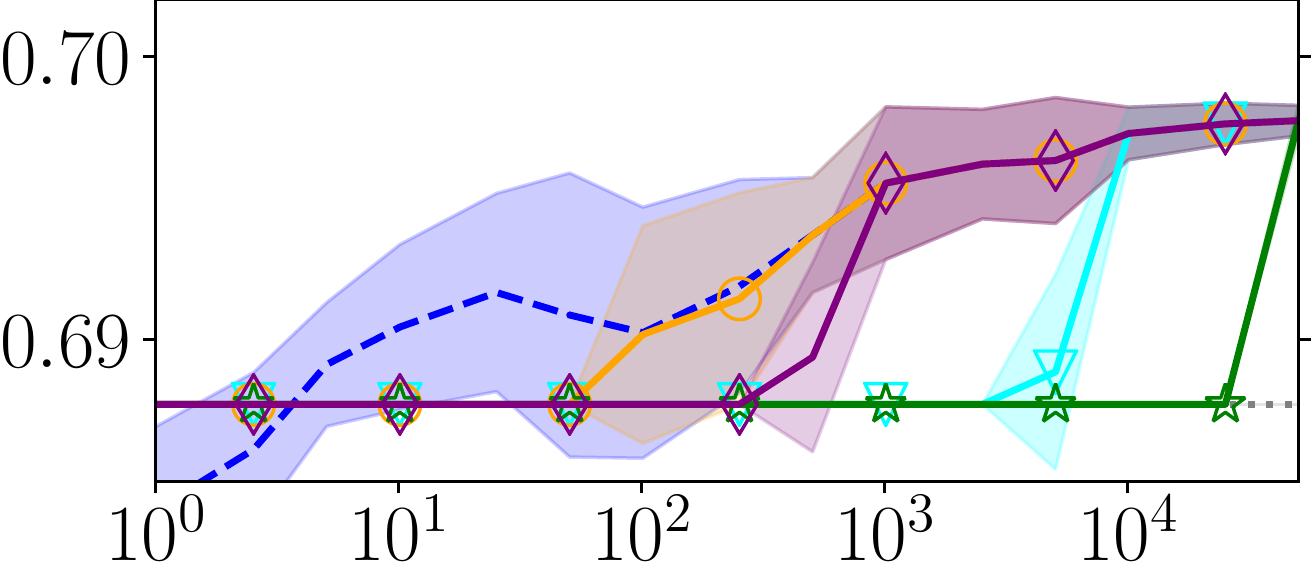}
\\
& \multicolumn{1}{c}{\small \hspace{0.5em} Mean Number of Clicks per Query}
& \multicolumn{1}{c}{\small \hspace{0.5em} Mean Number of Clicks per Query}
\\
\\
 \multicolumn{3}{c}{
 \includegraphics[scale=0.4]{08-genspec/figures/legend_sea_train_ndcg}
 }
\end{tabular}
\caption{
\ac{GENSPEC} compared to a meta-policy using the \acs{SEA} bounds~(see Section~\ref{sec:SEAresults}), on clicks generated with $\alpha=0.025$.
Notation is the same as in Figure~\ref{fig:genspec:main:linear10}.
}
\label{fig:sea:linear03}
\end{figure}

\begin{figure}[t]
\centering
\begin{tabular}{c r r}
&
\multicolumn{1}{c}{\hspace{1.0em}  Clicks generated with $\alpha=0.2$.}
&
\multicolumn{1}{c}{\hspace{1.0em}  Clicks generated with $\alpha=0.025$.}
\\
\rotatebox[origin=lt]{90}{\hspace{0.0em} \footnotesize Yahoo!\ Webscope} &
\includegraphics[scale=0.4]{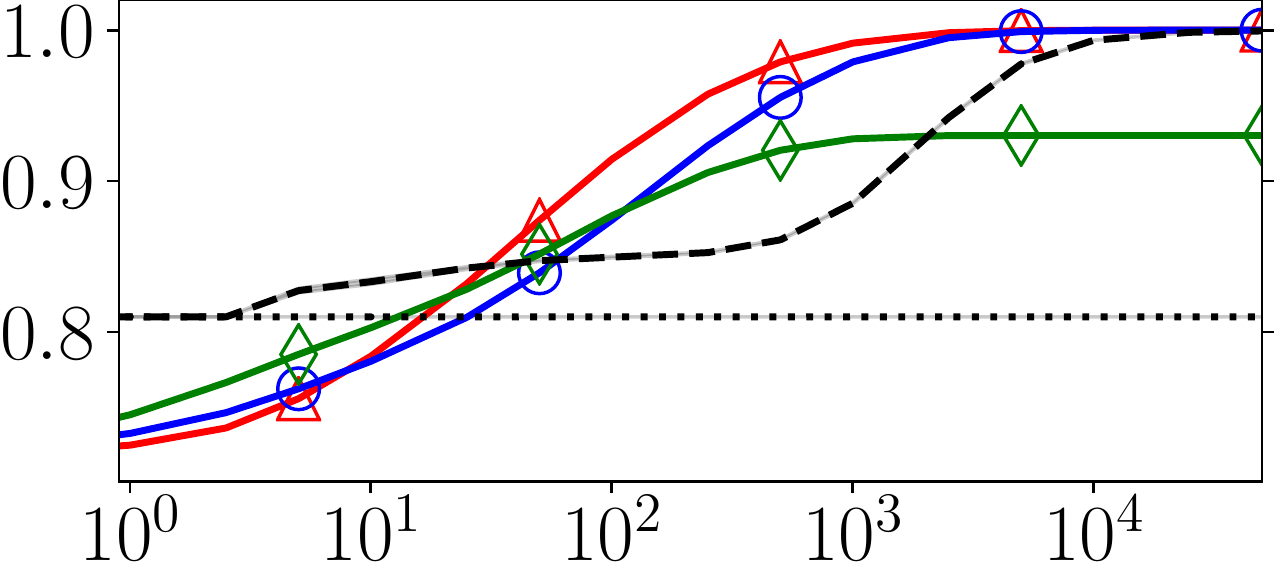} &
\includegraphics[scale=0.4]{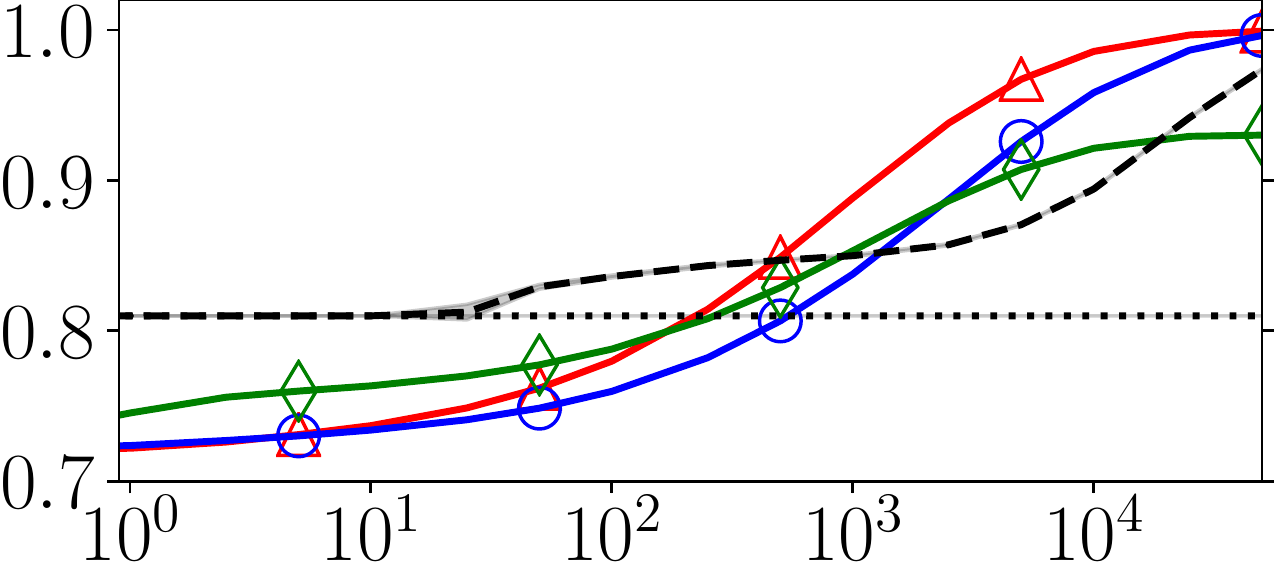}
\\
\rotatebox[origin=lt]{90}{\hspace{0.4em} \footnotesize MSLR-WEB30k} &
\includegraphics[scale=0.4]{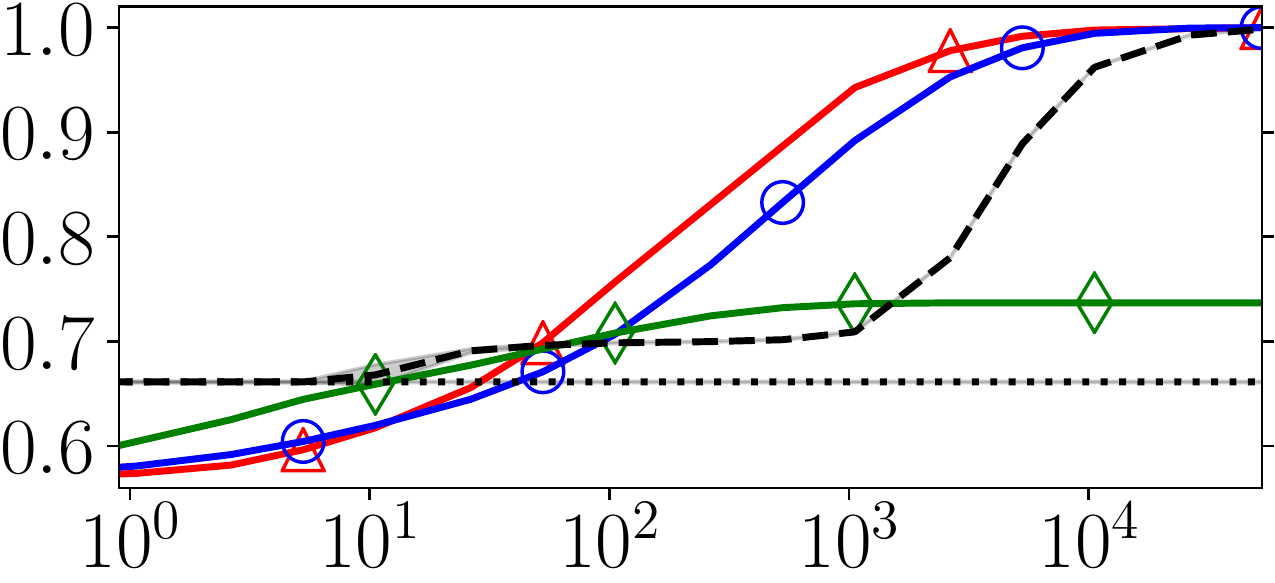} &
\includegraphics[scale=0.4]{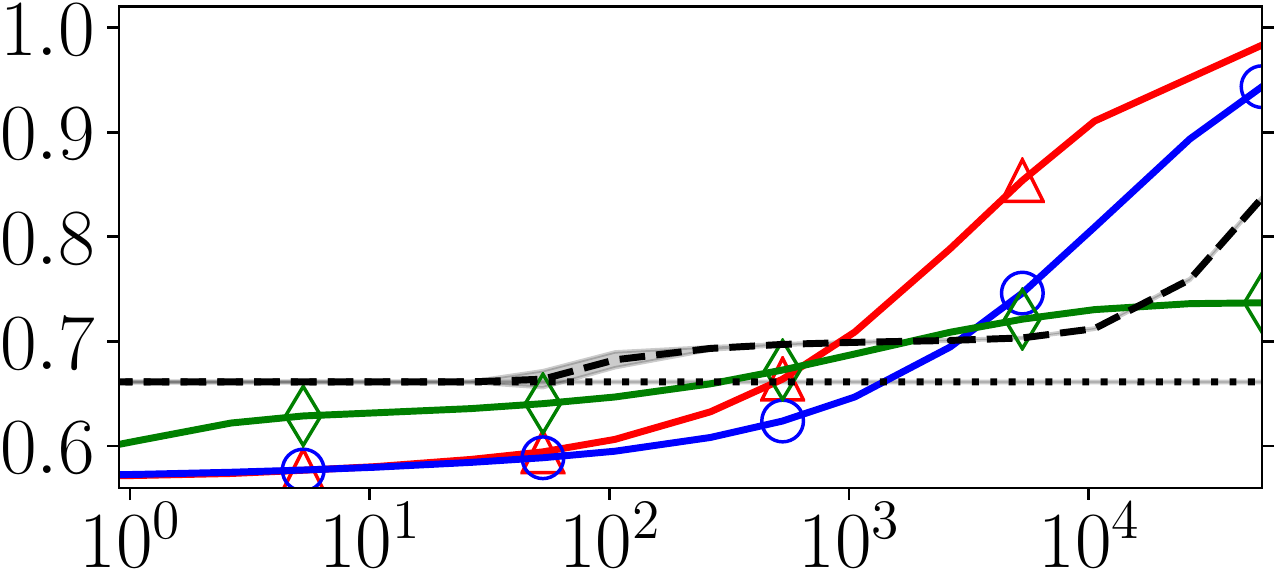}
\\
\rotatebox[origin=lt]{90}{\hspace{2.3em} \footnotesize Istella} &
\includegraphics[scale=0.4]{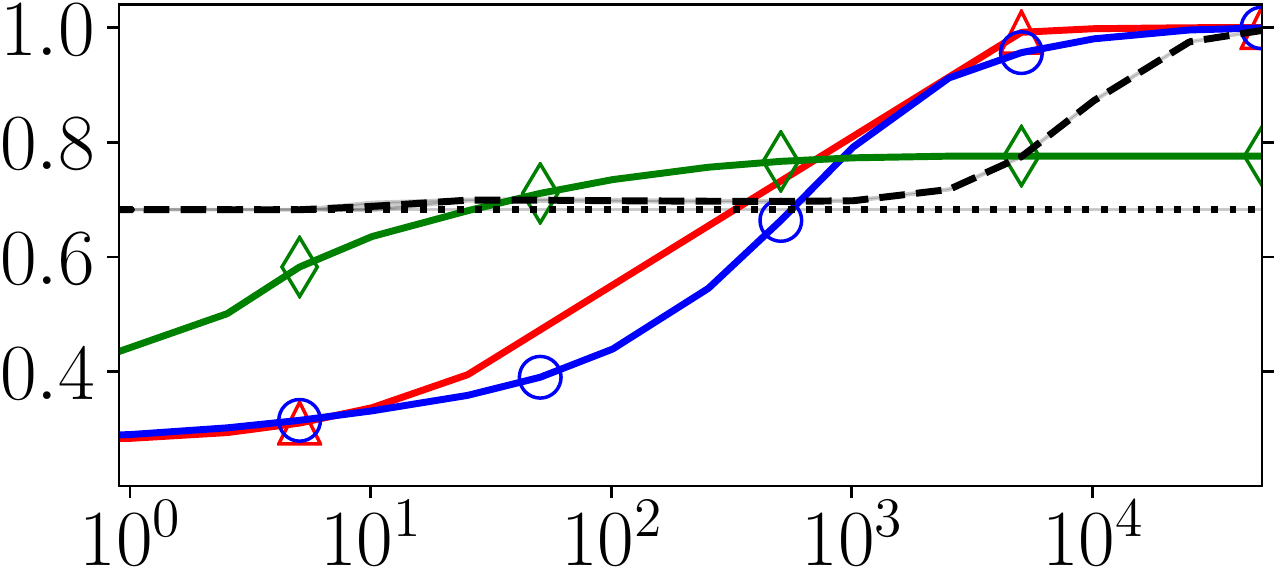} &
\includegraphics[scale=0.4]{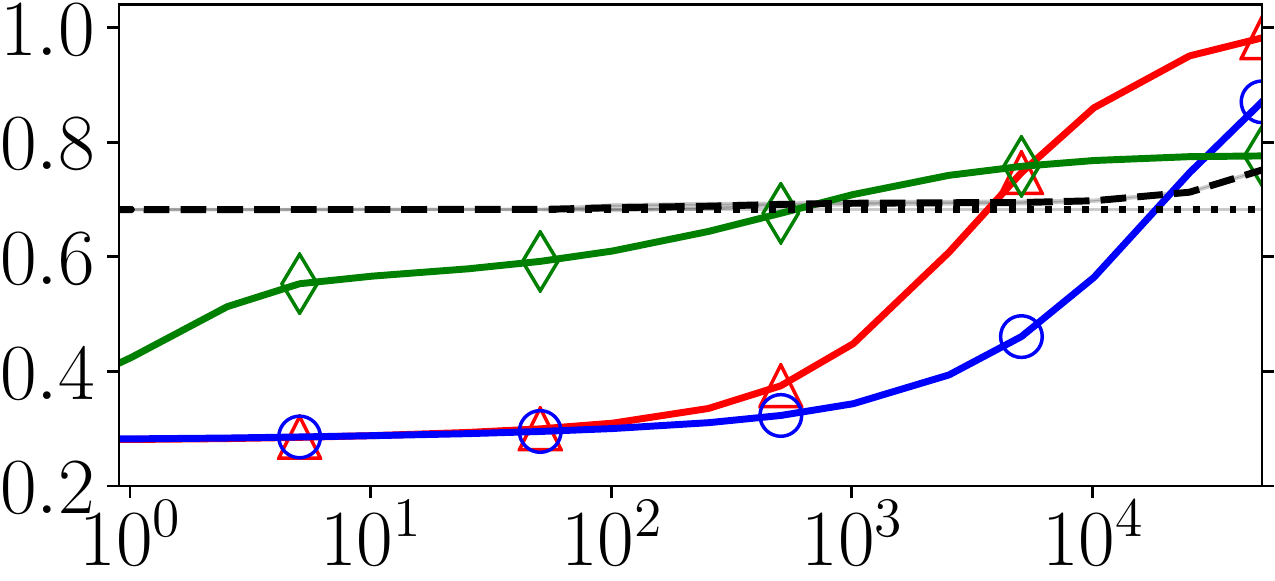}
\\
& \multicolumn{1}{c}{\small \hspace{0.5em} Mean Number of Clicks per Query}
& \multicolumn{1}{c}{\small \hspace{0.5em} Mean Number of Clicks per Query}
\\
\\
 \multicolumn{3}{c}{
 \includegraphics[scale=0.4]{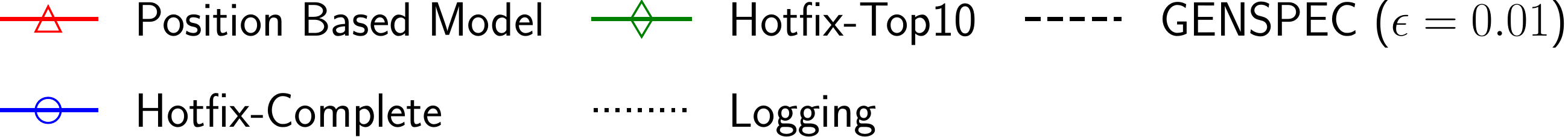}
 }
\end{tabular}
\caption{
\ac{GENSPEC} compared to various online \ac{LTR} bandits (see Section~\ref{sec:onlineltr}).
Notation is the same as in Figure~\ref{fig:genspec:main:linear10}.
}
\label{fig:bandits}
\end{figure}

\subsection{Behavior of the \acs{GENSPEC} meta-policy}
First, we consider the behavior of \ac{GENSPEC} compared with pure generalization or pure specialization policies.
Figures~\ref{fig:genspec:main:linear10} and~\ref{fig:genspec:main:linear03} show the performance of 
\begin{enumerate*}[label=(\roman*)]
\item \ac{GENSPEC} with different levels of confidence for its bounds ($\epsilon$),
along with that of 
\item the logging policy, 
\item the pure generalization policy, and 
\item the pure specialization policies between which the \ac{GENSPEC} meta-policy chooses.
\end{enumerate*}
We see that pure generalization requires few clicks to improve over the logging policy but is not able to reach optimal levels of performance.
The performance of pure specialization, on the other hand, is initially far below the logging policy.
However, after enough clicks have been gathered, performance increases until the optimal ranking is found; when click noise is limited ($\alpha = 0.2$) it reaches perfect performance on all three datasets (Train-NDCG).
On the unseen queries where there are no clicks (Test-NDCG), the specialization policy is unable to learn anything and provides random performance (not displayed in Figures~\ref{fig:genspec:main:linear10} and~\ref{fig:genspec:main:linear03}).
The initial period of poor performance can be very detrimental to queries that do not receive a large number of clicks.
Prior work has found that web-search queries follow a long-tail distribution~\citep{silverstein1999analysis, spink2002us}; \citet{white2007studying} found that 97\% of queries received fewer than $10$ clicks over six months.
For such queries, users may only experience the initial poor performance of pure specialization, and never see the improvements it brings at convergence.
This possibility can be a large deterrent from applying pure specialization in practice~\citep{wu2016conservative}.

Finally, the \ac{GENSPEC} policy combines properties of both: after a few clicks it deploys the generalization policy and thus outperforms the logging policy; as more clicks are gathered,  specialization policies are activated, further improving performance.
With $\alpha = 0.2$ the \ac{GENSPEC} policy with $\epsilon \leq 0.5$ reaches perfect Train-NDCG performance on all three datasets, similar to the pure specialization policy.
However, unlike pure specialization the performance of \ac{GENSPEC} (with $\epsilon > 0$) never drops below the logging policy.
Moreover, we never observe the situation where an increase in the number of clicks results in a decrease in mean performance.
There is a delay between when the pure specialization policy is the optimal choice and when \ac{GENSPEC} activates specialization policies.
Thus, while the usage of confidence bounds prevents the performance from dropping below the level of the logging policy, it does so at the cost of this delay.
When \ac{GENSPEC} does not use any bounds it deploys specialized policies earlier, however, in some cases these deployments result in worse performance than the logging policy, albeit less than pure specialization.
In all our observed results, a confidence of $\epsilon = 0.01$ was enough to prevent any decreases in performance.

To conclude, our experimental results show that the \ac{GENSPEC} meta-policy combines the high-performance at convergence of specialization and the safe robustness of generalization.
In contrast to pure specialization, which results in very poor performance when not enough clicks have been gathered, \ac{GENSPEC} effectively avoids incorrect deployment and under our tested conditions it never performs worse than the logging policy.
Meanwhile, \ac{GENSPEC} achieves considerable gains in performance at convergence, in contrast with pure generalization.
Therefore, we conclude that \ac{GENSPEC} is the best choice in situations where periods of poor performance have to be avoided~\citep{wu2016conservative} or when not all queries receive large numbers of clicks~\citep{white2007studying}.

\subsection{Effectiveness of relative bounding}
\label{sec:SEAresults}

\ac{GENSPEC} is not the first method that deploys policies based on confidence bounds. 
As discussed in Section~\ref{sec:genspecltr}, \citet{jagerman2020safety} previously introduced the \ac{SEA} algorithm.
\ac{SEA} chooses between deploying a generalizing policy or keeping the logging policy in deployment, by bounding both the performance of the logging and generalization policy.
When the upper bound of the logging policy is less than the lower bound of the generalizing policy, \ac{SEA} deploys the latter.
The big differences with \ac{GENSPEC} are that \ac{SEA} 
\begin{enumerate*}[label=(\roman*)]
\item uses two bounds to confidently estimate if one policy outperforms another, and 
\item does not consider specialization policies.
\end{enumerate*}
Because \ac{GENSPEC} directly bounds relative performance, its comparisons only use a single bound and thus we expect it to be more efficient w.r.t.\ the number of clicks required than SEA (see Appendix~\ref{sec:theoryrelativebounds} for a formal analysis).

For a fair comparison, we adapt \ac{SEA} to choose between the same policies as \ac{GENSPEC} and provide it with the same click data.
Figure~\ref{fig:sea:linear10} and~\ref{fig:sea:linear03} display the results of this comparison.
Across all settings, \ac{GENSPEC} deploys policies much earlier than \ac{SEA} with the same level of confidence.
While they converge at the same levels of performance, \ac{GENSPEC} requires considerably less data, e.g., on the Istella dataset with $\alpha=0.025$ \ac{GENSPEC} deploys with $10$ times less data.
Thus, we conclude that the relative bounds of \ac{GENSPEC} are much more efficient than the existing bounding approach of \ac{SEA}.

\subsection{Comparison with counterfactual \acs{LTR}}
\label{sec:comparisoncltr}
Obvious baselines for our experiments are methods from the counterfactual \ac{LTR} field~\citep{wang2016learning, joachims2017unbiased, ai2018unbiased}.
In our setting where the observance probabilities are given, all these methods reduce to~\citet{oosterhuis2020topkrankings}'s method (see Chapter~\ref{chapter:04-topk}), i.e., the method used to optimize the pure generalization policy in Figures~\ref{fig:genspec:main:linear10} and~\ref{fig:genspec:main:linear03}.
Thus, the comparison between \ac{GENSPEC} and pure generalization is effectively a comparison between \ac{GENSPEC} and state-of-the-art counterfactual \ac{LTR}.
As expected, we see that \ac{GENSPEC} reaches the same performance on previously unseen queries (Test-NDCG); but on queries with clicks (Train-NDCG) \ac{GENSPEC} outperforms standard counterfactual \ac{LTR} by enormous amounts once many clicks have been gathered.
Again, there is a small delay between the moment the generalization policy outperforms the logging policy and when \ac{GENSPEC} deploys it.
Since this observed delay is very short, this downside seems to be heavily outweighed by the large increases in performance in Train-NDCG.
Thus, we conclude that \ac{GENSPEC} is preferable over the existing counterfactual \ac{LTR} approaches, due to its ability to incorporate highly specialized models in its policy.

\subsection{Comparison with online \acs{LTR} bandits}
\label{sec:onlineltr}
Other related methods are online \ac{LTR} bandit algorithms~\citep{kveton2015cascading, katariya2016dcm}.
Unlike counterfactual \ac{LTR}, these bandit methods learn using online interventions: at each timestep they choose which ranking to display to users.
Thus, they have some control over the interactions they receive, and attempt to display rankings that will benefit the learning process the most.
As baselines we use the hotfix algorithm~\citep{zoghi2016click} and the \ac{PBM}~\citep{lagree2016multiple}.
The hotfix algorithm is a very general approach, it completely randomly shuffles the top-$n$ items and ranks them based on pairwise preferences inferred from clicks.
The main downside of the hotfix approach is that it can be very detrimental to the user experience due to the randomization.\footnote{We only report the performance of the ranking produced by the hotfix baseline, not of the randomized rankings used to gather clicks.}
We apply two versions of the hotfix algorithm, one for top-10 reranking to minimize randomization and another for reranking the complete ranking.
\ac{PBM} is perfectly suited for our task as it makes the same assumptions about user behavior as our experimental setting.
We apply PMB-PIE~\citep{lagree2016multiple}, which results in PBM always displaying the ranking it expects to perform best, thus attempting to maximize the user experience during learning.
These methods are very similar to our specialization policies: the bandit baselines memorize the best rankings and do not depend on features at all.
Consequently, their learned policies cannot be applied to previously unseen queries.
 
Figure~\ref{fig:bandits} displays the results for this comparison.
We see that when $\alpha = 0.2$ Hotfix-Complete, \ac{PBM} and \ac{GENSPEC} all reach perfect Train-NCDG; however, Hotfix-Complete and \ac{PBM} reach convergence much earlier than \ac{GENSPEC}.
We attribute this difference to three causes:
\begin{enumerate*}[label=(\roman*)]
\item the online interventions of the bandit baselines,
\item \ac{GENSPEC} only uses 70\% of the available data for training ($\traindata$) whereas the bandit baselines use everything, and
\item the delay in deployment added by \ac{GENSPEC}'s usage of confidence bounds.
\end{enumerate*}
Similar to the pure specialization policies, the earlier moment of convergence of the bandit baselines comes at the cost of an initial period of very poor performance.
We conclude that if only the moment of reaching optimal performance matters, \ac{PBM} is the best choice of method.
However, if periods of poor performance should be avoided~\citep{wu2016conservative}, or if some queries may not receive large numbers of clicks~\citep{white2007studying}, \ac{GENSPEC} is the better choice.
An additional advantage is that \ac{GENSPEC} is a counterfactual method and does not have to be applied online like the bandit baselines.

\subsection{Feature-based online \acs{LTR}}
Besides the bandit baselines discussed in Section~\ref{sec:onlineltr}, feature-based methods for online \ac{LTR} also exist~\citep{oosterhuis2018differentiable, schuth2016mgd, yue2009interactively, wang2019variance}.
A direct experimental comparison with these methods is beyond the scope of this chapter.
However, previous work has already compared these methods with each other~\citep{oosterhuis2019optimizing} and the state-of-the-art method with counterfactual \ac{LTR}~\citep{jagerman2019comparison}.
Based on the latter work by \citet{jagerman2019comparison} we do not expect considerable differences between these online \ac{LTR} methods and counterfactual \ac{LTR} in our settings.
Therefore, we expect that a comparison would lead to similar results as discussed in Section~\ref{sec:comparisoncltr}.

%% file: 08-genspec/sections/07-discussion.tex
\section{GENSPEC for Contextual Bandits}
\label{sec:genspeccontextualbandit}

So far we have discussed \ac{GENSPEC} for counterfactual \ac{LTR}. 
We will now show that it is also applicable to the broader contextual bandit problem.
Instead of a query $q$, we now keep track of an arbitrary context $z \in \{1,2,\ldots\}$ where
\begin{equation}
x_i, z_i \sim P(x, z).
\end{equation}
Data is gathered using a logging policy $\pi_0$:
\begin{equation}
a_i \sim \pi_0(a \mid x_i, z_i).
\end{equation}
However, unlike the \ac{LTR} case, the rewards $r_i$ are observed directly:
\begin{equation}
r_i \sim P(r \mid a_i,  x_i, z_i).
\end{equation}
With the propensities
\begin{equation}
\rho_i = \pi_0(a_i \mid x_i, z_i)
\end{equation}
the data is:
\begin{equation}
\mathcal{D} = \big \{(r_i, a_i, \rho_i, x_i, z_i)\big \}^N_{i=1};
\end{equation}
for specialization the data is filtered per context $z$:
\begin{equation}
\mathcal{D}_z = \big \{ (r_i, a_i, \rho_i, x_i, z_i) \in \mathcal{D}\mid z_i = z \big\}.
\end{equation}
Again, data for training $\traindata$ and for policy selection $\bounddata$ are separated.
The reward is estimated with an \ac{IPS} estimator:
\begin{equation}
\hat{\mathcal{R}}(\pi \mid \mathcal{D}) = \frac{1}{|\mathcal{D}|} \sum_{i \in \mathcal{D}} \frac{r_i}{\rho_i} \, \pi(a_i \mid x_i, z_i).
\end{equation}
With the policy spaces $\Pi_g$ and $\Pi_z$, the policy for generalization is:
\begin{equation}
\pi_g = \argmax_{\pi \in \Pi_g} \, \hat{\mathcal{R}}(\pi \mid \traindata);
\end{equation}
per context $z$, the specialization policy is:
\begin{equation}
\pi_z = \argmax_{\pi \in \Pi_z} \, \hat{\mathcal{R}}(\pi \mid \traindata_z).
\end{equation}
The difference between two policies is estimated by:
\begin{equation}
\hat{\delta}(\pi_1, \pi_2 \mid \mathcal{D}) = \hat{\mathcal{R}}(\pi_1 \mid \mathcal{D}) - \hat{\mathcal{R}}(\pi_2 \mid \mathcal{D}).
\end{equation}
We differ from the \ac{LTR} approach by estimating the bounds using:
\begin{equation}
R_i =  \frac{r_i}{\rho_i} \, \big(\pi_1(a_i \mid x_i, z_i) - \pi_2(a_i \mid x_i, z_i) \big).
\end{equation}
Following~\citet{thomas2015high}, the confidence bounds are:
\begin{equation}
\begin{split}
\textit{CB}(\pi_1,\, \pi_2 \mid \mathcal{D})
= \,&\, \frac{7b\ln\big(\frac{2}{1-\epsilon}\big)}{3(|\mathcal{D}|-1)}  \\
& + \frac{1}{|\mathcal{D}|} \sqrt{\frac{ 2 |\mathcal{D}| \ln\big(\frac{2}{1-\epsilon}\big)}{|\mathcal{D}|-1} \sum_{i  \in \mathcal{D}} \big(R_i - \hat{\delta}(\pi_1, \pi_2 \mid \mathcal{D}) \big)^2},
\end{split}
\end{equation}
where $b$ is the maximum possible value for $R_i$.
This results in the lower bound:
\begin{equation}
\textit{LCB}(\pi_1, \pi_2 \mid \mathcal{D}) = \hat{\delta}(\pi_1, \pi_2 \mid \mathcal{D}) - \textit{CB}(\pi_1, \pi_2 \mid \mathcal{D}),
\end{equation}
which is used by the \ac{GENSPEC} meta-policy:
\begin{equation}
\begin{split}
\pi_{GS}(a |\, x, z)
= \hspace{-1.1cm} &
\\
&  
\begin{cases}
\pi_z(a |\, x, z),&\text{if } \big( \textit{LCB}(\pi_z, \pi_g | \bounddata_z) > 0
 \land \, \textit{LCB}(\pi_z, \pi_0 | \bounddata_z) > 0 \big), \\
\pi_g(a |\, x, z),& \text{if } \big( \textit{LCB}(\pi_z, \pi_g |  \bounddata_z) \leq 0 \land \, \textit{LCB}(\pi_g, \pi_0 | \bounddata) > 0 \big), \\
\pi_0(a |\, x, z),& \text{otherwise}. 
\end{cases}
\end{split}
\end{equation}
As such, \ac{GENSPEC} can be applied to the contextual bandit problem for any arbitrary choice of context $z$.

%% file: 08-genspec/sections/08-conclusion.tex
\section{Conclusion}
\label{sec:Conclusion}

In this chapter we have introduced the \acf{GENSPEC} framework for contextual bandit problems.
For an arbitrary choice of contexts it simultaneously learns a general policy to perform well across all contexts, and many specialized policies each optimized for a single context.
Then, per context the \ac{GENSPEC} meta-policy uses high-confidence bounds to choose between deploying the logging policy, the general policy, or a specialized policy.
As a result, \ac{GENSPEC} combines the robust safety of a general policy with the high-performance of a successfully specialized policy.

We have shown how \ac{GENSPEC} can be applied to query-special\-ization for counterfactual \ac{LTR}.
Our results show that \ac{GENSPEC} combines the high performance of specialized policies on queries with sufficiently many interactions, with the robust performance on queries that were previously unseen or where little data is available.
Thus, it avoids the low performance at convergence of feature-based models underlying the general policy, and the initial poor performance of the tabular models underlying the specialized policies.
We expect that \ac{GENSPEC} can be used for other types of specialization by choosing different context divisions, i.e., personalization for \ac{LTR} is a promising choice.

With these findings we can answer thesis research question \ref{thesisrq:genspec} positively:
Using \ac{GENSPEC} we can combine the specialization ability of bandit-style online \ac{LTR} with the robust generalization of feature-based \ac{LTR}.
As a result, the choice between specialization and generalization can now be made in a principled, theoretically-grounded manner.
For the \ac{LTR} field this means that bandit-style \ac{LTR} and feature-based \ac{LTR} can now be seen as complementary, instead of a mutually exclusive choice.

Future work could explore other contextual bandit problems and choices for context.
Additionally, we hope that the robust safety of \ac{GENSPEC} further incites the application of bandit algorithms in practice.

While this chapter considered \ac{GENSPEC} for counterfactual \ac{LTR},
Chapter~\ref{chapter:06-onlinecounterltr} introduces a novel method that is both effective at counterfactual \ac{LTR} and online \ac{LTR}.
With only small adaptations the contributions of both chapters could be combined, thus potentially resulting in \ac{GENSPEC} for both online and counterfactual \ac{LTR}.
Future work could investigate the effectiveness of this possible combined approach.

%% file: 08-genspec/sections/09-appendix.tex
\section{Proof of Unbiasedness for Counterfactual Learning to Rank}
\label{sec:counterfactualproof}

This section will prove that the \ac{IPS} estimate $\hat{\mathcal{R}}$~(Eq.~\ref{eq:rewardestimate}) can be used to unbiasedly optimize the true reward $\mathcal{R}$~(Eq.~\ref{eq:genspec:truereward}), as claimed in Section~\ref{sec:counterfactualltr}.
For this proof we rely on the following assumptions:
\begin{enumerate*}[label=(\roman*)]
\item \ac{LTR} metrics are linear combinations of item relevances (Eq.~\ref{eq:ltrtrueest}),
\item the assumption that clicks never occur on unobserved items (Eq.~\ref{eq:unobserved}), and
\item click probabilities (conditioned on observance) are proportional to relevance (Eq.~\ref{eq:clickprop}).
\end{enumerate*}

First, we consider the expected value for an observed click $c_i(d)$ using Eq.~\ref{eq:unobserved}; for brevity we write $r(d) = r(d \mid x_i, q_i)$:
\begin{equation}
\begin{split}
\mathbb{E}_{o_i,a_i}\big[&c_i(d)\big] 
\\
&= \mathbb{E}_{a_i}\Big[P\big(c_i(d) = 1 | o_i(d) = 1, r(d)\big) \cdot P\big(o_i(d)=1 | a_i\big)\Big]  \\
&= P\big(c_i(d) = 1 \mid o_i(d) = 1, r(d)\big) \cdot \Bigg(\sum_{a \in \pi_0}  P\big(o_i(d)=1 \mid a\big) \cdot \pi_0(a \mid x_i, q_i) \Bigg) \\
&= \rho_i(d) \cdot P\big(c_i(d) = 1 \mid o_i(d) = 1, r(d)\big).
\end{split}
\end{equation}
Then, consider the expected value for the \ac{IPS} estimator, and note that $a_i$ is a historically observed action and that $a$ is the action being evaluated:
\begin{align}
\mathbb{E}_{o_i,a_i}\big[\hat{\Delta}(a \mid c_i, \rho_i)\big]
&= \mathbb{E}_{o_i,a_i}\Bigg[\sum_{d \in a }\frac{\lambda\big(\textit{rank}(d \mid a)\big) \cdot c_i(d)}{\rho_i(d)} \Bigg]
\nonumber\\ 
& = \sum_{d \in a }\frac{\rho_i(d)}{\rho_i(d)}\cdot \lambda\big(\textit{rank}(d \mid a)\big)
 \cdot P\big(c_i(d) = 1 \mid o_i(d) = 1, r(d)\big)
 \nonumber\\ 
& = \sum_{d \in a} \lambda\big(\textit{rank}(d \,|\, a)\big) \cdot P\big(c_i(d) = 1 \,|\, o_i(d) = 1, r(d )\big). 
\end{align}
This step assumes that $\rho_i(d) > 0$, i.e., that every item has a non-zero probability of being examined~\citep{joachims2017unbiased}. %
While $\mathbb{E}_{o_i,a_i}[\hat{\Delta}(a \mid c_i, \rho_i)]$ and $\Delta(a \mid x_i, q_i, r)$ are not necessarily equal, using Eq.~\ref{eq:clickprop} we see that they are proportional with some offset $C$:
\begin{align}
\mathbb{E}_{o_i,a_i}\big[\hat{\Delta}(a \mid c_i, \rho_i)\big] &\propto
\Big(\sum_{d \in a}\lambda\big(\textit{rank}(d \mid a)\big) \cdot r(d) \Big) + C \nonumber \\
&= \Delta(a \mid x_i, q_i, r) + C,
\end{align}
where $C$ is a constant: $C = \big(\sum^K_{i=1}\lambda(i)\big) \cdot \mu$.
Therefore, in expectation $\hat{\mathcal{R}}$ and $\mathcal{R}$ are also proportional with the same constant offset:
\begin{equation}
\mathbb{E}_{o_i,a_i}\big[\hat{\mathcal{R}}(\pi \mid \mathcal{D})\big] \propto \mathcal{R}(\pi) + C.
\end{equation}
Consequently, the estimator can be used to unbiasedly estimate the preference between two policies:
\begin{equation}
\mbox{}\hspace*{-2mm}
\mathbb{E}_{o_i,a_i}\big[\hat{\mathcal{R}}(\pi_1 \mid \mathcal{D})\big] < \mathbb{E}_{o_i,a_i}\big[\hat{\mathcal{R}}(\pi_2 \mid \mathcal{D})\big]
\Leftrightarrow \mathcal{R}(\pi_1) < \mathcal{R}(\pi_2).
\hspace*{-2mm}\mbox{}
 \label{eq:relativeproof}
\end{equation}
Moreover, this implies that maximizing the estimated performance unbiasedly optimizes the actual reward:
\begin{equation}
 \argmax_{\pi} \mathbb{E}_{o_i,a_i}\big[\hat{\mathcal{R}}(\pi \mid \mathcal{D})\big] = \argmax_{\pi} \mathcal{R}(\pi).
 \label{eq:maxproof}
\end{equation}
This concludes our proof. We have shown that $\hat{\mathcal{R}}$ is suitable for counterfactual evaluation since it can unbiasedly identify if a policy outperforms another (Eq.~\ref{eq:relativeproof}) and, furthermore, that $\hat{\mathcal{R}}$ can be used for unbiased \ac{LTR}, i.e., it can be used to find the optimal policy~(Eq.~\ref{eq:maxproof}).

\section{Efficiency of Relative Bounding by \acs{GENSPEC}}
\label{sec:theoryrelativebounds}

Our experimental results showed that \ac{GENSPEC} chooses between policies more efficiently than when using \ac{SEA} bounds~\citep{jagerman2020safety}.
In other words, when one policy has higher performance than another, the relative bounds of \ac{GENSPEC} require less data to be certain about this difference than the \ac{SEA} bounds.
In this section, we will prove that the relative bounds of \ac{GENSPEC} are more efficient than the \ac{SEA} bounds when the covariance between the reward estimates of two policies is positive:
\begin{equation}
\text{cov}\big(\hat{\mathcal{R}}(\pi_1 | \mathcal{D}), \hat{\mathcal{R}}(\pi_2 | \mathcal{D})\big) > 0.
\end{equation}
This means that \ac{GENSPEC} will deploy a policy earlier than \ac{SEA} if there is high covariance, since both estimates are based on the same interaction data $\mathcal{D}$ a high covariance is very likely.

Let us first consider when \ac{GENSPEC} deploys a policy:
Deployment by \ac{GENSPEC} depends on whether a relative confidence bound is greater than the estimated difference in performance~(cf.\ Eq.~\ref{eq:genspec}).
For two policies $\pi_1$ and $\pi_2$ deployment happens when:
\begin{equation}
\hat{\mathcal{R}}(\pi_1 \mid \mathcal{D}) - \hat{\mathcal{R}}(\pi_2 \mid \mathcal{D}) - \textit{CB}(\pi_1, \pi_2 \mid \mathcal{D})  > 0.
\end{equation}
Thus the bound has to be smaller than the estimated performance difference:
\begin{equation}
\textit{CB}(\pi_1, \pi_2 \mid \mathcal{D}) <
\hat{\mathcal{R}}(\pi_1 \mid \mathcal{D}) - \hat{\mathcal{R}}(\pi_2 \mid \mathcal{D}).
\label{eq:relativeboundreq}
\end{equation}
In contrast, \ac{SEA} does not use a single bound, but bounds the performance of both policies.
For clarity, we reformulate the \ac{SEA} bound in our notation. First we have $R^{\pi_j}_{i,d}$ the observed reward for an item $d$ at interaction $i$ for policy $\pi_j$:
\begin{equation}
R^{\pi_j}_{i,d} =  \frac{c_i(d)}{\rho_i(d)} \sum_{a \in \pi_j} \pi_j(a | x_i, q_i)
 \cdot  \lambda\big(\textit{rank}(d | a)\big).
\end{equation}
Then we have a $\nu^{\pi_j}$ for each policy:
\begin{align}
\nu^{\pi_j} =  \frac{ 2  |\mathcal{D}|  K  \ln\big(\frac{2}{1-\epsilon}\big)}{|\mathcal{D}| K-1}  \sum_{(i,d) \in \mathcal{D}} \big(K \cdot R^{\pi_j}_{i,d} -  \hat{\mathcal{R}}(\pi_j \mid \mathcal{D})  \big)^2,
\nonumber
\end{align}
which we use to note the confidence bound for a single policy $\pi_j$:
\begin{align}
\textit{CB}
(\pi_j \mid \mathcal{D})
= \frac{7 K b\ln\big(\frac{2}{1-\epsilon}\big)}{3(|\mathcal{D}|  K-1)} + \frac{1}{|\mathcal{D}|  K}  
 \cdot \sqrt{\nu^{\pi_j}}.
\end{align}
We note that the $b$ parameter has the same value for both the relative and single confidence bounds.
\ac{SEA} chooses between policies by comparing their upper and lower confidence bounds:
\begin{equation}
\hat{\mathcal{R}}(\pi_1 \mid \mathcal{D}) - \textit{CB}(\pi_1 \mid \mathcal{D}) > \hat{\mathcal{R}}(\pi_2 \mid \mathcal{D}) + \textit{CB}(\pi_2 \mid \mathcal{D}).
\end{equation}
In this case, the summation of the bounds has to be smaller than the estimated performance difference:
\begin{equation}
\textit{CB}(\pi_1 \mid \mathcal{D})  + \textit{CB}(\pi_2 \mid \mathcal{D})  < \hat{\mathcal{R}}(\pi_1 \mid \mathcal{D}) - \hat{\mathcal{R}}(\pi_2 \mid \mathcal{D}).
\label{eq:singleboundreq}
\end{equation}
We can now formally describe under which condition \ac{GENSPEC} is more efficient than \ac{SEA}:
by combining Eq.~\ref{eq:relativeboundreq} and Eq.~\ref{eq:singleboundreq}, we see that relative bounding is more efficient when:
\begin{equation}
\textit{CB}(\pi_1, \pi_2 \mid \mathcal{D}) <
\textit{CB}(\pi_1 \mid \mathcal{D})  + \textit{CB}(\pi_2 \mid \mathcal{D}).
\end{equation}
We notice that $\mathcal{D}$, $K$, $b$ and $\epsilon$ have the same value for both confidence bounds, thus we only require:
\begin{align}
\sqrt{\nu} < \sqrt{\nu^{\pi_1}} + \sqrt{\nu^{\pi_2}}.
\end{align}
If we assume that $\mathcal{D}$ is sufficiently large, we see that $\sqrt{\nu}$ approximates the standard deviation scaled by some constant:
\begin{align}
\sqrt{\nu} \approx C \cdot \sqrt{\text{var}\big(\hat{\delta}(\pi_1, \pi_2 | \mathcal{D})\big)},
\end{align}
where the constant is:
$C = \sqrt{\frac{ 2  |\mathcal{D}|^2  K^2  \ln\big(\frac{2}{1-\epsilon}\big)}{|\mathcal{D}| K-1}} $.
Since the purpose of the bounds is to prevent deployment until enough certainty has been gained, we think it is safe to assume that $\mathcal{D}$ is large enough  for this approximation before any deployment takes place. 

To keep our notation concise, we use the following:
$\hat{\delta} = \hat{\delta}(\pi_1, \pi_2 | \mathcal{D})$,
$\hat{\mathcal{R}}_1 = \hat{\mathcal{R}}(\pi_1 | \mathcal{D})$,
and
$\hat{\mathcal{R}}_2 = \hat{\mathcal{R}}(\pi_2 | \mathcal{D})$.
Using the same approximations for $\sqrt{\nu^{\pi_1}}$ and $\sqrt{\nu^{\pi_2}}$ we get:
\begin{align}
\sqrt{\text{var}(\hat{\delta})} < \sqrt{\text{var}(\hat{\mathcal{R}}_1)} + \sqrt{\text{var}(\hat{\mathcal{R}}_2)}.
\end{align}
By making use of the Cauchy-Schwarz inequality, we can derive the following lower bound:
\begin{align}
\sqrt{\text{var}(\hat{\mathcal{R}}_1) + \text{var}(\hat{\mathcal{R}}_2)}
\leq \sqrt{\text{var}(\hat{\mathcal{R}}_1)} + \sqrt{\text{var}(\hat{\mathcal{R}}_2)}.
\end{align}
Therefore, the relative bounding of \ac{GENSPEC} must be more efficient when the following is true:
\begin{align}
\text{var}(\hat{\delta})
< \text{var}(\hat{\mathcal{R}}_1) + \text{var}(\hat{\mathcal{R}}_2),
\end{align}
i.e. the variance of the relative estimator must be less than the sum of the variances of the estimators for the individual policies.
Finally, by rewriting $\text{var}(\hat{\delta})$ to:
\begin{equation}
\mbox{}\hspace*{-2mm}
\text{var}(\hat{\delta})
=
\text{var}(\hat{\mathcal{R}}_1 - \hat{\mathcal{R}}_2)
= \text{var}(\hat{\mathcal{R}}_1) + \text{var}(\hat{\mathcal{R}}_2) - 2\text{cov}(\hat{\mathcal{R}}_1, \hat{\mathcal{R}}_2),
\hspace*{-2mm}\mbox{}
\end{equation}
we see that the relative bounds of \ac{GENSPEC} are more efficient than the multiple bounds of \ac{SEA}
if the covariance between $\hat{\mathcal{R}}_1$ and $\hat{\mathcal{R}}_2$ is positive:
\begin{equation}
\text{cov}(\hat{\mathcal{R}}_1, \hat{\mathcal{R}}_2) > 0.
\end{equation}
Remember that both estimates are based on the same interaction data: $\hat{\mathcal{R}}_1 = \hat{\mathcal{R}}(\pi_1 | \mathcal{D})$,
and
$\hat{\mathcal{R}}_2 = \hat{\mathcal{R}}(\pi_2 | \mathcal{D})$.
Therefore, they are based on the same clicks and propensities scores, thus it is extremely likely that the covariance between the estimates is positive.
Correspondingly, it is also extremely likely that the relative bounds of \ac{GENSPEC} are more efficient than the bounds used by \ac{SEA}.

%% file: 08-genspec/notation.tex
\section{Notation Reference for Chapter~\ref{chapter:05-genspec}}
\label{notation:05-genspec}

\begin{center}
\begin{tabular}{l l}
 \toprule
\bf Notation  & \bf Description \\
\midrule
$K$ & the number of items that can be displayed in a single ranking \\
$i$ & an iteration number \\
$q$ & a user-issued query \\
$x$ & contextual information, i.e., additional features \\
$d$ & an item to be ranked\\
$a$ & a ranked list \\
$\pi$ & a ranking policy\\
$\pi(a \mid q)$ & the probability that policy $\pi$ displays ranking $a$ for query $q$ \\
$r(d \mid x, q)$ & the relevance of item $d$ w.r.t.\ query $q$ given context $x$\\
$\lambda\big(\textit{rank}(d \mid a)\big)$ & a metric function that weights items depending on their rank \\
$\mathcal{D}$ & the available interaction data\\
$c_i(d)$ & a function indicating item $d$ was clicked at iteration $i$ \\
$o_i(d)$ & a function indicating item $d$ was observed at iteration $i$ \\
\bottomrule
\end{tabular}
\end{center}

%% file: 09-onlinecountereval/main.tex
\chapter[Efficient and Unbiased Online Evaluation for Ranking]{Taking the Counterfactual Online: Efficient and Unbiased Online Evaluation for Ranking}
\label{chapter:06-onlinecountereval}

\footnote[]{This chapter was published as~\citep{oosterhuis2020taking}.
Appendix~\ref{notation:06-onlinecountereval} gives a reference for the notation used in this chapter.
}

Counterfactual evaluation can estimate \ac{CTR} differences between ranking systems based on historical interaction data, while mitigating the effect of position bias and item-selection bias.
In contrast, online evaluation methods, designed for ranking, estimate performance differences between ranking systems by showing interleaved rankings to users and observing their clicks.
We are curious to find out whether the online interventions of online evaluation methods truly result in more efficient evaluation, and additionally, whether the popular interleaving methods are truly unbiased w.r.t.\ biases such as position bias.
Accordingly this chapter will consider the following two thesis research questions:
\begin{itemize}
\item[\ref{thesisrq:onlineeval}] Can counterfactual evaluation methods for ranking be extended to perform efficient and effective online evaluation?
\item[\ref{thesisrq:interleaving}] Are existing interleaving methods truly capable of unbiased evaluation w.r.t.\ position bias?
\end{itemize}
We introduce the novel \acf{LogOpt}, which optimizes the policy for logging data so that the counterfactual estimate has minimal variance.
As minimizing variance leads to faster convergence, \ac{LogOpt} increases the data-efficiency of counterfactual estimation.
\ac{LogOpt} turns the counterfactual approach -- which is indifferent to the logging policy -- into an online approach, where the algorithm decides what rankings to display.
We prove that, as an online evaluation method, \ac{LogOpt} is unbiased w.r.t.\ position and item-selection bias, unlike existing interleaving methods.
Furthermore, we perform large-scale experiments by simulating comparisons between \emph{thousands} of rankers.
Our results show that while interleaving methods make systematic errors, \ac{LogOpt} is as efficient as interleaving without being biased.
Lastly, we provide a formal proof that shows interleaving methods are not unbiased w.r.t.\ position bias.

\input{09-onlinecountereval/sections/01-introduction}

\input{09-onlinecountereval/sections/02-ctr-estimators}

\input{09-onlinecountereval/sections/03-related-work}

\input{09-onlinecountereval/sections/04-method}

\input{09-onlinecountereval/sections/05-experimental-setup}
\input{09-onlinecountereval/sections/06-results}

\input{09-onlinecountereval/sections/07-conclusion}
\begin{subappendices}
\input{09-onlinecountereval/sections/08-appendix}

\input{09-onlinecountereval/notation}
\end{subappendices}

%% file: 09-onlinecountereval/sections/01-introduction.tex
\section{Introduction}
\label{sec:intro}

Evaluation is essential for the development of search and recommendation systems~\citep{hofmann-online-2016, kohavi2017online}.
Before any ranking model is widely deployed it is important to first verify whether it is a true improvement over the currently-deployed model.
A traditional way of evaluating relative differences between systems is through A/B testing, where part of the user population is exposed to the current system (``control") and the rest to the altered system (``treatment") during the same time period.
Differences in behavior between these groups can then indicate if the alterations brought improvements, e.g., if the treatment group showed a higher \ac{CTR} or more revenue was made with this system~\citep{chapelle2012large}.

Interleaving has been introduced in \ac{IR} as a more efficient alternative to A/B testing~\citep{joachims2003evaluating}.
Interleaving algorithms take the rankings produced by two ranking systems, and for each query create an interleaved ranking by combining the rankings from both systems.
Clicks on the interleaved rankings directly indicate relative differences.
Repeating this process over a large number of queries and averaging the results, leads to an estimate of which ranker would receive the highest \ac{CTR}~\citep{hofmann2013fidelity}.
Previous studies have found that interleaving requires fewer interactions than A/B testing, which enables them to make consistent comparisons in a much shorter timespan~\citep{schuth2015predicting, chapelle2012large}.

More recently, counterfactual evaluation for rankings has been proposed by \citet{joachims2017unbiased} to evaluate a ranking model based on clicks gathered using a different model.
By correcting for the position bias introduced during logging, the counterfactual approach can unbiasedly estimate the \ac{CTR} of a new model on historical data.
To achieve this, counterfactual evaluation makes use of \ac{IPS}, where clicks are weighted inversely to the probability that a user examined them during logging~\citep{wang2016learning}.
A big advantage compared to interleaving and A/B testing, is that counterfactual evaluation does not require online interventions.

In this chapter, we show that no existing interleaving method is truly unbiased: they are not guaranteed to correctly predict which ranker has the highest \ac{CTR}.
On two different industry datasets, we simulate a total of 1,000 comparisons between 2,000 different rankers.
In our setup, interleaving methods converge on the wrong answer for at least 2.2\% of the comparisons on both datasets.
A further analysis shows that existing interleaving methods are unable to reliably estimate \ac{CTR} differences of around 1\% or lower.
Therefore, in practice these systematic errors are expected to impact situations where rankers with a very similar \ac{CTR} are compared.

We propose a novel online evaluation algorithm: \acf{LogOpt}.
\ac{LogOpt} extends the existing unbiased counterfactual approach, and turns it into an online approach.
\ac{LogOpt} estimates which rankings should be shown to the user, so that the variance of its \ac{CTR} estimate is minimized.
In other words, it attempts to learn the logging-policy that leads to the fastest possible convergence of the counterfactual estimation.
Our experimental results indicate that our novel approach is as efficient as any interleaving method or A/B testing, without having a systematic error.
As predicted by the theory, we see that the estimates of our approach converge on the true \ac{CTR} difference between rankers.
Therefore, we have introduced the first online evaluation method that combines high efficiency with unbiased estimation.

The main contributions of this chapter are:
\begin{enumerate}[leftmargin=*]
\item The first logging-policy optimization method for minimizing the variance in counterfactual \ac{CTR} estimation.
\item The first unbiased online evaluation method that is as efficient  as state-of-the-art interleaving methods.
\item A large-scale analysis of existing online evaluation methods that reveals a previously unreported bias in interleaving methods.
\end{enumerate}

%% file: 09-onlinecountereval/sections/02-ctr-estimators.tex
\section{Preliminaries: Ranker Comparisons}

The overarching goal of ranker evaluation is to find the ranking model that provides the best rankings.
For the purposes of this chapter, we will define the quality of a ranker in terms of the number of clicks it is expected to receive.
Let $R$ indicate a ranking and let $\mathbbm{E}[\text{CTR}(R)] \in \mathbb{R}_{\geq0}$ be the expected number of clicks a ranking receives after being displayed to a user.
We consider ranking $R_1$ to be better than $R_2$ if in expectation it receives more clicks: $\mathbbm{E}[\text{CTR}(R_1)] > \mathbbm{E}[\text{CTR}(R_2)]$.
We will represent a ranking model by a policy $\pi$, with $\pi(R \mid q)$ as the probability that $\pi$ displays $R$ for a query $q$.
With $P(q)$ as the probability of a query $q$ being issued, the expected number of clicks received under a ranking model $\pi$ is:
\begin{equation}
\mathbbm{E}[\text{CTR}(\pi)]
=
\sum_{q} P(q) \sum_{R} \mathbbm{E}[\text{CTR}(R)] \pi(R \mid q).
\end{equation}
Our goal is to discover the $\mathbbm{E}[\text{CTR}]$ difference between two policies:
\begin{equation}
\Delta(\pi_1, \pi_2) 
= 
\mathbbm{E}[\text{CTR}(\pi_1)] - \mathbbm{E}[\text{CTR}(\pi_2)].
\label{eq:ctrdifference}
\end{equation}
We recognize that to correctly identify if one policy is better than another, we merely need a corresponding binary indicator:
\begin{equation}
\Delta_{bin}(\pi_1, \pi_2) 
= 
\text{sign}\big(
\Delta(\pi_1, \pi_2) 
\big).
\label{eq:binctrdifference}
\end{equation}
However, in practice the magnitude of the differences can be very important, for instance, if one policy is computationally much more expensive while only having a slightly higher $\mathbbm{E}[\text{CTR}]$, it may be preferable to use the other in production.
Therefore, estimating the absolute $\mathbbm{E}[\text{CTR}]$ difference is more desirable in practice.

\subsection{User behavior assumptions}
\label{sec:preliminary:assumptions}

Any proof regarding estimators using user interactions must rely on assumptions about user behavior.
In this chapter, we assume that only two forms of interaction bias are at play: position bias and item-selection bias.

Users generally do not examine all items that are displayed in a ranking but only click on examined items~\citep{chuklin-click-2015}.
As a result, a lower probability of examination for an item also makes it less likely to be clicked.
Position bias assumes that only the rank determines the probability of examination~\citep{craswell2008experimental}.
Furthermore, we will assume that given an examination only the relevance of an item determines the click probability.
Let $c(d) \in \{0,1\}$ indicate a click on item $d$ and $o(d) \in \{0,1\}$ examination by the user.
Then these assumptions result in the following assumed click probability:
\begin{equation}
\begin{split}
P(c(d) = 1 \mid R, q)
&= 
P(o(d) = 1 \mid R)
P(c(d) = 1 \mid o(d) = 1, q)
\\
&= 
\theta_{\text{rank}(d \mid R)}
\zeta_{d, q}.
\end{split}
\end{equation}
Here $\text{rank}(d \mid R)$ indicates the rank of $d$ in $R$; for brevity we use $\theta_{\text{rank}(d \mid R)}$ to denote the examination probability:
\begin{equation}
\theta_{\text{rank}(d \mid R)} = P(o(d) = 1 \mid R),
\end{equation}
and $\zeta_{d, q}$ for the conditional click probability:
\begin{equation}
\zeta_{d, q} = P(c(d) = 1 \mid o(d) = 1, q).
\end{equation}

We also assume that item-selection bias is present; this type of bias is an extreme form of position bias that results in zero examination probabilities for some items~\citep{ovaisi2020correcting, oosterhuis2020topkrankings}.
This bias is unavoidable in top-$k$ ranking settings, where only the $k \in \mathbb{N}_{> 0}$ highest ranked items are displayed.
Consequently, any item beyond rank $k$ cannot be observed or examined by the user:
$
\forall r \in \mathbb{N}_{> 0}\, (r > k \rightarrow \theta_{r} = 0)
$.
The distinction between item-selection bias and position bias is important because the original counterfactual evaluation method~\citep{joachims2017unbiased} is only able to correct for position bias when no item-selection bias is present~\citep{ovaisi2020correcting, oosterhuis2020topkrankings}. 

Based on these assumptions, we can now formulate the expected \ac{CTR} of a ranking:
\begin{equation}
\mathbbm{E}[\text{CTR}(R)] = \sum_{d \in R} P(c(d) = 1 \mid R, q) = \sum_{d \in R} \theta_{\text{rank}(d \mid R)}
\zeta_{d, q}.
\end{equation}
While we assume this model of user behavior, its parameters are still assumed unknown.
Therefore, the methods in this chapter will have to estimate $\mathbbm{E}[\text{CTR}]$ without prior knowledge of $\theta$ or $\zeta$.

\subsection{Goal: \acs{CTR}-estimator properties}
\label{sec:preliminary:properties}

Recall that our goal is to estimate the \ac{CTR} difference between rankers~(Eq.~\ref{eq:ctrdifference}); online evaluation methods do this based on user interactions.
Let $\mathcal{I}$ be the set of available user interactions, it contains $N$ tuples of a single (issued) query $q_i$, the corresponding displayed ranking $R_i$, and the observed user clicks $c_i$:
\begin{equation}
\mathcal{I} = \{(q_i, R_i, c_i)\}^N_{i=1}.
\end{equation}
Each evaluation method has a different effect on what rankings will be displayed to users.
Furthermore, each evaluation method converts each interaction into a single estimate using some function $f$:
\begin{equation}
x_i = f(q_i, R_i, c_i).
\end{equation}
The final estimate is simply the mean over these estimates:
\begin{equation}
\hat{\Delta}(\mathcal{I}) = \frac{1}{N} \sum_{i=1}^N x_i = \frac{1}{N} \sum_{i=1}^N f(q_i, R_i, c_i).
\end{equation}
This description fits all existing online and counterfactual evaluation methods for rankings.
Every evaluation method uses a different function $f$ to convert interactions into estimates; moreover, online evaluation methods also decide which rankings $R$ to display when collecting $\mathcal{I}$. 
These two choices result in different estimators.
Before we discuss the individual methods, we briefly introduce the three properties we desire of each estimator: consistency, unbiasedness and variance.

\begin{itemize}[leftmargin=*]
\item \textbf{Consistency} -- an estimator is \emph{consistent} if it converges as the number of issued queries $N$ increases.
All existing evaluation methods are consistent as their final estimates are means of bounded values. %
\item \textbf{Unbiasedness} -- an estimator is \emph{unbiased} if its estimate is equal to the true \ac{CTR} difference in expectation:
\begin{equation}
\text{Unbiased}(\hat{\Delta}) \Leftrightarrow  \mathbbm{E}\big[\hat{\Delta}(\mathcal{I})\big] = \Delta(\pi_1, \pi_2).
\end{equation}
If an estimator is both consistent and unbiased it is guaranteed to converge on the true $\mathbbm{E}[\text{CTR}]$ difference.
\item \textbf{Variance} -- the \emph{variance} of an estimator is the expected squared deviation between a single estimate $x$ and the mean $\hat{\Delta}(\mathcal{I})$:
\begin{equation}
\text{Var}(\hat{\Delta}) =  \mathbbm{E}\Big[ \big(x -  \mathbbm{E}[\hat{\Delta}(\mathcal{I})]\big)^2  \Big].
\end{equation}
Variance affects the rate of convergence of an estimator; for fast convergence it should be as low as possible.
\end{itemize}
In summary, our goal is to find an estimator, for the \ac{CTR} difference between two ranking models, that is consistent, unbiased and has minimal variance.

%% file: 09-onlinecountereval/sections/03-related-work.tex
\section{Existing Online and Counterfactual Evaluation Methods}
We describe three families of online and counterfactual evaluation methods for ranking.

\subsection{A/B testing}
A/B testing is a well established form of online evaluation to compare a system A with a system B~\citep{kohavi2017online}.
Users are randomly split into  two groups and during the same time period each group is exposed to only one of the systems.
In expectation, the only factor that differs between the groups is the exposure to the different systems.
Therefore, by comparing the behavior of each user group, the relative effect each system has can be evaluated.

We will briefly show that A/B testing is unbiased for $\mathbbm{E}[\text{CTR}]$ difference estimation.
For each interaction either $\pi_1$ or $\pi_2$ determines the ranking, let $A_i \in \{1, 2\}$ indicate the assignment and $A_i \sim P(A)$.
Thus, if $A_i = 1$, then $R_i \sim \pi_1(R \mid q)$ and if $A_i = 2$, then $R_i \sim \pi_2(R \mid q)$.
Each interaction $i$ is converted into a single estimate $x_i$ by $f_{\text{A/B}}$:
\begin{equation}
x_i =  f_{\text{A/B}}(q_i, R_i, c_i) = \left(\frac{\mathbbm{1}[A_i = 1]}{P(A = 1)} - \frac{\mathbbm{1}[A_i = 2]}{P(A = 2)}\right) \sum_{d \in R_i} c_i(d).
\end{equation}
We can prove that A/B testing is unbiased, since in expectation each individual estimate is equal to the \ac{CTR} difference:
\begin{align}
\mathbbm{E}[f_{\text{A/B}}(q_i, R_i, c_i)] &=
\sum_q P(q) \bigg(\frac{ P(A = 1) \sum_R \pi_1(R \mid q) \mathbbm{E}[\text{CTR}(R)] }{P(A = 1)} 
\nonumber \\
& \qquad \qquad \qquad - \frac{ P(A = 2) \sum_R \pi_2(R \mid q) \mathbbm{E}[\text{CTR}(R)] }{P(A = 2)} 
\bigg)
\nonumber \\
&= \sum_q P(q) \sum_R \mathbbm{E}[\text{CTR}(R)]\big(\pi_1(R \mid q) - \pi_2(R \mid q)\big)
\nonumber \\
&= \mathbbm{E}[\text{CTR}(\pi_1)] - \mathbbm{E}[\text{CTR}(\pi_2)] = \Delta(\pi_1, \pi_2).
\end{align}
Variance is harder to evaluate without knowledge of $\pi_1$ and $\pi_2$.
Unless $\Delta(\pi_1, \pi_2) = 0$, some variance is unavoidable since A/B testing alternates between estimating $\ac{CTR}(\pi_1)$ and $\ac{CTR}(\pi_2)$.

\subsection{Interleaving}
\label{sec:related:interleaving}
Interleaving methods were introduced specifically for evaluation in ranking, as a more efficient alternative to A/B testing~\citep{joachims2003evaluating}.
After a query is issued, interleaving methods take the rankings of two competing ranking systems and combine them into a single interleaved ranking.
Any clicks on the interleaved ranking can be interpreted as a preference signal between either ranking system.
Thus, unlike A/B testing, interleaving does not estimate the \ac{CTR} of individual systems but a relative preference; the idea is that this allows it to be more efficient than A/B testing.

Each interleaving method attempts to use randomization to counter position bias, without deviating too much from the original rankings so as to maintain the user experience~\citep{joachims2003evaluating}.
\emph{Team-draft interleaving} (TDI) randomly selects one ranker to place their top document first, then the other ranker places their top (unplaced) document next~\citep{radlinski2008does}.
Then it randomly decides the next two documents, and this process is repeated until all documents are placed in the interleaved ranking.
Clicks on the documents are attributed to the ranker that placed them.
The ranker with the most attributed clicks is inferred to be preferred by the user.
\emph{Probabilistic interleaving} (PI) treats each ranking as a probability distribution over documents; at each rank a distribution is randomly selected and a document is drawn from it~\citep{hofmann2011probabilistic}.
After clicks have been received, probabilistic interleaving computes the expected number of clicked documents per ranking system to infer preferences.
\emph{Optimized interleaving} (OI) casts the randomization as an optimization problem, and displays rankings so that if all documents are equally relevant no preferences are found~\citep{radlinski2013optimized}.

While every interleaving method attempts to deal with position bias, none is unbiased according to our definition (Section~\ref{sec:preliminary:properties}).
This may be confusing because previous work on interleaving makes claims of unbiasedness~\citep{hofmann2013fidelity, radlinski2013optimized, hofmann2011probabilistic}. 
However, they use different definitions of the term.
More precisely, TDI, PI, and OI provably converge on the correct outcome if all documents are equally relevant~\citep{hofmann2013fidelity, radlinski2013optimized, hofmann2011probabilistic, radlinski2008does}.
Moreover, if one assumes binary relevance and $\pi_1$ ranks all relevant documents equal to or higher than $\pi_2$, the binary outcome of PI and OI is proven to be correct in expectation~\citep{radlinski2013optimized, hofmann2013fidelity}.
However, beyond the confines of these unambiguous cases, we can prove that these methods do not meet our definition of unbiasedness: for every method one can construct an example where it converges on the incorrect outcome.
The rankers $\pi_1$, $\pi_2$ and position bias parameters $\theta$ can be chosen so that in expectation the wrong (binary) outcome is estimated; see Appendix~\ref{sec:appendix:bias} for a proof for each of the three interleaving methods. %
Thus, while more efficient than A/B testing, interleaving methods make systematic errors in certain circumstances and thus should not be considered to be unbiased w.r.t.\ \ac{CTR} differences.

We note that the magnitude of the bias should also be considered.
If the systematic error of an interleaving method is minuscule while the efficiency gains are very high, it may still be very useful in practice.
Our experimental results (Section~\ref{sec:biasinterleaving}) reveal that the systematic error of all three interleaving methods considered becomes very high when comparing systems with a \ac{CTR} difference of 1\% or smaller.

\subsection{Counterfactual evaluation}
\label{sec:related:counterfactual}

Counterfactual evaluation is based on the idea that if certain biases can be estimated well, they can also be adjusted
~\citep{joachims2017accurately, wang2016learning}.
While estimating relevance is considered the core difficulty of ranking evaluation, estimating the position bias terms $\theta$ is very doable.
By randomizing rankings, e.g., by swapping pairs of documents~\citep{joachims2017accurately} or exploiting data logged during A/B testing~\citep{agarwal2019estimating}, differences in \ac{CTR} for the same item on different positions can be observed directly.
Alternatively, using \ac{EM} optimization~\citep{wang2018position} or a dual learning objective~\citep{ai2018unbiased}, position bias can be estimated from logged data as well.
Once the bias terms $\theta$ have been estimated, logged clicks can be weighted so as to correct for the position bias during logging.
Hence, counterfactual evaluation can work with historically logged data.
Existing counterfactual evaluation algorithms do not dictate which rankings should be displayed during logging: they do not perform interventions and thus we do not consider them to be online methods.

Counterfactual evaluation assumes that the position bias $\theta$ and the logging policy $\pi_0$ are known, in order to correct for both position bias and item-selection bias.
Clicks are gathered with $\pi_0$ which decides which rankings are displayed to the user.
We follow \citet{oosterhuis2020topkrankings} (see Chapter~\ref{chapter:04-topk}) and use as propensity scores the probability of observance in expectation over the displayed rankings:
\begin{equation}
\rho(d \mid q) = \mathbbm{E}_R\big[P(o(d) = 1 \mid R) \mid \pi_0\big] = \sum_R \pi_0(R \mid q) P(o(d) = 1 \mid  R).
\label{eq:prop}
\end{equation}
Then we use $\lambda(d \mid \pi_1, \pi_2)$ to indicate the difference in observance probability under $\pi_1$ or $\pi_2$:
\begin{equation}
\begin{split}
\lambda(d \,|\, \pi_1, \pi_2)
&= \mathbbm{E}_R\big[P(o(d) = 1 \,|\, R) \,|\, \pi_1\big] - \mathbbm{E}_R\big[P(o(d) = 1 \,|\, R) \,|\, \pi_2\big]
\\
&= \sum_R \theta_{\text{rank}(d \mid R)}\big(\pi_1(R \mid q_i)- \pi_2(R \mid q_i)\big).
\label{eq:lambda}
\end{split}
\end{equation}
Then, the \ac{IPS} estimate function is formulated as:
\begin{equation}
x_i =  f_{\text{IPS}}(q_i, R_i, c_i) = \sum_{d :  \rho(d | q_i) > 0} \frac{c_i(d) }{\rho(d \,|\, q_i)}\lambda(d \mid \pi_1, \pi_2).
\label{eq:ipsestimator}
\end{equation}
Each click is weighted inversely to its examination probability, but items with a zero probability: $\rho(d \,|\, q_i) = 0$ are excluded.
We note that these items can never be clicked:
\begin{equation}
\forall q, d \, (\rho(d \,|\, q) = 0 \rightarrow c(d) = 0.
\end{equation}
Before we prove unbiasedness, we note that given $\rho(d \,|\, q_i) > 0$:
\begin{equation}
\begin{split}
\mathbbm{E}\bigg[\frac{c(d)}{\rho(d \,|\, q)}\bigg]
&=
\frac{\sum_R \pi_0(R \mid q) \theta_{\text{rank}(d \mid R)}
\zeta_{d, q}}{\rho(d \,|\, q_i)}
\\&=
\frac{\sum_R \pi_0(R \mid q) \theta_{\text{rank}(d \mid R)}
}{\sum_{R'} \pi_0(R' \,|\,  q) \theta_{\text{rank}(d \mid R')}} \zeta_{d, q}
= \zeta_{d, q}.
\end{split}
\end{equation}
This, in turn, can be used to prove unbiasedness:
\begin{equation}
\begin{split}
\mathbbm{E}[f_{\text{IPS}}(q_i, R_i, c_i)] &= \sum_q P(q) \sum_{d:\rho(d \,|\, q_i) > 0}  \zeta_{d, q}\lambda(d \mid \pi_1, \pi_2)
\\
&= \mathbbm{E}[\text{CTR}(\pi_1)] - \mathbbm{E}[\text{CTR}(\pi_2)] = \Delta(\pi_1, \pi_2).
\end{split}
\end{equation}
This proof is only valid under the following requirement:
\begin{equation}
\forall d, q\,  (\zeta_{d, q}\lambda(d \mid \pi_1, \pi_2) > 0 \rightarrow \rho(d \,|\, q) > 0).
\label{eq:requirement}
\end{equation}
In practice, this means that the items in the top-$k$ of either $\pi_1$ or $\pi_2$ need to have a non-zero examination probability under $\pi_0$, i.e., they must have a chance to appear in the top-$k$ under $\pi_0$. 

Besides Requirement~\ref{eq:requirement} the \ac{IPS} counterfactual evaluation method~\citep{joachims2017accurately, wang2016learning} is completely indifferent to $\pi_0$ and hence we do not consider it to be an online method.
In the next section, we will introduce an algorithm for choosing and updating $\pi_0$ during logging to minimize the variance of the estimator.
By doing so we turn counterfactual evaluation into an online method.

%% file: 09-onlinecountereval/sections/04-method.tex
\section{Logging Policy Optimization for Variance Minimization}
Next, we introduce a method aimed at finding a logging policy minimizes the variance of the estimates of the counterfactual estimator.

\subsection{Minimizing variance}

In Section~\ref{sec:related:counterfactual}, we have discussed counterfactual evaluation and established that it is unbiased as long as $\theta$ is known and the logging policy meets Requirement~\ref{eq:requirement}.
The variance of $\Delta_\mathit{IPS}$ depends on the position bias $\theta$, the conditional click probabilities $\zeta$, and the logging policy $\pi_0$.
In contrast to the user-dependent $\theta$ and $\zeta$, the way data is logged by $\pi_0$ is something one can have control over.
The goal of our method is to find the optimal policy that minimizes variance while still meeting Requirement~\ref{eq:requirement}:
\begin{equation}
\pi_0^* = \underset{\pi_0 : \, \pi_0 \text{ meets Req.~\ref{eq:requirement}}}{\argmin} \, \text{Var}\big(\hat{\Delta}_\mathit{IPS}^{\pi_0}\big),
\end{equation}
where $\hat{\Delta}_\mathit{IPS}^{\pi_0}$ is the counterfactual estimator based on data logged using $\pi_0$.

To formulate the variance, we first note that it is an expectation over queries:
\begin{equation}
\text{Var}(\hat{\Delta}) = \sum_q P(q) \text{Var}(\hat{\Delta} \mid q).
\end{equation}
To keep notation short, for the remainder of this section we will write: $\Delta = \Delta(\pi_1, \pi_2)$; $\theta_{d,R} = \theta_{\text{rank}(d \mid R)}$; $\zeta_{d} = \zeta_{d,q}$; $\lambda_d = \lambda(d \mid \pi_1, \pi_2)$; and $\rho_d = \rho(d \mid q, \pi_0)$. 
Next, we consider the probability of a click pattern $c$, this is simply a vector indicating a possible combination of clicked documents $c(d)=1$ and not-clicked documents $c(d)=0$:
\begin{equation}
\begin{split}
P(c \mid q) &= \sum_{R} \pi_0(R \mid q) \prod_{d : c(d) = 1} \theta_{d,R} \zeta_d \prod_{d : c(d) = 0} (1 - \theta_{d,R} \zeta_d)
\\
&= \sum_{R} \pi_0(R \mid q) P(c \mid R).
\end{split}
\label{eq:clickprob}
\end{equation}
Here, $\pi_0$ has some control over this probability: by deciding the distribution of displayed rankings it can make certain click patterns more or less frequent.
The variance added per query is the squared error of every possible click pattern weighted by the probability of each pattern. 
Let $\sum_c$ sum over every possible click pattern:
\begin{equation}
\text{Var}(\hat{\Delta}_\mathit{IPS}^{\pi_0} \mid q) = \sum_c P(c\mid q) \bigg(\Delta - \sum_{d : c(d) = 1}\frac{\lambda_d}{\rho_d}\bigg)^2.
\label{eq:ipsvariance}
\end{equation}
It is unknown whether there is a closed-form solution for $\pi_0^*$. 
However, the variance function is differentiable.
Taking the derivative reveals a trade-off between two potentially conflicting goals:
\begin{equation}
\begin{split}
\frac{\delta}{\delta \pi_0} \text{Var}(\hat{\Delta}_{IPS}^{\pi_0} \mid q) = \sum_c &
\overbrace{
\left[\frac{\delta}{\delta \pi_0}P(c \mid q)\right] \bigg(\Delta - \sum_{d : c(d) = 1}\frac{\lambda_d}{\rho_d}\bigg)^2
}^{\text{\footnotesize minimize frequency of high-error click patterns}}
\\  + &
\underbrace{
P(c \mid q) \left[\frac{\delta}{\delta \pi_0} \bigg(\Delta - \sum_{d : c(d) = 1}\frac{\lambda_d}{\rho_d}\bigg)^2\right]
}_{\text{\footnotesize minimize error of frequent click patterns}}.
\end{split}
\label{eq:gradient}
\end{equation}
On the one hand, the derivative reduces the frequency of click patterns that result in high error samples, i.e., by updating $\pi_0$ so that these are less likely to occur.
On the other hand, changing $\pi_0$ also affects the propensities $\rho_d$, i.e., if $\pi_0$ makes an item $d$ less likely to be examined, its corresponding value $\lambda_d/\rho_d$ becomes larger, which can lead to a higher error for related click patterns.
The optimal policy has to balance:
\begin{enumerate*}[label=(\roman*)]
\item avoiding showing rankings that lead to high-error click patterns; and
\item avoiding minimizing propensity scores, which increases the errors of corresponding click patterns.
\end{enumerate*}

Our method applies stochastic gradient descent to optimize the logging policy w.r.t.\ the variance.
There are two main difficulties with this approach:
\begin{enumerate*}[label=(\roman*)]
\item the parameters $\theta$ and $\zeta$ are unknown a priori; and
\item the gradients include summations over all possible rankings and all possible click patterns, both of which are computationally infeasible.
\end{enumerate*}
In the following sections, we will detail how \ac{LogOpt} solves both of these problems.

\subsection{Bias \& relevance estimation}

In order to compute the gradient in Eq.~\ref{eq:gradient}, the parameters $\theta$ and $\zeta$ have to be known.
\ac{LogOpt} is based on the assumption that accurate estimates of $\theta$ and $\zeta$ suffice to find a near-optimal logging policy.
We note that the counterfactual estimator only requires $\theta$ to be known for unbiasedness~(see Section~\ref{sec:related:counterfactual}).
Our approach is as follows. At given intervals during evaluation we use the available clicks to estimate $\theta$ and $\zeta$.
Then we use the estimated $\hat{\theta}$ to get the current estimate $\hat{\Delta}_\mathit{IPS}(\mathcal{I}, \hat{\theta})$ (Eq.~\ref{eq:ipsestimator}) and optimize w.r.t.\ the estimated variance (Eq.~\ref{eq:ipsvariance}) based on $\hat{\theta}$, $\hat{\zeta}$, and $\hat{\Delta}_\mathit{IPS}(\mathcal{I}, \hat{\theta})$.

For estimating $\theta$ and $\zeta$ we use the existing \ac{EM} approach by \citet{wang2018position}, because it works well in situations where few interactions are available and does not require randomization.
We note that previous work has found randomization-based approaches to be more accurate for estimating $\theta$~\citep{wang2018position, agarwal2019estimating, fang2019intervention}.
However, they require multiple interactions per query and specific types of randomization in their results; by choosing the \ac{EM} approach we avoid having these requirements.

\subsection{Monte-Carlo-based derivatives}
\label{sec:method:derivates}
Both the variance (Eq.~\ref{eq:ipsvariance}) and its gradient (Eq.~\ref{eq:gradient}), include a sum over all possible click patterns.
Moreover, they also include the probability of a specific pattern $P(c \mid q)$ that is based on a sum over all possible rankings (Eq.~\ref{eq:clickprob}).
Clearly, these equations are infeasible to compute under any realistic time constraints.
To solve this issue, we introduce gradient estimation based on Monte-Carlo sampling.
Our approach is similar to that of \citet{ma2020off}, however, we are estimating gradients of variance instead of general performance.

First, we assume that policies place the documents in order of rank and the probability of placing an individual document at rank $x$ only depends on the previously placed documents.
Let $R_{1:x-1}$ indicate the (incomplete) ranking from rank $1$ up to rank $x$, then $\pi_0(d \mid R_{1:x-1}, q)$ indicates the probability that document $d$ is placed at rank $x$ given that the ranking up to $x$ is $R_{1:x-1}$.
The probability of a ranking $R$ up to rank $k$ is thus:
\begin{equation}
\pi_0(R_{1:k} \mid q) = \prod_{x=1}^{k} \pi_0(R_x \mid R_{1:x-1}, q).
\end{equation}
Let $K$ be the length of a complete ranking $R$, the gradient of the probability of a ranking w.r.t.\ a policy is:
\begin{equation}
\frac{\delta \pi_0(R \mid q) }{\delta \pi_0}
 =
\sum_{x=1}^K
  \frac{\pi_0(R \mid q)}{\pi_0(R_x \mid R_{1:x},  q)}
  \left[\frac{\delta \pi_0(R_x \mid R_{1:x-1},  q) }{\delta \pi_0} \right].
\end{equation}
The gradient of the propensity w.r.t.\ the policy (cf.\ Eq.~\ref{eq:prop}) is:
\begin{equation}
\begin{split}
\frac{\delta \rho(d |\, q)}{\delta \pi_0} &=  \sum_{k=1}^K \theta_k \sum_R \pi_0(R_{1:k-1} |\,  q)   
\Bigg(
\left[ \frac{\delta \pi_0(d |\, R_{1:k-1}, q)}{\delta \pi_0} \right]
\\
& \qquad\quad   + 
\sum_{x=1}^{k-1} \frac{\pi_0(d |\, R_{1:k-1}, q)}{\pi_0(R_x |\, R_{1:x-1} , q)}\left[ \frac{\delta\pi_0(R_x |\,  R_{1:x-1} , q)}{\delta \pi_0} \right]
\Bigg).
\end{split}
\end{equation}
To avoid iterating over all rankings in the $\sum_R$ sum, we sample $M$ rankings:
$R^m \sim \pi_0(R \mid q)$, and a click pattern on each ranking: $c^m \sim P(c \mid R^m)$.
This enables us to make the following approximation:
\begin{equation}
\begin{split}
\widehat{\rho\text{-grad}}(d)
&=
 \frac{1}{M} \sum_{m=1}^{M} \sum_{k=1}^K \theta_k 
\Bigg(
\left[ \frac{\delta \pi_0(d |\, R^m_{1:k-1}, q)}{\delta \pi_0} \right] 
\\
& \quad \quad \, + 
\sum_{x=1}^{k-1} \frac{\pi_0(d |\, R^m_{1:k-1}, q)}{\pi_0(R^m_x  |\, R^m_{1:x-1} , q)}\left[ \frac{\delta\pi_0(R^m_x  |\, R^m_{1:x-1} , q)}{\delta \pi_0} \right]
\Bigg),
\end{split}
\end{equation}
since $\frac{\delta \rho(d |\, q)}{\delta \pi_0} \approx \widehat{\rho\text{-grad}}(d, q)$.
In turn, we can use this to approximate the second part of Eq.~\ref{eq:gradient}:
\begin{equation}
\widehat{\text{error-grad}}(c) =
2\bigg(\Delta - \sum_{d : c(d) = 1}\frac{\lambda_d}{\rho_d}\bigg) \sum_{d : c(d) = 1}\frac{\lambda_d}{\rho_d^2} \widehat{\rho\text{-grad}}(d).
\end{equation}
We approximate the first part of Eq.~\ref{eq:gradient} with:
\begin{equation}
\begin{split}
\widehat{\text{freq-grad}}&(R, c) = 
\\
& \bigg(\Delta - \sum_{d : c(d) = 1}\frac{\lambda_d}{\rho_d}\bigg)^2
 \sum_{x=1}^K
\frac{1}{\pi_0(R_x \mid R_{1:x-1},  q)}
\left[\frac{\delta \pi_0(R_x \mid R_{1:x-1},  q) }{\delta \pi_0} \right].
\end{split}
\end{equation}
Together, they approximate the complete gradient (cf.\ Eq.~\ref{eq:gradient}):
\begin{equation}
\frac{\delta \text{Var}(\hat{\Delta}_{IPS}^{\pi_0} \mid q)}{\delta \pi_0} 
\approx 
\frac{1}{M}\sum_{m=1}^M \widehat{\text{freq-grad}}(R^m, c^m) + \widehat{\text{error-grad}}(c^m).
\label{eq:approxgrad}
\end{equation}
Therefore, we can approximate the gradient of the variance w.r.t.\ a logging policy $\pi_0$, based on rankings sampled from $\pi_0$ and our current estimated click model $\hat{\theta}$, $\hat{\zeta}$, while staying computationally feasible.\footnote{For a more detailed description see Appendix~\ref{sec:appendix:approx}.}

\subsection{Summary}

We have summarized the \ac{LogOpt} method in Algorithm~\ref{alg:logopt}.
The algorithm requires a set of historical interactions $\mathcal{I}$ and two rankers $\pi_1$ and $\pi_2$ to compare.
Then by fitting a click model on $\mathcal{I}$ using an \ac{EM}-procedure (Line~\ref{line:clickmodel}) an estimate of observation bias $\hat{\theta}$ and document relevance $\hat{\zeta}$ is obtained.
Using $\hat{\theta}$, an estimate of the difference in observation probabilities $\hat{\lambda}$ is computed (Line~\ref{line:lambda} and cf.\ Eq~\ref{eq:lambda}), and an estimate of the CTR difference $\hat{\Delta}(\pi_1, \pi_2)$ (Line~\ref{line:delta} and cf.\ Eq~\ref{eq:ipsestimator}).
Then the optimization of a new logging policy $\pi_0$ begins:
A query is sampled from $\mathcal{I}$ (Line~\ref{line:query}), and for that query $M$ rankings are sampled from the current $\pi_0$ (Line~\ref{line:rankings}), then for each ranking a click pattern is sampled using $\hat{\theta}$ and $\hat{\zeta}$ (Line~\ref{line:clicks}).
Finally, using the sampled rankings and clicks, $\hat{\theta}$, $\hat{\lambda}$, and $\hat{\Delta}(\pi_1, \pi_2)$, the gradient is now approximated using Eq.~\ref{eq:approxgrad} (Line~\ref{line:grad}) and the policy $\pi_0$ is updated accordingly (Line~\ref{line:update}).
This process can be repeated for a fixed number of steps, or until the policy has converged.

This concludes our introduction of \ac{LogOpt}: the first method that optimizes the logging policy for faster convergence in counterfactual evaluation.
We argue that \ac{LogOpt} turns counterfactual evaluation into online evaluation, because it instructs which rankings should be displayed for the most efficient evaluation.
The ability to make interventions like this is the defining characteristic of an online evaluation method.

\begin{algorithm}[t]
\caption{\acf{LogOpt}} 
\label{alg:logopt}
\begin{algorithmic}[1]
\STATE \textbf{Input}: Historical interactions: $\mathcal{I}$; rankers to compare $\pi_1, \pi_2$.
\STATE $\hat{\theta}, \hat{\zeta} \leftarrow \text{infer\_click\_model}(\mathcal{I})$ \hfill \textit{\small // estimate bias using EM} \label{line:clickmodel}
\STATE $\hat{\lambda} \leftarrow \text{estimated\_observance}(\hat{\theta}, \pi_1, \pi_2)$ \hfill \textit{\small // estimate $\lambda$ cf.\ Eq~\ref{eq:lambda}} \label{line:lambda}
\STATE $\hat{\Delta}(\pi_1, \pi_2) \leftarrow \text{estimated\_CTR}(\mathcal{I}, \hat{\lambda}, \hat{\theta})$ \hfill \textit{\small // CTR diff. cf.\ Eq~\ref{eq:ipsestimator}} \label{line:delta}
\STATE $\pi_0 \leftarrow \text{init\_policy}()$ \hfill \textit{\small // initialize logging policy} 
\FOR{$j \in \{1,2,\ldots\}$}
\STATE $q \sim P(q\mid \mathcal{I})$ \hfill \textit{\small // sample a query from interactions} \label{line:query}
\STATE $\mathcal{R} \leftarrow \{R^1, R^2, \ldots, R^M\} \sim \pi_0(R \mid q)$ \hfill \textit{\small // sample $M$ rankings} \label{line:rankings}
\STATE $\mathcal{C} \leftarrow \{c^1, c^2, \ldots, c^M\} \sim P(c \mid \mathcal{R})$ \hfill \textit{\small // sample $M$ click patterns}  \label{line:clicks}
\STATE $\hat{\delta} \leftarrow \text{approx\_grad}(\mathcal{R}, \mathcal{C}, \hat{\lambda}, \hat{\theta}, \hat{\Delta}(\pi_1, \pi_2))$  \label{line:grad}
\hfill \textit{\small // using Eq.~\ref{eq:approxgrad}}
\STATE $\pi_0 \leftarrow \text{update}(\pi_0, \hat{\delta})$  \label{line:update}
\hfill \textit{\small // update using approx. gradient}
\ENDFOR
\RETURN $\pi_0$ 
\end{algorithmic}
\end{algorithm}

%% file: 09-onlinecountereval/sections/05-experimental-setup.tex
\section{Experimental Setup}

We ran semi-synthetic experiments that are prevalent in online and counterfactual evaluation~\citep{oosterhuis2020topkrankings, hofmann2011probabilistic, joachims2017unbiased}.
User-issued queries are simulated by sampling from learning to rank datasets; each dataset contains a preselected set of documents per query.
We use the Yahoo!\ Webscope~\citep{Chapelle2011} and MSLR-WEB30k~\citep{qin2013introducing} datasets; they both contain 5-grade relevance judgements for all preselected query-document pairs.
For each sampled query, we let the evaluation method decide which ranking to display and then simulate clicks on them using probabilistic click models.

To simulate position bias, we use the rank-based probabilities of \citet{joachims2017unbiased}:
\begin{equation}
P(o(d) = 1 \mid R, q) = \frac{1}{\text{rank}(d \mid R)}.
\end{equation}
If observed, the click probability is determined by the relevance label of the dataset (ranging from 0 to 4).
More relevant items are more likely to be clicked, yet non-relevant documents still have a non-zero click probability:
\begin{equation}
P(c(d) = 1 \mid o(d) = 1, q) =  0.225 \cdot \text{relevance\_label}(q,d) + 0.1.
\end{equation}
Spread over both datasets, we generated {2,000} rankers and created {1,000} ranker-pairs.
We aimed to generate rankers that are likely to be compared in real-world scenarios; unfortunately, no simple distribution of such rankers is available.
Therefore, we tried to generate rankers that have (at least) a decent \ac{CTR} and that span a variety of ranking behaviors.
Each ranker was optimized using LambdaLoss~\citep{wang2018lambdaloss} based on the labelled data of 100 sampled queries; each ranker is based on a linear model that only uses a random sample of 50\% of the dataset features.
Figure~\ref{fig:ctrdistribution} displays the resulting \ac{CTR} distribution; it appears to follow a normal distribution, on both datasets.

For each ranker-pair and method, we sample $3 \cdot 10^6$ queries and calculate their \ac{CTR} estimates for different numbers of queries.
We considered three metrics:
\begin{enumerate*}[label=(\roman*)]
\item The binary error: whether the estimate correctly predicts which ranker should be preferred.
\item The absolute error: the absolute difference between the estimate and the true $\mathbbm{E}[\text{CTR}]$ difference:
\end{enumerate*}
\begin{equation}
\text{absolute-error} = |\Delta(\pi_1, \pi_2) - \hat{\Delta}(\mathcal{I})|.
\end{equation}
And
\begin{enumerate*}[label=(\roman*),resume]
\item the mean squared error: the squared error \emph{per sample} (not the final estimate); if the estimator is unbiased this is equivalent to the empirical variance:
\end{enumerate*}
\begin{equation}
\text{mean-squared-error} = \frac{1}{N}\sum_{i=1}^N (\Delta(\pi_1, \pi_2) - x_i)^2.
\end{equation}
We compare \ac{LogOpt} with the following baselines:
\begin{enumerate*}[label=(\roman*)]
\item A/B testing (with equal probabilities for each ranker),
\item Team-Draft Interleaving,
\item Probabilistic Interleaving (with $\tau = 4$), and
\item Optimized Interleaving (with the inverse rank scoring function).
\end{enumerate*}
Furthermore, we compare \ac{LogOpt} with other choices of logging policies:
\begin{enumerate*}[label=(\roman*)]
\item uniform sampling,
\item A/B testing: showing either the ranking of A or B with equal probability, and
\item an Oracle logging policy: applying \ac{LogOpt} to the true relevances $\zeta$ and position bias $\theta$.
\end{enumerate*}
We also consider \ac{LogOpt} both in the case where $\theta$ is known \emph{a priori}, or where it has to be estimated still.
Because estimating $\theta$ and optimizing the logging policy $\pi_0$ is time-consuming, we only update $\hat{\theta}$ and $\pi_0$ after $10^3$, $10^4$, $10^5$ and $10^6$ queries.
The policy \ac{LogOpt} optimizes uses a neural network with 2 hidden layers consisting of 32 units each.
The network computes a score for every document, then a softmax is applied to the scores to create a distribution over documents.

\begin{figure}[t]
\centering
\begin{tabular}{l l}
 \multicolumn{1}{c}{ \footnotesize \hspace{1.8em} Yahoo Webscope}
&
 \multicolumn{1}{c}{ \footnotesize \hspace{1.8em} MSLR Web30k}
\\
\includegraphics[scale=0.42]{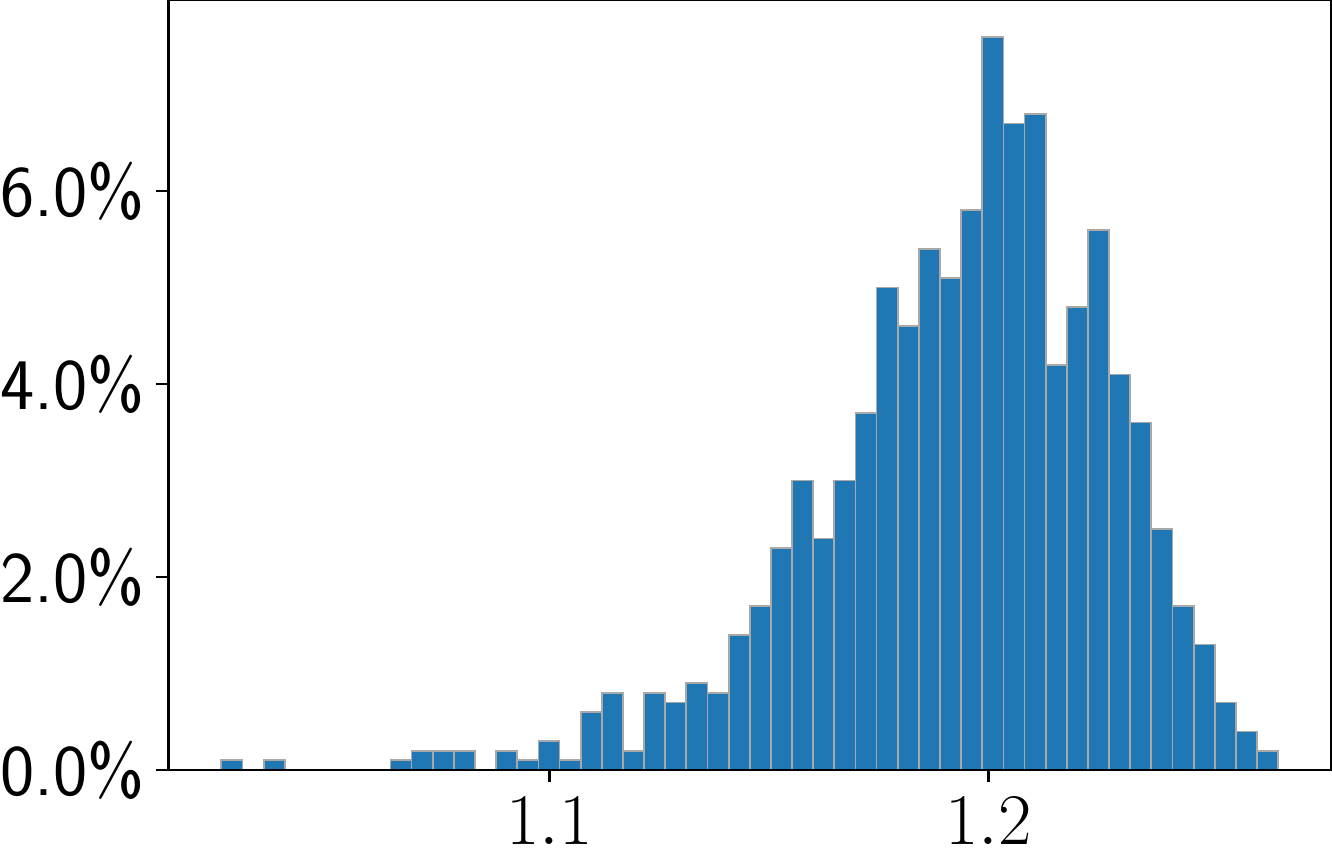} &
\includegraphics[scale=0.42]{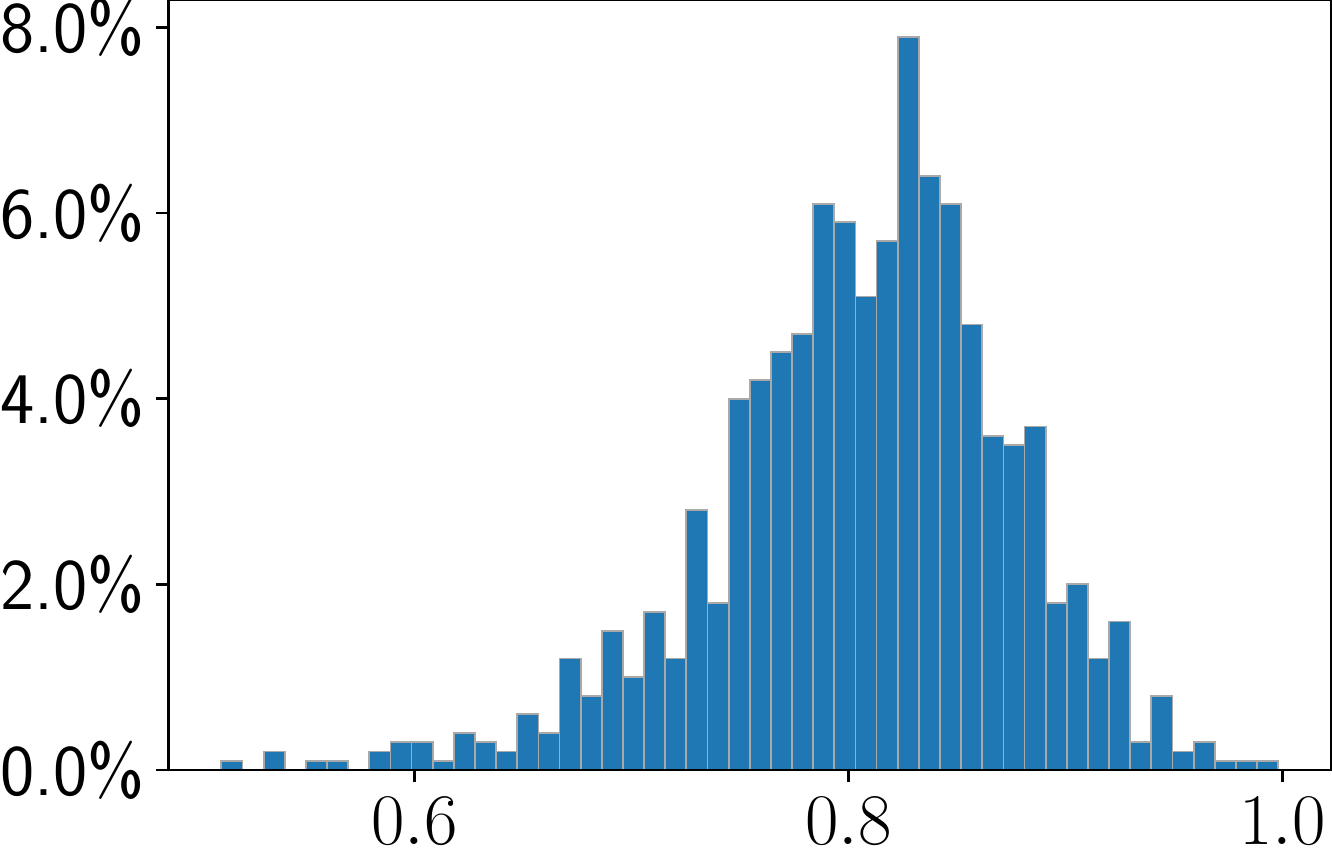} 
\end{tabular}
\caption{
The \ac{CTR} distribution of the 2000 generated rankers, 1000 were generated per dataset.
}
\label{fig:ctrdistribution}
\end{figure}

%% file: 09-onlinecountereval/sections/06-results.tex

\begin{figure*}[t]
\centering
\begin{tabular}{c r r}
&
 \multicolumn{1}{c}{\hspace{1.1em}\footnotesize Yahoo!\ Webscope}
&
 \multicolumn{1}{c}{\hspace{1.9em}\footnotesize MSLR-Web30k}
 \\
\rotatebox[origin=lt]{90}{\hspace{1.6em} \footnotesize Binary Error} &
\includegraphics[scale=0.4]{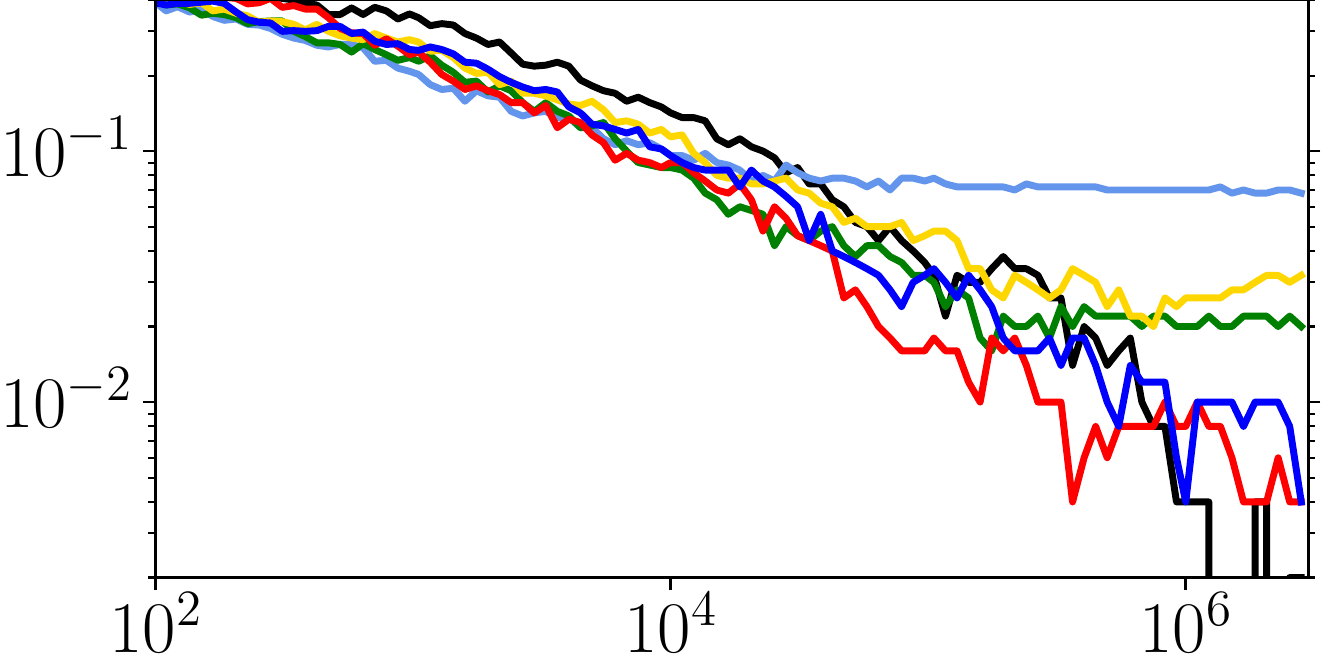} &
\includegraphics[scale=0.4]{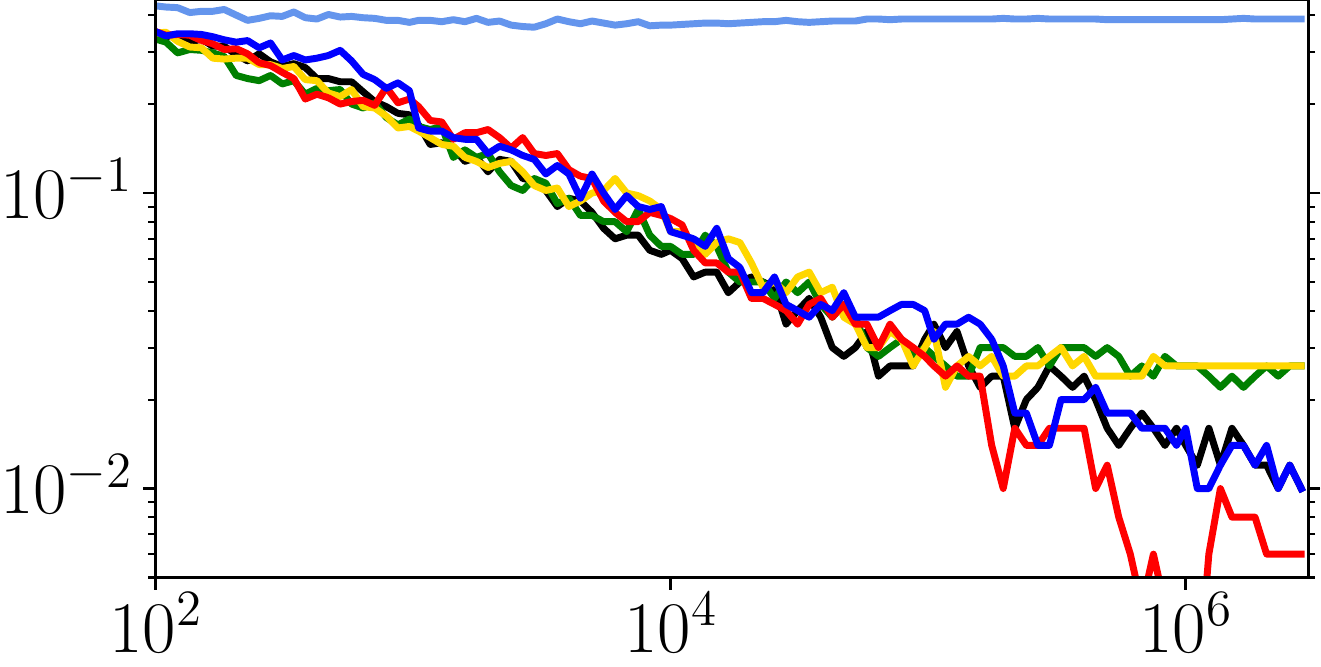}
\\
\rotatebox[origin=lt]{90}{\hspace{1.5em} \footnotesize Absolute Error} &
\includegraphics[scale=0.4]{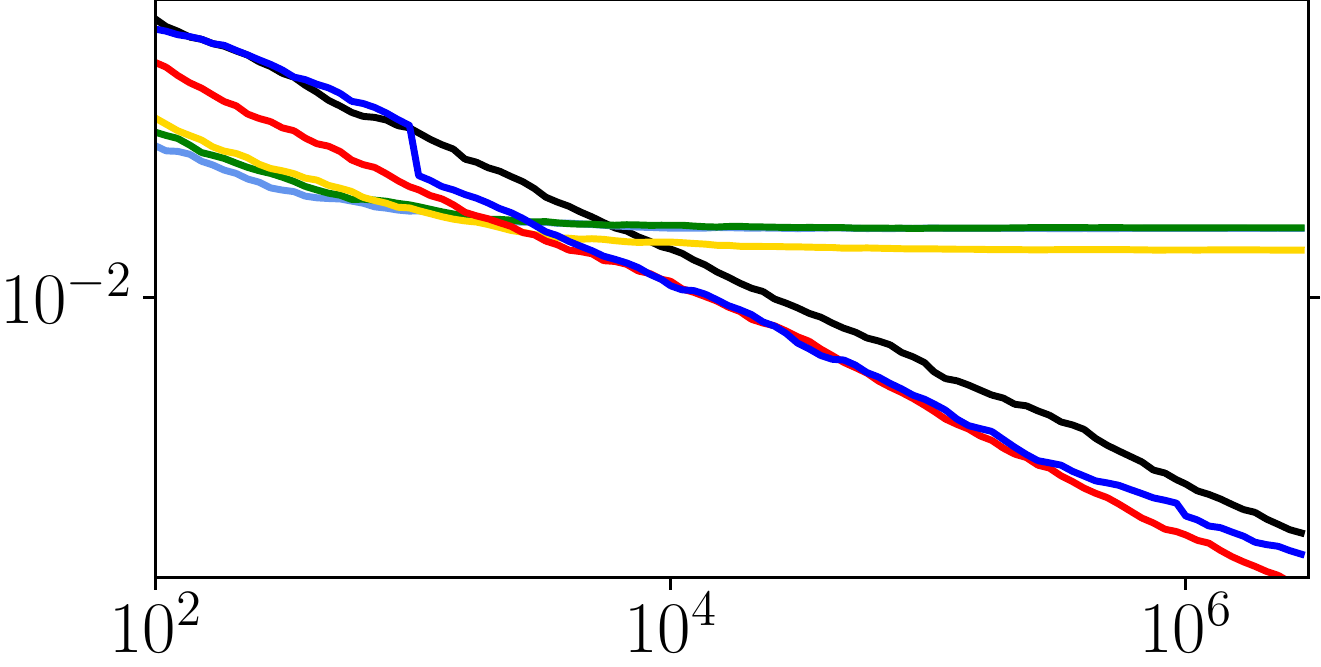} &
\includegraphics[scale=0.4]{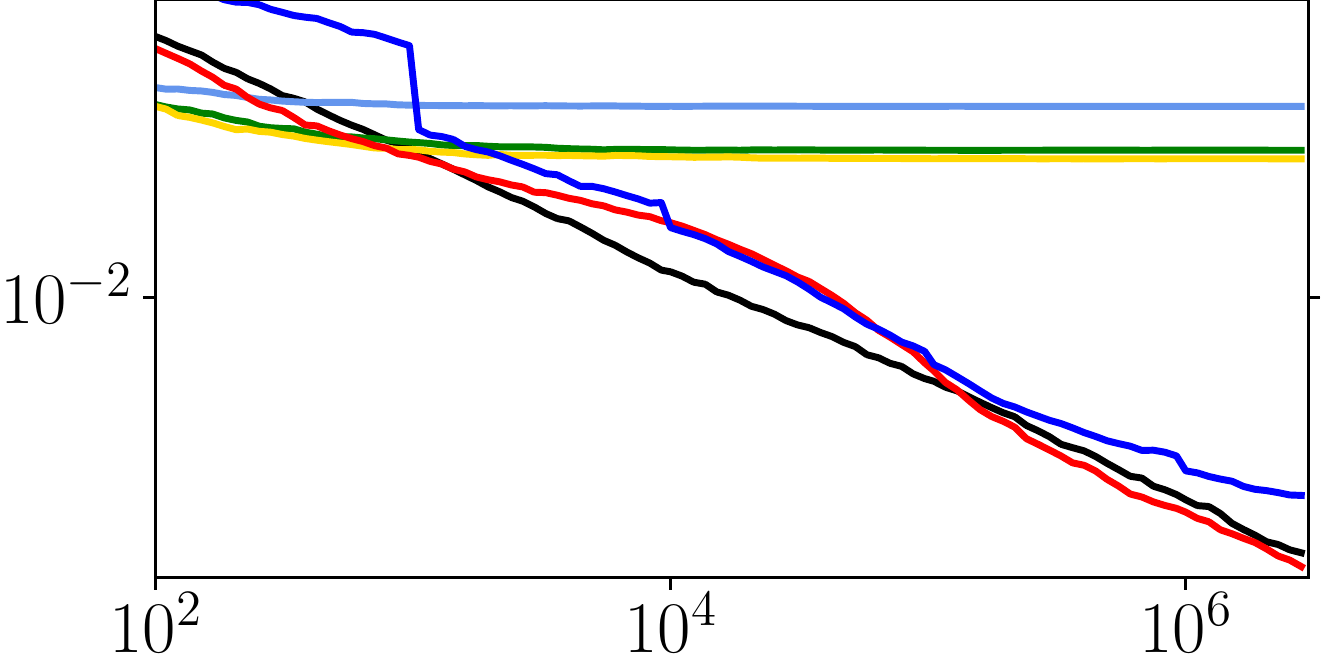} 
\\
\rotatebox[origin=lt]{90}{\hspace{0.1em} \footnotesize Mean Squared Error} &
\includegraphics[scale=0.4]{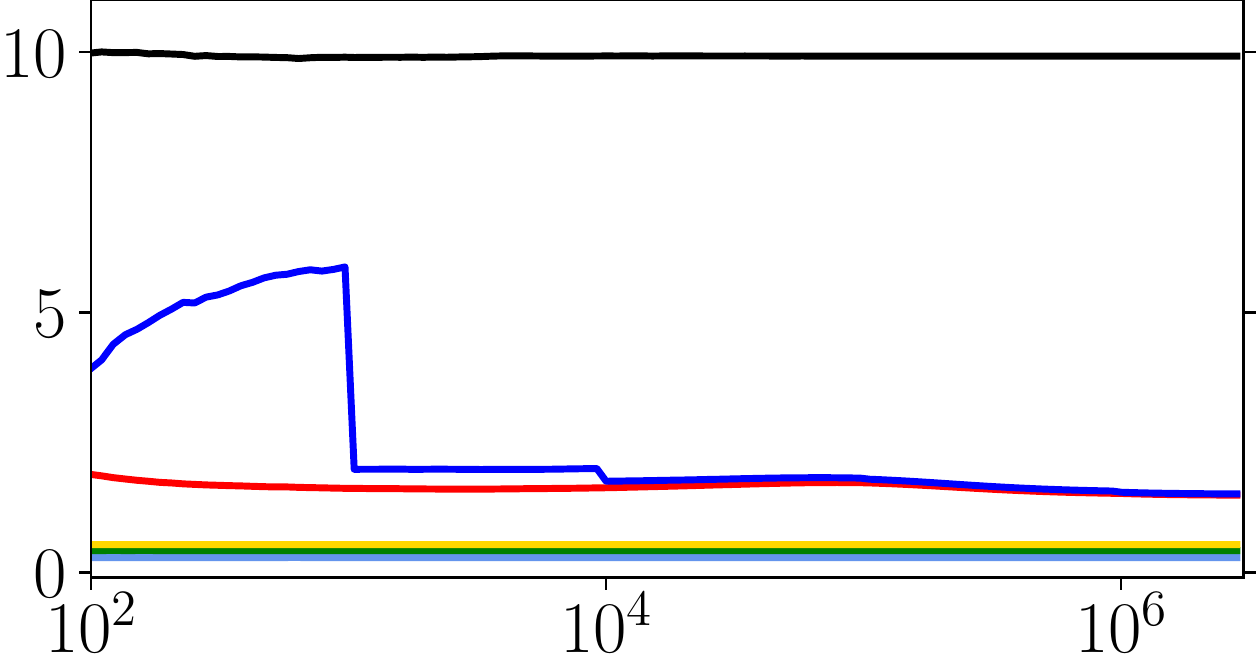} &
\includegraphics[scale=0.4]{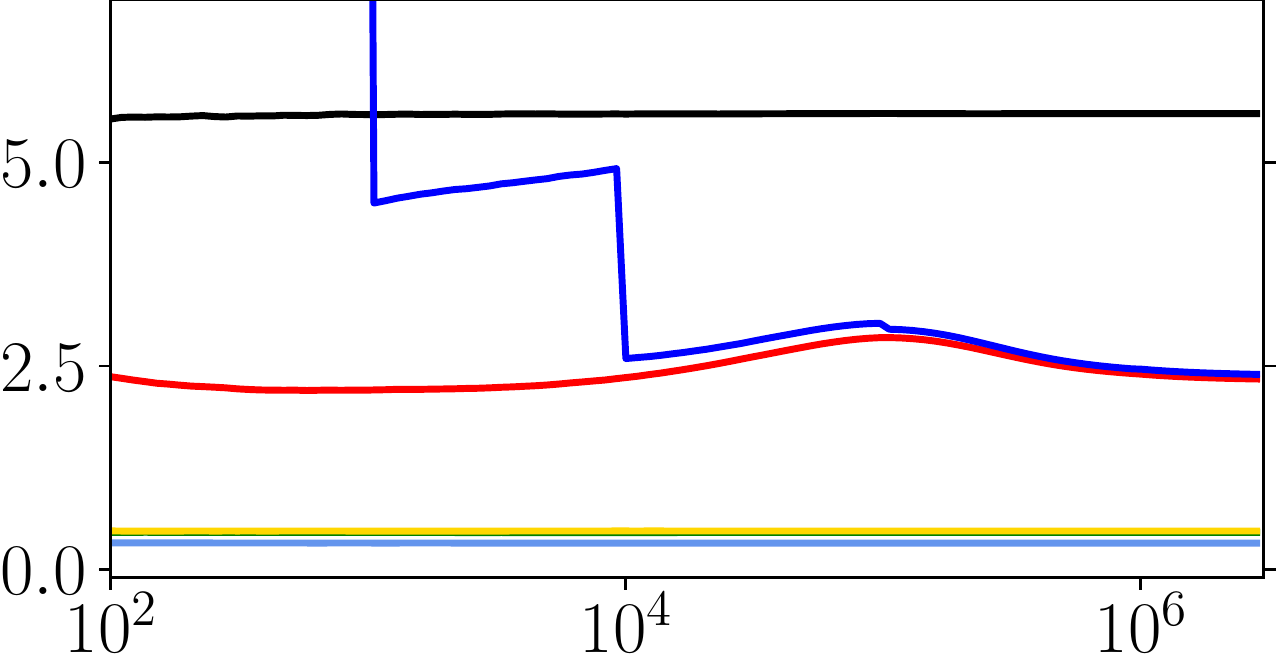}
\\
& \multicolumn{1}{c}{\hspace{1.4em} \footnotesize Number of Queries Issued}
& \multicolumn{1}{c}{\hspace{1.4em} \footnotesize Number of Queries Issued}
\\
 \multicolumn{3}{c}{
 \includegraphics[scale=.42]{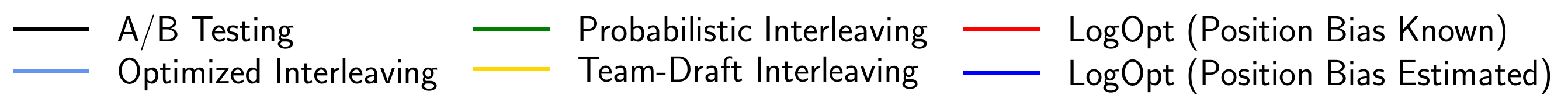}
} 
\end{tabular}
\caption{
Comparison of \ac{LogOpt} with other online methods; displayed results are an average over 500 comparisons.
}
\label{fig:baselines}
\end{figure*}

\begin{figure*}[t]
\centering
\begin{tabular}{c r r }
&
 \multicolumn{1}{c}{\hspace{1.1em}\footnotesize Yahoo!\ Webscope}
&
 \multicolumn{1}{c}{\hspace{1.9em}\footnotesize MSLR-Web30k}
 \\
\rotatebox[origin=lt]{90}{\hspace{1.6em} \footnotesize Binary Error} &
\includegraphics[scale=0.4]{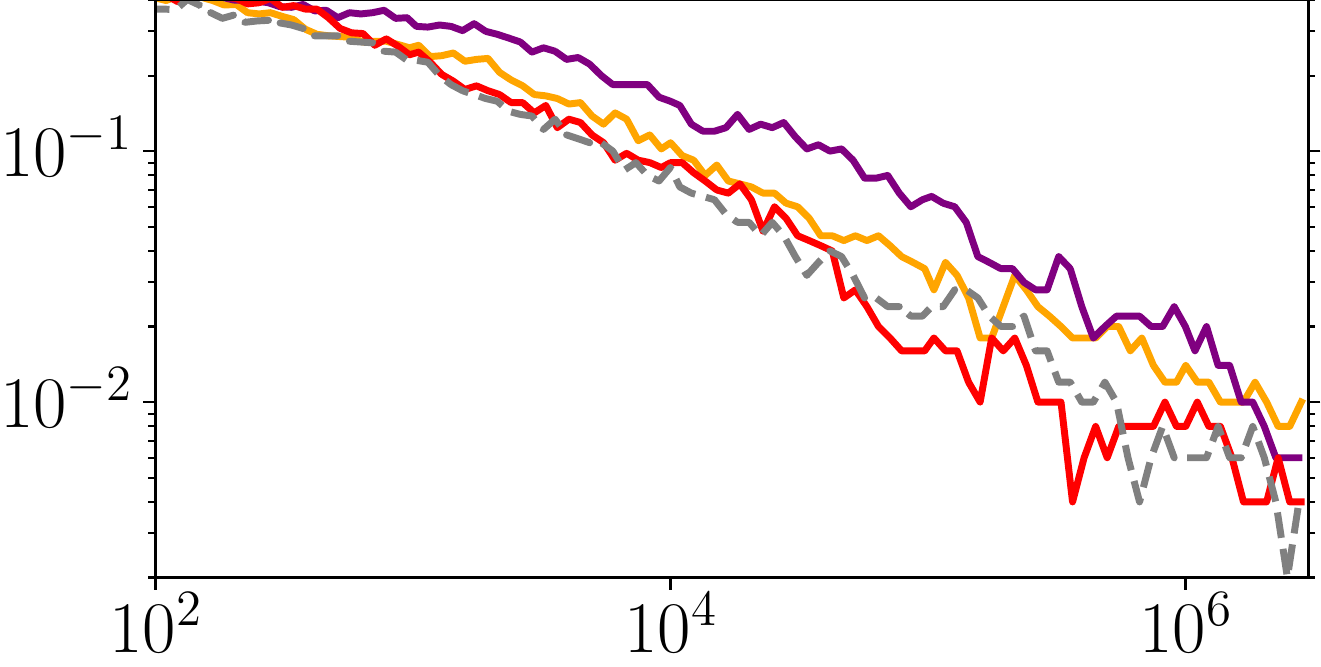} &
\includegraphics[scale=0.4]{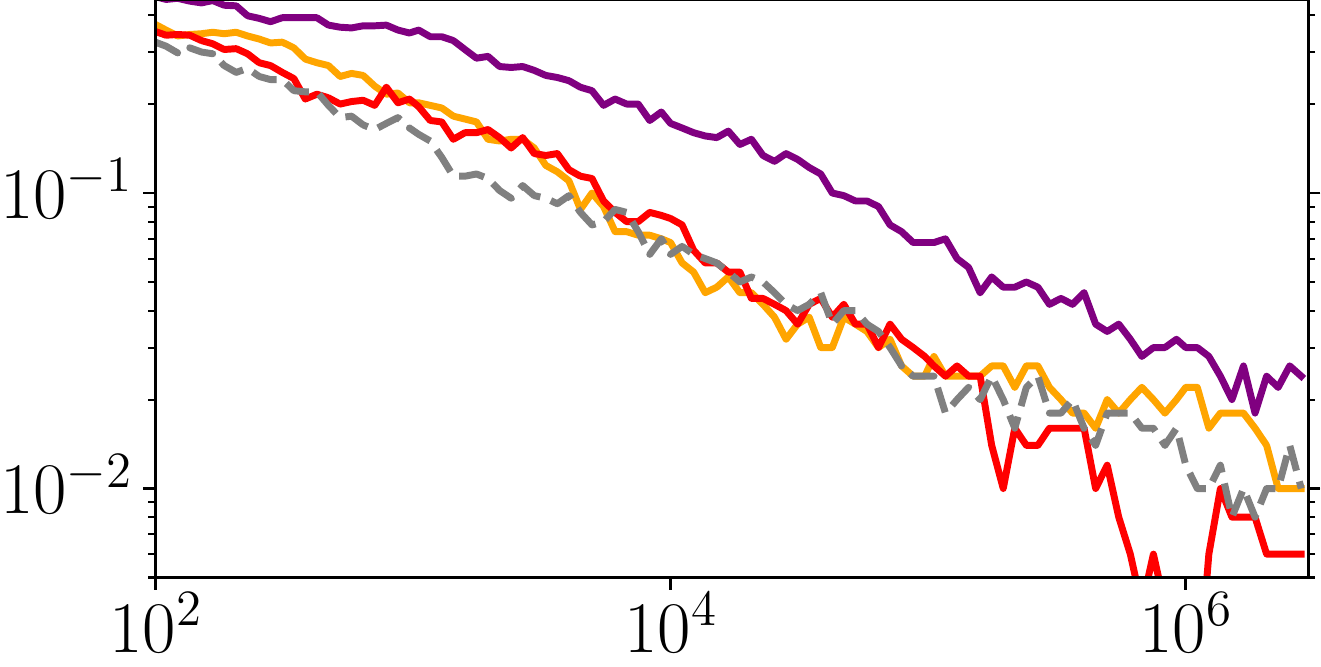}
\\
\rotatebox[origin=lt]{90}{\hspace{1.5em} \footnotesize Absolute Error} &
\includegraphics[scale=0.4]{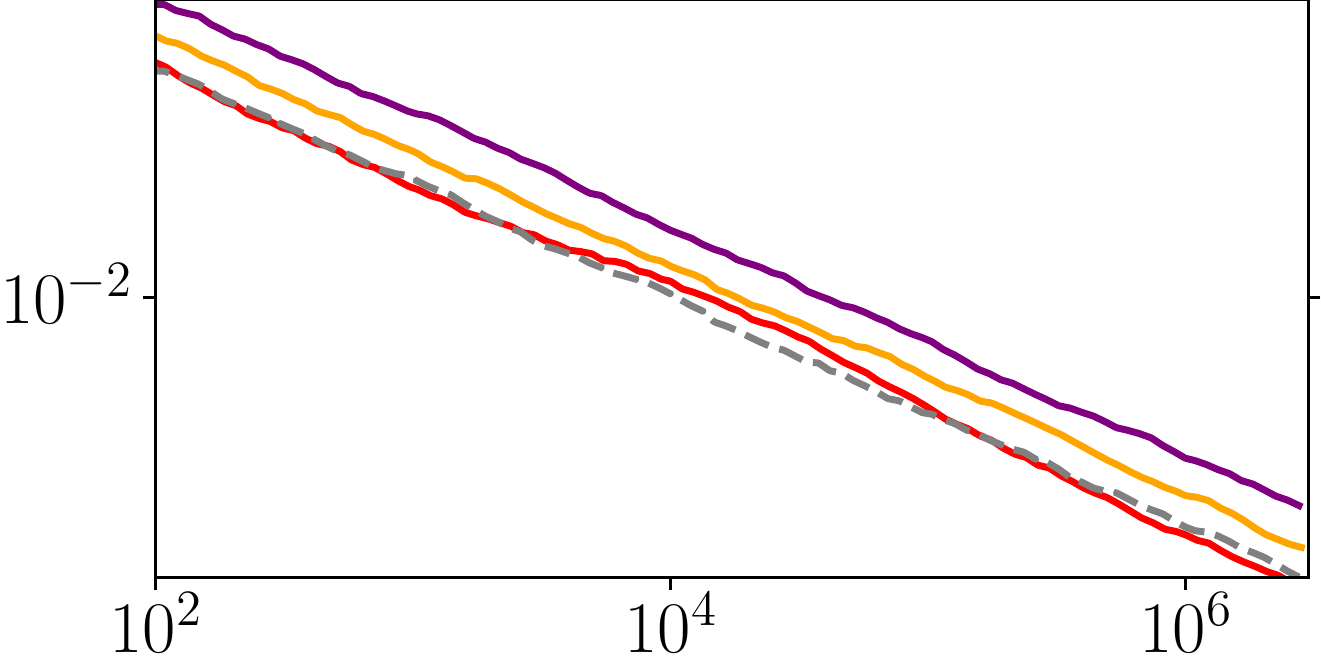} &
\includegraphics[scale=0.4]{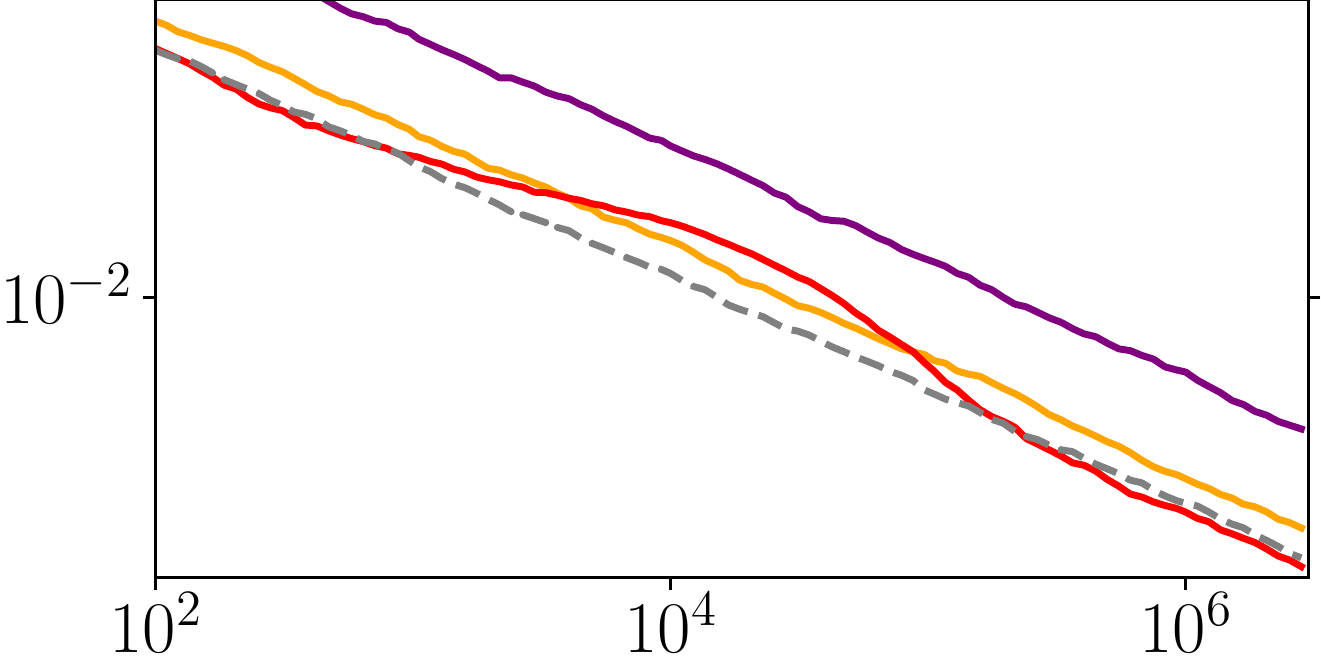}
\\
\rotatebox[origin=lt]{90}{\hspace{0.1em} \footnotesize Mean Squared Error} &
\includegraphics[scale=0.4]{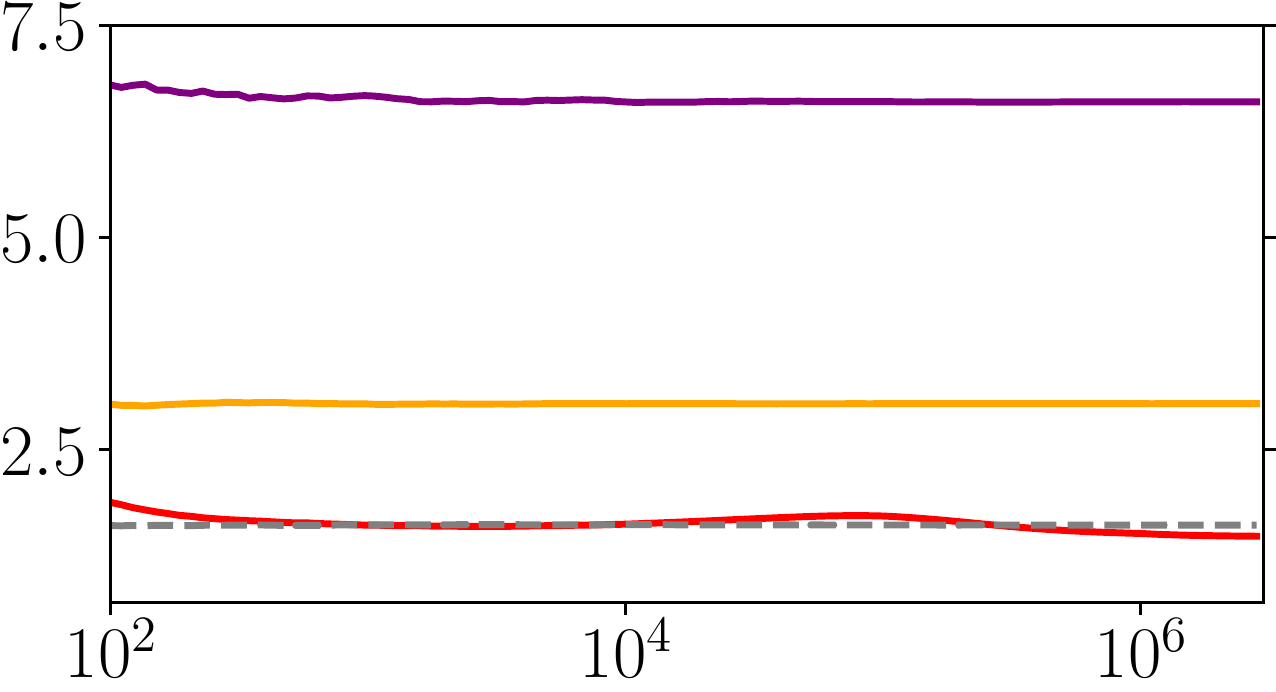} &
\includegraphics[scale=0.4]{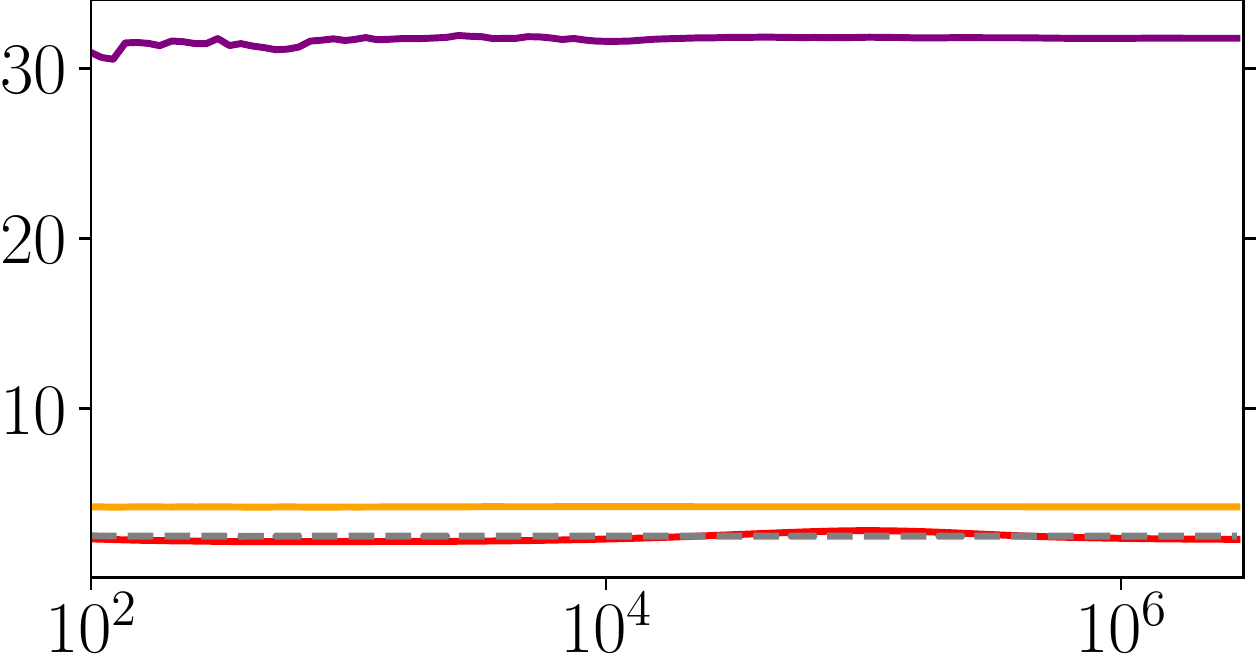} 
\\
& \multicolumn{1}{c}{\hspace{1.4em} \footnotesize Number of Queries Issued}
& \multicolumn{1}{c}{\hspace{1.4em} \footnotesize Number of Queries Issued}
 \vspace{0.5em}
\\
 \multicolumn{3}{c}{
 \includegraphics[scale=.4]{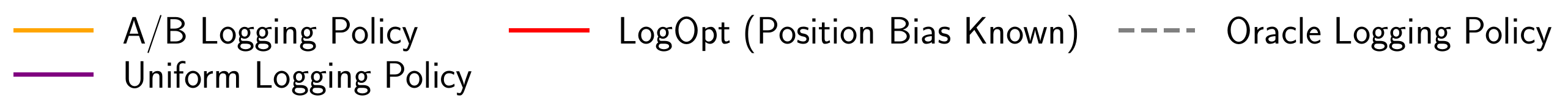}
} 
\end{tabular}
\caption{
Comparison of logging policies for counterfactual evaluation; displayed results are an average over 500 comparisons.
}
\label{fig:loggingpolicy}
\end{figure*}

\section{Results}
Our results are displayed in Figures~\ref{fig:baselines},~\ref{fig:loggingpolicy},~and~\ref{fig:errordistribution}.
Figure~\ref{fig:baselines} shows the results comparing \ac{LogOpt} with other online evaluation methods;
Figure~\ref{fig:loggingpolicy} compares \ac{LogOpt} with counterfactual evaluation using other logging policies;
and finally, Figure~\ref{fig:errordistribution} shows the distribution of binary errors for each method after $3 \cdot 10^6$ sampled queries.

\subsection{Performance of \acs{LogOpt}}

In Figure~\ref{fig:baselines} we see that, unlike interleaving methods, counterfactual evaluation with \ac{LogOpt} continues to decrease both its binary error and its absolute error as the number of queries increases.
While interleaving methods converge at a binary error of at least 2.2\% and an absolute error greater than $0.01$, \ac{LogOpt} appears to converge towards zero errors for both.
This is expected as \ac{LogOpt} is proven to be unbiased when the position bias is known. Interestingly, we see similar behavior from \ac{LogOpt} with estimated position bias.
Both when bias is known or estimated, \ac{LogOpt} has a lower error than the interleaving methods after $2 \cdot 10^3$ queries.
Thus we conclude that interleaving methods converge faster and have an initial period where their error is lower, but are biased.
In contrast, by being unbiased, \ac{LogOpt} converges on a lower error eventually. 

If we use Figure~\ref{fig:baselines} to compare \ac{LogOpt} with A/B testing, we see that on both datasets \ac{LogOpt} has a considerably smaller mean squared error.
Since both methods are unbiased, this means that \ac{LogOpt} has a much lower variance and thus is expected to converge faster.
On the Yahoo dataset we observe this behavior, both in terms of binary error and absolute error and regardless of whether the bias is estimated, \ac{LogOpt} requires half as much data as A/B testing to reach the same level or error.
Thus, on Yahoo \ac{LogOpt} is roughly twice as data-efficient as A/B testing.
On the MSLR dataset it is less clear whether \ac{LogOpt} is noticeably more efficient: after $10^4$ queries the absolute error of \ac{LogOpt} is twice as high, but after $10^5$ queries it has a lower error than A/B testing.
We suspect that the relative drop in performance around $10^4$ queries is due to \ac{LogOpt} overfitting on incorrect $\hat{\zeta}$ values, however, we were unable to confirm this.
Hence, \ac{LogOpt} is just as efficient as, or even more efficient than, A/B testing, depending on the circumstances.

Finally, when we use Figure~\ref{fig:loggingpolicy} to compare \ac{LogOpt} with other logging policy choices, we see that \ac{LogOpt} mostly approximates the optimal Oracle logging policy.
In contrast, the uniform logging policy is very data-inefficient; on both datasets it requires around ten times the number of queries to reach the same level or error as \ac{LogOpt}.
The A/B logging policy is a better choice than the uniform logging policy, but apart from the dip in performance on the MSLR dataset, it appears to require twice as many queries as \ac{LogOpt}.
Interestingly, the performance of \ac{LogOpt} is already near the Oracle when only $10^2$ queries have been issued.
With such a small number of interactions, accurately estimating the relevances $\zeta$ should not be possible, thus it appears that in order for \ac{LogOpt} to find an efficient logging policy the relevances $\zeta$ are not important.
This must mean that only the differences in behavior between the rankers (i.e., $\lambda$) have to be known for \ac{LogOpt} to be efficient.
Overall, these results show that \ac{LogOpt} can greatly increase the efficiency of counterfactual estimation.

\subsection{Bias of interleaving}
\label{sec:biasinterleaving}
Our results in Figure~\ref{fig:baselines} clearly illustrate the bias of interleaving methods: each of them systematically infers incorrect preferences in (at least) 2.2\% of the ranker-pairs.
These errors are systematic since increasing the number of queries from $10^5$ to $3\cdot10^6$ does not remove any of them.
Additionally, the combination of the lowest mean-squared-error with a worse absolute error than A/B testing after $10^4$ queries, indicates that interleaving results in a low variance at the cost of bias.
To better understand when these systematic errors occur, we show the distribution of binary errors w.r.t.\ the \ac{CTR} differences of the associated ranker-pairs in Figure~\ref{fig:errordistribution}.
Here we see that most errors occur on ranker-pairs where the \ac{CTR} difference is smaller than 1\%, and that of all comparisons the percentage of errors greatly increases as the \ac{CTR} difference decreases below 1\%.
This suggests that interleaving methods are unreliable to detect preferences when differences are 1\% \ac{CTR} or less.

It is hard to judge the impact this bias may have in practice.
On the one hand, a 1\% \ac{CTR} difference is far from negligible: generally a 1\% increase in \ac{CTR} is considered an impactful improvement in the industry~\citep{richardson2007predicting}.
On the other hand, our results are based on a single click model with specific values for position bias and conditional click probabilities.
While our results strongly prove that interleaving is biased, we should be careful not to generalize the size of the observed systematic error to all other ranking settings.

Previous work has performed empirical studies to evaluate various interleaving methods with real users.
\citet{chapelle2012large} applied interleaving methods to compare ranking systems for three different search engines, and found team-draft interleaving to highly correlate with absolute measures such as \ac{CTR}.
However, we note that in the study by \citet{chapelle2012large} no more than six rankers were compared, thus such a study would likely miss a systematic error of 2.2\%.
In fact, \citet{chapelle2012large} note themselves that they cannot confidently claim team-draft interleaving is completely unbiased.
\citet{schuth2015predicting} performed a larger comparison involving 38 ranking systems, but again, too small to reliably detect a small systematic error.

It appears that the field is missing a large scale comparison that involves a large enough number of rankers to observe small systematic errors.
If such an error is found, the next step is to identify if certain types of ranking behavior are erroneously and systematically disfavored.
While these questions remain unanswered, we are concerned that the claims of unbiasedness in previous work on interleaving (see Section~\ref{sec:related:interleaving}) give practitioners an unwarranted sense of reliability in interleaving.



\begin{figure}[t]
\centering
\begin{tabular}{c l l}
&
 \multicolumn{1}{c}{\hspace{0.5em}\footnotesize Yahoo!\ Webscope}
&
 \multicolumn{1}{c}{\hspace{1.5em}\footnotesize MSLR-Web30k}
\\
\rotatebox[origin=lt]{90}{\hspace{0.1em} \footnotesize Team-Draft Interleaving}
&
\includegraphics[scale=0.4]{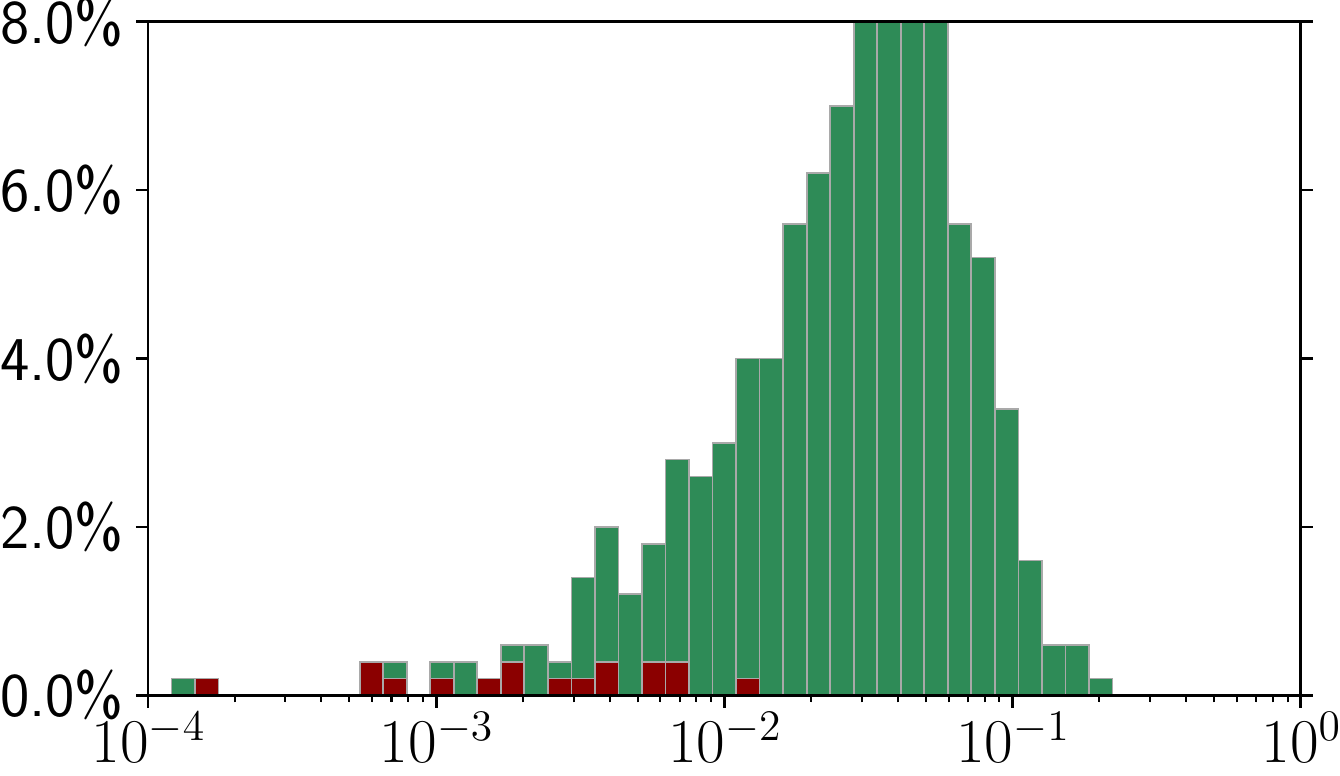} &
\includegraphics[scale=0.4]{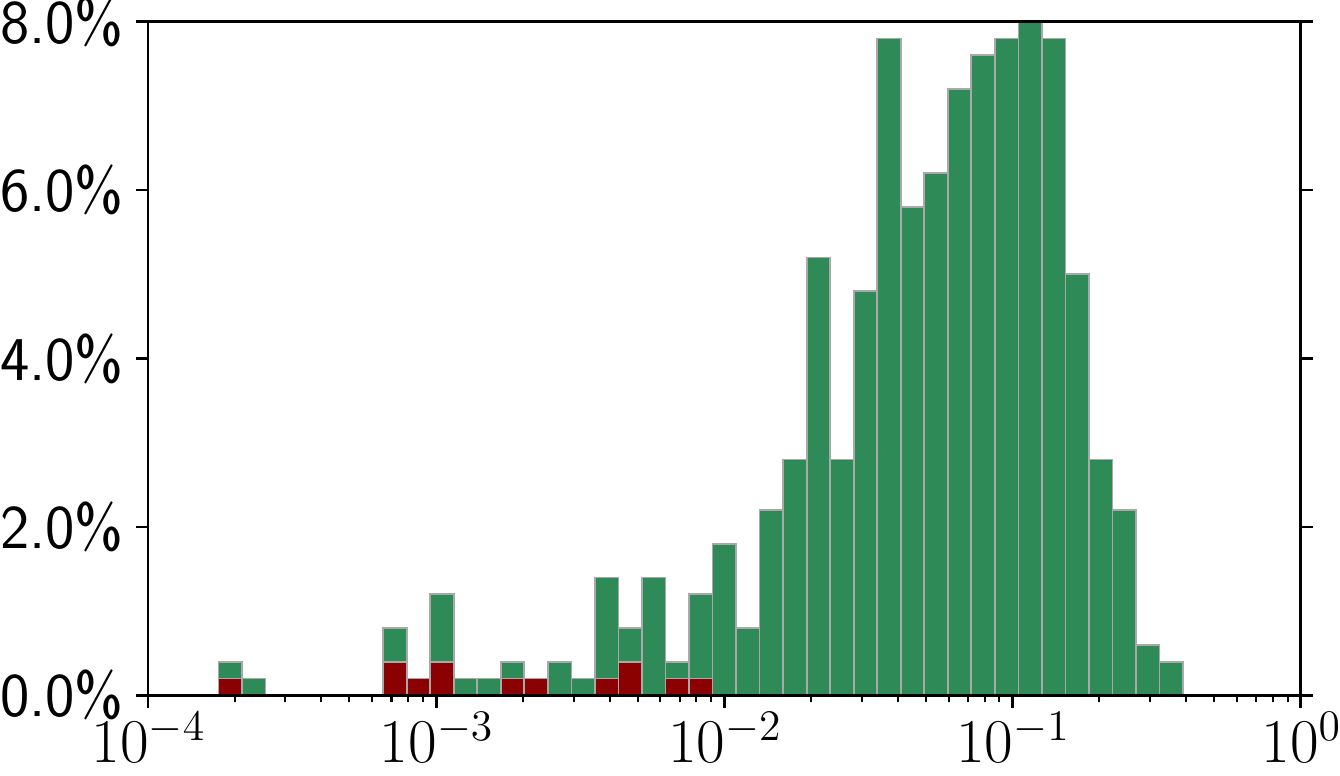} 
\\
\rotatebox[origin=lt]{90}{\hspace{0.0em} \footnotesize  Probabilistic Interleaving}
&
\includegraphics[scale=0.4]{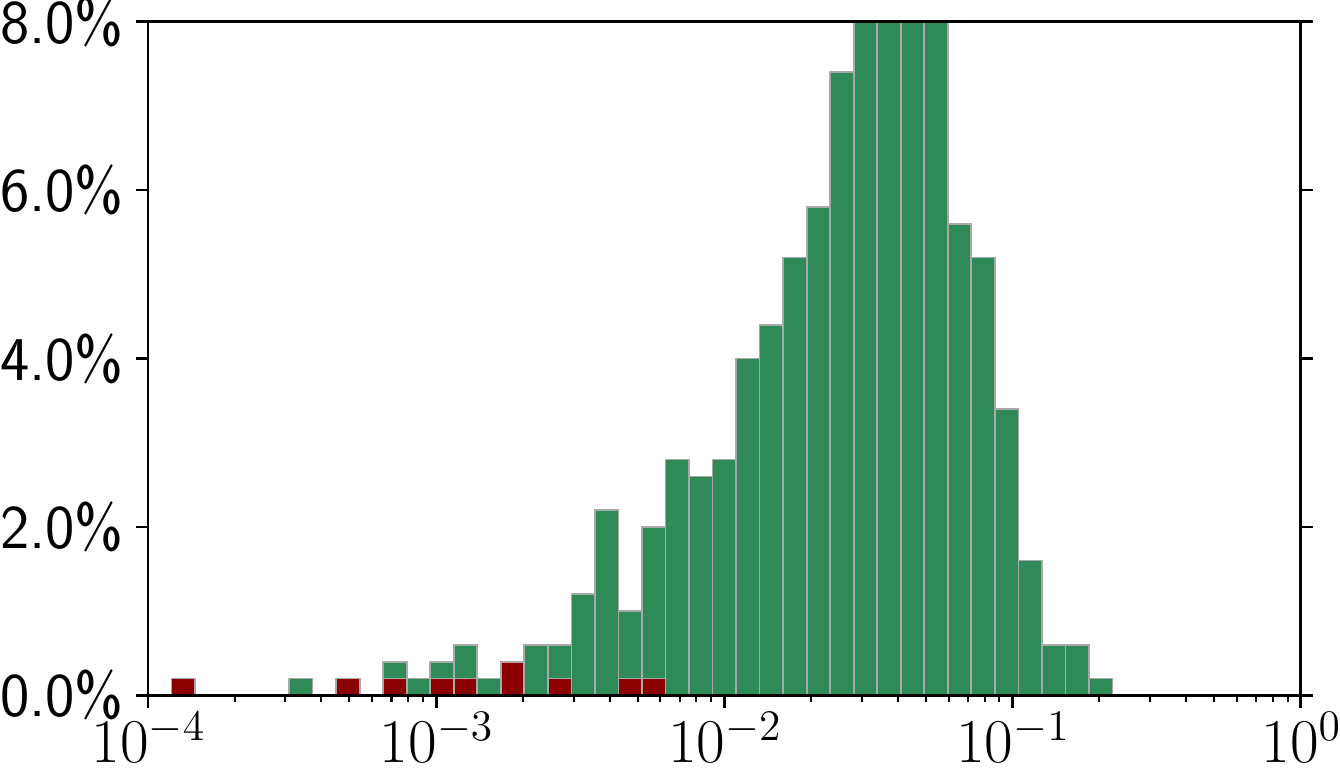} &
\includegraphics[scale=0.4]{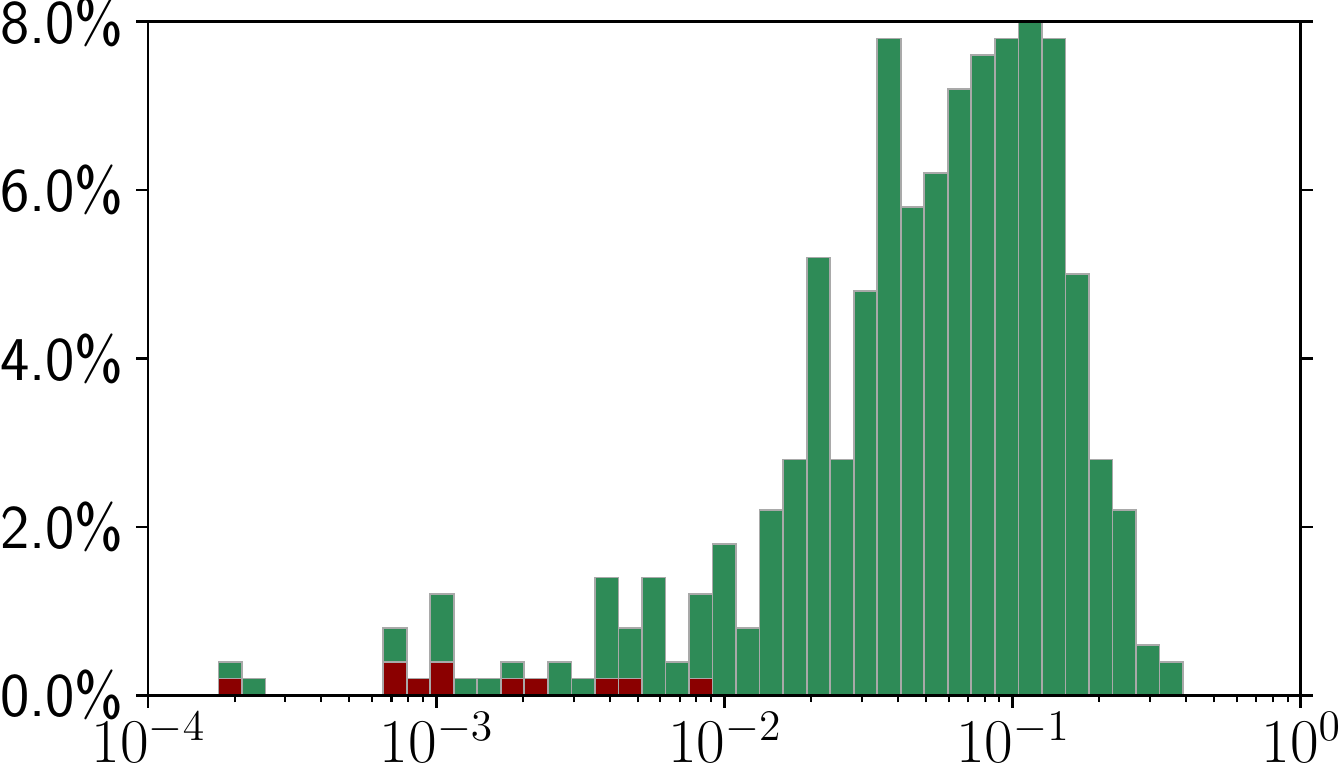} 
\\
\rotatebox[origin=lt]{90}{\hspace{0.4em} \footnotesize  Optimized Interleaving}
&
\includegraphics[scale=0.4]{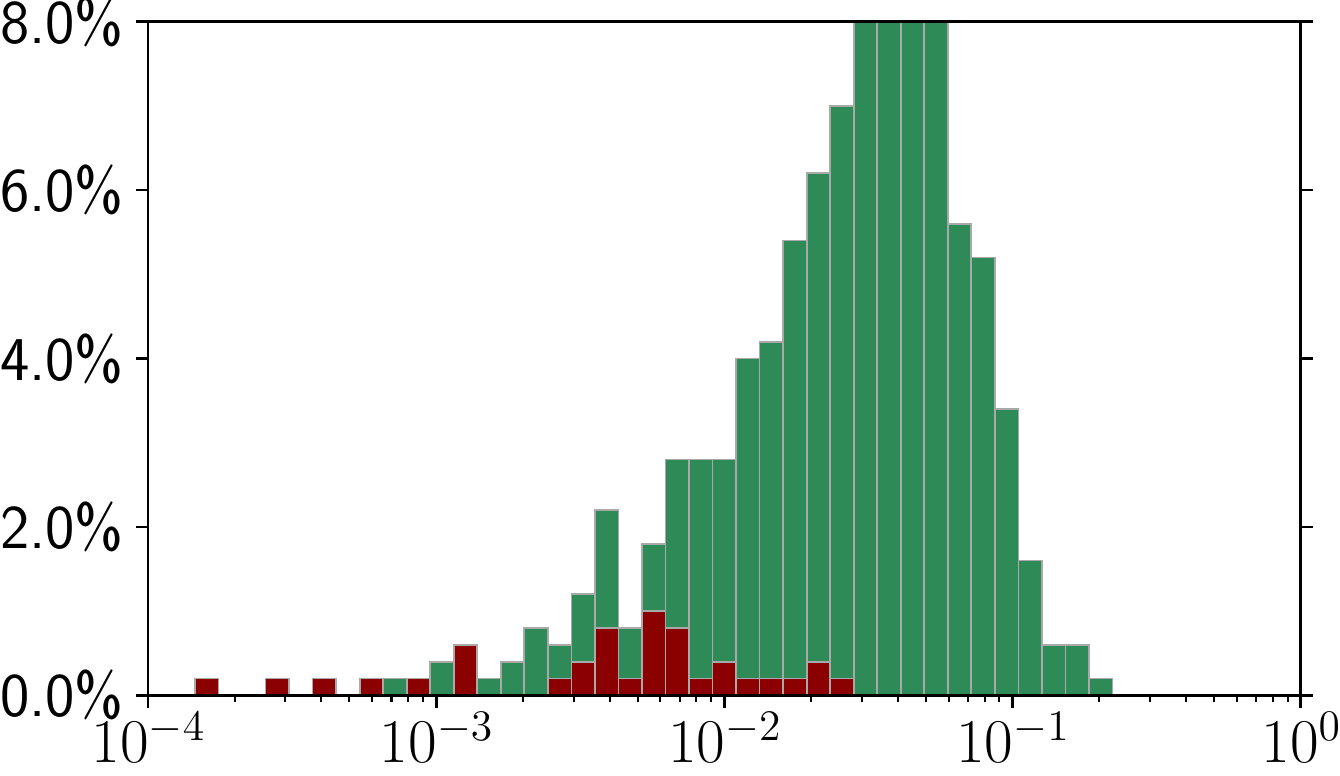} &
\includegraphics[scale=0.4]{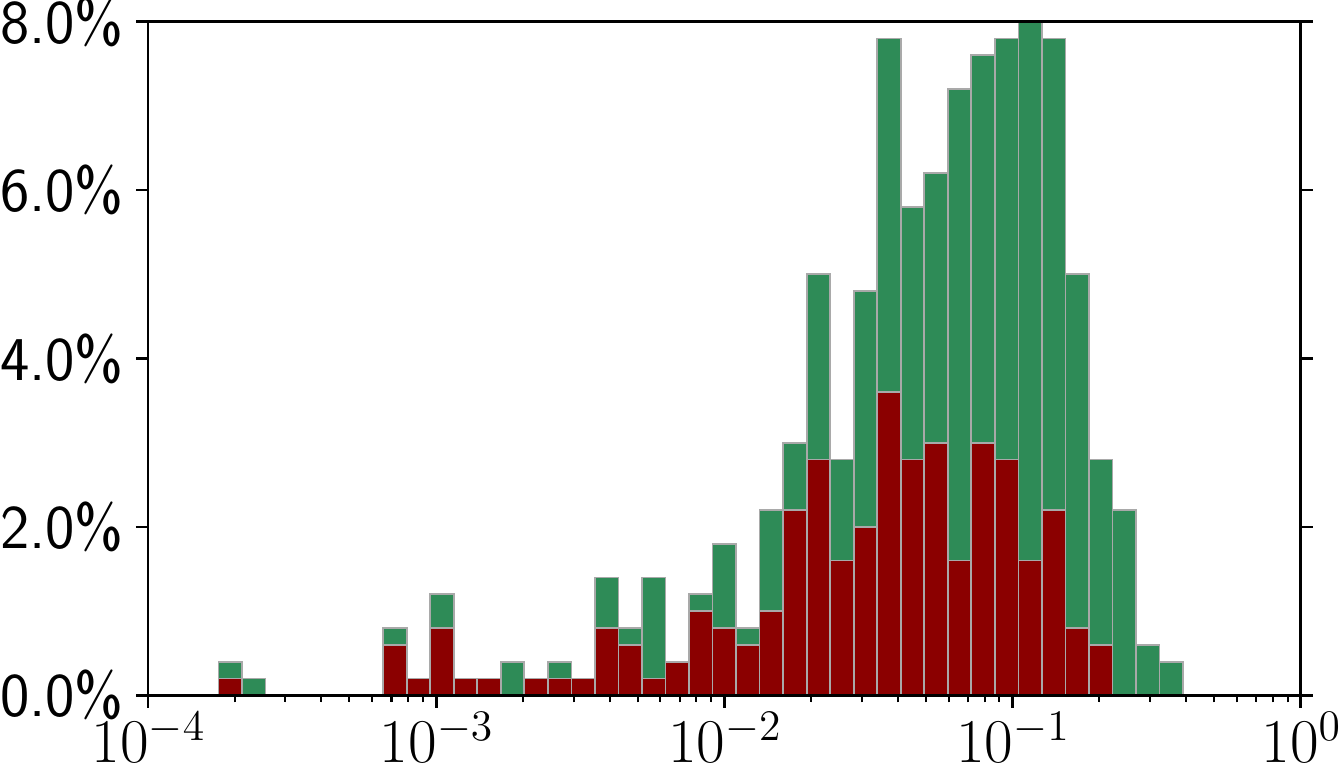} 
\\
\rotatebox[origin=lt]{90}{\hspace{2.2em} \footnotesize   A/B Testing}
&
\includegraphics[scale=0.4]{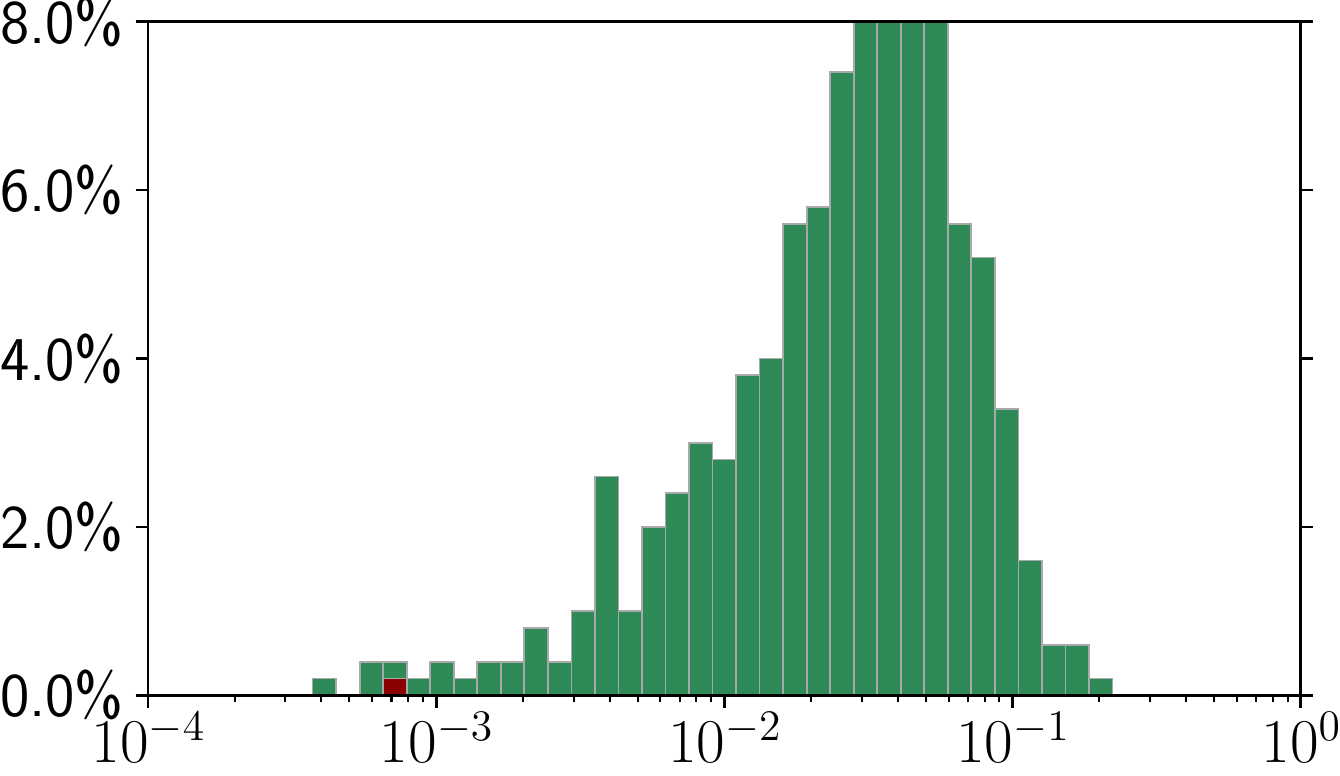} &
\includegraphics[scale=0.4]{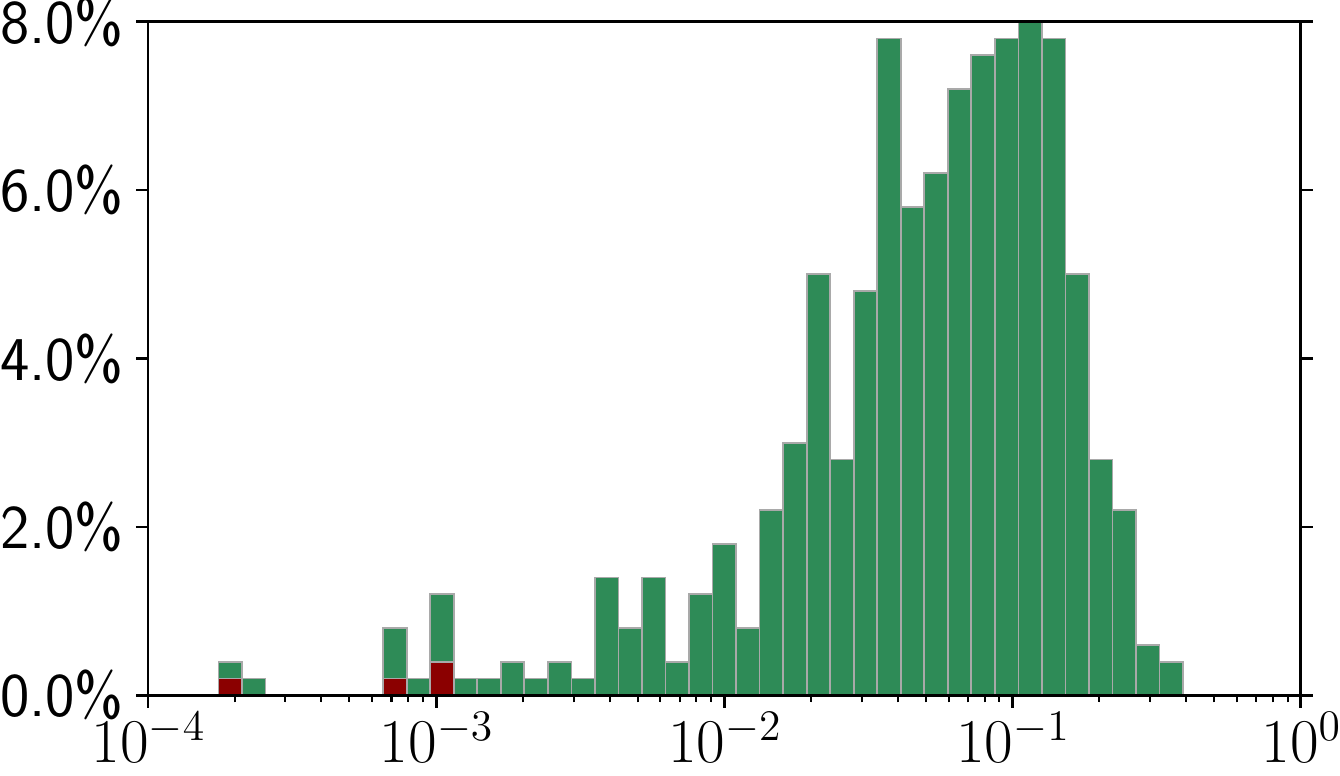} 
\\
\rotatebox[origin=lt]{90}{\hspace{0.1em} \footnotesize  LogOpt (Bias Estimated)}
&
\includegraphics[scale=0.4]{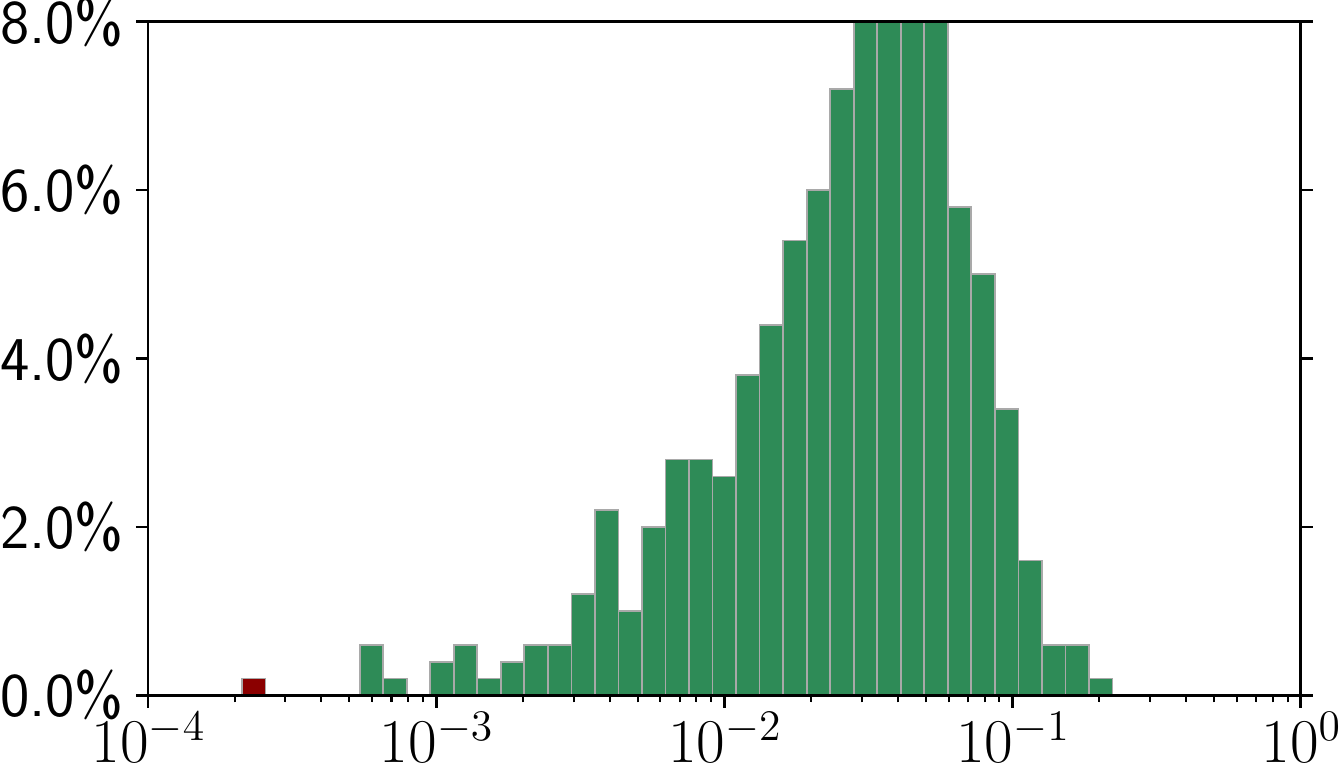} &
\includegraphics[scale=0.4]{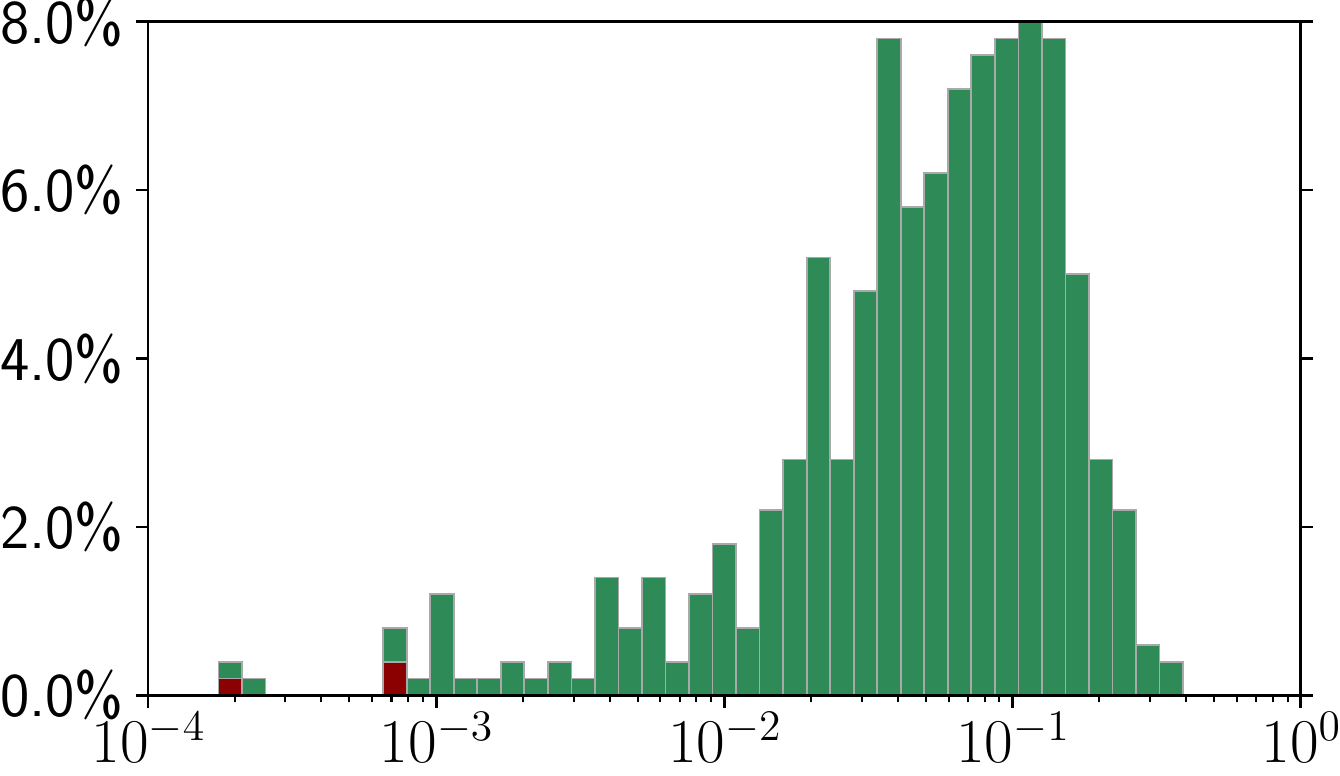} 
\\
&
 \multicolumn{1}{c}{\small  \hspace{1.8em} CTR difference}
&
 \multicolumn{1}{c}{\small  \hspace{1.8em} CTR difference}
\end{tabular}
\vspace{0.3\baselineskip}
\caption{
Distribution of errors over the \ac{CTR} differences of the rankers in the comparison; red indicates a binary error; green indicates a correctly inferred binary preference; results are on estimates based on $3 \cdot 10^6$ sampled queries.
}
\label{fig:errordistribution}
\end{figure}

%% file: 09-onlinecountereval/sections/07-conclusion.tex
\section{Conclusion}

In this chapter, we considered thesis research question \ref{thesisrq:onlineeval}: whether
 counterfactual evaluation methods for ranking can be extended to perform efficient and effective online evaluation.
Our answer is positive: we have introduced the \acf{LogOpt}: the first method that optimizes a logging policy for minimal variance counterfactual evaluation.
Counterfactual evaluation is proven to be unbiased w.r.t.\ position bias and item-selection bias under a wide range of logging policies.
With the introduction of \ac{LogOpt}, we now have an algorithm that can decide which rankings should be displayed for the fastest convergence.
Therefore, we argue that \ac{LogOpt} turns the \ac{IPS}-based counterfactual evaluation approach -- which is indifferent to the logging policy -- into an online approach -- which instructs the logging policy.
Our experimental results show that \ac{LogOpt} can lead to a better data-efficiency than A/B testing, while also showing that interleaving is biased.

This brings us to the second thesis research question that this chapter addressed, \ref{thesisrq:interleaving}: 
whether interleaving methods are truly unbiased w.r.t.\ position bias.
We answer this question negatively:
Our experimental results clearly reveal a systematic error in interleaving, moreover, in Appendix~\ref{sec:appendix:bias} we formally prove that cases exist where interleaving is affected by position bias.
In other words, interleaving should not be considered unbiased under the most common definition of bias in counterfactual evaluation.

While our findings are mostly theoretical, they do suggest that future work should further investigate the bias in interleaving methods.
Our results suggest that all interleaving methods make systematic errors, in particular when rankers with a similar \ac{CTR} are compared.
Furthermore, to the best of our knowledge, no empirical studies have been performed that could measure such a bias; our findings strongly show that such a study would be highly valuable to the field.
Finally, \ac{LogOpt} shows that in theory an evaluation method that is both unbiased and efficient is possible; if future work finds that these theoretical findings match empirical results with real users, this could be the start of a new line of theoretically-justified online evaluation methods.

Inspired by the success of this chapter to find a method effective at both online and counterfactual evaluation for ranking, Chapter~\ref{chapter:06-onlinecounterltr} introduces a method that is effective at both online and counterfactual \ac{LTR}.
Together, these chapters show that the divide between online and counterfactual optimization/evaluation can be bridged.

%% file: 09-onlinecountereval/sections/08-appendix.tex
\setlength{\tabcolsep}{1em}
\renewcommand{\arraystretch}{1.0}

\section{Proof of Bias in Interleaving}
\label{sec:appendix:bias}

Section~\ref{sec:related:interleaving} claimed that for the discussed interleaving methods, an example can be constructed so that in expectation the wrong binary outcome is estimated w.r.t.\ the actual expected \ac{CTR} differences.
These examples are enough to prove that these interleaving methods are biased w.r.t.\ \ac{CTR} differences.
In the following sections we will introduce a single example for each interleaving method.

For clarity, we will keep these examples as basic as possible.
We consider a ranking setting where only a single query $q_1$ occurs, i.e.\ $P(q_1) = 1$, furthermore, there are only three documents to be ranked: $A$, $B$, and $C$.
The two policies $\pi_1$ and $\pi_2$ in the comparison are both deterministic so that:
$\pi_1([A, B, C] \mid q_1) = 1$
and
$\pi_2([B, C, A] \mid q_1) = 1$.
Thus $\pi_1$ will always display the ranking: $[A, B, C]$, and $\pi_2$ the ranking: $[B, C, A]$.
Furthermore, document $B$ is completely non-relevant: $\zeta_B = 0$, consequently, $B$ can never receive clicks; this will make our examples even simpler.

The true $\mathbbm{E}[\text{CTR}]$ difference is thus:
\begin{equation}
\Delta(\pi_1, \pi_2) = (\theta_1 - \theta_3) \zeta_A + (\theta_3 - \theta_2) \zeta_C.
\end{equation}
For each interleaving method, will now show that position bias parameters $\theta_1$, $\theta_2$, and $\theta_3$ and relevances $\zeta_A$ and $\zeta_C$ exist where the wrong binary outcome is estimated.

\subsection{Team-Draft Interleaving}
Team-Draft Interleaving~\citep{radlinski2008does} lets rankers take turns to add their top document and keeps track which ranker added each document.
In total there are four possible interleaving and assignment combinations, each is equally probable:
\FloatBarrier
\begin{table}[h]
\centering
\begin{tabular}{c c c c}
\toprule 
Interleaving & Ranking & Assignments & Probability \\
\midrule
$R_1$ & A, B, C & 1, 2, 1 & 1/4 \\
$R_2$ & A, B, C & 1, 2, 2 & 1/4 \\
$R_3$ & B, A, C & 2, 1, 1 & 1/4 \\
$R_4$ & B, A, C & 2, 1, 2 & 1/4 \\
\bottomrule
\end{tabular}
\end{table}
\FloatBarrier
\noindent
Per issued query Team-Draft Interleaving produces a binary outcome, this is based on which ranker had most of its assigned documents clicked.
To match our \ac{CTR} estimate, we use $1$ to indicate $\pi_1$ receiving more clicks, and $-1$ for $\pi_2$.
Per interleaving we can compute the probability of each outcome:
\begin{align*}
P(\text{outcome} = \phantom{-}1 \mid R_1) &= \theta_1\zeta_A + (1- \theta_1\zeta_A)\theta_3\zeta_C,
\\
P(\text{outcome} = \phantom{-}1 \mid R_2) &= \theta_1\zeta_A(1- \theta_3\zeta_C),
\\
P(\text{outcome} = \phantom{-}1 \mid R_3) &= \theta_2\zeta_A + (1- \theta_2\zeta_A)\theta_3\zeta_C,
\\
P(\text{outcome} = \phantom{-}1 \mid R_4) &= \theta_2\zeta_A (1- \theta_3\zeta_C),
\\
P(\text{outcome} = -1 \mid R_1) &= 0,
\\
P(\text{outcome} = -1 \mid R_2) &= (1-\theta_1\zeta_A)\theta_3\zeta_C,
\\
P(\text{outcome} = -1 \mid R_3) &= 0,
\\
P(\text{outcome} = -1 \mid R_4) &= (1- \theta_2\zeta_A)\theta_3\zeta_C.
\end{align*}
Since every interleaving is equally likely, we can easily derive the unconditional probabilities:
\begin{align*}
P(\text{outcome} = \phantom{-}1) &= \frac{1}{4}\Big(
\theta_1\zeta_A + (1- \theta_1\zeta_A)\theta_3\zeta_C
+  \theta_1\zeta_A(1- \theta_3\zeta_C)
\\ & \quad \,\,
+ \theta_2\zeta_A + (1- \theta_2\zeta_A)\theta_3\zeta_C
+ \theta_2\zeta_A (1- \theta_3\zeta_C)
\Big),
\\
P(\text{outcome} = -1) &= \frac{1}{4}\Big(
(1-\theta_1\zeta_A)\theta_3\zeta_C
+ (1- \theta_2\zeta_A)\theta_3\zeta_C
\Big).
\end{align*}
With these probabilities, the expected outcome is straightforward to calculate:
\begin{align*}
\mathbbm{E}[\text{outcome}] &=
 P(\text{outcome} = 1) - P(\text{outcome} = -1)
 \\ &=
 \frac{1}{4}\Big(
\theta_1\zeta_A
+  \theta_1\zeta_A(1- \theta_3\zeta_C)
+ \theta_2\zeta_A 
+ \theta_2\zeta_A (1- \theta_3\zeta_C)
\Big)
> 0.
\end{align*}
Interestingly, without knowing the values for $\theta$, $\zeta_A$ and $\zeta_C$, we already know that the expected outcome is positive.
Therefore, we can simply choose values that lead to a negative \ac{CTR} difference, and the expected outcome will be incorrect.
For this example, we choose the position bias:
$\theta_1 = 1.0$,
$\theta_2 = 0.9$, and
$\theta_3 = 0.8$;
and the relevances:
$\zeta_1 = 0.1$, and
$\zeta_3 = 1.0$.
As a result, the expected binary outcome of Team-Draft Interleaving will not match the true $\mathbbm{E}[\text{CTR}]$ difference:
\begin{equation}
\Delta(\pi_1, \pi_2) < 0 \land \mathbbm{E}[\text{outcome}]  > 0.
\end{equation}
Therefore, we have proven that Team-Draft Interleaving is biased w.r.t.\ \ac{CTR} differences.

\subsection{Probabilistic Interleaving}
Probabilistic Interleaving~\citep{hofmann2011probabilistic} treats rankings as distributions over documents, we follow the soft-max approach of~\citet{hofmann2011probabilistic} and use $\tau=4.0$ as suggested.
Probabilistic Interleaving creates interleavings by sampling randomly from one of the rankings, unlike Team-Draft Interleaving it does not remember which ranking added each document.
Because rankings are treated as distributions, every possible permutation is a valid interleaving, leading to six possibilities with different probabilities of being displayed.
When clicks are received, every possible assignment is considered and the expected outcome is computed over all possible assignments.
Because there are 36 possible rankings and assignment combinations, we only report every possible ranking and the probabilities for documents $A$ or $C$ being added by $\pi_1$:
\FloatBarrier
\setlength{\tabcolsep}{0.4em}
\begin{table}[!]
\centering
\begin{tabular}{c c c c c}
\toprule 
Interleaving & Ranking & $P(\text{add}(A) = 1)$ & $P(\text{add}(C) = 1)$ & Probability \\
\midrule
$R_1$ & A, B, C & 0.9878 & 0.4701 & 0.4182 \\
$R_2$ & A, C, B & 0.9878 & 0.4999 & 0.0527 \\
$R_3$ & B, A, C & 0.8569 & 0.0588 & 0.2849 \\
$R_4$ & B, C, A & 0.5000 & 0.0588 & 0.2094 \\
$R_5$ & C, A, B & 0.9872 & 0.5000 & 0.0166 \\
$R_6$ & C, B, A & 0.5000 & 0.0562 & 0.0182 \\
\bottomrule
\end{tabular}
\end{table}
\FloatBarrier
\noindent
These probabilities are enough to compute the expected outcome, similar as the procedure we used for Team-Draft Interleaving.
We will not display the full calculation here as it is extremely long; we recommend using some form of computer assistance to perform these calculations.
While there are many possibilities, we choose the following position bias:
$\theta_1 = 1.0$,
$\theta_2 = 0.9$, and
$\theta_3 = 0.3$;
and relevance:
$\zeta_1 = 0.5$, and
$\zeta_3 = 1.0$.
This leads to the following erroneous result:
\begin{equation}
\Delta(\pi_1, \pi_2) < 0 \land \mathbbm{E}[\text{outcome}]  > 0.
\end{equation}
Therefore, we have proven that Probabilistic Interleaving is biased w.r.t.\ \ac{CTR} differences.

\subsection{Optimized Interleaving}
Optimized Interleaving casts interleaving as an optimization problem~\citep{radlinski2013optimized}.
Optimized Interleaving works with a credit function: each clicked document produces a positive or negative credit.
The sum of all credits is the final estimated outcome.
We follow \citet{radlinski2013optimized} and use the linear rank difference, resulting in the following credits per document:
$\text{click-credit}(A) = 2$,
$\text{click-credit}(B) = -1$, and
$\text{click-credit}(C) = -1$.
Then the set of allowed interleavings is created, these are all the rankings that do not contradict a pairwise document preference that both rankers agree on.
Given this set of interleavings, a distribution over them is found so that if every document is equally relevant then no preference is found.\footnote{\citet{radlinski2013optimized} state that if clicks are not correlated with relevance then no preference should be found, in their click model (and ours) these two requirements are actually equivalent.}
For our example, the only valid distribution over interleavings is the following:
\setlength{\tabcolsep}{0.4em}
\FloatBarrier
\begin{table}[h]
\centering
\begin{tabular}{c c c}
\toprule 
Interleaving & Ranking & Probability \\
\midrule
$R_1$ & A, B, C & $1/3$ \\
$R_2$ & B, A, C & $1/3$ \\
$R_3$ & B, C, A & $1/3$ \\
\bottomrule
\end{tabular}
\end{table}
\FloatBarrier
\noindent
The expected credit outcome shows us which ranker will be preferred in expectation:
\begin{equation}
\mathbbm{E}[\text{credit}] = \frac{1}{3}\big( 2(\theta_1 + \theta_2 + \theta_3)\zeta_A - (\theta_2 + 2\theta_3)\zeta_C \big).
\end{equation}
We choose the position bias:
$\theta_1 = 1.0$,
$\theta_2 = 0.9$, and
$\theta_3 = 0.9$;
and the relevances:
$\zeta_1 = 0.5$,
$\zeta_3 = 1.0$.
As a result, the true $\mathbbm{E}[\text{CTR}]$ difference is positive, but optimized interleaving will prefer $\pi_2$ in expectation:
\begin{equation}
\Delta(\pi_1, \pi_2) > 0 \land \mathbbm{E}[\text{credit}] < 0.
\end{equation}
Therefore, we have proven that Optimized Interleaving is biased w.r.t.\ \ac{CTR} differences.

\section{Expanded Explanation of Gradient Approximation}
\label{sec:appendix:approx}

This section describes our Monte-Carlo approximation of the variance gradient in more detail.
We repeat the steps described in Section~\ref{sec:method:derivates} and include some additional intermediate steps; this should make it easier for a reader to verify our theory.

First, we assume that policies place the documents in order of rank and the probability of placing an individual document at rank $x$ only depends on the previously placed documents.
Let $R_{1:x-1}$ indicate the (incomplete) ranking from rank $1$ up to rank $x$, then $\pi_0(d \mid R_{1:x-1}, q)$ indicates the probability that document $d$ is placed at rank $x$ given that the ranking up to $x$ is $R_{1:x-1}$.
The probability of a ranking $R$ of length $K$ is thus:
\begin{equation}
\pi_0(R \mid q) = \prod_{x=1}^{K} \pi_0(R_x \mid R_{1:x-1}, q).
\end{equation}
The probability of a ranking $R$ up to rank $k$ is:
\begin{equation}
\pi_0(R_{1:k} \mid q) = \prod_{x=1}^{k} \pi_0(R_x \mid R_{1:x-1}, q).
\end{equation}
Therefore the propensity (cf.\ Eq.~\ref{eq:prop}) can be rewritten to:
\begin{equation}
\rho(d \,|\, q) = \sum_{k=1}^K \theta_k \sum_R \pi_0(R_{1:k-1} \mid q) \pi_0(d \,|\, R_{1:k-1}, q).
\end{equation}
Before we take the gradient of the propensity, we note that the gradient of the probability of a single ranking is:
\begin{equation}
\frac{\delta \pi_0(R \mid q) }{\delta \pi_0}
 =
\sum_{x=1}^K
  \frac{\pi_0(R \mid q)}{\pi_0(R_x \mid R_{1:x},  q)}
  \left[\frac{\delta \pi_0(R_x \mid R_{1:x-1},  q) }{\delta \pi_0} \right].
  \label{eq:appendix:rankingprob}
\end{equation}
Using this gradient, we can derive the gradient of the propensity w.r.t.\ the policy: %
\begin{align}
\frac{\delta \rho(d |\, q)}{\delta \pi_0} &=  \sum_{k=1}^K \theta_k \sum_R \pi_0(R_{1:k-1} |\,  q) 
\Bigg(\left[ \frac{\delta \pi_0(d |\, R_{1:k-1}, q)}{\delta \pi_0} \right] \nonumber
\\
& \quad   + 
\sum_{x=1}^{k-1} \frac{\pi_0(d |\, R_{1:k-1}, q)}{\pi_0(R_x |\, R_{1:x-1} , q)}\left[ \frac{\delta\pi_0(R_x |\,  R_{1:x-1} , q)}{\delta \pi_0} \right]
\Bigg).
\end{align}
To avoid iterating over all rankings in the $\sum_R$ sum, we sample $M$ rankings:
$R^m \sim \pi_0(R \mid q)$, and a click pattern on each ranking: $c^m \sim P(c \mid R^m)$.
This enables us to make the following approximation:
\begin{equation}
\begin{split}
\widehat{\rho\text{-grad}}(d)
&=
 \frac{1}{M} \sum_{m=1}^{M} \sum_{k=1}^K \theta_k 
\Bigg(
\left[ \frac{\delta \pi_0(d |\, R^m_{1:k-1}, q)}{\delta \pi_0} \right] 
\\
& \qquad \quad \,\,\, + 
\sum_{x=1}^{k-1} \frac{\pi_0(d |\, R^m_{1:k-1}, q)}{\pi_0(R^m_x  |\, R^m_{1:x-1} , q)}\left[ \frac{\delta\pi_0(R^m_x  |\, R^m_{1:x-1} , q)}{\delta \pi_0} \right]
\Bigg),
\end{split}
\end{equation}
since $\frac{\delta \rho(d |\, q)}{\delta \pi_0} \approx \widehat{\rho\text{-grad}}(d, q)$.
The second part of Eq.~\ref{eq:gradient} is:
\begin{equation}
\Bigg[ \frac{\delta}{\delta \pi_0} \bigg(\Delta - \sum_{d : c(d) = 1}\frac{\lambda_d}{\rho_d}\bigg)^2 \Bigg]
= 2\bigg(\Delta - \sum_{d : c(d) = 1}\frac{\lambda_d}{\rho_d}\bigg) \sum_{d : c(d) = 1}\frac{\lambda_d}{\rho_d^2} \left[ \frac{\delta \rho_d}{\delta \pi_0} \right],
\end{equation}
using $\widehat{\rho\text{-grad}}(d)$ we get the approximation: 
\begin{equation}
\widehat{\text{error-grad}}(c) =
2\bigg(\Delta - \sum_{d : c(d) = 1}\frac{\lambda_d}{\rho_d}\bigg) \sum_{d : c(d) = 1}\frac{\lambda_d}{\rho_d^2} \widehat{\rho\text{-grad}}(d).
\end{equation}
Next, we consider the gradient of a single click pattern: 
\begin{equation}
\frac{\delta}{\delta \pi_0} P(c \mid q) = \sum_R P(c \mid R) \left[\frac{\delta \pi_0(R \mid q)}{\delta \pi_0} \right].
\end{equation}
This can then be used to reformulate the first part of Eq.~\ref{eq:gradient}:
\begin{equation}
\begin{split}
\sum_c
\left[\frac{\delta}{\delta \pi_0}P(c \mid q)\right] &\bigg(\Delta - \sum_{d : c(d) = 1}\frac{\lambda_d}{\rho_d}\bigg)^2 \\ 
&= \sum_c
\sum_R P(c \mid R) \left[\frac{\delta \pi_0(R \mid q)}{\delta \pi_0}\right] \bigg(\Delta - \sum_{d : c(d) = 1}\frac{\lambda_d}{\rho_d}\bigg)^2
\end{split}
\end{equation}
Making use of Eq.~\ref{eq:appendix:rankingprob}, we approximate this with:
\begin{equation}
\begin{split}
\widehat{\text{freq-grad}}&(R, c) = \\
& \bigg(\Delta - \sum_{d : c(d) = 1}\frac{\lambda_d}{\rho_d}\bigg)^2
 \sum_{x=1}^K
\frac{1}{\pi_0(R_x \mid R_{1:x-1},  q)}
\left[\frac{\delta \pi_0(R_x \mid R_{1:x-1},  q) }{\delta \pi_0} \right].
\end{split}
\end{equation}
Combining the approximation of both parts of Eq.~\ref{eq:gradient}, allows us to approximate the complete gradient:
\begin{equation}
\frac{\delta \text{Var}(\hat{\Delta}_{IPS}^{\pi_0} \mid q)}{\delta \pi_0} 
\approx 
\frac{1}{M}\sum_{m=1}^M \widehat{\text{freq-grad}}(R^m, c^m) + \widehat{\text{error-grad}}(c^m).
\end{equation}
This completes our expanded description of the gradient approximation.
We have shown that we can approximate the gradient of the variance w.r.t.\ a logging policy $\pi_0$, based on rankings sampled from $\pi_0$ and our current estimated click model $\hat{\theta}$, $\hat{\zeta}$, while staying computationally feasible.

%% file: 09-onlinecountereval/notation.tex
\section{Notation Reference for Chapter~\ref{chapter:06-onlinecountereval}}
\label{notation:06-onlinecountereval}

\begin{center}
\begin{tabular}{l l}
 \toprule
\bf Notation  & \bf Description \\
\midrule
$k$ & the number of items that can be displayed in a single ranking \\
$i$ & an iteration number \\
$q$ & a user-issued query \\
$d$ & an item to be ranked\\
$R$ & a ranked list \\
$R_{1:x}$ & the subranking in $R$ from index $1$ up to and including index $x$ \\
$\pi$ & a ranking policy\\
$\pi(R \mid q)$ & the probability that policy $\pi$ displays ranking $R$ for query $q$ \\
$\pi(R_x |\, R_{1:x-1}, q)$ & probability of $\pi$ adding item $R_x$ given $R_{1:x-1}$ is already placed \\
$\mathcal{I}$ & the available interaction data\\
$c$ & a click pattern: a vector indicating a combination of clicked \\ & and not-clicked items\\
$\sum_c$ & a summation over every possible click pattern \\
$c(d)$ & a function indicating item $d$ was clicked in click pattern $c$\\
$o(d)$ & a function indicating item $d$ was observed at iteration $i$ \\
$x_i$ & the estimate for a single interaction $i$ \\
$f(q_i, R_i, c_i)$ & the method-specific function that converts a single interaction \\ & into an estimate $x_i$\\
$\theta_{\text{rank}(d \mid R)}$  & the observation probability: $P(o(d) = 1 \mid R)$ \\
$\zeta_{d, q}$ & the conditional click probability: $P(c(d) = 1 \mid o(d) = 1, q)$ \\
\bottomrule
\end{tabular}
\end{center}

%% file: 10-onlinecounterltr/main.tex
\chapter{Unifying Online and Counterfactual Learning to Rank}
\label{chapter:06-onlinecounterltr}

\footnote[]{This chapter was submitted as~\citep{oosterhuis2021onlinecounterltr}.
Appendix~\ref{notation:06-onlinecounterltr} gives a reference for the notation used in this chapter.
}

In Chapter~\ref{chapter:06-onlinecountereval}, we introduced the \ac{LogOpt} algorithm that turns a counterfactual ranking evaluation method into an online evaluation method.
Thus, the contributions of Chapter~\ref{chapter:06-onlinecountereval} are a significant step in bridging the divide between online and counterfactual ranking evaluation.
Inspired by this contribution, this chapter will consider whether something similar can be done for the gap between online and counterfactual \ac{LTR}.
Accordingly, in this chapter the following question will be addressed: 
\begin{itemize}
\item[\ref{thesisrq:onlinecounterltr}] Can the counterfactual \ac{LTR} approach be extended to perform highly effective online \ac{LTR}?
\end{itemize}
In contrast with Chapter~\ref{chapter:06-onlinecountereval}, which looked at finding the best logging policy, this chapter will consider a novel counterfactual estimator;
we propose the novel \emph{intervention-aware estimator} for both counterfactual and online \ac{LTR}.
The estimator corrects for the effect of position bias, trust bias, and item-selection bias using corrections based on the behavior of the logging policy and online interventions: changes to the logging policy made during the gathering of click data.
Our experimental results show that, unlike existing counterfactual \ac{LTR} methods, the intervention-aware estimator can greatly benefit from online interventions.
In contrast, existing online methods are hindered without online interventions and thus should not be applied counterfactually.
With the introduction of the intervention-aware estimator, we aim to bridge the online/counterfactual \ac{LTR} division as it is shown to be highly effective in both online and counterfactual scenarios.

\input{10-onlinecounterltr/sections/01-introduction}

\input{10-onlinecounterltr/sections/02-user-interactions}

\input{10-onlinecounterltr/sections/03-background}

\input{10-onlinecounterltr/sections/04-related-work}

\input{10-onlinecounterltr/sections/06-estimator}

\input{10-onlinecounterltr/sections/07-experimental-setup}

\input{10-onlinecounterltr/sections/08-results}

\input{10-onlinecounterltr/sections/09-conclusion}
\begin{subappendices}
\input{10-onlinecounterltr/notation}
\end{subappendices}

%% file: 10-onlinecounterltr/sections/01-introduction.tex
\section{Introduction}
\label{sec:intro}

Ranking systems form the basis for most search and recommendation applications~\citep{liu2009learning}.
As a result, the quality of such systems can greatly impact the user experience, thus it is important that the underlying ranking models perform well.
The \ac{LTR} field considers methods to optimize ranking models. 
Traditionally this was based on expert annotations.
Over the years the limitations of expert annotations have become apparent; some of the most important ones are:
\begin{enumerate*}[label=(\roman*)]
\item they are expensive and time-consuming to acquire~\citep{qin2013introducing, Chapelle2011};
\item in privacy-sensitive settings expert annotation is unethical, e.g., in email or private document search~\citep{wang2018position}; and
\item often expert annotations appear to disagree with actual user preferences~\citep{sanderson2010}.
\end{enumerate*}

User interaction data solves some of the problems with expert annotations:
\begin{enumerate*}[label=(\roman*)]
\item interaction data is virtually free for systems with active users;
\item it does not require experts to look at potentially privacy-sensitive content;
\item interaction data is indicative of users' preferences.
\end{enumerate*}
For these reasons, interest in \ac{LTR} methods that learn from user interactions has increased in recent years.
However, user interactions are a form of implicit feedback and generally also affected by other factors than user preference~\citep{joachims2017accurately}.
Therefore, to be able to reliably learn from interaction data, the effect of factors other than preference has to be corrected for.
In clicks on rankings three prevalent factors are well known:
\begin{enumerate*}[label=(\roman*)]
\item \emph{position bias}: users are less likely to examine, and thus click, lower ranked items~\citep{craswell2008experimental};
\item \emph{item-selection bias}: users cannot click on items that are not displayed~\citep{ovaisi2020correcting, oosterhuis2020topkrankings}; and
\item \emph{trust bias}: because users trust the ranking system, they are more likely to click on highly ranked items that they do not actually prefer~\citep{agarwal2019addressing, joachims2017accurately}.
\end{enumerate*}
As a result of these biases, which ranking system was used to gather clicks can have a substantial impact on the clicks that will be observed.
Current \ac{LTR} methods that learn from clicks can be divided into two families:
\emph{counterfactual approaches}~\citep{joachims2017unbiased} -- that learn from historical data, i.e., clicks that have been logged in the past -- and \emph{online approaches}~\citep{yue2009interactively} -- that can perform interventions, i.e., they can decide what rankings will be shown to users.
Recent work has noticed that some counterfactual methods can be applied as an online method~\citep{jagerman2019comparison}, or vice versa~\citep{zhuang2020counterfactual, ai2020unbiased}.
Nonetheless, every existing method was designed for either the online or counterfactual setting, never both.

In this chapter, we propose a novel estimator for both counterfactual and online \ac{LTR} from clicks: the \emph{intervention-aware estimator}.
The intervention-aware estimator builds on ideas that underlie the latest existing counterfactual methods: the policy-aware estimator~\citep{oosterhuis2020topkrankings} and the affine estimator~\citep{vardasbi2020trust}; and expands them to consider the effect of online interventions.
It does so by considering how the effect of bias is changed by an intervention, and utilizes these differences in its unbiased estimation.
As a result, the intervention-aware estimator is both effective when applied as a counterfactual method, i.e., when learning from historical data, and as an online method where online interventions lead to enormous increases in efficiency.
In our experimental results the intervention-aware estimator is shown to reach state-of-the-art \ac{LTR} performance in both online and counterfactual settings, and it is the only method that reaches top-performance in both settings.

The main contributions of this chapter are:
\begin{enumerate}[align=left,leftmargin=*]
\item A novel intervention-aware estimator that corrects for position bias, trust bias, item-selection bias, and the effect of online interventions.
\item An investigation into the effect of online interventions on state-of-the-art counterfactual and online \ac{LTR} methods.
\end{enumerate}

%% file: 10-onlinecounterltr/sections/02-user-interactions.tex
\section{Interactions with Rankings}
\label{sec:userinteractions}

The theory in this chapter assumes that three forms of interaction bias occur: position bias, item-selection bias, and trust bias.

\emph{Position bias} occurs because users only click an item after examining it, and users are more likely to examine items displayed at higher ranks~\citep{craswell2008experimental}.
Thus the rank (a.k.a.\ position) at which an item is displayed heavily affects the probability of it being clicked.
We model this bias using $P(E = 1 \mid k)$: the probability that an item $d$ displayed at rank $k$ is examined by a user $E$~\citep{wang2018position}.

\emph{Item-selection bias} occurs when some items have a zero probability of being examined in some displayed rankings~\citep{ovaisi2020correcting}.
This can happen because not all items are displayed to the user, or if the ranked list is so long that no user ever considers the entire list.
We model this bias by stating:
\begin{equation}
\exists k, \forall k', \, (k' > k \rightarrow P(E = 1 \mid k') = 0),
\end{equation}
i.e., there exists a rank $k$ such that items ranked lower than $k$ have no chance of being examined.
The distinction between position bias and item-selection bias is important because some methods can only correct for the former if the latter is not present~\citep{oosterhuis2020topkrankings}.

Finally, \emph{trust bias} occurs because users trust the ranking system and, consequently, are more likely to perceive top ranked items as relevant even when they are not~\citep{joachims2017accurately}.
We model this bias using: $P(C = 1 \mid k, R, E)$: the probability of a click conditioned on the displayed rank $k$, the relevance of the item $R$, and examination $E$.

To combine these three forms of bias into a single click model, we follow \citet{agarwal2019addressing} and write:
\begin{equation}
\begin{split}
P(C=1 \mid d, k, q) \hspace{-1.3cm} &
\\
&= P(E = 1 \mid k)\big(P(C = 1 \mid k, R = 0, E=1)P(R=0 \mid d, q)\\
 & \hspace{2.5cm} + P(C = 1 \mid k, R = 1, E=1)P(R=1 \mid d, q)\big),
\end{split}
\label{eq:click-probability}
\end{equation}
where $P(R=1 \mid d, q)$ is the probability that an item $d$ is deemed relevant w.r.t.\ query $q$ by the user.
An analysis on real-world interaction data performed by \citet{agarwal2019addressing}, showed that this model better captures click behavior than models that only capture position bias~\citep{wang2018position} on search services for retrieving cloud-stored files and emails.

To simplify the notation, we follow \citet{vardasbi2020trust} and adopt:
\begin{equation}
\begin{split}
\alpha_k &= P(E = 1 \mid k)\big(P(C = 1 \mid k, R = 1, E=1)
\\ & \phantom{= P(E = 1 \mid k)\big(P} -  P(C = 1 \mid k, R = 0, E=1)\big),
\\
\beta_k &= P(E = 1 \mid k)P(C = 1 \mid k, R = 0, E=1).
\end{split}
\end{equation}
This results in a compact notation for the click probability \eqref{eq:click-probability}:
\begin{equation}
P(C=1 \mid d, k, q) = \alpha_k P(R=1 \mid d, q) + \beta_k.
\label{eq:clickmodelused}
\end{equation}
For a single ranking $y$, let $k$ be the rank at which item $d$ is displayed in $y$; we denote $\alpha_k = \alpha_{d,y}$ and $\beta_k = \beta_{d,y}$.
This allows us to specify the click probability conditioned on a ranking $y$:
\begin{equation}
P(C=1 \mid d, y, q) = \alpha_{d,y} P(R=1 \mid d, q) + \beta_{d,y}.
\label{eq:trustbiasmodel}
\end{equation}
Finally, let $\pi$ be a ranking policy used for logging clicks, where $\pi(y \mid q)$ is the probability of $\pi$ displaying ranking $y$ for query $q$, then the click probability conditioned on $\pi$ is:
\begin{equation}
P(C=1 \mid d, \pi, q) = \sum_{y} \pi(y \mid  q) \mleft( \alpha_{d,y} P(R=1 \mid d, q) + \beta_{d,y} \mright).
\label{eq:clickpolicy}
\end{equation}
The proofs in the remainder of this chapter will assume this model of click behavior.

%% file: 10-onlinecounterltr/sections/03-background.tex
\section{Background}

In this section we cover the basics on \ac{LTR} and counterfactual \ac{LTR}.

\subsection{Learning to rank}
\label{sec:background:LTR}

The field of \ac{LTR} considers methods for optimizing ranking systems w.r.t. ranking metrics.
Most ranking metrics are additive w.r.t.\ documents; let $P(q)$ be the probability that a user-issued query is query $q$, then the metric reward $\mathcal{R}$ commonly has the form:
\begin{equation}
\mathcal{R}(\pi)
=
\sum_{q} P(q) \sum_{d \in D_q} \lambda(d \mid D_q, \pi, q) P(R = 1 \mid d, q).
\label{eq:truereward}
\end{equation}
Here, the $\lambda$ function scores each item $d$ depending on how $\pi$ ranks $d$ when given  the preselected item set $D_q$; $\lambda$ can be chosen to match a desired metric, for instance, the common \ac{DCG} metric~\citep{jarvelin2002cumulated}:
\begin{equation}
\lambda_{\text{DCG}}(d \mid D_q, \pi, q)
=
\sum_y \pi(y \mid q) \mleft( \log_2(\text{rank}(d \mid y) + 1) \mright)^{-1}.
\label{eq:dcglambda}
\end{equation}
Supervised \ac{LTR} methods can optimize $\pi$ to maximize $\mathcal{R}$ if relevances $P(R = 1 \,|\, d, q)$ are known~\citep{wang2018lambdaloss, liu2009learning}. However in practice, finding these relevance values is not straightforward.

\subsection{Counterfactual learning to rank}
\label{sec:background:counterfactual}

Over time, limitations of the supervised \ac{LTR} approach have become apparent. 
Most importantly, finding accurate relevance values $P(R = 1 \mid d, q)$ has proved to be impossible or infeasible in many practical situations~\citep{wang2016learning}.
As a solution, \ac{LTR} methods have been developed that learn from user interactions instead of relevance annotations.
Counterfactual \ac{LTR} concerns approaches that learn from historical interactions.
Let $\mathcal{D}$ be a set of collected interaction data over $T$ timesteps; for each timestep $t$ it contains the user issued query $q_t$, the logging policy $\pi_t$ used to generate the displayed ranking $\bar{y}_t$, and the clicks $c_t$ received on the ranking:
\begin{equation}
\mathcal{D} = \{(\pi_t, q_t, \bar{y}_t, c_t)\}_{t=1}^{T},
\end{equation}
where $c_t(d) \in \{0,1\}$ indicates whether item $d$ was clicked at timestep $t$.
While clicks are indicative of relevance they are also affected by several forms of bias, as discussed in Section~\ref{sec:userinteractions}.

Counterfactual \ac{LTR} methods utilize estimators that correct for bias to unbiasedly estimate the reward of a policy $\pi$.
The prevalent methods introduce a function $\hat{\Delta}$ that transforms a single click signal to correct for bias.
The general estimate of the reward is:
\begin{equation}
\hat{\mathcal{R}}(\pi \mid \mathcal{D}) =  \frac{1}{T} \sum_{t=1}^T \sum_{d \in D_{q_t}}  \lambda(d \mid D_{q_t}, \pi, q) \hat{\Delta}(d \mid \pi_t, q_t, \bar{y}_t, c_t).
\label{eq:estimatedreward}
\end{equation}
We note the important distinction between the policy $\pi$ for which we estimate the reward, and the policy $\pi_t$ that was used to gather interactions at timestep $t$.
During optimization only $\pi$ is changed in order to maximize the estimated reward.

The original \ac{IPS} based estimator introduced by \citet{wang2016learning} and \citet{joachims2017unbiased} weights clicks according to examination probabilities:
\begin{equation}
\hat{\Delta}_{\text{IPS}}(d \mid \bar{y}_t, c_t) =  \frac{c_t(d)}{P(E = 1 \mid \bar{y}_t, d)}.
\end{equation}
This estimator results in unbiased optimization under two requirements.
First, every relevant item must have a non-zero examination probability in all displayed rankings:
\begin{equation}
\forall t,  \forall d \in D_{q_t}\, \left(P(R=1 \mid d, q_t) > 0
\rightarrow P(E = 1 \mid \bar{y}_t, d) > 0\right).
\end{equation}
Second, the click probability conditioned on relevance on examined items should be the same on every rank:
\begin{equation}
\forall k, k'\, \left(P(C \mid k, R, E=1) = P(C \mid k', R, E=1)\right),
\label{eq:notrustreq}
\end{equation}
i.e., no trust bias is present.
These requirements illustrate that this estimator can only correct for position bias, and is biased when item-selection bias or trust bias is present.
For a proof we refer to previous work by \citet{joachims2017unbiased} and \citet{vardasbi2020trust}.

\citet{oosterhuis2020topkrankings} (Chapter~\ref{chapter:04-topk}) adapt the \ac{IPS} approach to correct for item-selection bias as well.
They weight clicks according to examination probabilities conditioned on the logging policy, instead of the single displayed ranking on which a click took place.
This results in the \emph{policy-aware} estimator:
\begin{equation}
\begin{split}
\hat{\Delta}_{\text{aware}}(d \mid \pi_t, q_t, c_t) &= \frac{c_t(d)}{P(E = 1 \mid \pi_t, q_t, d)}
\\
&= 
\frac{c_t(d)}{\sum_{y} \pi(y \mid q_t)P(E = 1 \mid y, d, q_t)}.
\end{split}
\end{equation}
This estimator can be used for unbiased optimization under two assumptions.
First, every relevant item must have a non-zero examination probability under the logging policy:
\begin{equation}
\mbox{}\hspace*{-2mm}
\forall t, \forall d \in D_{q_t} \left(P(R=1 \mid, d, q_t) > 0 \rightarrow
P(E = 1 \,|\, \pi_t, d, q_t) > 0\right).
\hspace*{-2mm}
\end{equation}
Second, no trust bias is present as described in Eq.~\ref{eq:notrustreq}.
Importantly, this first requirement can be met under item-selection bias, since a stochastic ranking policy can always provide every item a non-zero probability of appearing in a top-$k$ ranking.
Thus, even when not all items can be displayed at once, a stochastic policy can provide non-zero examination probabilities to all items.
For a proof of this claim we refer to previous work by \citet{oosterhuis2020topkrankings}.

Lastly, \citet{vardasbi2020trust} prove that \ac{IPS} cannot correct for trust bias.
As an alternative, they introduce an estimator based on affine corrections.
This \emph{affine} estimator penalizes an item displayed at rank $k$ by $\beta_k$ while also reweighting inversely w.r.t.\ $\alpha_k$:
\begin{equation}
\hat{\Delta}_{\text{affine}}(d \mid \bar{y}_t, c_t) =  \frac{c_t(d) - \beta_{d,\bar{y}_t}}{\alpha_{d,\bar{y}_t}}.
\end{equation}
The $\beta$ penalties correct for the number of clicks an item is expected to receive due to its displayed rank, instead of its relevance.
The affine estimator is unbiased under a single assumption, namely that the click probability of every item must be correlated with its relevance in every displayed ranking:
\begin{equation}
\forall t, \forall d \in D_{q_t},  \,
\alpha_{d,\bar{y}_t} \not= 0.
\end{equation}
Thus, while this estimator can correct for position bias and trust bias, it cannot correct for item-selection bias.
For a proof of these claims we refer to previous work by \citet{vardasbi2020trust}.

We note that all of these estimators require knowledge of the position bias ($P(E=1 \mid k)$) or trust bias ($\alpha$ and $\beta$).
A lot of existing work has considered how these values can be inferred accurately~\citep{agarwal2019addressing, wang2018position, fang2019intervention}.
The theory in this chapter assumes that these values are known.

This concludes our description of existing counterfactual estimators on which our method expands.
To summarize, each of these estimators corrects for position bias, one also corrects for item-selection bias, and another also for trust bias. 
Currently, there is no estimator that corrects for all three forms of bias together.

%% file: 10-onlinecounterltr/sections/04-related-work.tex
\section{Related Work}

One of the earliest approaches to \ac{LTR} from clicks was introduced by \citet{Joachims2002}.
It infers pairwise preferences between items from click logs and uses pairwise \ac{LTR} to update an SVM ranking model.
While this approach had some success, in later work \citet{joachims2017unbiased} notes that position bias often incorrectly pushes the pairwise loss to flip the ranking displayed during logging.
To avoid this biased behavior, \citet{joachims2017unbiased} proposed the idea of counterfactual \ac{LTR}, in the spirit of earlier work by \citet{wang2016learning}.
This led to estimators that correct for position bias using \ac{IPS} weighting (see Section~\ref{sec:background:counterfactual}).
This work sparked the field of counterfactual \ac{LTR} which has focused on both capturing interaction biases and optimization methods that can correct for them.
Methods for measuring position bias are based on EM optimization~\citep{wang2018position}, a dual learning objective~\citep{ai2018unbiased}, or randomization~\citep{agarwal2019estimating, fang2019intervention};
for trust bias only an EM-based approach is currently known~\citep{agarwal2019addressing}.
\citet{agarwal2019counterfactual} showed how counterfactual \ac{LTR} can optimize neural networks and \ac{DCG}-like methods through upper-bounding.
\citet{oosterhuis2020topkrankings} introduced an \ac{IPS} estimator that can correct for item-selection bias (see Section~\ref{sec:background:counterfactual} and Chapter~\ref{chapter:04-topk}), while also showing that the LambdaLoss framework~\citep{wang2018lambdaloss} can be applied to counterfactual \ac{LTR} (see Chapter~\ref{chapter:04-topk}).
Lastly, \citet{vardasbi2020trust} proved that \ac{IPS} estimators cannot correct for trust bias and introduced an affine estimator that is capable of doing so (see Section~\ref{sec:background:counterfactual}).
There is currently no known estimator that can correct for position bias, item selection bias, and trust bias simultaneously.

The other paradigm for \ac{LTR} from clicks is online \ac{LTR}~\citep{yue2009interactively}.
The earliest method, \ac{DBGD}, samples variations of a ranking model and compares them using online evaluation~\citep{hofmann2011probabilistic}; if an improvement is recognized the model is updated accordingly.
Most online \ac{LTR} methods have increased the data-efficiency of \ac{DBGD}~\citep{schuth2016mgd, wang2019variance, hofmann2013reusing}; later work found that \ac{DBGD} is not effective at optimizing neural models~\citep{oosterhuis2018differentiable} (Chapter~\ref{chapter:02-pdgd}) and often fails to find the optimal linear-model even in ideal scenarios~\citep{oosterhuis2019optimizing} (Chapter~\ref{chapter:03-oltr-comparison}).
To these limitations, alternative approaches for online \ac{LTR} have been proposed. \ac{PDGD} takes a pairwise approach but weights pairs to correct for position bias~\citep{oosterhuis2018differentiable} (Chapter~\ref{chapter:02-pdgd}).
While \ac{PDGD} was found to be very effective and robust to noise~\citep{jagerman2019comparison, oosterhuis2019optimizing} (Chapter~\ref{chapter:03-oltr-comparison}), it can be proven that its gradient estimation is affected by position bias, thus we do not consider it to be unbiased.
In contrast, \citet{zhuang2020counterfactual} introduced \ac{COLTR}, which takes the \ac{DBGD} approach but uses a form of counterfactual evaluation to compare candidate models.
Despite making use of counterfactual estimation, \citet{zhuang2020counterfactual} propose the method solely for online \ac{LTR}.

Interestingly, with \ac{COLTR} the line between online and counterfactual \ac{LTR} methods starts to blur.
Recent work by \citet{jagerman2019comparison} applied the original counterfactual approach~\citep{joachims2017unbiased} as an online method and found that it lead to improvements.
Furthermore, \citet{ai2020unbiased} noted that with a small adaptation \ac{PDGD} can be applied to historical data.
Although this means that some existing methods can already be applied both online and counterfactually, no method has been found that is the most reliable choice in both scenarios.

%% file: 10-onlinecounterltr/sections/06-estimator.tex
\section{An Estimator Oblivious to Online Interventions}

Before we propose the main contribution of this chapter, the intervention-aware estimator, we will first introduce an estimator that simultaneously corrects for position bias, item-selection bias, and trust bias, without considering the effects of interventions.
Subsequently, the resulting intervention-oblivious estimator will serve as a method to contrast the intervention-aware estimator with.

Section~\ref{sec:background:counterfactual} described how the policy-aware estimator corrects for item-selection bias by taking into account the behavior of the logging policy used to gather clicks~\citep{oosterhuis2020topkrankings}.
Furthermore, Section~\ref{sec:background:counterfactual} also detailed how the affine estimator corrects for trust bias by applying an affine transformation to individual clicks~\citep{vardasbi2020trust}.
We will now show that a single estimator can correct for both item-selection bias and trust bias simultaneously, by combining the approaches of both these existing estimators.

First we note the probability of a click conditioned on a single logging policy $\pi_t$ can be expressed as:
\begin{equation}
\begin{split}
P(C=1 |\, d, \pi_t, q)
&= \sum_{\bar{y}} \pi_t(\bar{y} \mid q) \mleft( \alpha_{d,\bar{y}} P(R=1 \mid d, q) + \beta_{d,\bar{y}} \mright)
\\
&= \mathbbm{E}_{\bar{y}}[\alpha_{d} \,|\, \pi_t , q] P(R=1 \,|\, d, q) + \mathbbm{E}_{\bar{y}}[\beta_{d} \,|\, \pi_t, q].
\end{split}
\label{eq:clicktimeoblivious}
\end{equation}
where the expected values of $\alpha$ and $\beta$ conditioned on $\pi_t$ are:
\begin{equation}
\begin{split}
\mathbbm{E}_{\bar{y}}[\alpha_{d} \mid \pi_t, q] = \sum_{\bar{y}} \pi_t(\bar{y} \mid q) \alpha_{d,\bar{y}},
\\
\mathbbm{E}_{\bar{y}}[\beta_{d} \mid \pi_t, q] = \sum_{\bar{y}} \pi_t(\bar{y} \mid q) \beta_{d,\bar{y}}.
\end{split}
\end{equation}
By reversing Eq.~\ref{eq:clicktimeoblivious} the relevance probability can be obtained from the click probability.
We introduce our \emph{intervention-oblivious estimator}, which applies this transformation to correct for bias:
\begin{equation}
\hat{\Delta}_{\text{IO}}(d \mid q_t, c_t) =  \frac{c_t(d) - \mathbbm{E}_{\bar{y}}[\beta_{d} \mid \pi_t, q_t]}{\mathbbm{E}_{\bar{y}}[\alpha_{d} \mid \pi_t, q_t] }.
\label{eq:timeoblivious}
\end{equation}
The intervention-oblivious estimator brings together the policy-aware and affine estimators: on every click it applies an affine transformation based on the logging policy behavior.
Unlike existing estimators, we can prove that the intervention-oblivious estimator is unbiased w.r.t.\ our assumed click model (Section~\ref{sec:userinteractions}).

\begin{theorem}
\label{theorem:oblivious}
The estimated reward $\hat{\mathcal{R}}$ (Eq.~\ref{eq:estimatedreward}) using the intervention-oblivious estimator (Eq.~\ref{eq:timeoblivious}) is unbiased w.r.t.\ the true reward $\mathcal{R}$ (Eq.~\ref{eq:truereward}) under two assumptions: (1)~our click model (Eq.~\ref{eq:trustbiasmodel}), and (2)~the click probability on every item, conditioned on the logging policies per timestep $\pi_t$, is correlated with relevance:
\begin{equation}
\forall t, \forall d \in D_{q_t}, \quad \mathbbm{E}_{\bar{y}}[\alpha_{d} \mid \pi_t, q_t] \not= 0.
\label{eq:oblivious:correlationassumption}
\end{equation}
\end{theorem}

\begin{proof}
Using Eq.~\ref{eq:clicktimeoblivious} and Eq.~\ref{eq:oblivious:correlationassumption} the relevance probability can be derived from the click probability by:
\begin{equation}
P(R=1 \mid d, q) = \frac{P(C=1 \mid d, \pi_t, q) - \mathbbm{E}_{\bar{y}}[\beta_{d} \mid \pi_t, q]}{\mathbbm{E}_{\bar{y}}[\alpha_{d} \mid \pi_t, q] }.
\label{eq:relderivoblivious}
\end{equation}
Eq.~\ref{eq:relderivoblivious} can be used to show that $\hat{\Delta}_{\text{IO}}$ is an unbiased indicator of relevance:
\begin{equation}
\begin{split}
\mathbbm{E}_{\bar{y},c}\big[\hat{\Delta}_{\text{IO}}(d \mid q_t, c_t) \mid \pi_t \big]
&=
\mathbbm{E}_{\bar{y},c}\mleft[\frac{c_t(d) - \mathbbm{E}_{t,\bar{y}}[\beta_{d} \mid \pi_t, q_t]}{\mathbbm{E}_{\bar{y}}[\alpha_{d} \mid \pi_t, q_t] }\mid \pi_t, q_t \mright]
\\
&=  \frac{ \mathbbm{E}_{\bar{y},c}\mleft[ c_t(d)  \mid \pi_t, q_t\mright] - \mathbbm{E}_{\bar{y}}[\beta_{d} \mid \pi_t, q_t]}{\mathbbm{E}_{\bar{y}}[\alpha_{d} \mid \pi_t, q_t] } 
\\
&=  \frac{P(C=1 \mid d, \pi_t, q_t) - \mathbbm{E}_{\bar{y}}[\beta_{d} \mid \pi_t, q_t]}{\mathbbm{E}_{\bar{y}}[\alpha_{d} \mid \pi_t, q_t] }
 \\
& = P(R=1 \mid d, q_t).
\end{split}
\label{eq:singleproofoblivious}
\end{equation}
Finally, combining Eq.~\ref{eq:truereward} with Eq.~\ref{eq:estimatedreward} and Eq.~\ref{eq:singleproofoblivious} reveals that $\hat{\mathcal{R}}$ based on the intervention-oblivious estimator $\hat{\Delta}_{\text{IO}}$ is unbiased w.r.t.\  $\mathcal{R}$:
\begin{align}
\mathbbm{E}_{t,q,\bar{y},c}
\mleft[ \hat{\mathcal{R}}(\pi \mid \mathcal{D}) \mright] \hspace{-1.5cm}&
\\
& = 
\sum_{q} P(q) \sum_{d \in D_q}  \lambda(d \mid D_{q}, \pi, q) \frac{1}{T} \sum_{t=1}^T \mathbbm{E}_{\bar{y},c}\mleft[ \hat{\Delta}_{\text{IO}}(d \mid c, q) \,|\, \pi_t, q\mright] \nonumber
 \\
& = 
\sum_{q} P(q) \sum_{d \in D_q} \lambda(d \mid D_q, \pi, q) P(R = 1 \mid d, q)
= \mathcal{R}(\pi).
\qedhere
\end{align}
\end{proof}

\subsection{Example with an online intervention}

\begin{figure}[t]
\centering
{\renewcommand{\arraystretch}{0.2}
\begin{tabular}{c c}
\includegraphics[scale=0.72]{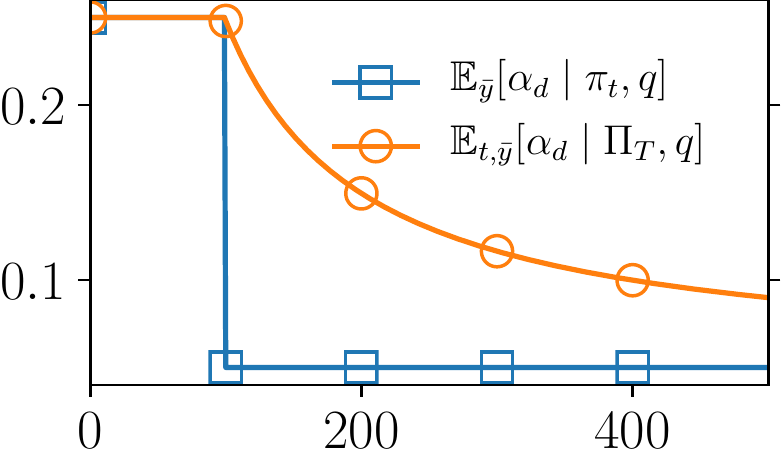} 
&
\includegraphics[scale=0.72]{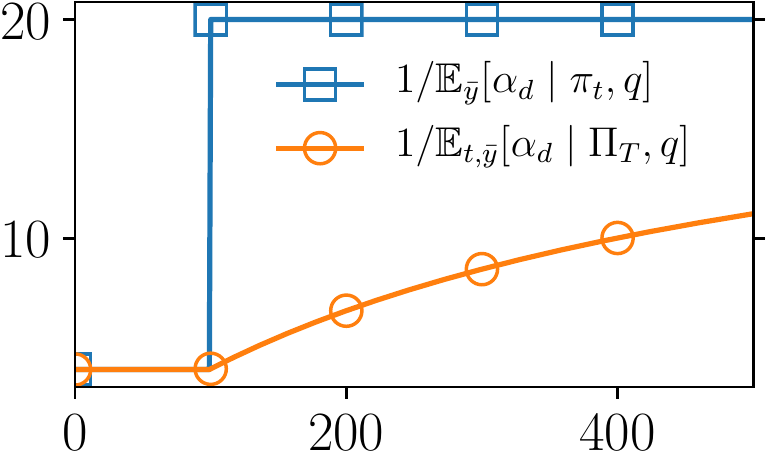}
\\
\small Timestep $T$
&
\small Timestep $T$
\end{tabular}
}
\caption{
Example of an online intervention and the weights used by the intervention-oblivious and intervention-aware estimators for a single item as more data is gathered.
}
\label{fig:intuition}
\end{figure}

Existing estimators for counterfactual \ac{LTR} are designed for a scenario where the logging policy is static:
\begin{equation}
\forall (\pi_t,\pi_{t'})\in \mathcal{D}, \quad \pi_t = \pi_{t'}.
\end{equation}
However, we note that if an online intervention takes place~\citep{jagerman2019comparison}, meaning that the logging policy was updated during the gathering of data:
\begin{equation}
\exists (\pi_t,\pi_{t'})\in \mathcal{D}, \quad \pi_t \not = \pi_{t'},
\end{equation}
the intervention-oblivious estimator is still unbiased.
This was already proven in Theorem~\ref{theorem:oblivious} because its assumptions cover both scenarios where online interventions do and do not take place.

However, the individual corrections of the intervention-oblivious estimator are only based on the single logging policy that was deployed at the timestep of each specific click.
It is completely oblivious to the logging policies applied at different timesteps.
Although this does not lead to bias in its estimation, it does result in unintuitive behavior.
We illustrate this behavior in Figure~\ref{fig:intuition}, here a logging policy that results in $\mathbbm{E}[\alpha_d \mid \pi_t, q] = 0.25$ for an item $d$ is deployed during the first $t \leq 100$ timesteps.
Then an online intervention takes place and the logging policy is updated so that for $t > 100$,  $ \mathbbm{E}[\alpha_d \mid \pi_t, q] = 0.05$.
The intervention-oblivous estimator weights clicks inversely to $\mathbbm{E}[\alpha_d \mid \pi_t]$; so clicks for $t \leq 100$ will be weighted by $1/0.25=4$ and for $t > 100$ by $1/0.05=20$.
Thus, there is a sharp and sudden difference in how clicks are treated before and after $t=100$.
What is unintuitive about this example is that the way clicks are treated after $t=100$ is completely independent of what the situation was before $t=100$.
For instance, consider another item $d'$ where $\forall t, \mathbbm{E}[\alpha_{d'} \mid \pi_t, q] = 0.05$.
If both $d$ and $d'$ are clicked on timestep $t = 101$, these clicks would both be weighted by $20$, despite the fact that $d$ has so far been treated completely different than $d'$.
One would expect that in such a case the click on $d$ should be weighted less, to compensate for the high $\mathbbm{E}[\alpha_d \mid \pi_t, q]$ it had in the first 100 timesteps.
The question is whether such behavior can be incorporated in an estimator without introducing bias.

\section{The Intervention-Aware Estimator}

Our goal for the intervention-aware estimator is to find an estimator whose individual corrections are not only based on single logging policies, but instead consider the entire collection of logging policies used to gather the data $\mathcal{D}$.
Importantly, this estimator should also be unbiased w.r.t.\ position bias, item-selection bias and trust bias.

For ease of notation, we use $\Pi_T$ for the set of policies that gathered the data in $\mathcal{D}$:
$\Pi_T = \{ \pi_1, \pi_2, \ldots, \pi_T\}$.
The probability of a click can be conditioned on this set:
\begin{equation}
\begin{split}
P(C=1 \mid d, \Pi_T, q) 
& = \frac{1}{T}\sum_{t=1}^T \sum_{\bar{y}} \pi_t(\bar{y} \mid q) \mleft( \alpha_{d,\bar{y}} P(R=1 \mid d, q) + \beta_{d,\bar{y}} \mright)
\\
&= \mathbbm{E}_{t,\bar{y}}[\alpha_{d} \mid \Pi_T , q] P(R=1 \mid d, q) + \mathbbm{E}_{t,\bar{y}}[\beta_{d} \mid \Pi_T, q],
\end{split}
\label{eq:clicktimeaware}
\end{equation}
where the expected values of $\alpha$ and $\beta$ conditioned on $\Pi_T$ are:
\begin{equation}
\begin{split}
\mathbbm{E}_{t,\bar{y}}[\alpha_{d} \mid \Pi_T, q] =\frac{1}{T} \sum_{t=1}^T \sum_{\bar{y}} \pi_t(\bar{y} \mid q) \alpha_{d,\bar{y}},
\\
\mathbbm{E}_{t,\bar{y}}[\beta_{d} \mid \Pi_T, q] =\frac{1}{T} \sum_{t=1}^T \sum_{\bar{y}} \pi_t(\bar{y} \mid q) \beta_{d,\bar{y}}.
\end{split}
\end{equation}
Thus $P(C=1 \mid d, \Pi_T, q)$ gives us the probability of a click given that any policy from $\Pi_T$ could be deployed.
We propose our \emph{intervention-aware estimator} that corrects for bias using the expectations conditioned on $\Pi_T$:
\begin{equation}
\hat{\Delta}_{\text{IA}}(d \mid q_t, c_t) =  \frac{c_t(d) - \mathbbm{E}_{t,\bar{y}}[\beta_{d} \mid \Pi_T, q_t]}{\mathbbm{E}_{t,\bar{y}}[\alpha_{d} \mid \Pi_T, q_t] }.
\label{eq:timeaware}
\end{equation}
The salient difference with the intervention-oblivious estimator is that the expectations are conditioned on $\Pi_T$, all logging policies in $\mathcal{D}$, instead of an individual logging policy $\pi_t$.
While the difference with the intervention-oblivious estimator seems small, our experimental results show that the differences in performance are actually quite sizeable.
Lastly, we note that when no interventions take place the intervention-oblivious estimator and intervention-aware estimators are equivalent.
Because the intervention-aware estimator is the only existing counterfactual \ac{LTR} estimator whose corrections are influenced by online interventions, we consider it to be a step that helps to bridge the gap between counterfactual and online \ac{LTR}.

Before we revisit our online intervention example with our novel intervention-aware estimator, we prove that it is unbiased w.r.t.\ our assumed click model (Section~\ref{sec:userinteractions}).

\begin{theorem}
The estimated reward $\hat{\mathcal{R}}$ (Eq.~\ref{eq:estimatedreward}) using the intervention-aware estimator (Eq.~\ref{eq:timeaware}) is unbiased w.r.t.\ the true reward $\mathcal{R}$ (Eq.~\ref{eq:truereward}) under two assumptions: (1)~our click model (Eq.~\ref{eq:trustbiasmodel}), and (2)~the click probability on every item, conditioned on the set of logging policies $\Pi_T$, is correlated with relevance:
\begin{equation}
\forall q, \forall d \in D_q, \quad  \mathbbm{E}_{t,\bar{y}}[\alpha_{d} \mid \Pi_T, q] \not= 0.
\label{eq:correlationassumption}
\end{equation}
\end{theorem}

\begin{proof}
Using Eq.~\ref{eq:clicktimeaware} and Eq.~\ref{eq:correlationassumption} the relevance probability can be derived from the click probability by:
\begin{equation}
P(R=1 \mid d, q) = \frac{P(C=1 \mid d, \Pi_T, q) - \mathbbm{E}_{t,\bar{y}}[\beta_{d} \mid \Pi_T, q]}{\mathbbm{E}_{t,\bar{y}}[\alpha_{d} \mid \Pi_T, q] }.
\label{eq:relderiv}
\end{equation}
Eq.~\ref{eq:relderiv} can be used to show that $\hat{\Delta}_{\text{IA}}$ is an unbiased indicator of relevance:
\begin{equation}
\begin{split}
\mathbbm{E}_{t,\bar{y},c}\big[\hat{\Delta}_{\text{IA}}(d \mid q_t, c_t) \mid \Pi_T \big]
&=
\mathbbm{E}_{t,\bar{y},c}\mleft[\frac{c_t(d) - \mathbbm{E}_{t,\bar{y}}[\beta_{d} \mid \Pi_T, q_t]}{\mathbbm{E}_{t,\bar{y}}[\alpha_{d} \mid \Pi_T, q_t] }\mid \Pi_T, q_t \mright]
\\
&=  \frac{ \mathbbm{E}_{t,\bar{y},c}\mleft[ c_t(d)  \mid \Pi_T, q_t\mright] - \mathbbm{E}_{t,\bar{y}}[\beta_{d} \mid \Pi_T, q_t]}{\mathbbm{E}_{t,\bar{y}}[\alpha_{d} \mid \Pi_T, q_t] }
 \\
&=  \frac{P(C=1 \mid d, \Pi_T, q_t) - \mathbbm{E}_{t,\bar{y}}[\beta_{d} \mid \Pi_T, q_t]}{\mathbbm{E}_{t,\bar{y}}[\alpha_{d} \mid \Pi_T, q_t] } 
 \\
&= P(R=1 \mid d, q_t).
\end{split}
\label{eq:singleproof}
\end{equation}
Finally, combining Eq.~\ref{eq:singleproof} with Eq.~\ref{eq:estimatedreward} and Eq.~\ref{eq:truereward} reveals that $\hat{\mathcal{R}}$ based on the intervention-aware estimator $\hat{\Delta}_{\text{IA}}$ is unbiased w.r.t.\  $\mathcal{R}$:
\begin{align}
\mathbbm{E}_{t,q,\bar{y},c}
\mleft[ \hat{\mathcal{R}}(\pi \mid \mathcal{D}) \mright] \hspace{-1cm} &
 \\
& = 
\sum_{q} P(q) \sum_{d \in D_q}  \lambda(d \mid D_{q}, \pi, q) \mathbbm{E}_{t,\bar{y},c}\mleft[ \hat{\Delta}_{\text{IA}}(d \mid c, q) \mid \Pi_T, q\mright]
\nonumber \\
& = 
\sum_{q} P(q) \sum_{d \in D_q} \lambda(d \mid D_q, \pi, q) P(R = 1 \mid d, q)
\nonumber \\
& = \mathcal{R}(\pi). \qedhere
\end{align}
\end{proof}

\subsection{Online intervention example revisited}

We will now revisit the example in Figure~\ref{fig:intuition}, but this time consider how the intervention-aware estimator treats item $d$.
Unlike the intervention-oblivious estimator, clicks are weighted by $\mathbbm{E}[\alpha_d \mid \Pi_T]$ which means that the exact timestep $t$ of a click does not matter, as long as $t < T$.
Furthermore, the weight of a click can change as the total number of timesteps $T$ increases.
In other words, as more data is gathered, the intervention-aware estimator retroactively updates the weights of all clicks previously gathered.

We see that this behavior avoids the sharp difference in weights of clicks occurring before the intervention $t \leq 100$ and after $t > 100$.
For instance, for a click on $d$ occuring at $t=101$ while $T=400$, results in $\mathbbm{E}[\alpha_d \mid \Pi_T] = 0.1$ and thus a weight of $1/0.1 = 10$.
This is much lower than the intervention-oblivious weight of $1/0.05=20$, because the intervention-aware estimator is also considering the initial period where $\mathbbm{E}[\alpha_{d} \mid \pi_t, q]$ was high.
Thus we see that the intervention-aware estimator has the behavior we intuitively expected: it weights clicks based on how the item was treated throughout all timesteps.
In this example, it leads weights considerably smaller than those used by the intervention-oblivious estimator.
In \ac{IPS} estimators, low propensity weights are known to lead to high variance~\citep{joachims2017unbiased}, thus we may expect that the intervention-aware estimator reduces variance in this example.

\subsection{An online and counterfactual approach}
\label{sec:onlinecounterapproach}

While the intervention-aware estimator takes into account the effect of interventions, it does not prescribe what interventions should take place.
In fact, it will work with any interventions that result in Eq.~\ref{eq:correlationassumption} being true, including the situation where no intervention takes place at all.
For clarity, we will describe the intervention approach we applied during our experiments here.
Algorithm~\ref{alg:online} displays our online/counterfactual approach.
As input it requires a starting policy ($\pi_0$), a choice for $\lambda$, the $\alpha$ and $\beta$ parameters, a set of intervention timesteps ($\Phi$), and the final timestep $T$.

The algorithm starts by initializing an empty set to store the gathered interaction data (Line~\ref{line:initdata}) and initializes the logging policy with the provided starting policy $\pi_0$.
Then for each timestep $i$ in $\Phi$ the dataset is expanded using the current logging policy so that $|\mathcal{D}| = i$ (Line~\ref{line:interventiongather}).
In other words, for $i - |\mathcal{D}|$ timesteps $\pi$ is used to display rankings to user-issued queries, and the resulting interactions are added to $\mathcal{D}$.
Then a policy is optimized using the available data in $\mathcal{D}$ which becomes the new logging policy.
For this optimization, we split the available data in training and validation partitions in order to do early stopping to prevent overfitting.
We use stochastic gradient descent where we use $\pi_0$ as the initial model; this practice is based on the assumption that $\pi_0$ has a better performance than a randomly initialized model.
Thus, during optimization, gradient calculation uses the intervention-aware estimator on the training partition of $\mathcal{D}$, and after each epoch, optimization is stopped if the intervention-aware estimator using the validation partition of $\mathcal{D}$ suspects overfitting.
Each iteration results in an intervention as the resulting policy replaces the logging policy, and thus changes the way future data is logged.
After iterating over $\Phi$ is completed, more data is gathered so that $|\mathcal{D}| = T$ and optimization is performed once more.
The final policy is the end result of the procedure.

We note that, depending on $\Phi$, our approach can be either online, counterfactual, or somewhere in between.
If $\Phi = \emptyset$ the approach is fully counterfactual since all data is gathered using the static $\pi_0$.
Conversely, if $\Phi = \{1,2,3,\ldots,T\}$ it is fully online since at every timestep the logging policy is updated.
In practice, we expect a fully online procedure to be infeasible as it is computationally expensive and user queries may be issued faster than optimization can be performed.
In our experiments we will investigate the effect of the number of interventions on the approach's performance.

%% file: 10-onlinecounterltr/sections/07-experimental-setup.tex
\section{Experimental Setup}

Our experiments aim to answer the following research questions:
\begin{enumerate}[align=left, label={\bf RQ\arabic*},leftmargin=*]
\item Does the intervention-aware estimator lead to higher performance than existing counterfactual \ac{LTR} estimators when online interventions take place?
\item Does the intervention-aware estimator lead to performance comparable with existing online \ac{LTR} methods?
\end{enumerate}
We use the semi-synthetic experimental setup that is common in existing work on both online \ac{LTR}~\citep{oosterhuis2018differentiable, oosterhuis2019optimizing, hofmann2013reusing, zhuang2020counterfactual} and counterfactual \ac{LTR}~\citep{vardasbi2020trust, ovaisi2020correcting, joachims2017unbiased}.
In this setup, queries and documents are sampled from a dataset based on commercial search logs, while user interactions and rankings are simulated using probabilistic click models.
The advantage of this setup is that it allows us to investigate the effects of online interventions on a large scale while also being easy to reproduce by researchers without access to live ranking systems.

\begin{algorithm}[t]
\caption{Our Online/Counterfactual \ac{LTR} Approach} 
\label{alg:online}
\begin{algorithmic}[1]
\STATE \textbf{Input}: Starting policy: $\pi_0$; Metric weight function: $\lambda$;\\
\qquad\quad
Inferred bias parameters: $\alpha$ and $\beta$;\\
\qquad\quad
Interventions steps: $\Phi$; End-time: $T$.
\STATE $\mathcal{D} \leftarrow \{\}$ \hfill \textit{\small // initialize data container} \label{line:initdata}
\STATE $\pi \gets \pi_0$ \hfill \textit{\small // initialize logging policy}  \label{line:initpolicy}
\FOR{$i \in \Phi$}
\STATE $\mathcal{D} \leftarrow \mathcal{D} \cup \text{gather}(\pi, i-|\mathcal{D}|)$  \hfill \textit{\small // observe $i-|\mathcal{D}|$ timesteps}  \label{line:interventiongather}
\STATE $\pi \leftarrow \text{optimize}(\mathcal{D}, \alpha, \beta, \pi_0)$ \hfill \textit{\small // optimize based on available data} \label{line:interventionoptimize}
\ENDFOR
\STATE $\mathcal{D} \leftarrow \mathcal{D} \cup \text{gather}(\pi, T-|\mathcal{D}|)$  \hfill \textit{\small // expand data to $T$} \label{line:finalgather}
\STATE $\pi \leftarrow \text{optimize}(\mathcal{D}, \alpha, \beta, \pi_0)$ \hfill \textit{\small // optimize based on final data} \label{line:finaloptimize}
\RETURN $\pi$ 
\end{algorithmic}
\end{algorithm}

We use the publicly-available Yahoo Webscope dataset~\citep{Chapelle2011}, which consists of \numprint{29921} queries with, on average, 24 documents preselected per query.
Query-document pairs are represented by 700 features and five-grade relevance annotations ranging from not relevant (0) to perfectly relevant (4). The queries are divided into training, validation and test partitions.

At each timestep, we simulate a user-issued query by uniformly sampling from the training and validation partitions.
Subsequently, the preselected documents are ranked according to the logging policy, and user interactions are simulated on the top-5 of the ranking using a probabilistic click model.
We apply Eq.~\ref{eq:clickmodelused} with $\alpha = [0.35, 0.53, 0.55, 0.54, 0.52]$ and $\beta = [0.65, 0.26, 0.15, 0.11, 0.08]$; the relevance probabilities are based on the annotations from the dataset: $P(R = 1 \mid d,q) = 0.25 \cdot \text{relevance\_label}(d,q)$.
The values of $\alpha$ and $\beta$ were chosen based on those reported by \citet{agarwal2019addressing} who inferred them from real-world user behavior.
In doing so, we aim to emulate a setting where realistic levels of position bias, item-selection bias, and trust bias are present.

All counterfactual methods use the approach described in Section~\ref{sec:onlinecounterapproach}.
To simulate a production ranker policy, we use supervised \ac{LTR} to train a ranking model on 1\% of the training partition~\citep{joachims2017unbiased}.
The resulting production ranker has much better performance than a randomly initialized model, yet still leaves room for improvement.
We use the production ranker as the initial logging policy.
The size of $\Phi$ (the intervention timesteps) varies per run, and the timesteps in $\Phi$ are evenly spread on an exponential scale.
All ranking models are neural networks with two hidden layers, each containing 32 hidden units with sigmoid activations.
Gradients are calculated using a Monte-Carlo method following \citet{oosterhuis2020taking} (Chapter~\ref{chapter:06-onlinecountereval}).
All policies apply a softmax to the document scores produced by the ranking models to obtain a probability distribution over documents.
Clipping is only applied on the training clicks, denominators of any estimator are clipped by $10/\sqrt{|\mathcal{D}|}$ to reduce variance.
Early stopping is applied based on counterfactual estimates of the loss using (unclipped) validation clicks.

The following methods are compared:
\begin{enumerate*}[label=(\roman*)]
\item The intervention-aware estimator.
\item The intervention-oblivious estimator.
\item The policy-aware estimator~\citep{oosterhuis2020topkrankings} (Chapter~\ref{chapter:04-topk}).
\item The affine estimator~\citep{vardasbi2020trust}.
\item \ac{PDGD}~\citep{oosterhuis2018differentiable} (Chapter~\ref{chapter:02-pdgd}), we apply \ac{PDGD} both online and as a counterfactual method.
As noted by \citet{ai2020unbiased}, this can be done by separating the logging models from the learned model and, basing the debiasing weights on the logging function.
\item Biased \ac{PDGD}, identical to \ac{PDGD} except that we do not apply the debiasing weights.
\item \ac{COLTR}~\citep{zhuang2020counterfactual}.
\end{enumerate*}
We compute the \ac{NDCG} of both the logging policy and of a policy trained on all available data.
Every reported result is the average of 20 independent runs, figures plot the mean, shaded areas indicate the standard deviation.

%% file: 10-onlinecounterltr/sections/08-results.tex
\begin{figure}[t]
\centering
\begin{tabular}{r c}
\rotatebox[origin=lt]{90}{\small \hspace{4.3em} NDCG} &
\includegraphics[scale=0.58]{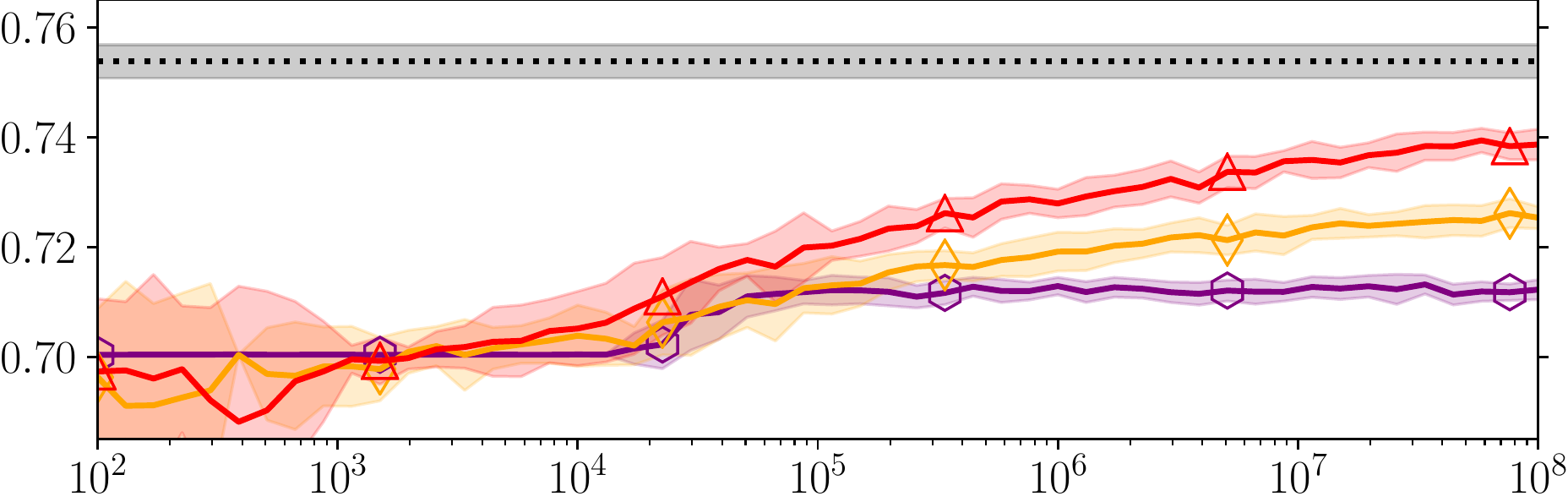}
\\
\rotatebox[origin=lt]{90}{\small \hspace{4.3em} NDCG} &
\includegraphics[scale=0.58]{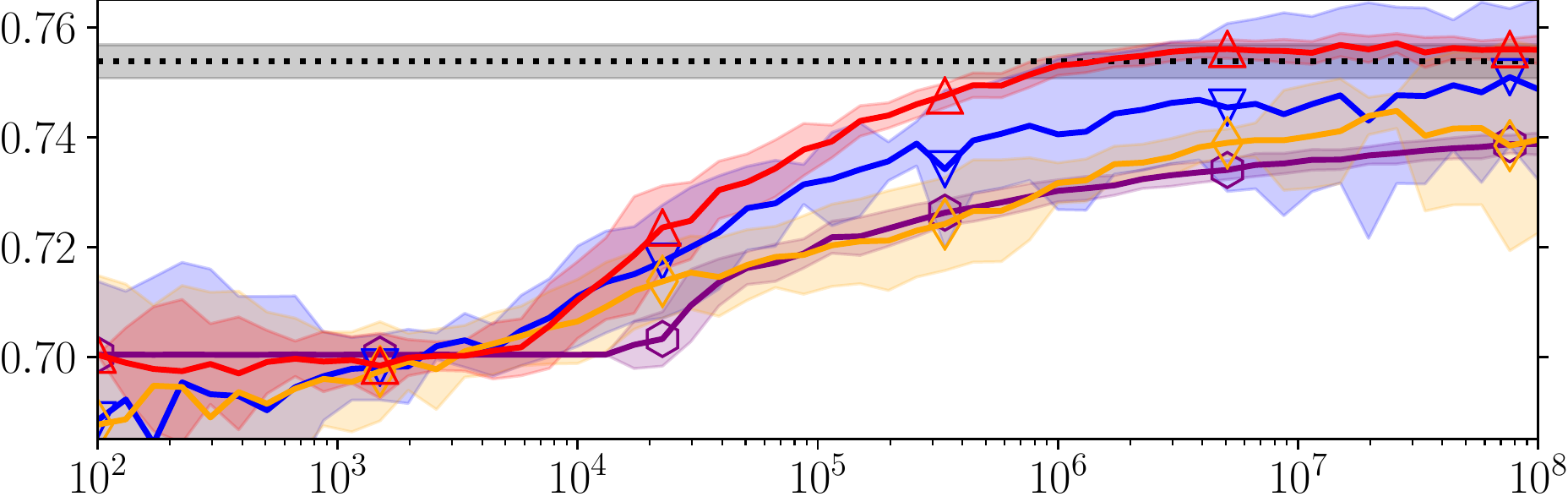}
\\
& \footnotesize \hspace{1em} Number of Logged Queries
\end{tabular}
\includegraphics[scale=0.53]{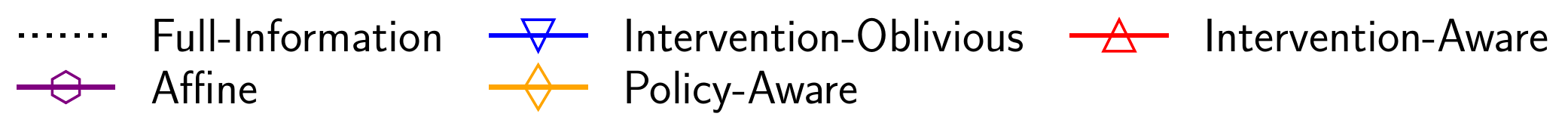}
\caption{
Comparison of counterfactual \ac{LTR} estimators. Top: Counterfactual runs (no interventions); Bottom: Online runs (50 interventions).
}
\label{fig:estimators}
\end{figure}

\section{Results and Discussion}

\subsection{Comparison with counterfactual \acs{LTR}}

To answer the first research question: \emph{whether the intervention-aware estimator leads to higher performance than existing counterfactual LTR estimators when online interventions take place},
we consider Figure~\ref{fig:estimators} which displays the performance of \ac{LTR} using different counterfactual estimators.

First we consider the top of Figure~\ref{fig:estimators} which displays performance in the counterfactual setting where the logging policy is static.\footnote{Since under a static logging policy the intervention-aware and the intervention-oblivious estimators are equivalent, our conclusions apply to both in this setting.}
We clearly see that the affine estimator converges at a suboptimal point of convergence, a strong indication of bias.
The most probable cause is that the affine estimator is heavily affected by the presence of item-selection bias.
In contrast, neither the policy-aware estimator nor the intervention-aware estimator have converged after $10^8$ queries.
However, very clearly the intervention-aware estimator quickly reaches a higher performance.
While the theory guarantees that it will converge at the optimal performance, we were unable to observe the number of queries it requires to do so.
From the result in the counterfactual setting, we conclude that by correcting for position-bias, trust-bias, and item-selection bias the intervention-aware estimator already performs better without online interventions.

Second, we turn to the bottom of Figure~\ref{fig:estimators} which considers the online setting where the estimators perform 50 online interventions during logging.
We see that online interventions have a positive effect on all estimators; leading to a higher performance for the affine and policy-aware estimators as well.
However, interventions also introduce an enormous amount of variance for the policy-aware and intervention-oblivious estimators.
In stark contrast, the amount of variance of the intervention-aware estimator hardly increases while it learns much faster than the other estimators.

Thus we answer the first research question positively:
the inter\-vention-aware estimator leads to higher performance than existing estimators, moreover, its data-efficiency becomes even greater when online interventions take place.

\begin{figure}[t]
\centering
\begin{tabular}{r c}
\rotatebox[origin=lt]{90}{\small \hspace{4.3em} NDCG} &
\includegraphics[scale=0.58]{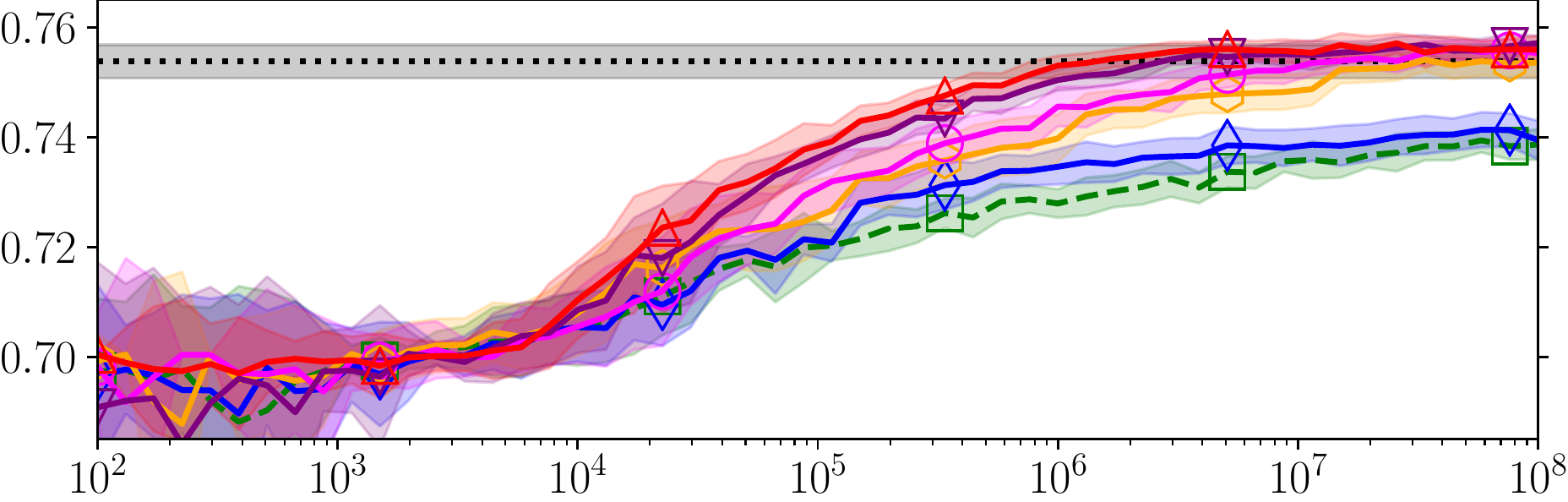}\\
\rotatebox[origin=lt]{90}{\small \hspace{0.1em} Logging Policy NDCG} &
\includegraphics[scale=0.58]{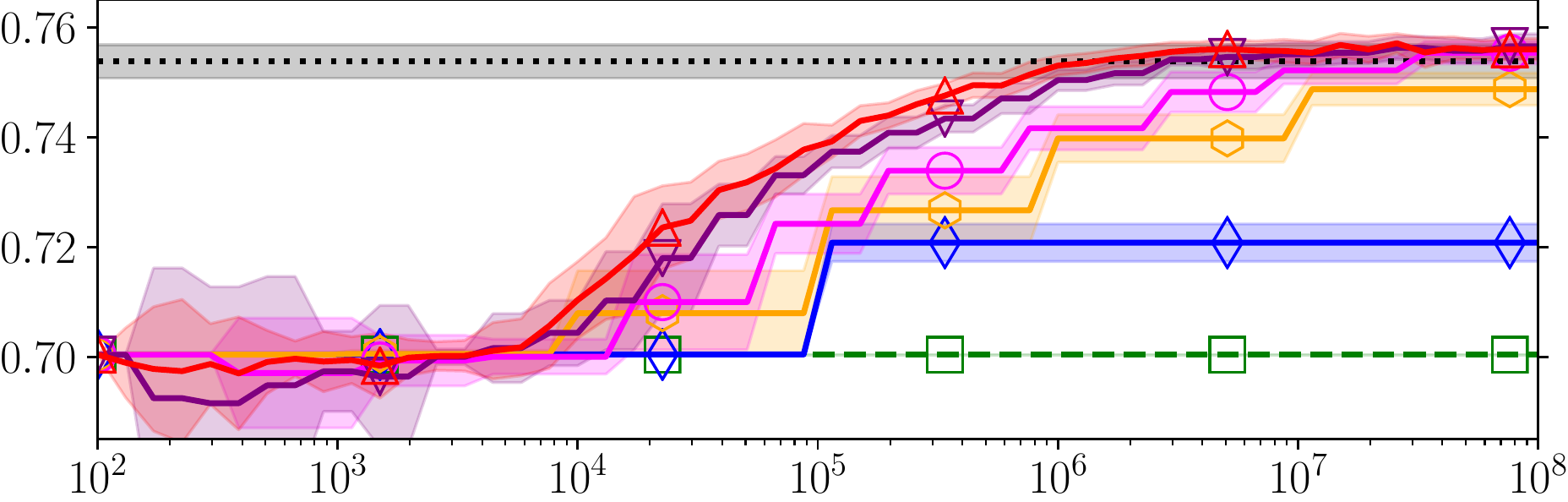}
\\
& \footnotesize \hspace{1em} Number of Logged Queries
\end{tabular}
\includegraphics[scale=0.52]{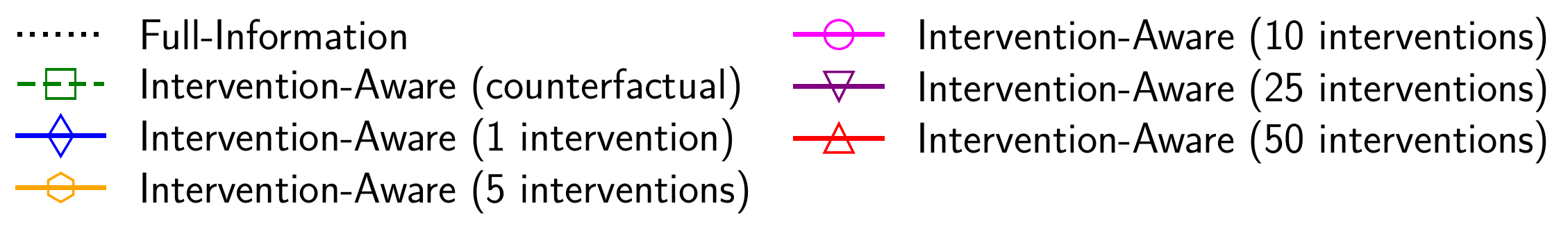}
\caption{
Effect of online interventions on \ac{LTR} with the intervention-aware estimator.
}
\label{fig:interventions}
\end{figure}

\subsection{Effect of interventions}

To better understand how much the intervention-aware estimator benefits from online interventions, we compared its performance under varying numbers of interventions in Figure~\ref{fig:interventions}.
It shows both the performance of the resulting model when training from the logged data (top), as the performance of the logging policy which reveals when interventions take place (bottom).
When comparing both graphs, we see that interventions lead to noticeable immediate improvements in data-efficiency.
For instance, when only 5 interventions take place the intervention-aware estimator needs more than 20 times the amount of data to reach optimal performance as with 50 interventions. 
Despite these speedups there are no large increases in variance.
From these observations, we conclude that the intervention-aware estimator can effectively and reliably utilize the effect of online interventions for optimization, leading to enormous increases in data-efficiency.

\begin{figure}[t]
\centering
\begin{tabular}{r c}
\rotatebox[origin=lt]{90}{\small \hspace{4.3em} NDCG} &
\includegraphics[scale=0.55]{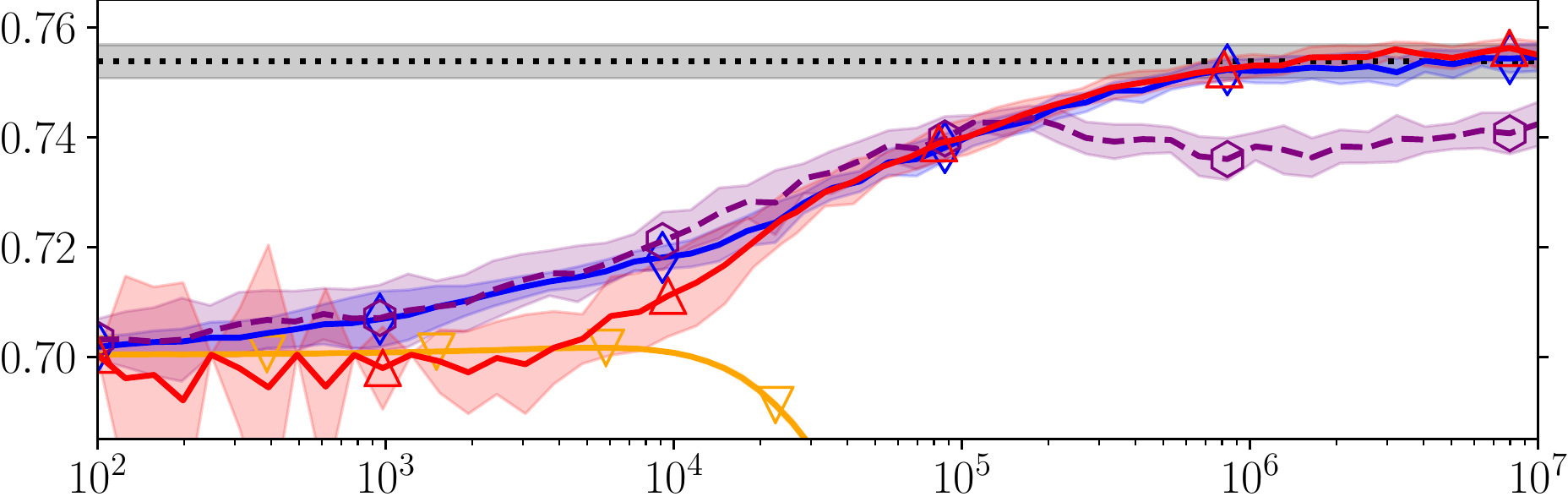}\\
\rotatebox[origin=lt]{90}{\small \hspace{0.1em} Logging Policy NDCG} &
\includegraphics[scale=0.55]{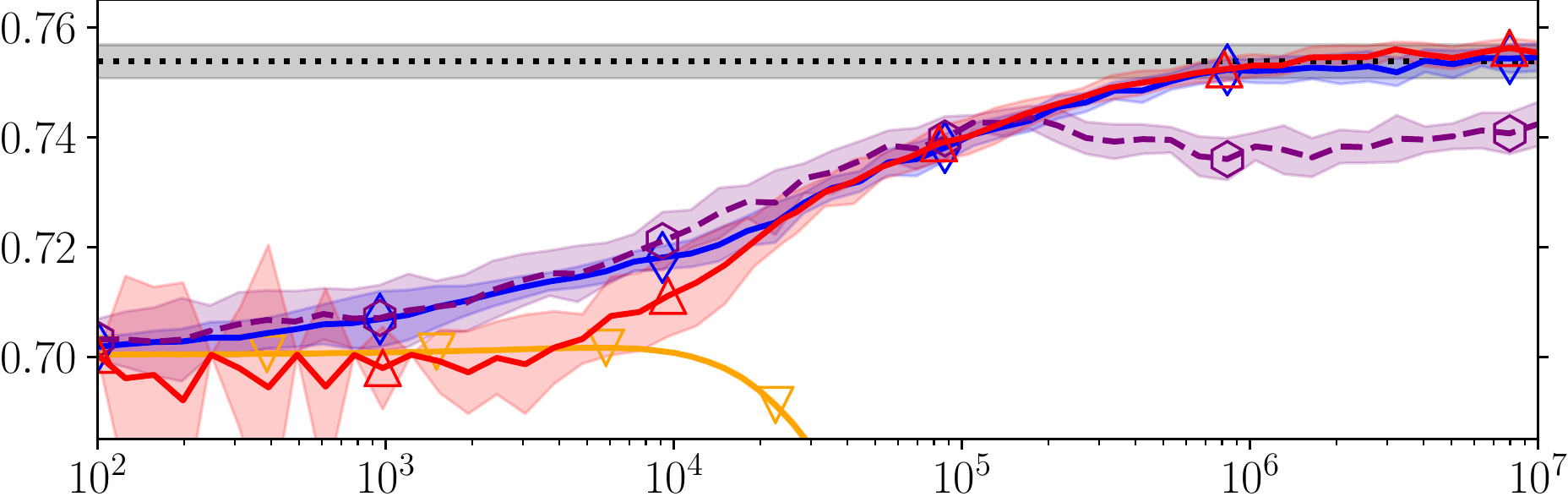}
\\
& \footnotesize \hspace{1em} Number of Logged Queries
\end{tabular}
\includegraphics[scale=0.52]{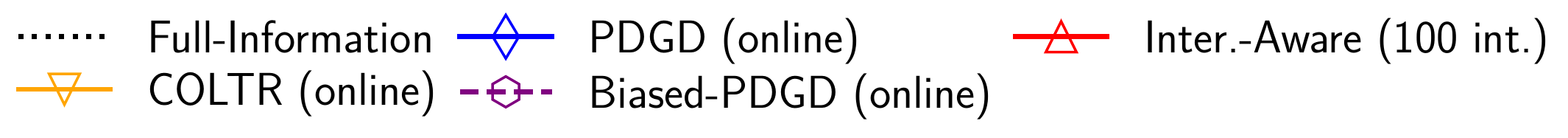} 
\caption{
Comparison with online \ac{LTR} methods.
}
\label{fig:pdgdcomp}
\end{figure}

\subsection{Comparison with online \acs{LTR}}

In order to answer the second research question: \emph{whether the inter\-vention-aware estimator leads to performance comparable with existing online \ac{LTR} methods},
we consider Figure~\ref{fig:pdgdcomp} which displays the performance of two online \ac{LTR} methods: \ac{PDGD} and \ac{COLTR} and the intervention-aware estimator with 100 online interventions.

First, we notice that \ac{COLTR} is unable to outperform its initial policy, moreover, we see its performance drop as the number of iterations increase.
We were unable to find hyper-parameters for \ac{COLTR} where this did not occur.
It seems likely that \ac{COLTR} is unable to deal with trust-bias, thus causing this poor performance.
However, we note that \citet{zhuang2020counterfactual} already show \ac{COLTR} performs poorly when no bias or noise is present, suggesting that it is perhaps an unstable method overall.

Second, we see that the difference between \ac{PDGD} and the inter\-vention-aware estimator becomes negligible after $2\cdot10^4$ queries.
Despite \ac{PDGD} running fully online, and the intervention-aware estimator only performing 100 interventions in total.
We do note that \ac{PDGD} initially outperforms the intervention-aware estimator, thus it appears that \ac{PDGD} works better with low numbers of interactions.
Additionally, we should also consider the difference in overhead: while \ac{PDGD} requires an infrastructure that allows for fully online learning, the intervention-aware estimator only requires 100 moments of intervention, yet has comparable performance after a short initial period.
By comparing Figure~\ref{fig:pdgdcomp} to Figure~\ref{fig:estimators}, we see that the intervention-aware estimator is the first counterfactual \ac{LTR} estimator that leads to stable performance while being comparably efficient with online \ac{LTR} methods.

Thus we answer the second research question positively:
besides an initial period of lower performance, the intervention-aware estimator has comparable performance to online \ac{LTR} and only requires 100 online interventions to do so.
To the best of our knowledge, it is the first counterfactual \ac{LTR} method that can achieve this feat.

\begin{figure}[t]
\centering
\begin{tabular}{r c}
\rotatebox[origin=lt]{90}{\small \hspace{4.3em} NDCG} &
\includegraphics[scale=0.58]{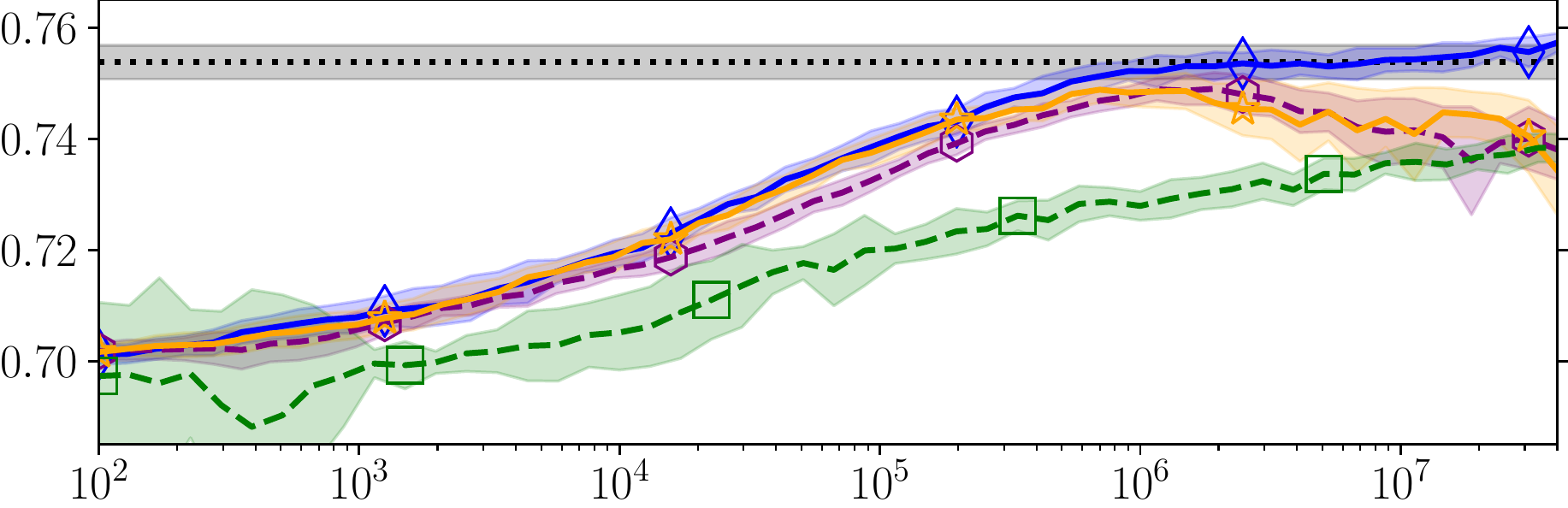}\\
\rotatebox[origin=lt]{90}{\small \hspace{0.1em} Logging Policy NDCG} &
\includegraphics[scale=0.58]{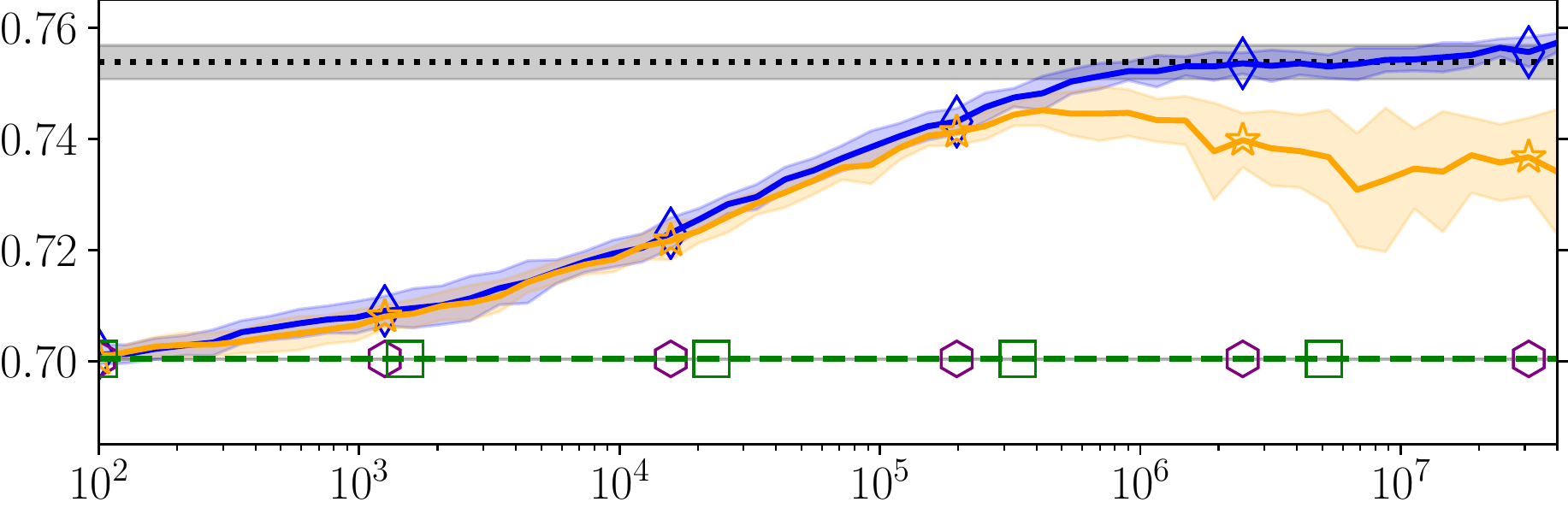}
\\
& \footnotesize \hspace{1em} Number of Logged Queries
\end{tabular}
\includegraphics[scale=0.55]{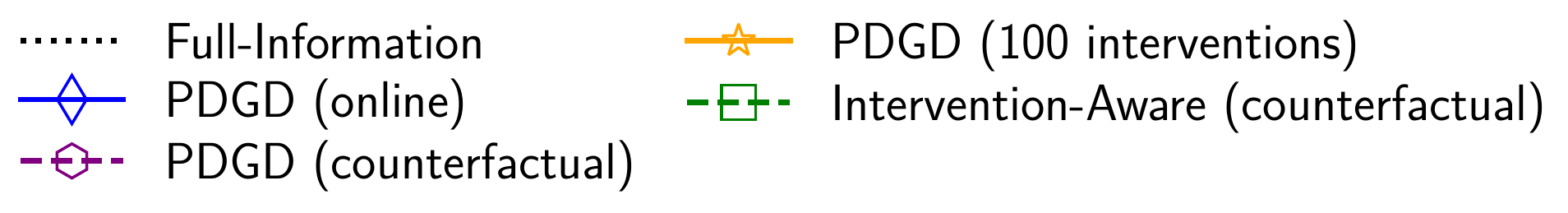} 
\caption{
Effect of online interventions on \ac{PDGD}.
}
\label{fig:pdgdcount}
\end{figure}

\subsection{Understanding the effectiveness of \acs{PDGD}}
Now that we concluded that the intervention-aware estimator reaches performance comparable to \ac{PDGD} when enough online interventions take place, the opposite question seems equally interesting:
\emph{Does \ac{PDGD} applied counterfactually provide performance comparable to existing counterfactual \ac{LTR} methods?}

To answer this question, we ran \ac{PDGD} in a counterfactual way following \citet{ai2020unbiased}, both fully counterfactual and with only 100 interventions.
The results of these runs are displayed in Figure~\ref{fig:pdgdcount}.
Quite surprisingly, the performance of \ac{PDGD} ran counterfactually and with 100 interventions, reaches much higher performance than the intervention-aware estimator without interventions.
However, after a peak in performance around $10^6$ queries, the \ac{PDGD} performance starts to drop.
This drop cannot be attributed to overfitting, since online \ac{PDGD} does not show the same behavior.
Therefore, we must conclude that \ac{PDGD} is biased when not ran fully online.
This conclusion does not contradict the existing theory, since in Chapter~\ref{chapter:02-pdgd} we only proved it is unbiased w.r.t.\ \emph{pairwise} preferences.
In other words, \ac{PDGD} is not proven to unbiasedly optimize a ranking metric, and therefore also not proven to converge on the optimal model.
This drop is particularly unsettling because \ac{PDGD} is a continuous learning algorithm: there is no known early stopping method for \ac{PDGD}.
Yet these results show there is a great risk in running \ac{PDGD} for too many iterations if it is not applied fully online.
To answer our \ac{PDGD} question: although \ac{PDGD} reaches high performance when run counterfactually and appears to have great data-efficiency initially, it appears to converge at a suboptimal biased model.
Thus we cannot conclude that \ac{PDGD} is a reliable method for counterfactual \ac{LTR}.

To better understand \ac{PDGD}, we removed its debiasing weights resulting in the performance shown in Figure~\ref{fig:pdgdcomp} (Biased-PDGD).
Clearly, \ac{PDGD} needs these weights to reach optimal performance. 
Similarly, from Figure~\ref{fig:pdgdcount} we see it also needs to be run fully online.
This makes the choice between the intervention-aware estimator and \ac{PDGD} complicated: on the one hand, \ac{PDGD} does not require us to know the $\alpha$ and $\beta$ parameters, unlike the intervention-aware estimator; furthermore, \ac{PDGD} has better initial data-efficiency even when not run fully online.
On the other hand, there are no theoretical guarantees for the convergence of \ac{PDGD}, and we have observed that not running it fully online can lead to large drops in performance.
It seems the choice ultimately depends on what guarantees a practitioner prefers.

%% file: 10-onlinecounterltr/sections/09-conclusion.tex
\section{Conclusion}

In this chapter, we have introduced an intervention-aware estimator: an extension of existing counterfactual approaches that corrects for position-bias, trust-bias, and item-selection bias, while also considering the effect of online interventions.
Our results show that the intervention-aware estimator outperforms existing counterfactual \ac{LTR} estimators, and greatly benefits from online interventions in terms of data-efficiency.
With only 100 interventions it is able to reach a performance comparable to state-of-the-art online \ac{LTR} methods.
These findings allow us to answer thesis research question \ref{thesisrq:onlinecounterltr}: whether the counterfactual \ac{LTR} approach be extended to perform highly effective online \ac{LTR}.
From our experimental results, it appears that the answer is positive: using the intervention-aware estimator and 100  online interventions the performance of state-of-the-art online \ac{LTR} methods can be matched.

With the introduction of the intervention-aware estimator, we hope to further unify the fields of online \ac{LTR} and counterfactual \ac{LTR} as it appears to be the most reliable method for both settings.
Future work could investigate what kind of interventions work best for the intervention-aware estimator.
Since we have already seen in Chapter~\ref{chapter:06-onlinecountereval} that such an approach is effective for counterfactual/online ranking evaluation.

In retrospect, this chapter has put many findings from previous chapters in a different perspective.
Chapter~\ref{chapter:02-pdgd} introduced the concept of unbiasedness w.r.t.\ pairwise preferences and proved \ac{PDGD} had this property.
The experimental results of this chapter have shown that unbiasedness w.r.t.\ pairwise preferences is not enough to guarantee convergence at an optimal level of \ac{NDCG}.
Furthermore, Chapter~\ref{chapter:03-oltr-comparison} showed \ac{PDGD} is very robust to noise and bias, but with the results of this chapter we now know that \ac{PDGD} needs to be run online for this robustness.
The policy-aware estimator in Chapter~\ref{chapter:04-topk} is a precursor to the intervention-aware estimator of this chapter.
While Chapter~\ref{chapter:04-topk} realized that taking the logging policy into account is beneficial to counterfactual estimation, this chapter showed that taking the idea further, by accounting for all logging policies, provides even more benefits.
Lastly, Chapter~\ref{chapter:06-onlinecountereval} looked at bridging the divide between online and counterfactual evaluation; in retrospect, the results of Chapter~\ref{chapter:06-onlinecountereval} might have been even better had it used the intervention-aware estimator.
Together, Chapter~\ref{chapter:06-onlinecountereval} and this chapter suggest that an online method should both optimize its logging policy and use an intervention-aware estimator to learn, thus leaving a potentially very fruitful direction for future work.

%% file: 10-onlinecounterltr/notation.tex
\section{Notation Reference for Chapter~\ref{chapter:06-onlinecounterltr}}
\label{notation:06-onlinecounterltr}

\begin{center}
\begin{tabular}{l l}
 \toprule
\bf Notation  & \bf Description \\
\midrule
$k$ & the number of items that can be displayed in a single ranking \\
$t$ & a timestep number \\
$T$ & the total number of timesteps (gathered so far) \\
$\mathcal{D}$ & the available data \\
$\mathcal{R}(\pi)$ & the metric reward of a policy $\pi$ \\
$\hat{\mathcal{R}}(\pi \mid \mathcal{D})$ & an estimate of the metric reward of a policy $\pi$ \\
$q$ & a user-issued query \\
$D_q$ & the set of items to be ranked for query $q$ \\
$d$ & an item to be ranked\\
$y$ & a ranked list \\
$\pi$ & a ranking policy\\
$\pi(R \mid q)$ & the probability that policy $\pi$ displays ranking $R$ for query $q$ \\
$\pi(R_x |\, R_{1:x-1}, q)$ & probability of $\pi$ adding item $R_x$ given $R_{1:x-1}$ is already placed \\
$\Pi_T$ & the set of logging policies deployed up to timestep $T$\\
$ \lambda(d \mid D_q, \pi, q)$ & a metric function that weights items depending on their rank \\
$c(d)$ & a function indicating item $d$ was clicked in click pattern $c$\\
$o(d)$ & a function indicating item $d$ was observed at iteration $i$ \\
\bottomrule
\end{tabular}
\end{center}

%% file: 11-conclusions/main.tex
\chapter{Conclusions}
\label{chapter:conclusions}

In Section~\ref{section:introduction:rqs} we stated the overarching question that we aim to answer in this thesis:
\begin{itemize}
\item[] \em Could there be a single general theoretically-grounded approach that has competitive performance for both evaluation and \ac{LTR} from user clicks on rankings, in both the counterfactual and online settings?
\end{itemize}
The thesis has explored this question by looking at both online and counterfactual families of \ac{LTR} methods, and in particular, to see if one of these approaches can be extended to be effective at both the online and counterfactual \ac{LTR} scenarios.
In this final chapter, we will summarize the findings of the thesis and discuss how they reflect on our overarching thesis question.
Finally, we will consider future research directions for the field of \ac{LTR} from user clicks.

\section{Main Findings}

This section look back at the thesis research questions posed in Section~\ref{section:introduction:rqs}.
We divide our discussion in two parts discussing online methods and counterfactual methods for \ac{LTR} and evaluation, respectively.

\subsection{Novel methods for online learning and evaluation for ranking}

The first part of the thesis focussed on online \ac{LTR} methods.
Chapter~\ref{chapter:01-online-evaluation} looked at multileaving methods~\citep{Schuth2014a} for comparing multiple ranking systems at once and asked:
\begin{enumerate}[label=\textbf{RQ\arabic*},ref=\textbf{RQ\arabic*}]
\item[\ref{thesisrq:multileaving}] Does the effectiveness of online ranking evaluation methods scale to large comparisons? 
\end{enumerate}
We introduced the novel \ac{PPM} algorithm, \ac{PPM} bases evaluation on inferred pairwise item preferences.
Furthermore, \ac{PPM} is proven to have \emph{fidelity} -- it is provably unbiased in unambiguous cases~\citep{hofmann2013fidelity} -- and \emph{considerateness} -- it is safe w.r.t.\ the user experience during the gathering of clicks.
From our theoretical analysis, we find that no other existing multileaving method manages to meet both criteria.
In addition, our empirical results indicate that using \ac{PPM} leads to a much lower number of errors, in particular when applied to large scale comparisons.
Therefore, we answered \ref{thesisrq:multileaving} positively: \ac{PPM} is shown to be effective at online ranking evaluation for large scale comparisons.

Besides in Chapter~\ref{chapter:01-online-evaluation}, online evaluation was also the subject of Chapter~\ref{chapter:06-onlinecountereval}, which addressed the question:
\begin{enumerate}[label=\textbf{RQ\arabic*},ref=\textbf{RQ\arabic*},resume]
\item[\ref{thesisrq:interleaving}] Are existing interleaving methods truly capable of unbiased evaluation w.r.t.\ position bias?
\end{enumerate}
We showed that under a basic rank-based model of position bias (common in counterfactual \ac{LTR}~\citep{joachims2017unbiased, wang2018position, agarwal2019estimating}), three of the most prevalent interleaving algorithms are not unbiased: \acl{TDI}~\citep{radlinski2008does}, \acl{PI}~\citep{hofmann2011probabilistic}, and \acl{OI}~\citep{radlinski2013optimized}.
For each of these three methods, we showed that situations exist where the binary outcome of the method does not agree with the expected binary difference in \ac{CTR}.
In other words, under a basic assumption of position bias, situations exist where these interleaving methods are expected to prefer one system over another, while the latter system has a higher expected \ac{CTR} than the former.
Thus, we answer \ref{thesisrq:interleaving} negatively: the most prevalent interleaving methods are not unbiased w.r.t.\ position bias.

This finding can be extended to the multileaving methods: \acl{TDM}~\citep{Schuth2014a}, \acl{PM}~\citep{schuth2015probabilistic}, and \acl{OM}~\citep{Schuth2014a}, since they are equivalent to their interleaving counterparts when only two systems are compared.
While we did not examine it in this thesis, it is likely that \ac{PPM} also fails to be unbiased under basic position bias.
Nonetheless, an evaluation method can still be effective despite being biased, for instance, if the systematic error is small or situations where bias occurs are rare.

Chapter~\ref{chapter:02-pdgd} looked at online \ac{LTR} methods.
Existing online \ac{LTR} methods have relied on sampling model variants and comparing them using online evaluation~\citep{yue2009interactively}.
In response to the existing online \ac{LTR} approach, Chapter~\ref{chapter:02-pdgd} considered the question:
\begin{enumerate}[label=\textbf{RQ\arabic*},ref=\textbf{RQ\arabic*},resume]
\item[\ref{thesisrq:pdgd}] Is online \ac{LTR} possible without relying on model-sampling and online evaluation?
\end{enumerate}
We answered this question positively by introducing \ac{PDGD}, an online \ac{LTR} method that learns from inferred pairwise preferences and uses a debiased pairwise loss.
Besides proving that \ac{PDGD} is unbiased w.r.t.\ pairwise preferences, our experimental results show that \ac{PDGD} greatly outperforms the previous state-of-the-art \ac{DBGD}~\citep{yue2009interactively} algorithm in terms of data-efficiency and convergence.
Furthermore, \ac{PDGD} is the first online \ac{LTR} method that can effectively optimize neural networks as ranking models.

Chapter~\ref{chapter:06-onlinecounterltr} took another look at \ac{PDGD}, in particular at conditions where \ac{PDGD} is no longer effective.
The results in Chapter~\ref{chapter:06-onlinecounterltr} show that \ac{PDGD} fails to reach optimal performance without debiasing weights or when not applied fully online.
A particular worrisome observation was that, when not applied fully online, the performance of \ac{PDGD} can degrade as more interactions are gathered.
While this behavior looks similar, it is not overfitting since \ac{PDGD} does not display it when applied online.
Instead, it appears that \ac{PDGD} becomes severely biased when not applied fully online.
Therefore, we can conclude that the fact that \ac{PDGD} is unbiased w.r.t.\ pairwise preferences is not enough to guarantee unbiased optimization.
It appears that we do not fully understand why \ac{PDGD} appears to be so effective when run online.

The results of Chapter~\ref{chapter:02-pdgd} had surprising implications on \ac{DBGD}, for instance, it appeared that \ac{DBGD} was not able to reach the performance of \ac{PDGD} at convergence.
Meanwhile, \ac{DBGD} forms the basis of most existing online \ac{LTR} methods.
This prompted us to further investigate \ac{DBGD} in Chapter~\ref{chapter:03-oltr-comparison}, where we asked:
\begin{enumerate}[label=\textbf{RQ\arabic*},ref=\textbf{RQ\arabic*},resume]
\item[\ref{thesisrq:dbgd}] Are \ac{DBGD} \ac{LTR} methods reliable in terms of theoretical soundness and empirical performance?
\end{enumerate}
By critically examining the theoretical assumptions underlying the \ac{DBGD} method, we found that they are impossible when optimizing a deterministic ranking model.
This means that the existing theoretical guarantees of \ac{DBGD} are unsound in a lot of previous work where such models were used~\citep{hofmann2013reusing, hofmann11:balancing, oosterhuis2018differentiable, oosterhuis2016probabilistic, schuth2016mgd, wang2018efficient, yue2009interactively, zhao2016constructing, wang2019variance}.
Moreover, our empirical analysis revealed that ideal circumstances exist where \ac{DBGD} is still unable to find the optimal model.
In other words, even in scenarios where optimization should be very easy, \ac{DBGD} was unable to get near optimal performance.
These findings lead us to answer \ref{thesisrq:dbgd} negatively: our empirical results show that \ac{DBGD} is very unreliable and its theoretical guarantees do not cover the most common \ac{LTR} ranking models.

\subsection{Extending the counterfactual approach to learning and evaluation for ranking}

The second part of the thesis considered counterfactual \ac{LTR} methods for optimization and evaluation.
In particular, we tried to widen the applicability of counterfactual \ac{LTR} methods and their effectiveness as online methods.

First, Chapter~\ref{chapter:04-topk} recognized that the original \ac{IPS} counterfactual method~\citep{joachims2017unbiased} is not unbiased when item selection bias occurs.
This bias occurs when not all items can be displayed in a single ranking; this bias is unavoidable in top-$k$ ranking settings where only $k$ items can be displayed.
One of the questions Chapter~\ref{chapter:04-topk} addressed is:
\begin{enumerate}[label=\textbf{RQ\arabic*},ref=\textbf{RQ\arabic*},resume]
\item[\ref{thesisrq:topk}] Can counterfactual \ac{LTR} be extended to top-$k$ ranking settings?
\end{enumerate}
We showed that one can correct for item selection bias by basing propensity weights on both the position bias of the user and the stochastic ranking behavior of the logging policy.
Our novel policy-aware estimator uses this idea to extend the original \ac{IPS} approach by taking into account the logging policy behavior.
We prove that, assuming rank-based position bias, the policy-aware estimator is unbiased as longs as the logging policy gives every relevant item a non-zero probability of appearing in the top-$k$ of a ranking.
Furthermore, in our experimental results the policy-aware estimator approximates optimal performance regardless of the amount of item-selection bias present.
Therefore, we answer \ref{thesisrq:topk} positively: with the introduction of the policy-aware estimator the applicability of counterfactual \ac{LTR} has been extended to top-$k$ ranking settings.

Besides learning from top-$k$ feedback, Chapter~\ref{chapter:04-topk} also considered optimizing for top-$k$ metrics.
Interestingly, the existing counterfactual \ac{LTR} methods~\citep{hu2019unbiased, agarwal2019counterfactual} for optimizing \ac{DCG} metrics are very dissimilar from the state-of-the-art in supervised \ac{LTR}~\citep{wang2018lambdaloss, burges2010ranknet}.
To address this dissimilarity, Chapter~\ref{chapter:04-topk} posed the following question:
\begin{enumerate}[label=\textbf{RQ\arabic*},ref=\textbf{RQ\arabic*},resume]
\item[\ref{thesisrq:lambdaloss}] Is it possible to apply state-of-the-art supervised \ac{LTR} methods to the counterfactual \ac{LTR} problem?
\end{enumerate}
We answer this question positively by showing that, with some small adjustments, the LambdaLoss framework~\citep{wang2018lambdaloss} can be applied to counterfactual \ac{LTR} losses,
thus enabling the application of state-of-the-art supervised \ac{LTR} to counterfactual \ac{LTR}.
The implication of this finding is that there does not need to be a division between state-of-the-art supervised \ac{LTR} and counterfactual \ac{LTR}.
In other words, counterfactual \ac{LTR} methods can build on the best methods from the supervised \ac{LTR} field.

Chapter~\ref{chapter:05-genspec} takes a look at tabular and feature-based \ac{LTR} methods.
Tabular methods optimize a tabular ranking model~\citep{kveton2015cascading, lagree2016multiple, lattimore2019bandit, Komiyama2015, zoghi-online-2017}, which  remembers the optimal ranking,
in contrast with feature-based methods that optimize models that use the features of items to predict the optimal ranking.
The tabular models are extremely expressive and can capture any possible ranking, making them always capable of converging on the optimal ranking~\citep{zoghi2016click}.
However, their learned behavior does not generalize to previously unseen circumstances.
Conversely, the learned behavior of feature-based models can generalize very well to previously unseen circumstances~\citep{liu2009learning, bishop2006pattern}.
But feature-based models can also be limited by the available features, because often the available features do not provide enough information to predict the optimal ranking.
Thus feature-based \ac{LTR} generalizes very well to unseen circumstances, whereas tabular \ac{LTR} can specialize extremely well in specific circumstances.
Inspired by this tradeoff, we asked the following question in Chapter~\ref{chapter:05-genspec}:
\begin{enumerate}[label=\textbf{RQ\arabic*},ref=\textbf{RQ\arabic*},resume]
\item[\ref{thesisrq:genspec}] Can the specialization ability of tabular online \ac{LTR} be combined with the robust feature-based approach of counterfactual \ac{LTR}?
\end{enumerate}
Our answer is in the form of the novel \ac{GENSPEC} algorithm, a method for combining the behavior of a single robust generalized model and numerous specialized models.
\ac{GENSPEC} optimizes a single feature-based ranking model for performance across all queries, and many  tabular ranking models each specialized for a single query.
Then \ac{GENSPEC} applies a meta-policy that uses high-confidence bounds to safely decide per query which model to deploy.
Consequently, for previously unseen queries \ac{GENSPEC} chooses the generalized model which utilizes robust feature-based prediction. 
For other queries, it can decide to deploy a specialized model, i.e., if it has enough data to confidently determine that the tabular model has found the better ranking.
Our experimental results show that \ac{GENSPEC} successfully combines robust performance on unseen queries with extremely high performance at convergence.
Accordingly, we answer \ref{thesisrq:genspec} positively: using \ac{GENSPEC} we can combine the specialization properties of tabular \ac{LTR} with the robust generalization of feature-based \ac{LTR}.
For the \ac{LTR} field, the introduction of  \ac{GENSPEC} shows that specialization does not need to be unique to tabular online \ac{LTR}, instead it can be a property of counterfactual \ac{LTR} as well.

As discussed above, Chapter~\ref{chapter:06-onlinecountereval} proved that several prominent interleaving methods are biased w.r.t.\ a basic model of position bias.
Nonetheless, empirical results suggest that these online ranking evaluation methods are still very effective.
This leaves a gap for a theoretically-grounded online ranking evaluation method that is also very effective.
To address this gap, Chapter~\ref{chapter:06-onlinecountereval} considers counterfactual ranking evaluation, which has strong theoretical guarantees, and asks:
\begin{enumerate}[label=\textbf{RQ\arabic*},ref=\textbf{RQ\arabic*},resume]
\item[\ref{thesisrq:onlineeval}] Can counterfactual evaluation methods for ranking be extended to perform efficient and effective online evaluation?
\end{enumerate}
We realized that with the introduction of the policy-aware estimator in Chapter~\ref{chapter:04-topk}, the logging policy has an important role in counterfactual estimation.
Using the policy-aware estimator as a starting point, we introduce the \ac{LogOpt} that optimizes the logging policy to minimize the variance of the policy-aware estimator.
\ac{LogOpt} can be deployed during the gathering of data, periodically or fully online, and thus changes the logging behavior through an intervention.
As such, it turns the counterfactual evaluation approach with the policy-aware estimator into an online approach.
Our experimental results show that applying \ac{LogOpt} increases the data-efficiency of counterfactual evaluation with the policy-aware estimator.
The performance with \ac{LogOpt} is comparable to A/B testing and interleaving, but in contrast with interleaving, the policy-aware estimator applied with \ac{LogOpt} does not have a systematic error.
Therefore, we answer \ref{thesisrq:onlineeval} positively: by optimizing the logging policy with \ac{LogOpt}, counterfactual evaluation can perform effective and data-efficient online evaluation.

Inspired by how Chapter~\ref{chapter:06-onlinecountereval} bridges part of the gap between online and counterfactual ranking evaluation, Chapter~\ref{chapter:06-onlinecounterltr} addressed our final question:
\begin{enumerate}[label=\textbf{RQ\arabic*},ref=\textbf{RQ\arabic*},resume]
\item[\ref{thesisrq:onlinecounterltr}] Can the counterfactual \ac{LTR} approach be extended to perform highly effective online \ac{LTR}?
\end{enumerate}
The motivation is similar to the previous chapter: we would like to find a theoretically-grounded method that is effective at both counterfactual \ac{LTR} and online \ac{LTR}.
Since counterfactual \ac{LTR} has strong theoretical guarantees, we used it as a starting point.
Then we introduced the novel intervention-aware estimator which does not assume a stationary logging policy.
As a result, the estimator takes into account the fact that an online intervention may change the logging policy during the gathering of data.
Thus when applied online, the intervention-aware estimator does not only consider the logging policy used when a click was logged but also all the other logging policies applied at all other timesteps.
In addition, the intervention-aware estimator also combines the theoretical properties of recent counterfactual \ac{LTR} estimators: it is the first estimator that can correct for both position bias, item-selection bias, and trust bias.
Our experimental results show that the intervention-aware estimator results in much lower variance, than an equivalent estimator that ignores the effect of interventions.
Furthermore, in our experimental setting it outperformed all existing counterfactual estimators, with especially large differences when online interventions take place.
Importantly, we observed that the intervention-aware estimator matches the performance of \ac{PDGD} with as few as $100$ interventions during learning.
Besides a small initial period, \ac{LTR} with the intervention-aware estimator was able to reach the performance of the most effective online \ac{LTR} methods.
Therefore, we answer \ref{thesisrq:onlinecounterltr} positively: the intervention-aware estimator extends the counterfactual \ac{LTR} approach to perform highly effective online \ac{LTR}.
For the \ac{LTR} field, this demonstrates that methods do not have to be either part of counterfactual \ac{LTR} or online \ac{LTR}.
By designing them for both applications at once, they can be highly effective in both scenarios.

Finally, we note the complementary nature of the findings in the second part of the thesis.
Many of the contributions of earlier chapters were used in later chapters.
For instance, the methods introduced in Chapter~\ref{chapter:05-genspec} and Chapter~\ref{chapter:06-onlinecountereval} made use of the policy-aware estimator proposed in Chapter~\ref{chapter:04-topk}, and Chapter~\ref{chapter:06-onlinecounterltr} built on the policy-aware estimator to introduce the intervention-aware estimator.
Similarly, the adaptation of LambdaLoss for counterfactual \ac{LTR} introduced in Chapter~\ref{chapter:04-topk} was applied in the experiments of Chapter~\ref{chapter:05-genspec} and Chapter~\ref{chapter:06-onlinecountereval}.
While not explored in the thesis, many of the later contributions can also be applied to methods in earlier chapters.
For instance, the intervention-aware estimator from Chapter~\ref{chapter:06-onlinecounterltr} is completely compatible with the LambdaLoss adaption from Chapter~\ref{chapter:04-topk} and \ac{GENSPEC} from Chapter~\ref{chapter:05-genspec}.
In particular, it could be applied in combination with \ac{LogOpt} from Chapter~\ref{chapter:06-onlinecountereval}, potentially leading to even more effective online ranking evaluation.
Together, the contributions of the second part can be combined into a single framework for counterfactual \ac{LTR} and ranking evaluation, where our contributions complement each other.
Importantly, this framework bridges several gaps between supervised \ac{LTR}, online \ac{LTR}, and counterfactual \ac{LTR}.

\section{Summary of Findings}

The overarching question this thesis aimed to answer considered
whether \emph{there could be a single general theoretically-grounded approach that has competitive performance for both evaluation and \ac{LTR} from user clicks on rankings, in both the counterfactual and online settings.}

We have looked at the family of online methods for \ac{LTR}~\citep{yue2009interactively, hofmann2013reusing, wang2019variance} and ranking evaluation~\citep{radlinski2013optimized, joachims2003evaluating, hofmann2013fidelity, Schuth2014a}, which traditionally avoid making strong assumptions about user behavior, i.e., that a model of position bias is known~\citep{wang2018position}.
While this makes their theory widely applicable, the theoretical guarantees of these methods are relatively weak.
For instance, some interleaving and multileaving methods are proven to converge on correct outcomes if clicks are uncorrelated with relevance and thus every ranker performs equally well~\citep{radlinski2013optimized, hofmann2011probabilistic}.
Though such guarantees are valuable, they only cover a small group of unambiguous situations and thus leave most situations without theoretical guarantees.
Online \ac{LTR} methods are often motivated by empirical results from semi-synthetic experiments, where they are tested in settings with varying levels of noise and bias~\citep{schuth2016mgd, hofmann12:balancing, oosterhuis2017balancing, wang2018efficient}.
The fundamental question with this type of empirical motivation is how well the results generalize, in particular, whether a method is still effective if the experimental conditions change slightly.
This thesis has presented four examples of online methods that showed surprisingly poor performance when tested in new conditions:
\begin{enumerate*}[label=(\roman*)]
\item On several datasets \ac{DBGD}~\citep{yue2009interactively} did not get close to optimal performance after $10^6$ issued-queries while learning from clicks without noise or position bias (Chapter~\ref{chapter:03-oltr-comparison}).
\item \acl{TDI}~\citep{radlinski2008does}, \acl{PI}~\citep{hofmann2011probabilistic}, and \acl{OI}~\citep{radlinski2013optimized} make systematic errors in some ranking comparisons when tested under rank-based position bias (Chapter~\ref{chapter:06-onlinecountereval}).
\item The performance of the \acs{COLTR} algorithm~\citep{zhuang2020counterfactual} dropped severely when tested under position bias, item-selection bias, and trust bias (Chapter~\ref{chapter:06-onlinecounterltr}).
\item \ac{PDGD} no longer converged to near-optimal performance when we ran it counterfactually or only provided it with $100$ online interventions, and instead resulted in a large drop in performance (Chapter~\ref{chapter:06-onlinecounterltr}).
\end{enumerate*}
While these online methods \ac{LTR} and evaluation have also shown great performance in previous work~\citep{radlinski2008does, hofmann2011probabilistic, radlinski2013optimized, jagerman2019comparison, zhuang2020counterfactual, yue2009interactively, schuth2016mgd}, these problematic examples illustrate why we cannot conclude that these online \ac{LTR} methods are reliable.
For instance, the performance of a method like \ac{PDGD} was thought to be very robust to noise and bias~\citep{jagerman2019comparison} (Chapter~\ref{chapter:02-pdgd} and~\ref{chapter:03-oltr-comparison}), until tested without constant online interventions (Chapter~\ref{chapter:06-onlinecounterltr}).
Without strong theoretical guarantees, we cannot know whether there are more currently-unknown conditions required for the robust performance of \ac{PDGD}.
In general, it is unclear how robust online \ac{LTR} methods are in practice; this thesis has shown that there is a potential risk for detrimental performance if real-world circumstances do not match the tested experimental settings.
Therefore, we conclude that online \ac{LTR} methods should not be used as a basis for a single general approach for \ac{LTR} and ranking evaluation from user clicks.

In the second part of the thesis, we considered the family of counterfactual methods for \ac{LTR} and ranking evaluation~\citep{joachims2017unbiased, wang2016learning}, which consist of theoretically-grounded methods that use explicit assumptions about user behavior.
In contrast with the online family, counterfactual methods are less widely applicable: they only provide guarantees when the assumed models of user behavior are true.
For instance, the original counterfactual \ac{LTR} method assumes clicks are only affected by relevance and rank-based position bias~\citep{joachims2017unbiased, wang2016learning}.
Despite their limited applicability, counterfactual methods have very strong theoretical guarantees.
In contrast to most online \ac{LTR} methods, counterfactual \ac{LTR} methods guarantee convergence at the same performance as supervised \ac{LTR}, given that their assumptions about user behavior are true.
The findings of this thesis indicate that the strong guarantees with limited applicability of counterfactual \ac{LTR} are preferable over the weak guarantees with wide applicability of online \ac{LTR}.
This is mainly because widening the applicability of counterfactual \ac{LTR} proved very doable.
In this thesis, we have expanded the applicability of counterfactual \ac{LTR} and evaluation to
\begin{enumerate*}[label=(\roman*)]
\item top-$k$ settings with item-selection bias (Chapter~\ref{chapter:04-topk}), and
\item ranking settings where both trust bias and item-selection bias occur (Chapter~\ref{chapter:06-onlinecounterltr}).
\end{enumerate*}
Besides expanding the settings where counterfactual \ac{LTR} methods can be applied, we expanded the methods that perform counterfactual \ac{LTR}, including:
\begin{enumerate*}[label=(\roman*),resume]
\item the state-of-the-art LambdaLoss supervised \ac{LTR} framework~\citep{wang2018lambdaloss} (Chapter~\ref{chapter:04-topk}),
\item tabular models for extremely specialized rankings (Chapter~\ref{chapter:05-genspec}), and
\item a meta-policy that safely chooses between generalized feature-based models and specialized tabular models (Chapter~\ref{chapter:05-genspec}).
\end{enumerate*}
Moreover, this thesis also found novel algorithms to increase the effectivity of counterfactual \ac{LTR} methods for
\begin{enumerate*}[label=(\roman*),resume]
\item online ranking evaluation (Chapter~\ref{chapter:06-onlinecountereval}), and 
\item online \ac{LTR} (Chapter~\ref{chapter:06-onlinecounterltr}), even with a limited number of online interventions.
\end{enumerate*}
Together, these contributions have widened the applicability of counterfactual \ac{LTR} while maintaining its strong theoretical guarantees.
As a direct result of this thesis, counterfactual \ac{LTR} is applicable to more settings, more \ac{LTR} methods can be applied to the counterfactual \ac{LTR} problem, and counterfactual \ac{LTR} methods are more effective in both the counterfactual and online \ac{LTR} scenarios.

In conclusion, based on the findings of this thesis, it appears that counterfactual \ac{LTR} could form the basis of a general approach for \ac{LTR} from user clicks.
In our experimental results, counterfactual \ac{LTR} provided competitive performance to online \ac{LTR} methods in both the counterfactual and online settings.
While the theory of counterfactual \ac{LTR} does rely on stronger assumptions regarding user behavior than existing online \ac{LTR} methods, counterfactual \ac{LTR} provides far stronger theoretical guarantees.
In contrast, it is currently unclear under what conditions online \ac{LTR} methods are effective, making their performance very unpredictable.
Therefore, we answer our overarching thesis question positively: the counterfactual \ac{LTR} framework proposed in this thesis provides a unified approach for effective and reliable \ac{LTR} from user clicks.
For the \ac{LTR} field, the counterfactual \ac{LTR} framework bridges many gaps between areas of online \ac{LTR}, counterfactual \ac{LTR}, and supervised \ac{LTR}, and as such, it unifies many of the most effective methods for \ac{LTR} from user clicks.

\section{Future Work}

We will conclude the thesis with promising research directions for future work.

The most obvious direction is to widen the applicability of the counterfactual \ac{LTR} framework.
This means that estimators are introduced that are unbiased under other assumptions about user behavior.
\citet{joachims2017unbiased} mentioned that the original counterfactual method is unbiased as long as click probabilities decompose into observation and relevance probabilities.
For example, \citet{vardasbi2020cascade} looked at the performance of counterfactual \ac{LTR} when assuming cascading user behavior, an alternative to rank-based position bias.
Additionally, \citet{fang2019intervention} looked at context-dependent position bias, where the degree of bias varies per query.
It seems natural to continue this trend to more complex models of user behavior.
The challenge for future work is two-fold: find \ac{LTR} methods that are proven to be unbiased under more complex user behavior models; and introduce methods that can reliably find the parameters of these behavior models.

Besides learning from more complex user behavior, there is a big need for \ac{LTR} based on user clicks that optimizes for more complex goals.
Some existing work has already looked at complex goals: for instance, \citet{Radlinski2008} introduced a bandit algorithm for tabular \ac{LTR} that optimizes for both relevance and diversity within a ranking.
Thus, using user clicks to find a ranking that has relevant items, as well as having variety in the items within the ranking.
Another example comes from \citet{marco2020control}, who use counterfactual \ac{LTR} to optimize for relevance and ranking fairness.
Ranking fairness metrics are based on the amount of exposure different items receive, for example, some fairness metrics measure whether certain groups of items receive similar amounts of exposure.
Other areas of \ac{LTR} also optimize for computational efficiency to ensure that ranking systems can process queries in minimal amounts of time~\citep{gallagher2019joint}.
Future work could investigate if counterfactual \ac{LTR} can be used for complex goals like these and combinations of them.

Surprisingly, the experimental results in this thesis showed that \ac{PDGD} is no longer effective when not applied fully online, and similarly, we observed very poor performance for the \acs{COLTR} algorithm~\citep{zhuang2020counterfactual}.
However, we could not find theoretically proven conditions that guarantee that \ac{PDGD} or \acs{COLTR} is or is not effective.
It appears that we lack a theoretical approach to understand the limits of online \ac{LTR} methods.
If such an approach could be found, we may be able to correct for the faults in some online \ac{LTR} methods, or understand when they can be applied reliably.
Thus it may be very valuable if future work reconsidered the theory behind existing online \ac{LTR} methods.

Finally, most of the existing work on \ac{LTR} from user interactions only consideres user clicks.
Existing work has already looked at additional signals that are useful for learning~\citep{schuth2015predicting, kharitonov2015generalized}.
Novel methods that learn from other interactions in addition to user clicks have the potential to better understand user preferences.
However, the main challenge this direction of research may be the availability of such data.
Perhaps this direction of research mostly needs a publicly available source of data and methods to share such data in a privacy-respecting way.

Overall, our main advice for future work is to focus on methods that forge connections between advances in the larger field of \ac{LTR};
that is, methods that combine the best of different areas, as our proposed framework does for online \ac{LTR}, counterfactual \ac{LTR}, and supervised \ac{LTR}.

%% file: thesis-back.tex
\renewcommand{\bibsection}{\chapter{Bibliography}}
\renewcommand{\bibname}{Bibliography}
\markboth{Bibliography}{Bibliography}
\renewcommand{\bibfont}{\footnotesize}
\setlength{\bibsep}{0pt}

\bibliographystyle{abbrvnat}
\bibliography{thesis}

\input{12-summaries/summary}
\input{12-summaries/samenvatting}

%% file: 12-summaries/summary.tex
\chapter{Summary}

Ranking systems form the basis for online search engines and recommendation services.
They process large collections of items, for instance web pages or e-commerce products, and present the user with a small ordered selection.
The goal of a ranking system is to help a user find the items they are looking for with the least amount of effort.
Thus the rankings they produce should place the most relevant or preferred items at the top of the ranking.
Learning to rank is a field within machine learning that covers methods which optimize ranking systems w.r.t.\ this goal.
Traditional supervised learning to rank methods utilize expert-judgements to evaluate and learn, however, in many situations such judgements are impossible or infeasible to obtain.
As a solution, methods have been introduced that perform learning to rank based on user clicks instead.
The difficulty with clicks is that they are not only affected by user preferences, but also by what rankings were displayed.
Therefore, these methods have to prevent being biased by other factors than user preference.
This thesis concerns learning to rank methods based on user clicks and specifically aims to unify the different families of these methods.

The first part of the thesis consists of three chapters that look at online learning to rank algorithms which learn by directly interacting with users.
Its first chapter considers large scale evaluation and shows existing methods do not guarantee correctness and user experience, we then introduce a novel method that can guarantee both.
The second chapter proposes a novel pairwise method for learning from clicks that contrasts with the previous prevalent dueling-bandit methods.
Our experiments show that our pairwise method greatly outperforms the dueling-bandit approach.
The third chapter further confirms these findings in an extensive experimental comparison, furthermore, we also show that the theory behind the dueling-bandit approach is unsound w.r.t.\ deterministic ranking systems.

The second part of the thesis consists of four chapters that look at counterfactual learning to rank algorithms which learn from historically logged click data.
Its first chapter takes the existing approach and makes it applicable to top-$k$ settings where not all items can be displayed at once.
It also shows that state-of-the-art supervised learning to rank methods can be applied in the counterfactual scenario.
The second chapter introduces a method that combines the robust generalization of feature-based models with the high-performance specialization of tabular models.
The third chapter looks at evaluation and introduces a method for finding the optimal logging policy that collects click data in a way that minimizes the variance of estimated ranking metrics.
By applying this method during the gathering of clicks, one can turn counterfactual evaluation into online evaluation.
The fourth chapter proposes a novel counterfactual estimator that considers the possibility that the logging policy has been updated during the gathering of click data.
As a result, it can learn much more efficiently when deployed in an online scenario where interventions can take place.
The resulting approach is thus both online and counterfactual, our experimental results show that its performance matches the state-of-the-art in both the online and the counterfactual scenario.

As a whole, the second part of this thesis proposes a framework that bridges many gaps between areas of online, counterfactual, and supervised learning to rank.
It has taken approaches, previously considered independent, and unified them into a single methodology for widely applicable and effective learning to rank from user clicks.

%% file: 12-summaries/samenvatting.tex
\chapter{Samenvatting}

Rankingsystemen vormen de basis voor online zoekmachines en aanbevelingsdiensten.
Ze verwerken grote verzamelingen van bijvoorbeeld webpagina's of web-winkel producten, en presenteren een kleine geordende selectie aan de gebruiker.
Voor de beste gebruikerservaring, moeten de resulterende rankings de meest relevante of geprefereerde items bovenaan plaatsen. 
Het learning-to-rank veld omvat methodes die rankingssystemen optimaliseren voor dit doel.
Traditionele learning-to-rank methoden maken gebruik van supervisie: annotaties van deskundigen. %
Omdat het verkrijgen van dergelijke annotaties vaak onmogelijk is, zijn methoden ontwikkeld die leren op basis van gebruikersclicks.
Helaas worden clicks niet alleen be\"invloed door gebruikersvoorkeuren, maar ook door welke rankings worden weergegeven.
Om werkelijk de gebruikersvoorkeuren te leren moeten deze methoden dus de invloed van zulke andere factoren vermijden.
Dit proefschrift betreft learning-to-rank methoden op basis van gebruikersclicks en heeft specifiek het doel de verschillende families van deze methoden te verenigen.

Het eerste deel van dit proefschrift bestaat uit drie hoofdstukken die kijken naar online learning-to-rank algoritmen die leren d.m.v.\ directe interactie met gebruikers.
Het eerste hoofdstuk behandelt evaluatie op grote schaal waar we de eerste methode introduceren die garanties biedt voor zowel juiste resultaten en goede gebruikerservaring.
Het tweede hoofdstuk stelt een nieuwe paarsgewijze methode voor om te leren van clicks die in contrast staat met de eerdere dueling-bandits methoden.
Onze experimenten tonen aan dat onze paarsgewijze methode veel beter presteert dan de dueling-bandits methoden.
Het derde hoofdstuk bevestigt deze bevinding en laat ook zien dat de theorie achter de dueling-bandits methoden incorrect is t.o.v.\ deterministische rankingsystemen.

Het tweede deel van het proefschrift bestaat uit vier hoofdstukken die kijken naar counterfactual learning-to-rank algoritmen die leren van eerder verzamelde clickdata.
Het eerste hoofdstuk maakt deze algoritmen toepasbaar op top-k rankings die niet alle items tegelijkertijd kunnen weergeven.
Verder laat het ook zien dat state-of-the-art supervised learning-to-rank methoden toepasbaar zijn in het counterfactual scenario.
Het tweede hoofdstuk introduceert een methode die de robuuste generalisatie van `feature-based' modellen combineert met de hoge prestaties van gespecialiseerde tabelmodellen.
Het derde hoofdstuk behandelt evaluatiemethodes en introduceert een methode voor het vinden van de optimale logging-policy die clickdata verzamelt op een manier die de variantie van geschatte rankingmetrieken minimaliseert.
Door deze methode toe te passen tijdens het verzamelen van clicks wordt counterfactual evaluatie omgezet in online evaluatie.
Het vierde hoofdstuk presenteert een nieuwe counterfactual methode die rekening houdt met de mogelijkheid dat de logging policy niet constant is.
Als gevolg hiervan kan het veel effici\"enter leren in een online scenario waar interventies plaatsvinden.
Onze experimentele resultaten laten zien dat de prestaties van deze methode overeenkomen met de state-of-the-art in zowel het online als het counterfactual scenario.

In zijn geheel stelt het tweede deel van dit proefschrift een raamwerk voor dat veel hiaten overbrugt tussen de online, counterfactual en supervised learning-to-rank gebieden.
Het heeft methodes die voorheen als onafhankelijk werden beschouwd, verenigd in \'e\'en enkele methodologie voor breed toepasbare en effectieve learning-to-rank op basis van gebruikersclicks.